\documentclass[showpacs,onecolumn,aps,pra,longbibliography,superscriptaddress,notitlepage]{revtex4-2}
\usepackage{qcircuit}
\usepackage[dvips]{graphicx}
\usepackage{amsmath,amssymb,amsthm,mathrsfs,amsfonts,dsfont}
\usepackage{subfigure, epsfig}
\usepackage{braket}
\usepackage{svg}
\usepackage{bm}
\usepackage{enumerate}
\usepackage{algorithm}
\usepackage{algpseudocode}
\usepackage{diagbox}
\usepackage{physics}
\usepackage{color}
\usepackage{multirow}
\usepackage[marginal]{footmisc}
\usepackage{comment}
\usepackage{tikz}
\usepackage{booktabs}
\usepackage[]{qcircuit}
\usepackage{makecell}
\usetikzlibrary{arrows}
\usetikzlibrary{shapes,fadings,snakes}
\usetikzlibrary{decorations.pathmorphing,patterns}
\usetikzlibrary{calc}
\usetikzlibrary{positioning}
\usepackage[colorlinks = true]{hyperref}
\usepackage[nameinlink,capitalize]{cleveref}
\newtheorem{remark}{Remark}

\graphicspath{{./figure/}}

\renewcommand{\sec}[1]{\hyperref[sec:#1]{Section~\ref*{sec:#1}}}
\newcommand{\app}[1]{\hyperref[app:#1]{Appendix~\ref*{app:#1}}}
\newcommand{\thm}[1]{\hyperref[thm:#1]{Theorem~\ref*{thm:#1}}}
\newcommand{\lem}[1]{\hyperref[lem:#1]{Lemma~\ref*{lem:#1}}}
\newcommand{\cor}[1]{\hyperref[cor:#1]{Corollary~\ref*{cor:#1}}}
\newcommand{\fgr}[1]{\hyperref[fgr:#1]{Figure~\ref*{fgr:#1}}}
\newcommand{\tab}[1]{\hyperref[tab:#1]{Table~\ref*{tab:#1}}}

\newtheorem{assumption}{Assumption}
\newtheorem{theorem}{Theorem}

\newtheorem{lemma}{Lemma}
\newtheorem{corollary}{Corollary}
\newtheorem{proposition}{Proposition}
\newtheorem*{proposition*}{Proposition}

\newtheorem{definition}{Definition}

\newcommand{\comments}[1]{}

\usepackage{enumitem}

\begin{document}
\title{Quantum Derivative Pricing for SPDEs via BDSDE Representation}

\begin{abstract}
We study quantum speedups of derivative pricing for
stochastic partial differential equation (SPDE) models through their
backward doubly stochastic differential equation (BDSDE)
representations.
We develop conditional and nested quantum-accelerated multilevel Monte
Carlo (QA-MLMC) methods for estimating the resulting conditional and
nested expectations, improving the sampling complexity of classical Monte Carlo methods from $\widetilde{\mathcal O}(\epsilon^{-2})$ to
$\widetilde{\mathcal O}(\epsilon^{-1})$ within additive error
$\epsilon$.
We apply the framework to derivative pricing and sensitivity
analysis, providing quantum-accelerated estimators for prices as well as
first-order and second-order Greeks, likelihood-ratio and Malliavin-weight representations for Greeks, and Heston-type stochastic-volatility models.
To enable efficient multilevel coupling, we construct a family of
Forward--Backward Taylor discretization schemes for the stochastic
integrals arising in the BDSDE representations and establish global
strong-error order one convergence for pricing and Greek estimators. Numerical experiments showcase our schemes for first-order and second-order Greeks can reach the required orders for the full quadratic quantum speedups.
\end{abstract}

\author{Xinmiao Li}
\affiliation{Yau Mathematical Sciences Center, Tsinghua University, Beijing 100084, China}
\affiliation{Qiuzhen College, Tsinghua University, Beijing 100084, China}

\author{Yanqiao Wang}
\affiliation{Yau Mathematical Sciences Center, Tsinghua University, Beijing 100084, China}
\affiliation{Qiuzhen College, Tsinghua University, Beijing 100084, China}
\affiliation{Institute for AI Industry Research, Tsinghua University, Beijing 100084, China}

\author{Rundi Lu}
\affiliation{Yau Mathematical Sciences Center, Tsinghua University, Beijing 100084, China}
\affiliation{{Department of Mathematics, Tsinghua University,  Beijing 100084, China}}

\author{Zhengwei Liu}
\email{liuzhengwei@tsinghua.edu.cn}
\affiliation{Yau Mathematical Sciences Center, Tsinghua University, Beijing 100084, China}
\affiliation{{Department of Mathematics, Tsinghua University,  Beijing 100084, China}}
\affiliation{Yanqi Lake Beijing Institute of Mathematical Sciences and Applications, Beijing 101408, China}

\author{Jin-Peng Liu}
\email{liujinpeng@tsinghua.edu.cn}
\affiliation{Yau Mathematical Sciences Center, Tsinghua University, Beijing 100084, China}
\affiliation{Yanqi Lake Beijing Institute of Mathematical Sciences and Applications, Beijing 101408, China}
\affiliation{Institute for Applied Mathematics, Tsinghua University, Beijing 100084, China}

\maketitle

\tableofcontents

\newpage

\section{Introduction}
\label{sec:introduction}
Volatility in modern financial markets is neither constant nor exogenous.
It evolves randomly across multiple time scales and interacts with other
sources of uncertainty, including liquidity conditions, interest rates,
and latent risk factors.
Under the classical Black--Scholes framework and its stochastic-volatility
extensions,
derivative prices can be characterized as solutions of partial differential
equations (PDEs)
~\cite{BlackScholes1973,Heston1993,FouquePapanicolaouSircar2000}.

When the market environment itself evolves randomly,
the pricing operator becomes stochastic,
leading naturally to stochastic partial differential equations (SPDEs).
Such SPDE formulations arise in a variety of financial settings,
including stochastic volatility surface models~\cite{Bergomi2015},
stochastic term-structure models~\cite{HeathJarrowMorton1992,santa2001dynamics,cont2005modeling},
and stochastic forward-curve models in energy markets~\cite{benth2008stochastic}.
In these applications,
the solution of the SPDE typically represents the value of a contingent claim
under a random market environment.
Beyond computing option prices,
practical applications require the estimation of risk sensitivities,
commonly known as Greeks,
including Delta,
Vega,
Gamma,
and higher-order parameter sensitivities.
These quantities play a central role in dynamic hedging,
model calibration,
risk management,
and uncertainty quantification~\cite{hull2016options}.

The numerical approximation of SPDEs has been studied extensively over the
past several decades~\cite{printems2001discretization,larsson2003partial}.
A wide range of deterministic and stochastic discretization techniques have
been developed,
including finite-difference methods,
finite-element methods,
spectral methods,
and stochastic Galerkin approaches
~\cite{gyongy1999lattice,lord2004numerical,brenner2008mathematical}.
These methods have achieved considerable success in the simulation and
analysis of stochastic systems arising in physics,
engineering,
and finance.

However,
the direct numerical treatment of financial SPDEs remains challenging.
In multi-factor stochastic-volatility and term-structure models,
the effective state dimension can grow rapidly.
As a consequence,
grid-based discretizations often suffer from the curse of dimensionality,
resulting in substantial computational costs for both spatial discretization
and the solution of large-scale linear systems
~\cite{BungartzGriebel2004,hout2010adi}.

Moreover,
many quantities of practical interest,
including option prices,
risk measures,
and Greek sensitivities,
are naturally expressed as expectations or nested expectations of the
underlying stochastic system.
Accurate estimation of such quantities often requires substantial sampling
effort in addition to the numerical solution of the SPDE itself.

These challenges motivate the search for alternative formulations that avoid
direct discretization of the underlying SPDE.
An attractive approach is provided by probabilistic representations,
which reformulate the problem in terms of stochastic differential
equations and expectation estimation.

A major development in this direction was the introduction of backward
stochastic differential equations (BSDEs) by Pardoux and
Peng~\cite{Pardoux1990AdaptedSO}.
BSDEs provide a nonlinear extension of the classical Feynman--Kac formula
and establish a probabilistic representation for broad classes of
semilinear parabolic PDEs
~\cite{Pardoux1992BackwardSD,kobylanski2000backward}.
Since their introduction,
BSDEs have become an important tool in
stochastic control~\cite{peng1992generalized,yong1999stochastic},
mathematical finance~\cite{Karoui1997BackwardSD},
and nonlinear expectation theory
~\cite{peng1997backward,peng2004nonlinear}.

To extend the probabilistic correspondence between BSDEs and PDEs
to stochastic partial differential equations,
Pardoux and Peng introduced backward doubly stochastic differential
equations (BDSDEs)~\cite{Pardoux1994BackwardDS}.
They established a stochastic Feynman--Kac formula showing that
solutions of a broad class of quasilinear SPDEs can be represented by
solutions of BDSDEs driven simultaneously by a forward Brownian motion
and a backward Brownian motion.
For the class of stochastic option-pricing SPDEs considered in this work,
let $X^{t,x}$ denote the forward state process describing the underlying
risk factors, initialized from $x$ at time $t$.
Then the SPDE solution admits a BDSDE representation of the form
\begin{align*}
  u(s,X_s^{t,x})=Y_s^{t,x},\quad X_t^{t,x}=x,
  \quad s\in[t,T],
\end{align*}
where $Y_s^{t,x}$ solves the associated BDSDE.
This representation transforms the original SPDE problem into the
estimation of stochastic expectations and avoids direct spatial
discretization of the underlying equation.
Moreover, the coexistence of forward and backward sources of randomness
naturally leads to conditional and nested expectation structures,
providing a probabilistic foundation for Monte Carlo, multilevel Monte
Carlo, and quantum mean estimation methodologies.

The numerical solution of BSDEs and BDSDEs has attracted considerable
attention over the past several decades.
For BSDEs,
a variety of time-discretization,
regression-based,
and Monte Carlo methods have been 
developed~\cite{zhang2004numerical,bouchard2004discrete,gobet2005regression}.
These approaches provide practical algorithms for approximating BSDE
solutions and have led to a rich literature on numerical methods for
high-dimensional PDEs through probabilistic representations.
More recently,
BSDE formulations have inspired a variety of machine-learning approaches
for high-dimensional PDEs,
including Deep BSDE methods,
deep backward dynamic programming schemes,
and related neural-network-based algorithms
~\cite{han2017deep,han2018solving,beck2019machine}.
These methods substantially improve scalability in high-dimensional
settings and have become an active research direction at the interface of
scientific computing and machine learning.

For BDSDEs,
several numerical approximation schemes have been proposed,
including Euler-type and regression-based methods
~\cite{aman2013numerical,bachouch2016empirical}.
In particular,
\cite{bao2016first}
developed a first-order scheme based on a two-sided It\^o--Taylor
expansion.
Since many of these schemes rely on conditional expectations with respect to
the forward Brownian motion,
the forward randomness is integrated out at each time step.
As a consequence,
they do not directly provide the pathwise strong approximations required
for multilevel Monte Carlo coupling and Greek estimation.

Many quantities arising from BDSDE representations,
including option prices and Greek sensitivities,
are naturally expressed as stochastic expectations.
To estimate such quantities efficiently,
multilevel Monte Carlo (MLMC),
introduced by Giles~\cite{giles2008multilevel},
exploits strong couplings between successive discretization levels to
achieve substantial computational savings over standard Monte Carlo methods.

MLMC has been successfully applied to a wide range of problems in
computational finance,
including Greek estimation,
efficient risk measurement,
and basket option pricing
~\cite{giles2009multilevel,burgos2012computing,giles2018multilevel,giles2019multilevel}.
It has also been extended to stochastic partial differential 
equations~\cite{barth2013multilevel,iliev2017renormalization,chada2022improved}.
In particular,
\cite{giles2012stochastic}
developed a multilevel Monte Carlo framework based on a Milstein finite
difference discretization for SPDEs and demonstrated its effectiveness in
the pricing of basket credit derivatives.
To the best of our knowledge, however,
the combination of MLMC and BDSDE representations has received little
attention in the existing literature.

Recent advances in quantum computing have opened new possibilities for
further accelerating stochastic simulation.
Quantum amplitude estimation and related quantum mean estimation
algorithms reduce the sampling complexity of expectation estimation from
$\mathcal O(\epsilon^{-2})$
to
$\widetilde{\mathcal O}(\epsilon^{-1})$,
thereby providing a quadratic speedup over classical Monte Carlo 
methods~\cite{brassard2000quantum,heinrich2002quantum,Montanaro2015,Kothari2022MeanEW}.
Building on these developments,
quantum-accelerated multilevel Monte Carlo (QA-MLMC) methods have recently
emerged as a powerful framework for expectation estimation.
\cite{An2021quantumaccelerated}
combined quantum mean estimation with MLMC and established quantum speedups
for stochastic differential equations.
Subsequent developments further extended and refined this framework,
including quadratic speedups for nonlinear and nested expectation problems
and related stochastic simulation tasks
~\cite{blanchet2024quadratic,Blanchet2025NonlinearQM,ozgul2025quantum,li2026quantum}.

At the same time,
quantum algorithms have also been investigated for financial PDEs and
derivative pricing.
Examples include quantum algorithms for option valuation and financial
simulation~\cite{rebentrost2018quantum,stamatopoulos2020option,herman2026quantum},
as well as recent end-to-end quantum PDE frameworks for option pricing
under Black--Scholes and Heston-type models~\cite{guseynov2026end}.
These works demonstrate the potential of quantum computation for
high-dimensional problems arising in quantitative finance.
On the other hand,
quantum algorithms have also been explored for the numerical solution of
BSDEs.
For example,
\cite{fujita2024application}
introduced a quantum least-squares Monte Carlo approach for solving BSDEs.

However,
most existing quantum approaches mainly focus on PDEs, SDEs, or BSDEs.
To the best of our knowledge,
although quantum simulation algorithms for stochastic differential equations are recently studied~\cite{jin2025quantum,bravyi2025quantum,yang2025circuitefficientrandomizedquantumsimulation,li2026efficientquantumsimulationnonlinear,bravyi2026quantumalgorithmsstochasticnonlinear}, quantum algorithms for SPDEs based on BDSDE representations have not
yet been systematically studied.
The extension of these ideas to SPDEs faces several challenges.

A key difficulty is that achieving the optimal complexity
$\widetilde{\mathcal O}(\epsilon^{-1})$
requires not only quantum acceleration, but also effective multilevel coupling and sufficiently accurate pathwise discretizations.
In the standard multilevel complexity analysis,
if the bias, variance, and cost exponents are denoted by
$(\alpha,\beta,\gamma)$,
then achieving the optimal quantum complexity typically requires
$$\beta\ge 2\gamma,$$
which is stronger than the classical MLMC requirement
$\beta\ge\gamma$.
Equivalently,
if a discretization has strong convergence order $r$ so that the
level-difference variance behaves like $\mathcal O(h^{2r})$,
then the optimal QA-MLMC regime requires
$$r=\frac{\beta}{2}\ge \gamma.$$
For direct discretizations of a $d$-dimensional SPDE,
the cost exponent $\gamma$ is often large because each sample involves the spatial degrees
of freedom of the discretized random field. For instance, if a tensor-product
grid with mesh size \(h_\ell\) is used in each spatial coordinate and the
cost per grid point is uniformly bounded, then the number of spatial grid
points is proportional to \(h_\ell^{-d}\). Hence $\gamma_{\mathrm{SPDE}} = d$.
Thus,
achieving the optimal QA-MLMC complexity would require a strong
approximation order comparable to the spatial dimension ($r\ge \gamma_{\mathrm{SPDE}}= d$),
which is generally unrealistic in high-dimensional settings.

Even after passing to a BDSDE formulation,
a standard strong-error order-$1/2$ discretization (e.g., Euler discretization) is still insufficient for
the full quantum speedup.
When $\gamma=1$ and $r=1/2$,
the variance exponent is only $\beta=1$,
and the resulting QA-MLMC complexity is typically
$\widetilde{\mathcal O}(\epsilon^{-3/2})$
rather than
$\widetilde{\mathcal O}(\epsilon^{-1})$.
This is why the strong-error order one Forward--Backward Taylor discretization
developed in this work is essential:
it yields $r=1$,
reaching the critical regime $\beta=2\gamma$
and enabling the \emph{full quadratic quantum speedup}.
The comparison is summarized in Table~\ref{tab:spde-vs-bdsde}.

\begin{table}[ht]
\centering
\label{tab:spde-vs-bdsde}
\begin{tabular}{c | c | c | c | c}
\toprule
Approach
&
Cost exponent
&
Strong-error order
&
MLMC complexity
&
QA-MLMC complexity
\\
\midrule
Direct SPDE discretization
&
$\gamma_{\mathrm{SPDE}}= d$
&
$r\ge d$ (unrealistic)
&
$\widetilde{\mathcal O}(\epsilon^{-2})$
&
$\widetilde{\mathcal O}(\epsilon^{-1})$ 
\\
Standard BDSDE discretization
&
$\gamma_{\mathrm{BDSDE}}=1$
&
$r=\frac12$
&
$\widetilde{\mathcal O}(\epsilon^{-2})$
&
$\widetilde{\mathcal O}(\epsilon^{-3/2})$
\\
Our method
(Theorem~\ref{thm:QAMLMC for u when given B}
and Section~\ref{sec:Strong-Order-One Numerical Schemes for Pricing and Greek Estimators})
&
$\gamma_{\mathrm{BDSDE}}=1$
&
$r=1$
&
$\widetilde{\mathcal O}(\epsilon^{-2})$
&
$\widetilde{\mathcal O}(\epsilon^{-1})$
\\
\bottomrule
\end{tabular}
\caption{Classical and quantum MLMC algorithms for the direct discretizations of SPDEs and
the BDSDE-based discretizations. Here $r$ denotes the strong convergence
order, $\gamma$ denotes the cost exponent per sample, and $\epsilon$ denotes the additive error.}
\end{table}

These considerations indicate that extending quantum multilevel methods
to SPDEs requires more than a direct application of existing QA-MLMC
theory.

Existing QA-MLMC frameworks are primarily designed for expectation
estimation problems arising from SDE-type path simulations.
In contrast,
the BDSDE representation of SPDE solutions introduces an additional
backward source of randomness,
which naturally gives rise to \emph{conditional} and \emph{nested} expectation
structures.
These structures require quantum multilevel estimators that are adapted
to the two sources of randomness.

Moreover,
most existing multilevel and quantum approaches in this area focus
primarily on pricing problems.
By contrast,
efficient estimators for first- and second-order Greeks under
stochastic-environment models are much less developed.
This leaves open the construction of multilevel and quantum estimators
for sensitivity analysis in the SPDE and BDSDE setting.

\subsection{Contributions}

In this work,
we develop a unified computational framework for stochastic derivative-pricing
SPDEs based on their BDSDE representations.
The framework combines multilevel Monte Carlo,
conditional and nested quantum estimators,
and new Forward--Backward Taylor discretization schemes with strong-error order one convergence,
allowing derivative pricing and Greek estimation to be treated within a common
probabilistic setting.\\

\paragraph{Modeling and algorithms.}

\begin{enumerate}

\item
\textbf{BDSDE-based framework for SPDE pricing and sensitivity analysis.}

Section~\ref{sec:SPDE Formulation and Probabilistic Representation}
introduces the SPDE models arising in stochastic-environment financial
markets and develops a BDSDE-based probabilistic representation.
This representation reformulates pricing and Greek estimation problems
as conditional and nested expectation estimation tasks, providing a
unified framework for direct pricing, first-order sensitivities, and
second-order sensitivities.
The resulting estimation framework serves as the foundation for all
subsequent quantum algorithms developed in this work.

\item
\textbf{Conditional quantum-accelerated multilevel Monte Carlo for SPDEs.}

Building on the BDSDE-based estimation framework,
we develop a conditional QA-MLMC methodology for approximating
\begin{align*}
 u(t,x;B)
=\mathbb E_W[P(t,x;B)\mid B]. 
\end{align*}
The key ingredients are the conditional quantum level-$\ell$
difference evaluator
(Algorithm~\ref{alg:Q-level-l-difference}),
which constructs quantum estimators for the level differences
appearing in the multilevel decomposition,
and the conditional quantum-accelerated MLMC estimator
$\widehat{u}(t,x;B)$
(Algorithm~\ref{alg:conditional-QA-MLMC}),
which combines these level differences across discretization levels
to produce a complete conditional estimator of
$u(t,x;B)$.

\item
\textbf{Nested quantum-accelerated multilevel Monte Carlo for SPDEs.}

Building on recent quantum algorithms for nonlinear and nested
expectation estimation
~\cite{Blanchet2025NonlinearQM,sun2026optimalquantumspeedupsrepeatedly},
we specialize the nested QA-MLMC framework to the BDSDE
representations arising from SPDEs.
This yields a nested quantum-accelerated estimator
(Algorithm~\ref{alg:nested-qmlmc})
for quantities of the form
\begin{align*}
\mathbb E_B\!\left[\varphi(u(t,x;B))\right].
\end{align*}
The resulting methodology combines the coupled level difference
evaluator
(Algorithm~\ref{alg: Q-Coupled Level Difference Evaluator})
with a nested multilevel quantum estimation procedure,
thereby extending quantum multilevel techniques to nested
expectation structures arising naturally in stochastic-environment
SPDE models.

\item
\textbf{Applications to derivative pricing and sensitivity analysis.}

In Section~\ref{sec:applications},
we specialize the conditional and nested QA-MLMC frameworks
to derivative pricing and sensitivity analysis.
Using the probabilistic representations developed in this work,
we construct quantum estimators for a broad class of first-order Greeks,
including Delta, spot Vega, Rho, and general parameter Greeks;
smooth second-order Greeks like Gamma, Vanna and Volga; extensions to nonsmooth payoffs via likelihood-ratio and Malliavin-weight representations; and Heston-type stochastic-volatility models.

\item
\textbf{Global strong-error order one analysis and numerical experiments.}

In Section~\ref{sec:Strong-Order-One Numerical Schemes for Pricing and Greek Estimators}, we establish a global strong-error framework for the BDSDE payoff
functionals arising in pricing,
first-order Greek estimation,
and second-order Greek estimation.
We then construct concrete Forward--Backward Taylor discretization
operators and prove that they satisfy the first-order consistency and
accumulated stability properties required by this framework.
Consequently,
the resulting estimators achieve global strong convergence of order one.

We further validate this theory through numerical experiments in Section~\ref{sec:experiment}.
For pricing,
first-order Greek estimation,
and second-order Greek estimation,
the observed multilevel exponents satisfy
\[
    (\alpha,\beta,\gamma)\approx(1,2,1),
\]
confirming the strong-error order one behavior predicted by the theory and the
conditions required for the quadratic quantum speedup.

\end{enumerate}

\paragraph{Theorems.}

\begin{enumerate}
\item
\textbf{BDSDE representation and QA-MLMC foundations.}

Theorem~\ref{thm:spde-representation},
following
\cite{Pardoux1994BackwardDS,BallyMatoussi2001},
provides the BDSDE representation for a class of linear backward SPDEs
and serves as the probabilistic foundation of our framework.
Theorem~\ref{thm:QA-MLMC},
following
\cite{An2021quantumaccelerated},
recalls the general quantum-accelerated MLMC methodology that underlies
our quantum complexity analysis.

\item
\textbf{Conditional quantum speedups for SPDE solution estimation.}

Theorem~\ref{thm:QAMLMC for u when given B}
shows that, conditioned on a realization of the environmental Brownian
motion $B$,
the SPDE solution
\begin{align*}
u(t,x;B)=\mathbb E_W[P(t,x;B)\mid B]
\end{align*}
can be estimated with additive error $\epsilon$
using computational cost
\begin{align*}
\widetilde{\mathcal O}(\epsilon^{-1}),
\end{align*}
yielding a quadratic quantum speedup over classical Monte Carlo methods.

\item
\textbf{Quantum speedups for nested SPDE expectations.}

Theorem~\ref{thm:nested-qmlmc}
extends the conditional framework to nested quantities of the form
\begin{align*}
\mathbb E_B[\varphi(u(t,x;B))].
\end{align*}

Under suitable regularity and moment assumptions,
the resulting nested QA-MLMC estimator achieves additive error
$\epsilon$
with overall complexity
\begin{align*}
\widetilde{\mathcal O}(\epsilon^{-1}),
\end{align*}
thereby extending quantum speedups to SPDE quantities involving
both forward and backward randomness.

\item
\textbf{Quantum estimation of first and second-order Greeks.}

Propositions~\ref{prop:app-first-order-greeks}
and~\ref{prop:app-second-order-greeks}
derive pathwise conditional representations for first-order and
second-order Greeks.
These representations express the Greeks as conditional expectations
of explicitly constructed payoff functionals involving first-order and
second-order variational processes.

Building on these representations,
Corollaries~\ref{cor:app-qmlmc-greeks}
and~\ref{cor:app-nested-qmlmc-greeks}
show that both conditional and unconditional Greeks can be estimated
with additive error $\epsilon$
using quantum-accelerated MLMC methods with complexity
\begin{align*}
\widetilde{\mathcal O}(\epsilon^{-1}).
\end{align*}

\item
\textbf{Global strong-error order one schemes via Forward--Backward Taylor discretization.}

Proposition~\ref{prop:error-u},
Proposition~\ref{prop:error-first-order-greek-payoff},
and
Proposition~\ref{prop:error-second-order-greek-payoff}
establish a general strong-error framework showing that
global strong-error order one follows whenever the constituent
discretization operators satisfy suitable first-order consistency and
stability properties.

To realize this framework, we introduce a new family of
Forward--Backward Taylor discretization schemes
for the stochastic integrals arising in the BDSDE
representations of pricing and Greek estimators.
The construction explicitly captures mixed forward--backward
iterated integrals and provides concrete discretization operators
for pricing, first-order Greeks, and second-order Greeks.

Propositions
~\ref{prop:int-strong-error},
~\ref{prop:int stable},
~\ref{prop:first-order-greek-strong-error},
~\ref{prop:first-order-greek accumulated stability},
~\ref{prop:second-order-greek-strong-error},
and
~\ref{prop:second-order-greek accumulated stability}
verify that the Forward--Backward Taylor discretization satisfies the required
strong-error order one consistency and stability conditions.
Consequently, all resulting pricing and Greek estimators achieve
global strong convergence of order one.

\end{enumerate}

\subsection{Organization}
Section~\ref{sec:introduction}
introduces the background and motivation of the work,
reviews the related literature,
and summarizes the main contributions.

Section~\ref{sec:SPDE Formulation and Probabilistic Representation}
presents the SPDE models considered in this paper and develops the
BDSDE-based probabilistic representation and estimation framework that
serves as the foundation of our methodology.

Section~\ref{sec:Quantum-accelerated MLMC for the SPDE}
develops the quantum-accelerated multilevel Monte Carlo framework.
Subsection~\ref{subsec:Conditional Quantum-Accelerated MLMC Estimator}
introduces the conditional QA-MLMC estimator for approximating
$u(t,x;B)=\mathbb E_W[P(t,x;B)\mid B]$,
while
Subsection~\ref{subsec:Nested Quantum-accelerated MLMC Estimator}
develops the nested QA-MLMC estimator for quantities of the form
$\mathbb E_B[\varphi(u(t,x;B))]$.

Section~\ref{sec:applications}
specializes the general framework to derivative pricing and sensitivity
analysis.
We derive conditional pricing representations,
first- and second-order Greek representations,
and the corresponding conditional and nested quantum estimators.
Several extensions, including nonsmooth payoffs and Heston-type
stochastic-volatility models, are also discussed.

Section~\ref{sec:Strong-Order-One Numerical Schemes for Pricing and Greek Estimators}
contains the main numerical analysis component of the paper.
We develop strong-error order one numerical schemes for pricing and Greek
estimators and establish the strong-convergence and stability results
required for the quantum-accelerated multilevel complexity analysis.

Section~\ref{sec:experiment}
presents numerical experiments validating the proposed Forward--Backward Taylor discretization and the multilevel convergence properties across multiple realizations of the backward Brownian motion.

Finally, Section~\ref{sec:discussion}
concludes the paper by summarizing the main contributions and outlining
possible directions for future work.

The appendices contain additional background material and technical
proofs.
In particular,
Appendix~\ref{sec:SPDE and BDSDE}
reviews the classical BDSDE representation of SPDEs,
while
Appendix~\ref{app:strong-error-framework}
contains the proofs of the strong-error framework and the numerical
analysis results developed in
Section~\ref{sec:Strong-Order-One Numerical Schemes for Pricing and Greek Estimators}.

\section{SPDE Formulation and Probabilistic Representation}
\label{sec:SPDE Formulation and Probabilistic Representation}

\subsection{From Stochastic Volatility Models to SPDEs}

We begin with the classical Heston stochastic model~\cite{Heston1993} under the
risk-neutral measure:
\begin{align}
  \label{eq:Heston S_t}
\mathrm dS_t&=rS_t\,\mathrm dt+
\sqrt{V_t}\,S_t\,\mathrm dW_t^{(S)},\\
\label{eq:Heston V_t}
\mathrm dV_t&=\kappa(\theta-V_t)\,\mathrm dt
+\xi\sqrt{V_t}\,\mathrm dW_t^{(V)},
\end{align}
where $S_t$ denotes the asset price,
$V_t$ is the instantaneous variance process,
$r$ is the risk-free interest rate,
$\kappa$ is the mean-reversion rate,
$\theta$ is the long-run variance level,
$\xi$ is the volatility-of-volatility parameter,
and
\begin{align}
  \label{eq:Heston cor}
\mathrm d\langle W^{(S)},W^{(V)}\rangle_t
=\rho\,\mathrm dt.
\end{align}

For a European payoff $g(S_T)$, the option price is
\begin{align*}
u(t,s,v)
=\mathbb E^{\mathbb Q}
\left[e^{-r(T-t)}g(S_T)
\mid S_t=s,\,V_t=v
\right].
\end{align*}

Under standard regularity assumptions,
the standard arbitrage-free pricing arguments~\cite{Merton1973RationalOptionPricing, BlackScholes1973} 
 imply that the option price $u(t,s,v)$ satisfies
 the partial differential equation (PDE)
\begin{align*}
\partial_t u+\mathcal L u-r u=0,
\end{align*}
with terminal condition
$$
u(T,s,v)=g(s),
$$
where the generator $\mathcal L$ associated with
\eqref{eq:Heston S_t}--\eqref{eq:Heston cor} is given by
\begin{align*}
\mathcal L
=rs\partial_s
+\kappa(\theta-v)\partial_v
+\frac12 vs^2\partial_{ss}
+\rho\xi sv\,\partial_{sv}
+\frac12\xi^2v\,\partial_{vv}.
\end{align*}

The above pricing equation is deterministic because all market factors are
assumed to be either deterministic or fully represented by the finite-dimensional
state variables $(S_t,V_t)$.

In modern financial markets, however, quantities such as volatility
surfaces, liquidity conditions, interest-rate environments,
and latent risk factors may themselves evolve randomly over time.
As a consequence, the coefficients of the pricing equation become random,
and the option value must be viewed as a random field.

Motivated by this observation, we consider stochastic option-pricing SPDEs
of the form
\begin{align}
\label{eq:spde0}
\mathrm du=\bigl(\mathcal L u
-r u+F\bigr)\,\mathrm dt+G\,\mathrm dB_t,
\end{align}
with terminal condition
\begin{align*}
u(T,s,v)=g(s).
\end{align*}
Here,
$\mathcal L$ denotes the stochastic-volatility generator,
$r$ is the discount rate,
$F$ is a source term,
and $G$ characterizes the stochastic forcing induced by the random market
environment.

Direct numerical discretization of
\eqref{eq:spde0}
is computationally demanding,
particularly in high-dimensional settings.
To facilitate efficient numerical approximation,
we employ a probabilistic representation based on backward doubly
stochastic differential equations (BDSDEs).
This representation transforms the SPDE problem into the estimation of
conditional expectations over stochastic paths and serves as the
foundation of the algorithms developed in this paper.

\subsection{BDSDE Representation and Estimation Framework}
\label{subsec:BDSDE Representation and Estimation Framework}

A key advantage of stochastic partial differential equations is that,
under suitable regularity conditions,
their solutions admit probabilistic representations through backward
doubly stochastic differential equations (BDSDEs).
Introduced by Pardoux and Peng~\cite{Pardoux1994BackwardDS},
BDSDEs extend the classical Feynman--Kac formula and provide a
probabilistic representation for a broad class of quasilinear SPDEs.

In this work,
the SPDE is treated through its associated BDSDE representation.
Since the SPDE--BDSDE correspondence is classical,
we only summarize the formulation required for the subsequent numerical
analysis.
Further details can be found in
Appendix~\ref{sec:SPDE and BDSDE}.

In our applications,
we are primarily interested in the linear setting.
Following~\cite{BallyMatoussi2001},
we consider the following class of linear backward SPDEs.

\begin{theorem}\label{thm:spde-representation}
  For the linear backward SPDE
  \begin{equation}\label{eq:linear SPDE(with BDSDE)0}
  \begin{cases}
\mathrm{d} u(t,x)
=
\Bigl[\mathcal{L}u(t,x)
+ F(t,x)
+ c(t)u(t,x)
+ \widetilde{c}(t)(\sigma^\top\nabla u)(t,x)\Bigr]\,\mathrm{d}t
+ \Bigl[H(t,x)+d(t)u(t,x)\Bigr]\,\mathrm{d}B_t,\quad t\in[0,T],\\
u(T,x)=G(x),
\end{cases}
\end{equation}
where $u:\mathbb{R_+}\times \mathbb{R}^d\rightarrow \mathbb{R}^k$ and
  $\mathcal{L}u=\left( Lu_1,\cdots ,Lu_k \right)^\top$
with $L=\frac{1}{2}\sum\limits_{i,j=1}^d(\sigma\sigma^\top)_{ij}\frac{\partial^2}{\partial x_i\partial x_j}+\sum\limits_{i=1}^db_i\frac{\partial}{\partial x_i}.$
Assume that $G\in C_p^3(\mathbb{R}^d)$,
$F, H\in C_b^3([0,T]\times\mathbb{R}^d)$ and $c,\widetilde{c},d$ are bounded deterministic functions.
Let $\{X_s^{t,x};t\leq s\leq T\}$ be the solution of the SDE
\begin{equation}\label{eq:SDE(with BDSDE)0}
\begin{cases}
\mathrm{d}X_s^{t,x}=b(X_s^{t,x})\,\mathrm{d}s
+\sigma(X_s^{t,x})\,\mathrm{d}W_s,\quad s\in[t,T],\\
X_t^{t,x}=x.
\end{cases}
\end{equation}
Define the stochastic exponential $\Phi(s,r)$ as
  \begin{align*}
    \Phi(s,\tau)=\exp\left(\int_s^\tau c(r)\,\mathrm{d}r 
    +\int_s^\tau d(r)\mathrm d \overleftarrow{B}_r
    +\int_s^\tau \widetilde{c}(r)\,\mathrm{d}W_r
    -\frac{1}{2}\int_s^\tau\left(|\widetilde{c}(r)|^2-|d(r)|^2\right)\,\mathrm{d}r\right).
  \end{align*}
Then the SPDE~\eqref{eq:linear SPDE(with BDSDE)0} has a unique solution  and it can 
  be written as
  \begin{align}\label{eq:u(t,x)0}
    u(t,x)=\mathbb{E}\left[\Phi(t,T)G(X_T^{t,x})
    +\int_t^T\Phi(t,r)\left(F(r,X_r^{t,x})+d(r) H(r,X_r^{t,x})\right)\,\mathrm{d}r 
    +\int_t^T\Phi(t,r)H(r,X_r^{t,x})\mathrm d \overleftarrow{B}_r\,|\,\mathcal{F}_{t,T}^B  \right].
  \end{align}
  For convenience in the subsequent discussion, when the Brownian motion $B$ is fixed, we denote
\begin{align}\label{eq:eq of P(t,x;B)}
P(t,x;B)=\Phi(t,T)G(X_T^{t,x})
    +\int_t^T\Phi(t,r)\left(F(r,X_r^{t,x})+d(r) H(r,X_r^{t,x})\right)\,\mathrm{d}r 
    +\int_t^T\Phi(t,r)H(r,X_r^{t,x})\mathrm d \overleftarrow{B}_r
\end{align}
and
\begin{align}\label{eq:u(t,x;B)=E[P]}
  u(t,x;B)=\mathbb{E}_W[P(t,x;B)].
\end{align}

\end{theorem}
\begin{remark}
  Here, the quantity $P$ depends on the solution of the SDE $X_r^{t,x}$, 
  and the expectation is taken with respect to the randomness arising from the SDE.
\end{remark}

Theorem~\ref{thm:spde-representation}
reduces the solution of the SPDE to the evaluation of a path functional
$P(t,x;B)$.
For a fixed realization of the backward Brownian motion $B$,
the only remaining randomness in
\eqref{eq:eq of P(t,x;B)}
comes from the forward Brownian motion $W$.
Consequently,
the SPDE solution admits the conditional expectation representation
\begin{align*}
u(t,x;B)
=\mathbb E_W
\bigl[P(t,x;B)\,\big|\,B\bigr].
\end{align*}
This observation transforms the original SPDE problem into the
estimation of a conditional expectation with respect to the forward
diffusion.

In many applications,
the quantity of ultimate interest is obtained after averaging over the
random environment generated by $B$.
This leads to the nested expectation
\begin{align}
U(t,x)
=\mathbb E_B
\bigl[u(t,x;B)\bigr]
=\mathbb E_B
\Big[\mathbb E_W
\bigl[P(t,x;B)\,\big|\,B\bigr]\Big].
\end{align}

Consequently,
the probabilistic representation naturally gives rise to two estimation
problems.
The first is the conditional estimation problem
\begin{align}
u(t,x;B)
=\mathbb E_W
\bigl[P(t,x;B)\,\big|\,B\bigr],
\end{align}
where the realization of $B$ is fixed.
The second is the nested estimation problem
\begin{align}
U(t,x)
=\mathbb E_B
\bigl[u(t,x;B)\bigr],
\end{align}
which requires averaging over both sources of randomness.

These two estimation problems form the basis of the multilevel Monte
Carlo and quantum algorithms developed in the subsequent sections.

\section{Quantum-accelerated MLMC for SPDEs}
\label{sec:Quantum-accelerated MLMC for the SPDE}

The probabilistic representation developed in the previous section
reduces the SPDE problem to the estimation of conditional and nested
expectations associated with the path functional
$P(t,x;B)$. 
To estimate efficiently,
we employ multilevel Monte Carlo (MLMC) 
together with quantum mean
estimation.

Suppose that a sequence of approximations
$\{\Phi_\ell\}_{\ell\ge0}$
to a random variable $\Phi$
is available and satisfies
\begin{align*}
\left|\mathbb E[\Phi]-\mathbb E[\Phi_\ell]\right|=\mathcal{O}(2^{-\alpha\ell}),
\quad
\operatorname{Var}(\Phi_\ell-\Phi_{\ell-1})
=\mathcal{O}(2^{-\beta\ell}),
\quad
\operatorname{Cost}(\Phi_\ell)
=\mathcal{O}(2^{\gamma\ell}),
\end{align*}
for some $\alpha,\beta,\gamma>0$.
The MLMC estimator is based on the telescoping identity
\begin{align*}
\mathbb E[\Phi_L]=\mathbb E[\Phi_0]
+\sum_{\ell=1}^{L}\mathbb E[\Phi_\ell-\Phi_{\ell-1}].
\end{align*}
Under the standard MLMC framework~\cite{giles2008multilevel,giles2015multilevel},
if $\beta>\gamma$,
an estimator with mean-square error
$\mathcal{O}(\epsilon^2)$
can be constructed with overall computational complexity
$\mathcal{O}(\epsilon^{-2})$.

Recent advances in quantum computing provide the possibility of further
accelerating expectation estimation.
In particular, quantum mean estimation algorithms reduce the sampling
complexity from
$\mathcal{O}(\epsilon^{-2})$
to
$\widetilde{\mathcal O}(\epsilon^{-1})$,
yielding a quadratic speedup over classical Monte Carlo methods.
The following result will be used as the basic quantum subroutine.

\begin{lemma}[Quantum mean estimation {\cite{Kothari2022MeanEW}}]
\label{lem:quantum mean estimation}
Suppose a random variable $\Phi$ admits an efficient quantum encoding.
Then there exists a quantum algorithm which estimates
$\mathbb E[\Phi]$
with additive error
$\sigma(\Phi)/n$
using $\mathcal{O}(n)$ oracle calls, with constant success probability.
\end{lemma}

Using the standard powering argument
{\cite{Jerrum1986RandomGO}},
the success probability can be amplified from a constant to
$1-\delta$
with an additional $\mathcal{O}(\log(1/\delta))$ overhead.

For notational simplicity, we write
\[
\operatorname{QME}(\Phi;\epsilon,\delta)
\]
for a quantum mean estimation subroutine which returns an estimate
$\widehat \mu$ of $\mathbb E[\Phi]$ such that
\[
\Pr\left(
|\widehat\mu-\mathbb E[\Phi]|\le \epsilon
\right)
\ge 1-\delta .
\]

Replacing the classical sample average on each MLMC level by QME
leads to quantum-accelerated MLMC (QA-MLMC).
Our algorithm is mainly based on the following quantum-accelerated Monte Carlo methods proposed in~\cite{An2021quantumaccelerated}.

\begin{theorem}
\label{thm:QA-MLMC}
Let $\Phi$ be a random variable, and let $\Phi_\ell~(\ell=0,1,\ldots,L)$ 
be a sequence of random variables approximating $\Phi$ at level $\ell$.
 Further define $\Phi_{-1}=0$.
Let $C_\ell$ be the cost of sampling from $\Phi_\ell$, 
and let $V_\ell$ be the variance of $\Phi_\ell-\Phi_{\ell-1}$.
If there exist positive constants $\alpha,
\beta,\gamma$
such that 
\begin{flalign*}
&
\begin{cases}
|\mathbb{E}[\Phi_\ell-\Phi]| = \mathcal{O}(2^{-\alpha \ell}),\\
V_\ell = \mathcal{O}(2^{-\beta \ell}),\\
C_\ell = \mathcal{O}(2^{\gamma \ell})
\end{cases}
&&
\end{flalign*}
Then for any $\epsilon< 1/e$ 
there is a quantum algorithm that 
estimates $\mathbb{E}[\Phi]$ up to additive error
$\epsilon$ with probability at least 0.99, and with cost
    \begin{align*}
        \begin{cases} 
        \widetilde{\mathcal{O}}\left(\epsilon^{-1}\right), & \beta \geq 2\gamma, \\
        \widetilde{\mathcal{O}}\left(\epsilon^{-1 - (\gamma - 0.5\beta)/\alpha}\right), & \beta < 2\gamma.
        \end{cases}
    \end{align*}
\end{theorem}

However, the complexity requirements for QA-MLMC are substantially more
restrictive.
In particular, achieving the optimal quantum complexity
$\widetilde{\mathcal O}(\epsilon^{-1})$
 requires
$\beta\ge 2\gamma$,
which is stronger than the classical MLMC requirement
$\beta\ge\gamma$.
For discretization-based SPDE solvers,
the cost exponent $\gamma$ is often large due to the curse of
dimensionality.
For example, under standard tensor-product spatial discretizations of a
$d$-dimensional SPDE, one typically has
$\gamma\approx d.$
Since the MLMC variance exponent typically satisfies
$\beta=2r$, where $r$ is the strong convergence order,
the condition $\beta\ge2\gamma$ would require
$r\ge d,$
namely, a strong convergence order at least equal to the spatial
dimension, which is generally unrealistic in practice.

The BDSDE representation developed in the previous section avoids
direct spatial discretization of the SPDE and converts the problem into
the estimation of stochastic path functionals.
This structure enables the construction of efficient multilevel
estimators with favorable variance decay properties and provides a
natural foundation for quantum acceleration.

In the remainder of this section,
we first consider the conditional estimation problem and then extend
the framework to the nested setting.

\subsection{Conditional Quantum-accelerated MLMC Estimator }
\label{subsec:Conditional Quantum-Accelerated MLMC Estimator}

The probabilistic representation introduced in
Subsection~\ref{subsec:BDSDE Representation and Estimation Framework}
reduces the pricing problem to the estimation of the conditional
expectation
\begin{align*}
u(t,x;B)
=\mathbb E_W[P(t,x;B)\mid B].
\end{align*}
To construct the conditional quantum-accelerated estimators,
we first develop suitable discretizations for the stochastic quantities
appearing in this representation.
Recalling that
\begin{align*}
P(t,x;B)
=\Phi(t,T)G(X_T^{t,x})
    +\int_t^T\Phi(t,r)\left(F(r,X_r^{t,x})+d(r) H(r,X_r^{t,x})\right)\,\mathrm{d}r 
    +\int_t^T\Phi(t,r)H(r,X_r^{t,x})\mathrm d \overleftarrow{B}_r,
\end{align*}
we observe that 
there are three components
that require discretization: the evolution of the SDE, the exponential weight
$\Phi$, and the integrals over the time interval $[t,T]$.

In classical numerical schemes, these components are often discretized using
Euler-type methods, which already provide sufficient accuracy.
However, in quantum algorithms, higher-order convergence is often required.
Therefore, we introduce a unified notation $\mathcal{S}$ to denote the
discretization operators associated with each component, allowing us to
systematically track and analyze their respective convergence properties.

We introduce three discretization operators. Let
\begin{align*}
\Delta W_k:=W_{t_{k+1}}-W_{t_k},
\quad
\Delta\overleftarrow B_k:=B_{t_k}-B_{t_{k+1}},
\quad t_k=t+kh .
\end{align*}

 (1) State discretization $\mathcal{S}_X$.
  Given the current state $X_k$, the time $t_k$, the stepsize $h$, and
  the forward Brownian increment $\Delta W_k$, we set
 \begin{align}\label{eq:def of S_X}
    X_{k+1}=\mathcal{S}_X(X_k,t_k,h;\Delta W_k).
\end{align}

 (2) Exponential weight discretization $\mathcal{S}_\Phi$.
  Given the current weight $\Phi_k$, the time $t_k$, the stepsize $h$,
  the forward Brownian increment $\Delta W_k$, and the backward Brownian
  increment $\Delta\overleftarrow B_k$, we set
  \begin{align}\label{eq:def of S_phi}
    \Phi_{k+1}=\mathcal{S}_\Phi(\Phi_k,t_k,h;
    \Delta W_k,\Delta\overleftarrow B_k).
\end{align}

  (3) Integral discretization $\mathcal{S}_{\mathrm{int}}$.
  Given the current term $Y_k$, weight $\Phi_k$, state $x_k$, the time
  $t_k$, the stepsize $h$, and the Brownian increments
  $(\Delta W_k,\Delta\overleftarrow B_k)$, we set
 \begin{align}\label{eq:def of S_int}
    Y_{k+1}=Y_k+
    \mathcal{S}_{\mathrm{int}}\left(
    \Phi_k,x_k,t_k,h;
    \Delta W_k,\Delta\overleftarrow B_k
    \right).
  \end{align}

Assume that reversible unitary implementations of the discretization operators
  $\mathcal{S}_X,\mathcal{S}_\Phi$ and $\mathcal{S}_{\mathrm{int}}$ are available in the following standard form:
  \begin{align*}
    O_X &: |x,s,h,\Delta W\rangle|0\rangle
     \mapsto |x,s,h,\Delta W \rangle\ket{\mathcal{S}_X(x,s,h;\Delta W)},\\
    O_\Phi &: |\Phi,s,h,\Delta W,\Delta\overleftarrow B\rangle|0\rangle
     \mapsto |\Phi,s,h,\Delta W,\Delta\overleftarrow B\rangle
    \ket{\mathcal{S}_\Phi(\Phi,s,h;\Delta W,\Delta\overleftarrow B)},\\
    O_\mathrm{int}&: \ket{Y,\Phi,x,s,h;\Delta W,\Delta\overleftarrow B}\ket{0}
    \mapsto \ket{Y,\Phi,x,s,h;\Delta W,\Delta\overleftarrow B}
    \ket{Y+\mathcal{S}_\mathrm{int}\left(\Phi,x,s,h;\Delta W,\Delta\overleftarrow B\right)}.
  \end{align*}

For notational simplicity, we combine these elementary reversible updates into a
single one-step oracle. Suppressing unchanged input registers, work registers,
and the standard uncomputation of ancillas, we write, for
$\eta\in\{h,2h\}$,
\[
\mathcal U_{\mathcal S,\eta}(s;\Delta W,\Delta\overleftarrow B):
\ket{X,\Phi,Y}
\mapsto
\ket{
\mathcal S_X(X,s,\eta;\Delta W),
\mathcal S_\Phi(\Phi,s,\eta;\Delta W,\Delta\overleftarrow B),
Y+\mathcal S_{\rm int}(\Phi,X,s,\eta;\Delta W,\Delta\overleftarrow B)
}.
\]

\begin{algorithm}[H]
\caption{Conditional Quantum Level-$\ell$ Difference Sample $\ket{\Delta_{\mathcal{S}}(h;B)}$}
\label{alg:Q-level-l-difference}
\begin{algorithmic}[1]
\Require
$t,T,x$, fine step size $h$, fixed backward Brownian path $B$,
and one-step quantum oracle $\mathcal U_{\mathcal S,\eta}$ for $\eta\in\{h,2h\}$.
\Ensure
A quantum state encoding $|\Delta_{\mathcal S}(h;B)\rangle$.

\State Set $N=(T-t)/(2h)$ and $s_j=t+jh$, $j=0,\ldots,2N$.
 For $j=0,\ldots,2N-1$, set
$\Delta\overleftarrow B_j^f
=
B_{s_j}-B_{s_{j+1}} .$

\State Prepare the fine-level randomness superposition
$\sum_{\boldsymbol{\Delta W}^f}
\sqrt{p(\boldsymbol{\Delta W}^f)}
\ket{\boldsymbol{\Delta W}^f}$,
$\boldsymbol{\Delta W}^f
=
(\Delta W_0^f,\ldots,\Delta W_{2N-1}^f).
$

\State Initialize
$\ket{X_0^f,\Phi_0^f,Y_0^f}
=
\ket{X_0^c,\Phi_0^c,Y_0^c}
=
\ket{x,1,0}.
$

\For{$k=0,\ldots,N-1$}
    \State Fine path: apply two fine steps
    $$    \ket{X_{2k+2}^f,\Phi_{2k+2}^f,Y_{2k+2}^f}
    \leftarrow
    \mathcal U_{\mathcal S,h}
    (s_{2k+1};\Delta W_{2k+1}^f,\Delta\overleftarrow B_{2k+1}^f)
    \mathcal U_{\mathcal S,h}
    (s_{2k};\Delta W_{2k}^f,\Delta\overleftarrow B_{2k}^f)
    \ket{X_{2k}^f,\Phi_{2k}^f,Y_{2k}^f}.$$

    \State Coarse path: apply one coarse step
    $$    \ket{X_{2k+2}^c,\Phi_{2k+2}^c,Y_{2k+2}^c}
    \leftarrow
    \mathcal U_{\mathcal S,2h}
    (s_{2k};\Delta W_{2k}^f+\Delta W_{2k+1}^f,
    \Delta\overleftarrow B_{2k}^f+\Delta\overleftarrow B_{2k+1}^f)
    \ket{X_{2k}^c,\Phi_{2k}^c,Y_{2k}^c}.$$
\EndFor

\State Compute
$\ket{P^f} \gets \ket{\Phi_{2N}^{f} G(x_{2N}^{f})+ Y_{2N}^f}$
and 
$\ket{P^c}\gets \ket{\Phi_{2N}^{c} G(x_{2N}^{c})+ Y_{2N}^c}. $

\State \Return
$\ket{\Delta_{\mathcal S}(h;B)}
=
\ket{P^f-P^c}.$
\end{algorithmic}
\end{algorithm}

The returned quantum state
$\ket{\Delta_{\mathcal S}(h;B)}$
encodes a coupled multilevel difference sample conditioned on $B$.
For notational uniformity, define
\[
\Delta_0(t,x;B):=P_0(t,x;B),
\quad
\Delta_\ell(t,x;B)
:=
P_\ell^{f}(t,x;B)-P_{\ell}^{c}(t,x;B),
\quad \ell\ge1.
\]
The fine and coarse trajectories are generated using the shared Brownian
coupling
\[
\Delta W_k^c
=
\Delta W_{2k}^f+\Delta W_{2k+1}^f,
\quad
\Delta\overleftarrow B_k^c
=
\Delta\overleftarrow B_{2k}^f
+
\Delta\overleftarrow B_{2k+1}^f.
\]

For a fixed realization of $B$, the conditional multilevel identity gives
\[
u_L(t,x;B)
=
\mathbb E_W[P_L(t,x;B)\mid B]
=
\sum_{\ell=0}^{L}
\mathbb E_W[\Delta_\ell(t,x;B)\mid B].
\]
By Lemma~\ref{lem:quantum mean estimation}, each conditional level mean
can be estimated by applying QME to the corresponding quantum encoding.

\begin{algorithm}[H]
\caption{Conditional Quantum-Accelerated MLMC Estimator $\widehat{u}(t,x;B)$}
\label{alg:conditional-QA-MLMC}
\begin{algorithmic}[1]
\Require
$t,T,x,L,h_0$, tolerances $\{(\epsilon_\ell,\delta_\ell)\}_{\ell=0}^L$,
fixed backward Brownian path $B$, and quantum oracles for $\ket{\Delta_{\mathcal{S}}(h;B)}$.

\Ensure
Estimator $\widehat{u}(t,x;B)$ of $u_L(t,x;B)$.

\State Initialize $\widehat{u}(t,x;B)\gets 0$.

\For{$\ell=0$ to $L$}

\State Compute
    $$
    \widehat{\mu}_\ell(B)
    \leftarrow
    \operatorname{QME}\!\left(
    \Delta_\ell(t,x;B),\epsilon_\ell,\delta_\ell
    \right)
    $$
    to estimate
$\mu_\ell(B) =
    \mathbb{E}_W\left[
    \Delta_\ell(t,x;B)|B
    \right]$.
    \State Update
    $\widehat{u}(t,x;B)\gets
    \widehat{u}(t,x;B)+\widehat{\mu}_\ell(B)$.
\EndFor

\State \Return $\widehat{u}(t,x;B)$.
\end{algorithmic}
\end{algorithm}

Algorithm~\ref{alg:conditional-QA-MLMC}
combines the conditional multilevel telescoping decomposition
with quantum mean estimation applied independently on each level.
The resulting estimator approximates the conditional quantity
$$
u_L(t,x;B)
=\mathbb E_W[P_L(t,x;B)\mid B].
$$

\begin{theorem}
\label{thm:QAMLMC for u when given B}
Assume that:

(1) \(G\in C_p^3(\mathbb R^d)\) is globally Lipschitz,
\(F,H\in C_b^3([0,T]\times\mathbb R^d)\), and
\(c,\widetilde c,d\) are bounded deterministic functions.

(2)
The discretization operators
\(\mathcal{S}_X\), \(\mathcal{S}_\Phi\), and
\(\mathcal{S}_{\mathrm{int}}\)
admit strong-error orders at least
\(1\), \(1\), and \(1\), respectively, in the sense of
Definitions~\ref{def:Strong error of S X}--\ref{def:Strong error of S int}.
Moreover,
\(\mathcal{S}_{\mathrm{int}}\)
satisfies the accumulated stability property in
Definition~\ref{def:accumulated stability of Sint}.

(3) There exists a constant $M>0$, independent of $h$, such that
\small
\[
\sup_{h\in(0,h_0]}
\left\|
\sup_{0\le k\le (T-t)/h}
\left|
\Phi_k^{(h)}
\right|
\right\|_{L^4_{W,B}}
+
\sup_{h\in(0,h_0]}
\left\|
\sup_{0\le k\le (T-t)/h}
\left|
G(X_{t_k}^{t,x})
\right|
\right\|_{L^4_{W,B}}
\le M,
\]
\normalsize
where $\Phi_k^{(h)}$ is defined in
Definition~\ref{def:Strong error of S Phi}.

Then, for the linear backward SPDE
\begin{equation*}
\begin{cases}
\mathrm{d} u(t,x)
=
\Bigl[
\mathcal{L}u(t,x)
+F(t,x)
+c(t)u(t,x)
+\widetilde c(t)(\sigma^\top\nabla u)(t,x)
\Bigr]\,\mathrm{d}t
+
\Bigl[
H(t,x)+d(t)u(t,x)
\Bigr]\,\mathrm{d}B_t,
\quad t\in[0,T],\\
u(T,x)=G(x),
\end{cases}
\end{equation*}
there exists a set $\Omega_B^\ast$ with
$\mathbb P_B(\Omega_B^\ast)=1$ such that, for every fixed realization
$B\in\Omega_B^\ast$, there exists a quantum algorithm
$\mathcal A(\epsilon;B)$ that estimates $u(t,x;B)=\mathbb{E}_W[P(t,x;B)\,|\,B]$ with additive error at most
$\epsilon$ and success probability at least $0.9$, at computational cost
\[
\widetilde{\mathcal O}(\epsilon^{-1}).
\]
\end{theorem}

\begin{proof}
We adopt the notation of Theorem~\ref{thm:QA-MLMC}. For level
$\ell$, set
$h_\ell=(T-t)2^{-\ell}$ and
$N_\ell=2^\ell$.
By Proposition~\ref{prop:error-u},
the level-$\ell$ approximation satisfies 
the joint strong-error estimate
\[
\left\|
P_\ell(t,x)-P(t,x)
\right\|_{L^2_{W,B}}
\le
C2^{-\ell}.
\]
Equivalently,
\[
\mathbb E_{W,B}
\left[
\left|
P_\ell(t,x)-P(t,x)
\right|^2
\right]
\le
C2^{-2\ell}.
\]

We now pass from the joint estimate to a conditional estimate with respect to
a fixed realization of the backward Brownian motion. Applying
Lemma~\ref{lem:EB to EW} with $p=2$ to
\[
\Pi_{\ell,0}:=P_\ell(t,x)-P(t,x),
\]
we obtain that, for every $a>1$, there exists a finite random constant
$C_a(B)<\infty$ for $\mathbb P_B$-almost every realization of $B$ such that
\[
\mathbb E_W
\left[
\left|
P_\ell(t,x;B)-P(t,x;B)
\right|^2
\,\middle|\,B
\right]
\le
C_a(B)(1+\ell)^a2^{-2\ell}.
\]
Hence, for every such fixed $B$,
\[
\left\|
P_\ell(t,x;B)-P(t,x;B)
\right\|_{L^2_W}
\le
C_a(B)^{1/2}(1+\ell)^{a/2}2^{-\ell},
\]
i.e.
\[
\left\|
P_\ell(t,x;B)-P(t,x;B)
\right\|_{L^2_W}
=
\widetilde{\mathcal O}(2^{-\ell}).
\]

Therefore the bias satisfies
\[
\left|
\mathbb E_W[P_\ell(t,x;B)-P(t,x;B)]
\right|
\le
\left\|
P_\ell(t,x;B)-P(t,x;B)
\right\|_{L^2_W}
=
\widetilde{\mathcal O}(2^{-\ell}).
\]
Thus the weak convergence parameter in Theorem~\ref{thm:QA-MLMC} is
$\alpha=1$
up to logarithmic factors.

Next, for $\ell\ge1$, the level difference satisfies
\begin{align*}
\left\|
P_\ell(t,x;B)-P_{\ell-1}(t,x;B)
\right\|_{L^2_W}
&\le
\left\|
P_\ell(t,x;B)-P(t,x;B)
\right\|_{L^2_W}
+
\left\|
P_{\ell-1}(t,x;B)-P(t,x;B)
\right\|_{L^2_W} =
\widetilde{\mathcal O}(2^{-\ell}).
\end{align*}
Consequently,
\[
\operatorname{Var}_W
\left(
P_\ell(t,x;B)-P_{\ell-1}(t,x;B)
\right)
\le
\left\|
P_\ell(t,x;B)-P_{\ell-1}(t,x;B)
\right\|_{L^2_W}^2
=
\widetilde{\mathcal O}(2^{-2\ell}).
\]
Thus the variance decay parameter is
$\beta=2$
up to logarithmic factors.

Finally,  the computational cost per sample at level $\ell$
satisfies
\[
C_\ell=\mathcal{O}(2^\ell),
\]
so $\gamma=1.$

Therefore the conditions of Theorem~\ref{thm:QA-MLMC} are satisfied, with
\[
(\alpha,\beta,\gamma)=(1,2,1)
\]
up to logarithmic factors. Applying Theorem~\ref{thm:QA-MLMC}, we
conclude that, for $\mathbb P_B$-almost every fixed realization of $B$,
$u(t,x;B)$ can be estimated with additive error at most $\epsilon$ and success
probability at least $0.9$, at total computational cost
$\widetilde{\mathcal O}(\epsilon^{-1}).$
\begin{remark}
  Although Theorem~\ref{thm:QA-MLMC} in \cite{An2021quantumaccelerated} does not explicitly cover the case
$(\alpha,\beta,\gamma)=(1,2,1)$
up to logarithmic factors, this does not affect the resulting computational complexity. Moreover, Theorem~5 of \cite{li2026quantum} explicitly treats this borderline case.
\end{remark}
\end{proof}

\subsection{Nested Quantum-accelerated MLMC Estimator }
\label{subsec:Nested Quantum-accelerated MLMC Estimator}

Furthermore, let $\varphi:\mathbb R\to\mathbb R$ be a globally
Lipschitz function.
In many applications, the quantity of interest is obtained by averaging
over the random environment generated by the backward Brownian motion
$B$.
This leads to the nested expectation
\begin{align*}
\mathbb E_B\bigl[\varphi(u(t,x;B))\bigr].
\end{align*}
Following the framework developed in
\cite{Blanchet2025NonlinearQM},
we construct quantum-accelerated estimators for this quantity.

\begin{algorithm}[H]
\caption{Coupled Level Difference Evaluator $\mathcal{Q}_\ell(B_0)$ }
\label{alg: Q-Coupled Level Difference Evaluator}
\begin{algorithmic}[1]
\Require
Given backward Brownian motion $B_0$
\Ensure
A coupled level-difference sample $\mathcal{Q}_\ell(B_0)$.

\If{$\ell=0$}
    \State Apply $\mathcal{A}(1/(2K);B_0)$ to estimate $u(t,x;B_0)$,
and amplify the success probability to $1-(8K^2V)^{-1}$
using the powering lemma
\State Clip the output into the region $\left[-\sqrt{V},\sqrt{V}\right]$,
denote the clipped output by $\widehat u_0(B_0)$
    \State \Return $\mathcal{Q}_0(B_0)\gets \varphi\left(\widehat u_0(B_0)\right)$
\Else
    \State Apply $\mathcal{A}(2^{-(\ell+1)}/K;B_0)$ to estimate $u(t,x;B_0)$,
and amplify the success probability to $1-2^{-(2\ell+1)}(4K^2V)^{-1}$
using the powering lemma, denote the outputs by $\widetilde u_\ell(B_0)$
\State Apply $\mathcal{A}(2^{-\ell}/K;B_0)$ to estimate $u(t,x;B_0)$,
and amplify the success probability to $1-2^{-(2\ell-1)}(4K^2V)^{-1}$
using the powering lemma, denote the outputs by $\widetilde u_{\ell-1}(B_0)$
\State Clip  $\widetilde u_{\ell}(B_0)$ and $\widetilde u_{\ell-1}(B_0)$ into the region $\left[-\sqrt{V},\sqrt{V}\right]$
and denote as $\widehat u_\ell(B_0)$ and $\widehat u_{\ell-1}(B_0)$ 
respectively
    \State \Return $\mathcal{Q}_\ell(B_0)\gets \varphi\left(\widehat u_\ell(B_0)\right)-\varphi\left(\widehat u_{\ell-1}(B_0)\right)$
\EndIf
\end{algorithmic}
\end{algorithm}

\begin{algorithm}[H]
\caption{Nested Quantum-accelerated MLMC Estimator for $\theta=\mathbb{E}_{B}\left[\varphi(u(t,x;B))\right]$}
\label{alg:nested-qmlmc}
\begin{algorithmic}[1]
\Require
 Target accuracy $\epsilon$
 \Ensure
An estimator $\widehat\theta$ such that $|\widehat\theta-\theta|\leq \epsilon$ with probability $\geq 0.8$.

\State Set $L= \left\lceil \log_2\!\left(\frac{2}{\epsilon}\right)\right\rceil$
\State $\widehat\theta \gets 0$
\For{$\ell=0$ to $L$}
    \State Apply quantum mean estimation to estimate $\mathbb{E}_B\left[\mathcal{Q}_\ell(B)\right]$
    with accuracy $\epsilon/(2L+2)$ and success probability at least $\max\{0.9,1-0.01^{\ell}\}$.
    Denote the output by $\widehat m_\ell$
     \State $\widehat\theta \gets \widehat\theta + \widehat m_\ell$
\EndFor

\State \Return $\widehat\theta$
\end{algorithmic}
\end{algorithm}

\begin{theorem}
\label{thm:nested-qmlmc}
Under the assumptions of
Theorem~\ref{thm:QAMLMC for u when given B},
assume further that
$\varphi:\mathbb{R}\to\mathbb{R}$ is globally Lipschitz, namely,
there exists $K>0$ such that
$|\varphi(x)-\varphi(y)|
\le K|x-y|$,
$x,y\in\mathbb{R},$
the outer variance satisfies
$\operatorname{Var}_{B}\!\left(\varphi(u(t,x;B))\right)\le S,$
and the conditional second moment of the payoff satisfies
$\mathbb{E}\!\left[|P(t,x;B_0)|^2\right]\le V$
for every realization $B_0$ of $B$.

Let $\mathcal{Q}_\ell(B)$ be defined in Algorithm~\ref{alg: Q-Coupled Level Difference Evaluator} and
let $\widehat\theta$ be the output of Algorithm~\ref{alg:nested-qmlmc}.
Then, for any $\epsilon\in(0,1)$,  $\widehat\theta$  is a quantum estimator of 
$\mathbb{E}\left[\varphi(u(t,x))\right]$($=\mathbb{E}_B\left[\varphi(u(t,x;B))\right]$) with additive error $\epsilon$ and
success probability at least 0.8, the total cost is $\widetilde{\mathcal{O}}(\epsilon^{-1})$.
\end{theorem}
\begin{remark}
Recall that $P(t,x;B)$ is defined in \eqref{eq:eq of P(t,x;B)}.
For a fixed realization of the backward Brownian motion $B$,
we have $u(t,x;B)=\mathbb{E}[P(t,x;B)]$,
where the expectation is taken over the stochasticity of the forward SDE.
\end{remark}

\begin{proof}
Following the notation of
Algorithm~\ref{alg: Q-Coupled Level Difference Evaluator}, we have
\begin{align*}
  \mathcal Q_\ell(B)=
\begin{cases}
\varphi(\widehat u_0(B)), & \ell=0,\\
\varphi(\widehat u_\ell(B))-\varphi(\widehat u_{\ell-1}(B)), & \ell\geq 1,
\end{cases}
\qquad
\widehat u_\ell(B)\in[-\sqrt{V},\sqrt{V}].
\end{align*}

Let
\begin{align*}
  m_\ell := \mathbb E[\mathcal Q_\ell(B)].
\end{align*}
By the guarantee of quantum mean estimation, with the success probability
specified in Algorithm~\ref{alg:nested-qmlmc}, we have
\begin{align*}
  |\widehat m_\ell-m_\ell|
  \leq \frac{\epsilon}{2L+2},
  \qquad 0\leq \ell\leq L ,
\end{align*}
and
\begin{align*}
  \sum_{\ell=0}^L |\widehat m_\ell-m_\ell|
  \leq
  (L+1)\frac{\epsilon}{2L+2}
  =
  \frac{\epsilon}{2}.
\end{align*}

By the construction of
Algorithm~\ref{alg: Q-Coupled Level Difference Evaluator},
\begin{align*}
  \sum_{\ell=0}^L m_\ell
  =
  \mathbb E\!\left[\varphi(\widehat u_L(B))\right].
\end{align*}
Therefore, for
$  \theta=\mathbb E_B\!\left[\varphi(u(t,x;B))\right],$
we obtain
\begin{align*}
|\widehat\theta-\theta|
&\leq
\left|
\sum_{\ell=0}^L \widehat m_\ell
-
\sum_{\ell=0}^L m_\ell
\right|
+
\left|
\sum_{\ell=0}^L m_\ell-\theta
\right| \\
&\leq
\sum_{\ell=0}^L |\widehat m_\ell-m_\ell|
+
\left|
\mathbb E\!\left[\varphi(\widehat u_L(B))\right]
-
\mathbb E_B\!\left[\varphi(u(t,x;B))\right]
\right| \\
&\leq
\frac{\epsilon}{2}
+
K\mathbb E\!\left[
|\widehat u_L(B)-u(t,x;B)|
\right].
\end{align*}

Next, we bound the second term. By the construction in
Algorithm~\ref{alg: Q-Coupled Level Difference Evaluator}, for $\ell\geq 0$,
the estimator $\widetilde u_\ell(B_0)$ satisfies
\begin{align*}
  \left|\widetilde u_\ell(B_0)-u(t,x;B_0)\right|
  \leq \frac{2^{-(\ell+1)}}{K}
\end{align*}
with failure probability at most
\begin{align*}
  p_\ell
  =
  2^{-(2\ell+1)}(4K^2V)^{-1}.
\end{align*}
Since
\begin{align*}
  |u(t,x;B_0)|
  =
  \left|\mathbb E_W[P(t,x;B_0)]\right|
  \leq
  \left(\mathbb E_W[|P(t,x;B_0)|^2]\right)^{1/2}
  \leq \sqrt V,
\end{align*}
and since the clipping map onto $[-\sqrt V,\sqrt V]$ is nonexpansive,
we have on the success event
\begin{align*}
  |\widehat u_\ell(B_0)-u(t,x;B_0)|
  \leq
  |\widetilde u_\ell(B_0)-u(t,x;B_0)|
  \leq
  \frac{2^{-(\ell+1)}}{K}.
\end{align*}
On the failure event, both $\widehat u_\ell(B_0)$ and $u(t,x;B_0)$ belong to
$[-\sqrt V,\sqrt V]$, and hence
\begin{align*}
  |\widehat u_\ell(B_0)-u(t,x;B_0)|
  \leq 2\sqrt V.
\end{align*}
Consequently,
\begin{align*}
  \mathbb E\!\left[
  |\widehat u_\ell(B)-u(t,x;B)|^2
  \right]
  &\leq
  \left(\frac{2^{-(\ell+1)}}{K}\right)^2
  +(2\sqrt V)^2
  \cdot
  2^{-(2\ell+1)}(4K^2V)^{-1} \\
  &\leq
  K^{-2}2^{-2\ell}.
\end{align*}
Therefore,
\begin{align*}
K\mathbb E\!\left[
|\widehat u_L(B)-u(t,x;B)|
\right]
&\leq
K
\left(
\mathbb E\!\left[
|\widehat u_L(B)-u(t,x;B)|^2
\right]
\right)^{1/2} \\
&\leq
2^{-L}
\leq
\frac{\epsilon}{2},
\end{align*}
where the last inequality follows from
$L=\lceil \log_2(2/\epsilon)\rceil$.
Thus, on the joint success event,
\begin{align*}
  |\widehat\theta-\theta|\leq \epsilon .
\end{align*}

For the success probability, the failure probability of the level-$0$
outer mean estimation is at most $0.1$, and the failure probabilities of
the remaining outer mean estimations are bounded by
$\sum_{\ell=1}^{\infty}0.01^\ell<0.02$.
Hence the overall success probability is at least
\begin{align*}
  1-0.1-\sum_{\ell=1}^{\infty}0.01^\ell
  >0.8 .
\end{align*}

Next, we consider the computational complexity of Algorithm~\ref{alg:nested-qmlmc}.
By the construction in Algorithm~\ref{alg: Q-Coupled Level Difference Evaluator},
we have
\begin{align*}
  \left|\widetilde u_\ell(B_0)-u(t,x;B_0)\right|\leq 2^{-(\ell+1)}/K
\end{align*}
with probability at least $1-2^{-(2\ell+1)}(4K^2V)^{-1}$.
Since $\mathbb{E}\left[|P(t,x;B_0)|^2 \right]\leq V$, we known
$$|u(t,x;B_0)|=\left|\mathbb{E}\left[P(t,x;B_0) \right]\right|\leq \sqrt{V}.$$

So
\begin{align*}
  \mathbb{E}\left[\left|\widehat u_{\ell}(B)-u(t,x;B)\right|^2\right]
  \leq \left(2^{-(\ell+1)}/K\right)^2+(2\sqrt{V})^2\cdot 2^{-(2\ell+1)}(4K^2V)^{-1}\leq K^{-2}2^{-2\ell}.
\end{align*}

For $\ell\geq 1$, since $\varphi$ is globally $K$-Lipschitz, we have
\begin{align*}
\operatorname{Var}\left(\mathcal Q_\ell(B)\right)
&\leq \mathbb{E}\left[\mathcal Q_\ell(B)^2\right]\\
&= \mathbb{E}\left[\left|\varphi(\widehat u_\ell(B))-\varphi(\widehat u_{\ell-1}(B))\right|^2\right] \\
&\leq 2\mathbb{E}\left[\left|\varphi(\widehat u_\ell(B))-\varphi( u(t,x;B))\right|^2\right]
+2\mathbb{E}\left[\left|\varphi(\widehat u_{\ell-1}(B))-\varphi( u(t,x;B))\right|^2\right]\\
&\leq 2K^2\mathbb{E}\left[\left|\widehat u_\ell(B)- u(t,x;B)\right|^2\right]
+2K^2\mathbb{E}\left[\left|\widehat u_{\ell-1}(B)- u(t,x;B)\right|^2\right]\\
&\leq 16\cdot2^{-2\ell} .
\end{align*}

By Theorem~\ref{thm:QAMLMC for u when given B}, Lemma~\ref{lem:quantum mean estimation} and powering lemma,
for $\ell\geq 1$, the cost of estimating $\mathbb{E}_B[\mathcal{Q}_\ell(B)]$ is
\begin{align*}
  \mathcal{O}((2L+2)/\epsilon\times 2^{-\ell}\times \log(100^\ell))\times
  \widetilde{\mathcal{O}} \left(2^{\ell+1}K\times \log(2^{-(2\ell+1)}(4K^2V)^{-1})\right)
  =\widetilde{\mathcal{O}} \left(KLl^2/\epsilon\right).
\end{align*}
Here, the first term corresponds to the number of queries required to estimate
$\mathbb{E}_B[\mathcal{Q}_\ell(B)]$, while the second term represents the
computational cost of a single query to $\mathcal{Q}_\ell(B)$, as given by
Theorem~\ref{thm:QAMLMC for u when given B}.

For $\ell=0$, the variance is bounded by $2S+1$ under our assumptions,
and hence the computational complexity is 
$\mathcal{O}(SL/\epsilon)$.

Summing up all the $L$ costs together yields a total cost of $\widetilde{\mathcal{O}}(\epsilon^{-1})$.
\end{proof}

\begin{remark}
Note that the conditional expectation framework of
Section~\ref{subsec:Conditional Quantum-Accelerated MLMC Estimator}
and the nested expectation framework of
Section~\ref{subsec:Nested Quantum-accelerated MLMC Estimator}
depend only on the existence of a representation of the form
\begin{align}
Q(t,x;B)
=
\mathbb E_W[\mathcal{P}(t,x;B)\mid B],
\end{align}
where \(\Pi\) is a stochastic path functional.

The specific structure of the underlying SPDE enters the framework
only through the choice of \(\Pi\).
Consequently, the conditional and nested QA-MLMC methodologies
developed above apply without modification to any quantity
admitting a representation of this form.

In the subsequent sections, we will construct such path
functionals \(\mathcal P\) for the SPDE solution itself and for the
associated first- and second-order Greek estimators.
\end{remark}

\section{Applications to Derivative Pricing and Sensitivity Analysis}\label{sec:applications}

This section specializes the conditional and nested QA-MLMC framework of
Sections~\ref{sec:SPDE Formulation and Probabilistic Representation}--\ref{sec:Quantum-accelerated MLMC for the SPDE} to derivative pricing and sensitivity
analysis.  The common structure is the following: after fixing the backward
Brownian path $B$, each target quantity is represented as
\[
\mathbb E_W[\mathcal P(t,x;B)\mid B],
\]
where $\mathcal P(t,x;B)$ is a scalar path functional of the forward diffusion.
Once such a representation is available, the same dyadic fine--coarse coupling
and the same quantum mean estimation routine used for the price can be applied
to the corresponding payoff register.

Throughout this section, unless otherwise stated, we work in the scalar case
$k=1$.  Vector-valued claims are obtained componentwise.  We write
$\mathbb E_W[\cdot\mid B]$ for expectation over the forward Brownian motion
after conditioning on the realization of the backward Brownian motion $B$.
The terminal function in the BDSDE representation is denoted by $G$; when the
terminal condition is written as $g$ elsewhere in the paper, one simply sets
$G=g$.

\subsection{Conditional Pricing under Common Noise}\label{sec:app-pricing}

By \eqref{eq:eq of P(t,x;B)}--\eqref{eq:u(t,x;B)=E[P]}, the conditional value
of the linear backward SPDE is
\begin{equation}\label{eq:app-price-representation}
u(t,x;B)
=
\mathbb E_W\!\left[P(t,x;B)\mid B\right],
\end{equation}
where
\begin{align}\label{eq:app-payoff-functional}
P(t,x;B)
&=
\Phi(t,T)\,G\!\left(X_T^{t,x}\right)
+\int_t^T
\Phi(t,s)
\Big(F(s,X_s^{t,x})+d(s)H(s,X_s^{t,x})\Big)\,\mathrm ds
+\int_t^T
\Phi(t,s)H(s,X_s^{t,x})\,\mathrm d \overleftarrow{B}_s.
\end{align}
Here $X^{t,x}$ is the forward diffusion driven by $W$, while the backward
stochastic integral is evaluated along the fixed path $B$.

\paragraph{Pure discounted pricing.}
For a European claim with discounted terminal payoff only, take
\begin{equation}\label{eq:app-pure-pricing}
G(x)=g(x),\quad
F\equiv0,\quad
H\equiv0,\quad
c(s)=-r(s),\quad
d(s)=\widetilde c(s)=0 .
\end{equation}
Then
\begin{equation}\label{eq:app-european-price}
u(t,x;B)
=
\mathbb E_W\!\left[
\exp\!\left(-\int_t^T r(s)\,\mathrm ds\right)
g\!\left(X_T^{t,x}\right)
\;\middle|\; B
\right].
\end{equation}
In this pure terminal-payoff case, the dependence on $B$ disappears unless the
forward model itself contains common-noise coefficients.  Nonzero $F$ covers
running cashflows, and nonzero $H$ encodes exposure to the backward common-noise
factor.

\paragraph{Compatibility with the QA-MLMC algorithms.}
For each level $\ell$, let $P_\ell(t,x;B)$ be the time-discretized version of
\eqref{eq:app-payoff-functional} generated by the operators
\[
(\mathcal S_X,\mathcal S_\Phi,\mathcal S_{\mathrm{int}})
\]
and by the same dyadic Brownian coupling as in
Algorithm~\ref{alg:Q-level-l-difference}.  The conditional telescoping identity
is
\begin{equation}\label{eq:app-price-telescoping}
\mathbb E_W[P_L(t,x;B)\mid B]
=
\mathbb E_W[P_0(t,x;B)\mid B]
+
\sum_{\ell=1}^L
\mathbb E_W[P_\ell(t,x;B)-P_{\ell-1}(t,x;B)\mid B].
\end{equation}
Therefore Algorithm~\ref{alg:conditional-QA-MLMC} estimates $u(t,x;B)$.
The unconditional price
\begin{equation}\label{eq:app-unconditional-price}
U(t,x):=\mathbb E_B[u(t,x;B)]
\end{equation}
is obtained by Algorithm~\ref{alg:nested-qmlmc} with $\varphi(z)=z$.

\subsection{First-order Greeks for Smooth Payoffs}
\label{sec:app-first-order-greeks}

We first treat pathwise first-order Greeks.  In this subsection the
coefficients $c,\widetilde c,d$ are deterministic functions of time, as in the
linear BDSDE representation.  Hence the exponential weight $\Phi(t,s)$ depends
on $W$ and $B$, but not on the initial state $x$.  If one allows
state-dependent discounting or state-dependent coefficients in $\Phi$, then the
corresponding spatial derivatives of $\Phi$ must be added to the formulas below.

\begin{proposition}[Pathwise representation of conditional first-order Greeks]
\label{prop:app-first-order-greeks}
Let the forward Brownian motion have dimension $d$, and let
$\sigma_{\cdot a}$ denote the $a$th column of $\sigma$.  Assume that
$b,\sigma\in C_b^2$ in the spatial variable, that
$G\in C_p^2(\mathbb R^d)$, and that
$F,H\in C_b^2([0,T]\times\mathbb R^d)$.  Assume also the usual moment bounds
that justify differentiation under $\mathbb E_W[\cdot\mid B]$.

Define the Jacobian flow
\begin{equation}\label{eq:app-jacobian}
J_s^{t,x}:=\nabla_x X_s^{t,x}\in\mathbb R^{d\times d},
\quad
J_t^{t,x}=I_d .
\end{equation}
Then $J^{t,x}$ solves
\begin{equation}\label{eq:app-jacobian-sde}
dJ_s^{t,x}
=
\nabla_x b(X_s^{t,x})J_s^{t,x}\,ds
+
\sum_{a=1}^{d}
\nabla_x\sigma_{\cdot a}(X_s^{t,x})J_s^{t,x}\,dW_s^a .
\end{equation}
For each coordinate direction $e_i$, define
\begin{equation}\label{eq:app-first-order-greek}
\mathcal G_i(t,x;B):=\partial_{x_i}u(t,x;B).
\end{equation}
Then
\begin{equation}\label{eq:app-first-order-greek-expectation}
\mathcal G_i(t,x;B)
=
\mathbb E_W\!\left[P^{(i)}(t,x;B)\mid B\right],
\end{equation}
where
\begin{align}\label{eq:app-first-order-greek-payoff}
P^{(i)}(t,x;B)
&=
\Phi(t,T)\,
\nabla G(X_T^{t,x})^\top J_T^{t,x}e_i
\nonumber\\
&\quad
+\int_t^T
\Phi(t,s)\,
\nabla_x\!\Big(F(s,X_s^{t,x})+d(s)H(s,X_s^{t,x})\Big)^\top
J_s^{t,x}e_i\,\mathrm ds
\nonumber\\
&\quad
+\int_t^T
\Phi(t,s)\,
\nabla_x H(s,X_s^{t,x})^\top
J_s^{t,x}e_i\,
\mathrm d \overleftarrow{B}_s.
\end{align}
\end{proposition}

\begin{proof}
Under the stated assumptions, the stochastic flow
$x\mapsto X_s^{t,x}$ is differentiable and its Jacobian solves
\eqref{eq:app-jacobian-sde}.  Since $\Phi$ is independent of the initial state
$x$ in the present linear setting, differentiating
\eqref{eq:app-payoff-functional} pathwise gives
\eqref{eq:app-first-order-greek-payoff}.  The assumed moment bounds justify
interchanging differentiation and conditional expectation, which proves
\eqref{eq:app-first-order-greek-expectation}.
\end{proof}

\subsubsection{Delta}\label{sec:app-delta}

Assume that the first state component is the spot, so that $x=(S,\chi)$, where
$\chi$ denotes all remaining factors.  The conditional Delta is
\begin{equation}\label{eq:app-delta}
\Delta(t,x;B)
:=
\partial_S u(t,x;B)
=
\mathbb E_W\!\left[P^{(S)}(t,x;B)\mid B\right].
\end{equation}
In the pure pricing case \eqref{eq:app-pure-pricing}, if
$G(x)=g(S)$ depends only on the terminal spot, then
\begin{equation}\label{eq:app-delta-pure}
P^{(S)}(t,x;B)
=
\Phi(t,T)\,
g'(S_T^{t,x})\,\partial_S S_T^{t,x}.
\end{equation}
Thus Delta is obtained by propagating only the tangent direction
$J_s^{t,x}e_S$.

\subsubsection{Spot Vega}\label{sec:app-spot-vega}

If the state contains an instantaneous variance or volatility factor $V$, the
spot Vega is the sensitivity with respect to that state variable:
\begin{equation}\label{eq:app-spot-vega}
\mathcal V_{\mathrm{spot}}(t,x;B)
:=
\partial_V u(t,x;B)
=
\mathbb E_W\!\left[P^{(V)}(t,x;B)\mid B\right].
\end{equation}
If $G(x)=g(S)$ depends only on the terminal spot, then
\begin{equation}\label{eq:app-spot-vega-pure}
P^{(V)}(t,x;B)
=
\Phi(t,T)\,
g'(S_T^{t,x})\,\partial_V S_T^{t,x}.
\end{equation}
The dependence on the initial variance or volatility is transmitted through the
tangent flow.

\subsubsection{Rho and general parameter Greeks}
\label{sec:app-rho-parameter-greeks}

State sensitivities are only one class of Greeks.  Let $\vartheta$ be a scalar
parameter entering the forward coefficients, the terminal payoff, the running
terms, or the exponential weight.  We write
\[
X^{\vartheta,t,x},\quad
\Phi^\vartheta,\quad
P^\vartheta(t,x;B),\quad
u^\vartheta(t,x;B)
=
\mathbb E_W[P^\vartheta(t,x;B)\mid B].
\]
Assume differentiability in $\vartheta$ and sufficient moment bounds for
differentiation under $\mathbb E_W[\cdot\mid B]$.

Define the parameter tangent process
\begin{equation}\label{eq:app-parameter-tangent}
U_s^\vartheta
:=
\partial_\vartheta X_s^{\vartheta,t,x},
\quad
U_t^\vartheta=0 .
\end{equation}
It solves
\begin{align}\label{eq:app-parameter-tangent-sde}
dU_s^\vartheta
&=
\Big[
\partial_\vartheta b^\vartheta(X_s^{\vartheta,t,x})
+
\nabla_x b^\vartheta(X_s^{\vartheta,t,x})U_s^\vartheta
\Big]\,ds
+
\sum_{a=1}^{d}
\Big[
\partial_\vartheta\sigma_{\cdot a}^\vartheta(X_s^{\vartheta,t,x})
+
\nabla_x\sigma_{\cdot a}^\vartheta(X_s^{\vartheta,t,x})U_s^\vartheta
\Big]\,dW_s^a .
\end{align}

When the coefficients $c^\vartheta,d^\vartheta,\widetilde c^\vartheta$ in
$\Phi^\vartheta$ depend on $\vartheta$, define
\begin{align}\label{eq:app-logPhi-derivative}
\Lambda_\vartheta(t,s)
&:=
\int_t^s \partial_\vartheta c^\vartheta(r)\,\mathrm dr
+
\int_t^s \partial_\vartheta d^\vartheta(r)\,
\mathrm d \overleftarrow{B}_r
+
\int_t^s \partial_\vartheta\widetilde c^\vartheta(r)\,\mathrm dW_r
-
\int_t^s
\Big(
\widetilde c^\vartheta(r)\partial_\vartheta\widetilde c^\vartheta(r)
-
d^\vartheta(r)\partial_\vartheta d^\vartheta(r)
\Big)\,\mathrm dr .
\end{align}
Then
\begin{equation}\label{eq:app-Phi-derivative}
\partial_\vartheta\Phi^\vartheta(t,s)
=
\Phi^\vartheta(t,s)\Lambda_\vartheta(t,s).
\end{equation}
For notational compactness, write
\[
A^\vartheta(s,x)
:=
F^\vartheta(s,x)+d^\vartheta(s)H^\vartheta(s,x).
\]
Then
\begin{equation}\label{eq:app-general-parameter-greek}
\partial_\vartheta u^\vartheta(t,x;B)
=
\mathbb E_W\!\left[Q^{(\vartheta)}(t,x;B)\mid B\right],
\end{equation}
where
\begin{align}\label{eq:app-general-parameter-greek-payoff}
Q^{(\vartheta)}(t,x;B)
&=
\partial_\vartheta\Phi^\vartheta(t,T)\,
G^\vartheta(X_T^{\vartheta,t,x})
+
\Phi^\vartheta(t,T)
\Big(
\partial_\vartheta G^\vartheta(X_T^{\vartheta,t,x})
+
\nabla_x G^\vartheta(X_T^{\vartheta,t,x})^\top U_T^\vartheta
\Big)
\nonumber\\
&\quad
+
\int_t^T
\Big[
\partial_\vartheta\Phi^\vartheta(t,s)\,
A^\vartheta(s,X_s^{\vartheta,t,x})
+
\Phi^\vartheta(t,s)
\Big(
\partial_\vartheta A^\vartheta(s,X_s^{\vartheta,t,x})
+
\nabla_x A^\vartheta(s,X_s^{\vartheta,t,x})^\top U_s^\vartheta
\Big)
\Big]\,\mathrm ds\notag\\
&\quad+
\int_t^T
\Big[
\partial_\vartheta\Phi^\vartheta(t,s)\,
H^\vartheta(s,X_s^{\vartheta,t,x})
+
\Phi^\vartheta(t,s)
\Big(
\partial_\vartheta H^\vartheta(s,X_s^{\vartheta,t,x})
+
\nabla_x H^\vartheta(s,X_s^{\vartheta,t,x})^\top U_s^\vartheta
\Big)
\Big]\,
\mathrm d \overleftarrow{B}_s.
\end{align}

\paragraph{Discount-only Rho.}
If the short rate enters only through the discount coefficient
$c^\eta(s)=-r^\eta(s)$, with
\[
r^\eta(s)=r(s)+\eta\psi(s),
\]
and all other coefficients, including the forward dynamics, are kept fixed,
then in the pure pricing setting \eqref{eq:app-pure-pricing},
\begin{equation}\label{eq:app-rho-discount-only}
\mathrm{Rho}^{\mathrm{disc}}_{\psi}(t,x;B)
:=
\left.\partial_\eta u^\eta(t,x;B)\right|_{\eta=0}
=
-
\mathbb E_W\!\left[
\left(\int_t^T\psi(s)\,\mathrm ds\right)
\Phi(t,T)G(X_T^{t,x})
\;\middle|\;B
\right].
\end{equation}
For a constant short-rate shift, $\psi\equiv1$,
\begin{equation}\label{eq:app-rho-discount-only-constant}
\mathrm{Rho}^{\mathrm{disc}}(t,x;B)
=
-(T-t)u(t,x;B).
\end{equation}
This identity is a discount-only identity.  In risk-neutral equity models where
$r$ also enters the drift of $X$, the full Rho must instead be computed from
\eqref{eq:app-general-parameter-greek-payoff}, including the parameter tangent
$U^\vartheta$.

\paragraph{Discrete parameter Greeks.}
For the augmented Milstein dynamics of Section~\ref{sec:Strong-Order-One Numerical Schemes for Pricing and Greek Estimators}, let
\[
\widetilde X_n=(S_n,V_n,D_n,Y_n,Z_n)
\]
and write one step as
\[
\widetilde X_{n+1}
=
\mathcal S(\widetilde X_n;\xi_n),
\]
where $\xi_n$ collects the Brownian increments and any iterated stochastic
integrals required by the scheme.  A discrete parameter Greek is obtained by
propagating
\begin{equation}\label{eq:app-discrete-parameter-tangent}
\dot{\widetilde X}_{n+1}
=
\nabla_{\widetilde X}\mathcal S(\widetilde X_n;\xi_n)
\dot{\widetilde X}_n
+
\partial_\vartheta\mathcal S(\widetilde X_n;\xi_n).
\end{equation}
The terminal Greek payoff is obtained by differentiating the terminal payoff
register, for example
\[
D_NG(S_N)+Y_N+Z_N .
\]
This covers full Rho as well as model-parameter sensitivities such as
$\partial_\kappa u$, $\partial_\xi u$, $\partial_\rho u$, and
$\partial_\theta u$.

\subsection{Second-order Greeks for Smooth Payoffs}
\label{sec:app-second-order-greeks}

Second-order Greeks are obtained by differentiating the pathwise first-order
representation once more.

\begin{proposition}[Pathwise representation of conditional second-order Greeks]
\label{prop:app-second-order-greeks}
Assume that $b,\sigma\in C_b^3$ in the spatial variable, that
$G\in C_p^3(\mathbb R^d)$, and that
$F,H\in C_b^3([0,T]\times\mathbb R^d)$.  Assume the corresponding moment bounds
needed for differentiating under $\mathbb E_W[\cdot\mid B]$.

For $i,j\in\{1,\ldots,d\}$, define
\begin{equation}\label{eq:app-second-tangent}
K_s^{(ij),t,x}
:=
\partial_{x_i x_j}^2X_s^{t,x},
\quad
K_t^{(ij),t,x}=0 .
\end{equation}
Let $J_s^i:=J_s^{t,x}e_i$.  Then
\begin{align}\label{eq:app-second-tangent-sde}
dK_s^{(ij),t,x}
&=
\Big[
\nabla_x b(X_s^{t,x})K_s^{(ij),t,x}
+
\nabla_x^2 b(X_s^{t,x})[J_s^i,J_s^j]
\Big]\,ds
+
\sum_{a=1}^{d}
\Big[
\nabla_x\sigma_{\cdot a}(X_s^{t,x})K_s^{(ij),t,x}
+\nabla_x^2\sigma_{\cdot a}(X_s^{t,x})[J_s^i,J_s^j]
\Big]\,dW_s^a .
\end{align}
Moreover,
\begin{equation}\label{eq:app-second-order-greek}
\mathcal G_{ij}^{(2)}(t,x;B)
:=
\partial_{x_i x_j}^2u(t,x;B)
=
\mathbb E_W\!\left[P^{(ij)}(t,x;B)\mid B\right],
\end{equation}
where
\begin{align}\label{eq:app-second-order-greek-payoff}
P^{(ij)}(t,x;B)
&=
\Phi(t,T)
\Big(
(J_T^j)^\top\nabla_x^2G(X_T^{t,x})J_T^i
+
\nabla_xG(X_T^{t,x})^\top K_T^{(ij),t,x}
\Big)
\nonumber\\
&\quad
+
\int_t^T
\Phi(t,s)
\Big(
(J_s^j)^\top
\nabla_x^2\!\big(F+dH\big)(s,X_s^{t,x})J_s^i
+
\nabla_x\!\big(F+dH\big)(s,X_s^{t,x})^\top
K_s^{(ij),t,x}
\Big)\,\mathrm ds
\nonumber\\
&\quad
+
\int_t^T
\Phi(t,s)
\Big(
(J_s^j)^\top\nabla_x^2H(s,X_s^{t,x})J_s^i
+
\nabla_xH(s,X_s^{t,x})^\top K_s^{(ij),t,x}
\Big)\,
\mathrm d \overleftarrow{B}_s.
\end{align}
Here $\nabla_x^2 a(x)[u,v]$ denotes the bilinear action of the Hessian of $a$
on the pair $(u,v)$.
\end{proposition}

\begin{proof}
Differentiate \eqref{eq:app-first-order-greek-payoff} with respect to the
initial coordinate $x_j$.  The $C_b^3/C_p^3$ regularity assumptions give the
second derivative of the stochastic flow and the SDE
\eqref{eq:app-second-tangent-sde}.  The chain rule yields
\eqref{eq:app-second-order-greek-payoff}, and the moment assumptions justify
passing the derivative through the conditional expectation.
\end{proof}

\subsubsection{Gamma}\label{sec:app-gamma}

When $x=(S,\chi)$ and the first component is the spot, the conditional Gamma is
\begin{equation}\label{eq:app-gamma}
\Gamma(t,x;B)
:=
\partial_{SS}u(t,x;B)
=
\mathbb E_W\!\left[P^{(SS)}(t,x;B)\mid B\right].
\end{equation}
In the pure pricing case \eqref{eq:app-pure-pricing}, if $G(x)=g(S)$, then
\begin{equation}\label{eq:app-gamma-pure}
P^{(SS)}(t,x;B)
=
\Phi(t,T)
\left[
g''(S_T^{t,x})(\partial_SS_T^{t,x})^2
+
g'(S_T^{t,x})\partial_{SS}^2S_T^{t,x}
\right].
\end{equation}

\subsubsection{Vanna and Volga}\label{sec:app-vanna-volga}

If the state contains a variance or volatility factor $V$, then
\begin{equation}\label{eq:app-vanna-volga}
\mathrm{Vanna}(t,x;B)
:=
\partial_{SV}u(t,x;B)
=
\mathbb E_W\!\left[P^{(SV)}(t,x;B)\mid B\right],
\quad
\mathrm{Volga}(t,x;B)
:=
\partial_{VV}u(t,x;B)
=
\mathbb E_W\!\left[P^{(VV)}(t,x;B)\mid B\right].
\end{equation}
If $G(x)=g(S)$ depends only on the terminal spot, then
\begin{align}
P^{(SV)}(t,x;B)
&=
\Phi(t,T)
\left[
g''(S_T^{t,x})
\partial_SS_T^{t,x}\,
\partial_VS_T^{t,x}
+
g'(S_T^{t,x})\partial_{SV}^2S_T^{t,x}
\right],
\label{eq:app-vanna-pure}
\\
P^{(VV)}(t,x;B)
&=
\Phi(t,T)
\left[
g''(S_T^{t,x})
(\partial_VS_T^{t,x})^2
+
g'(S_T^{t,x})\partial_{VV}^2S_T^{t,x}
\right].
\label{eq:app-volga-pure}
\end{align}

\subsection{Greek Estimators in the MLMC and QA-MLMC Pipeline}
\label{sec:app-greek-algorithms}

The formulas above identify the exact scalar random variables whose conditional
expectations are the Greeks.  To use them in the algorithms, one augments the
one-step discretization by tangent variables.

\paragraph{First-order Greeks.}
Let
\[
\xi_k
\]
collect all random inputs required by the one-step scheme, including
$\Delta W_k$, $\Delta\overleftarrow B_k$, and, when used, mixed iterated
integrals such as $J_{t_k,h}^{WB}$.  A first-order Greek discretization has the
form
\begin{equation}\label{eq:app-discrete-jacobian}
J_{k+1}
=
\mathcal S_J(J_k,X_k,t_k,h;\Delta W_{t_k},\Delta\overleftarrow B_{t_k}),
\end{equation}
together with a Greek-payoff update
\begin{equation}\label{eq:app-discrete-first-order-payoff}
Y_{k+1}^{(i)}
=
Y_k^{(i)}
+
\mathcal S_{\mathrm{int}}^{(i)}
\big(
\Phi_k,X_k,J_k,t_k,h;\Delta W_{t_k},\Delta\overleftarrow B_{t_k}
\big),
\end{equation}
whose terminal value approximates \eqref{eq:app-first-order-greek-payoff}.

\paragraph{Second-order Greeks.}
For second-order sensitivities, add
\begin{equation}\label{eq:app-discrete-second-tangent}
K_{k+1}^{(ij)}
=
\mathcal S_K\left(K_k^{(ij)},J_k,X_k,t_k,h;\Delta W_{t_k},\Delta\overleftarrow B_{t_k}\right),
\end{equation}
and
\begin{equation}\label{eq:app-discrete-second-order-payoff}
Y_{k+1}^{(ij)}
=
Y_k^{(ij)}
+
\mathcal S_{\mathrm{int}}^{(ij)}
\big(
\Phi_k,X_k,J_k,K_k^{(ij)},t_k,h;\Delta W_{t_k},\Delta\overleftarrow B_{t_k}
\big),
\end{equation}
whose terminal value approximates \eqref{eq:app-second-order-greek-payoff}.

Using the same dyadic Brownian coupling as in
Algorithm~\ref{alg:Q-level-l-difference}, define the level differences
\begin{equation}\label{eq:app-greek-level-differences}
\Delta_\ell^{(i)}
:=
P_\ell^{(i,f)}-P_{\ell-1}^{(i,c)},
\quad
\Delta_\ell^{(ij)}
:=
P_\ell^{(ij,f)}-P_{\ell-1}^{(ij,c)} .
\end{equation}
These level differences are fed into Algorithm~\ref{alg:conditional-QA-MLMC} in
exactly the same way as the price level differences.

\begin{corollary}[Conditional QA-MLMC for smooth Greeks]
\label{cor:app-qmlmc-greeks}
Assume the hypotheses of Proposition~\ref{prop:app-first-order-greeks}.  Assume
also that the augmented first-order discretization
\[
(\mathcal S_X,\mathcal S_\Phi,\mathcal S_J,
\mathcal S_{\mathrm{int}}^{(i)})
\]
has strong-error order at least $1$, satisfies the required accumulated stability
estimate, and has uniformly bounded moments of sufficiently high order.  Then
\[
\left|
\mathbb E_W[P_\ell^{(i)}-P^{(i)}\mid B]
\right|
=
\mathcal O(2^{-\ell}),
\quad
\operatorname{Var}_W(P_\ell^{(i)}-P_{\ell-1}^{(i)}\mid B)
=
\mathcal O(2^{-2\ell}),
\quad
C_\ell=\mathcal O(2^\ell).
\]
Consequently, Algorithm~\ref{alg:conditional-QA-MLMC}, with the price payoff
register replaced by the first-order Greek payoff register, estimates
$\mathcal G_i(t,x;B)$ with additive error $\epsilon$ and cost
$\widetilde{\mathcal O}(\epsilon^{-1})$.

If, in addition, the augmented second-order discretization
\[
(\mathcal S_X,\mathcal S_\Phi,\mathcal S_J,\mathcal S_K,
\mathcal S_{\mathrm{int}}^{(ij)})
\]
has strong-error order at least $1$, satisfies the analogous stability estimate, and
has the required moment bounds, then the same conclusion holds for
$\mathcal G_{ij}^{(2)}(t,x;B)=\partial_{x_i x_j}^2u(t,x;B)$.
\end{corollary}

\begin{proof}
The proof is the same as the proof of
Theorem~\ref{thm:QAMLMC for u when given B}.  The Greek payoff is still a scalar
path functional.  For any fixed Greek, the tangent variables enlarge the state
dimension only by a constant factor, so the cost exponent remains
$\gamma=1$.  The assumed strong-error order gives $\alpha=1$ and $\beta=2$ in
Theorem~\ref{thm:QA-MLMC}.  Therefore the conditional quantum MLMC cost is
$\widetilde{\mathcal O}(\epsilon^{-1})$.
\end{proof}

\begin{corollary}[Nested QA-MLMC for unconditional Greeks]
\label{cor:app-nested-qmlmc-greeks}
Let
\[
U(t,x):=\mathbb E_B[u(t,x;B)].
\]
In addition to the assumptions of Corollary~\ref{cor:app-qmlmc-greeks}, assume
that differentiation may be interchanged with the outer expectation over $B$,
for example by dominated convergence or by a uniform integrability condition on
the conditional Greek payoffs.  Assume also
\[
\sup_{B_0}
\mathbb E_W\!\left[|P^{(i)}(t,x;B_0)|^2\mid B_0\right]
<\infty,
\]
and, for second-order Greeks,
\[
\sup_{B_0}
\mathbb E_W\!\left[|P^{(ij)}(t,x;B_0)|^2\mid B_0\right]
<\infty .
\]
Then nested QA-MLMC yields estimators of
\[
\partial_{x_i}U(t,x)
=
\mathbb E_B[\mathcal G_i(t,x;B)]
\]
and
\[
\partial_{x_i x_j}^2U(t,x)
=
\mathbb E_B[\mathcal G_{ij}^{(2)}(t,x;B)]
\]
with additive error $\epsilon$ and total cost
$\widetilde{\mathcal O}(\epsilon^{-1})$.
\end{corollary}

\begin{proof}
Apply Theorem~\ref{thm:nested-qmlmc} with $\varphi(z)=z$, replacing the inner
conditional price estimator by the corresponding conditional Greek estimator.
The interchange condition identifies the outer expectation of the conditional
Greek with the Greek of the unconditional value function.
\end{proof}

\subsection{Nonsmooth Payoffs: Likelihood-ratio and Malliavin-weight Greeks}
\label{sec:app-nonsmooth-greeks}

The pathwise Greek formulas above were stated under differentiability
assumptions on the terminal function \(G\).  Hence they do not directly cover
nonsmooth payoffs.  For Lipschitz payoffs with isolated kinks, such as vanilla
calls, first-order pathwise estimators may still be justified under suitable
non-atomicity or density assumptions on \(X_T^{t,x}\).  However, discontinuous
payoffs, barrier-type path functionals, and higher-order Greeks generally
require additional smoothing, likelihood-ratio, or Malliavin
integration-by-parts arguments~\cite{BroadieGlasserman1996,Fournie1999,Giles2009Vibrato}.

For illustration, consider the pure terminal-value case \eqref{eq:app-pure-pricing}, so that \(F=H=0\), \(c=-r\), and \(d=\widetilde c=0\).  In this case \(\Phi(t,T)\) is independent of the forward Brownian motion after conditioning on \(B\).  Assume that the forward diffusion is sufficiently smooth and uniformly elliptic, with
\[
a(s,x):=\sigma(s,x)\sigma(s,x)^\top \succeq \lambda I_d
\]
for some \(\lambda>0\).

If the transition density \(p(t,x;T,y)\) exists and is differentiable with respect to the initial condition \(x\), then the likelihood-ratio representation has the form~\cite{BroadieGlasserman1996}
\begin{equation}\label{eq:app-lr-greek}
\partial_{x_i}u(t,x;B)
=
\mathbb E_W\!\left[
\Phi(t,T)G(X_T^{t,x})
\,\partial_{x_i}\log p(t,x;T,X_T^{t,x})
\;\middle|\;B
\right].
\end{equation}

Alternatively, one may use a Malliavin-weight representation.  Define
\begin{equation}\label{eq:app-bel-weight-general}
\mathcal W_i(t,T)
:=
\frac{1}{T-t}
\int_t^T
\left(
\sigma(s,X_s^{t,x})^\top
a(s,X_s^{t,x})^{-1}
J_s^i
\right)^\top \mathrm dW_s .
\end{equation}
When \(\sigma\) is square and invertible, this reduces to
\begin{equation}\label{eq:app-bel-weight-square}
\mathcal W_i(t,T)
=
\frac{1}{T-t}
\int_t^T
\left(
\sigma(s,X_s^{t,x})^{-1}J_s^i
\right)^\top \mathrm dW_s .
\end{equation}
The Bismut--Elworthy--Li formula then gives~\cite{Bismut1984,ElworthyLi1994,Fournie1999}
\begin{equation}\label{eq:app-bel-directional}
\partial_{x_i}u(t,x;B)
=
\mathbb E_W\!\left[
\Phi(t,T)G(X_T^{t,x})\,\mathcal W_i(t,T)
\;\middle|\;B
\right].
\end{equation}
Thus the derivative is again represented as the conditional expectation of a scalar path functional, with the payoff derivative replaced by a stochastic
weight.

The simple form \eqref{eq:app-bel-directional} uses the fact that, in the pure terminal-value setting, \(\Phi(t,T)\) is independent of \(W\).  If \(\widetilde c\neq0\), or more generally if the exponential factor or running payoff contains \(W\)-dependent terms, the Malliavin integration-by-parts formula must also account for the Malliavin derivative of those terms.  Such general weighted representations can still be incorporated into the same conditional-expectation framework, but the weight and payoff register must be modified accordingly.

The conditional and nested QA-MLMC architecture can be applied to these weighted-payoff estimators once the corresponding discretized weighted payoff registers satisfy the same bias, variance-decay, cost, and moment assumptions
used in Corollaries~\ref{cor:app-qmlmc-greeks} and
\ref{cor:app-nested-qmlmc-greeks}.  In particular, one needs
\[
\left|
\mathbb E_W[P_{\ell}^{\mathrm{w}}-P^{\mathrm{w}}\mid B]
\right|
=
O(2^{-\ell}),
\quad
\operatorname{Var}_W(P_{\ell}^{\mathrm{w}}-P_{\ell-1}^{\mathrm{w}}\mid B)
=
O(2^{-2\ell}),
\quad
C_\ell=O(2^\ell),
\]
together with sufficient moment or tail bounds for the likelihood-ratio or Malliavin weight.  If the weight is unbounded, clipping or truncation may be used only after controlling the induced bias at the target accuracy. Barrier and other path-dependent discontinuous payoffs require path-dependent Malliavin weights or conditional-smoothing arguments and are not covered by the terminal-value formula above~\cite{Fournie1999,Fournie2001,GobetKohatsuHiga2003}.

\subsection{Heston-type Stochastic-volatility Models}
\label{sec:app-heston}

As a canonical two-factor example, consider the Heston stochastic-volatility model~\cite{Heston1993}.  To match the
independent Brownian-input convention used in the algorithms, write
\begin{equation}\label{eq:app-heston}
\begin{aligned}
dS_s
&=
rS_s\,ds+\sqrt{V_s}S_s\,dW_s^{(1)},\\
dV_s
&=
\kappa(\theta-V_s)\,ds
+
\xi\sqrt{V_s}
\left(
\rho\,dW_s^{(1)}
+
\sqrt{1-\rho^2}\,dW_s^{(2)}
\right),
\end{aligned}
\end{equation}
where $W^{(1)}$ and $W^{(2)}$ are independent Brownian motions.  Equivalently,
one may use correlated Brownian motions with
$d\langle W^{(1)},W^{(2)}\rangle_s=\rho\,ds$.

With state vector $x=(S,V)$, the conditional sensitivities above become
\[
\Delta=\partial_Su,\quad
\mathcal V_{\mathrm{spot}}=\partial_Vu,\quad
\Gamma=\partial_{SS}u,\quad
\mathrm{Vanna}=\partial_{SV}u,\quad
\mathrm{Volga}=\partial_{VV}u.
\]
Rho and the model-parameter Greeks
\[
\partial_\kappa u,\quad
\partial_\theta u,\quad
\partial_\xi u,\quad
\partial_\rho u
\]
are computed using the parameter-tangent formulas
\eqref{eq:app-parameter-tangent-sde} or, at the discrete level,
\eqref{eq:app-discrete-parameter-tangent}.

\begin{remark} The map \(v\mapsto\sqrt v\) is not \(C^1\) at \(v=0\), and the
full-truncation map \(v\mapsto v^+\) used in many numerical schemes is also
nonsmooth~\cite{LordKoekkoekVanDijk2010}.  Therefore, the global pathwise assumptions of
Propositions~\ref{prop:app-first-order-greeks} and
\ref{prop:app-second-order-greeks} do not hold for the unregularized Heston
diffusion without additional localization or moment arguments.  A rigorous
application of the smooth-Greek corollaries should therefore use one of the
following routes:
\begin{itemize}
\item replace $\sqrt v$ by a smooth positive regularization, such as
$\sqrt{v+\delta}$ or another smooth approximation, and analyze the
regularization bias;
\item localize the process to the region $V_s\geq \epsilon$ and control the
exit error;
\item use likelihood-ratio or Malliavin-weight estimators for nonsmooth
payoffs or boundary-sensitive regimes;
\item use a discretization-specific tangent recursion, with a smooth
truncation if differentiability of the numerical map is required.
\end{itemize}
\end{remark}
Once the chosen model and discretization satisfy the bias, variance, stability,
and moment assumptions stated in Corollaries~\ref{cor:app-qmlmc-greeks} and
\ref{cor:app-nested-qmlmc-greeks}, the same conditional and nested QA-MLMC
complexity bounds apply to Delta, spot Vega, Gamma, Vanna, Volga, full Rho, and
the Heston parameter Greeks.

\section{Strong-Error Order One Schemes for Pricing and Greek Estimators}
\label{sec:Strong-Order-One Numerical Schemes for Pricing and Greek Estimators}

\cite{bao2016first} developed a first-order scheme for BDSDEs
based on a two-sided It\^o--Taylor expansion.
However, since the scheme relies on conditional expectations with respect to
the forward Brownian motion at each time step, it integrates out the forward
randomness and therefore does not directly provide the pathwise strong
approximations required by our conditional-on-$B$ MLMC framework.
Inspired by the underlying It\^o--Taylor expansion, we develop new
discretization schemes and prove conditional strong-error order one convergence. Consequently, the resulting multilevel differences
exhibit second-order variance decay, which is essential for retaining the
quadratic quantum speedup in the QA-MLMC framework.

The BDSDE representations derived in the previous section
provide probabilistic formulas for option prices and Greeks.
To enable efficient multilevel Monte Carlo simulation,
we construct strong-error order one discretization schemes
for direct pricing estimators,
first-order Greek estimators,
and second-order Greek estimators.

The presentation of this section follows a common pattern for direct
pricing estimators, first-order Greek estimators, and second-order Greek
estimators.

For each estimator, we first derive a sufficient condition under which a
generic discretization operator yields global strong convergence of order
one. We then introduce a concrete Forward--Backward Taylor discretization operator and prove that it satisfies the required conditions. As a result, all
three estimators admit globally first-order accurate approximations.

\subsection{Forward--Backward Taylor Discretization Framework}

A central contribution of our numerical scheme is a family of
Forward--Backward Taylor discretizations for the path functionals
appearing in the BDSDE representations of prices and Greeks.

The key challenge is that the backward stochastic integrals
\begin{align*}
  \int Q(r,X_r)\,\mathrm d \overleftarrow{B}_r,
\end{align*}
where $Q$ denotes a generic integrand and may represent different functions in the pricing and Greek representations.
Since the integrand depends on the forward diffusion $X$, which is itself driven by an independent Brownian motion $W$, pathwise approximations compatible with the conditional-on-$B$ MLMC framework must retain the dependence on both sources of randomness throughout the discretization.

Our approach is based on Taylor expansions of the stochastic integrands.
The resulting discretizations naturally involve the mixed
forward--backward iterated integrals
\begin{align}
J^{WB}_{s,h}
&=
\Big(J_{ij}^{WB}(s,h)\Big)_{1\le i\le \ell,\ 1\le j\le d},
\quad
J_{ij}^{WB}(s,h)
:=
\int_s^{s+h}(W_r^j-W_s^j)\,
\mathrm d \overleftarrow{B}_r^i,
\label{eq:def-JWB}
\\
J^{BB}_{s,h}
&=
\Big(J_{ij}^{BB}(s,h)\Big)_{1\le i,j\le \ell},
\quad
J_{ij}^{BB}(s,h)
:=
\int_s^{s+h}
\bigl(B_{s}^j-B_r^j\bigr)\,
\mathrm d \overleftarrow{B}_r^i,
\label{eq:def-JBB}
\end{align}
together with the usual Brownian increments.
These terms capture the interaction between the \emph{forward} Brownian motion
$W$ and the \emph{backwar}d Brownian motion $B$, and play a crucial role
in achieving strong-error order one.

Throughout this section,
we consider a uniform partition
\[
t_k=t+kh,\quad k=0,\ldots,N,
\]
with stepsize $h=(T-t)/N$.
We denote the forward and backward Brownian increments by
\[
\Delta W_{t_k}
=
W_{t_{k+1}}-W_{t_k},
\quad
\Delta\overleftarrow B_{t_k}
=
B_{t_k}-B_{t_{k+1}},
\]
respectively.

Unless stated otherwise,
all $L^q$ norms are taken with respect to the joint law of
$(W,B)$:
\[
\|\cdot\|_{L^q}
:=
\|\cdot\|_{L^q_{W,B}}.
\]
When conditioning on a fixed realization of the backward Brownian motion
$B$,
we write the conditional norm as
$\|\cdot\|_{L_W^q}$.

\subsection{Direct Pricing Estimator}
Recall that the discretization operators
$\mathcal S_X$,
$\mathcal S_\Phi$,
and
$\mathcal S_{\mathrm{int}}$
were introduced in
\eqref{eq:def of S_X},
\eqref{eq:def of S_phi},
and
\eqref{eq:def of S_int},
respectively.
The goal of this subsection is to establish the strong approximation
properties of these operators and to derive the accumulated stability
estimates required for the global strong-error order one analysis.

\begin{definition}[Strong-error order of $S_X$]
  \label{def:Strong error of S X}
The discretization operator $S_X$ is said to have strong-error order $p>0$
if there exists a constant $C_X>0$ such that 
\begin{align*}
\left\|
\sup_{0\leq k\leq N}
\left|X_{t_k}^{t,x}-X_k^{(h)}\right|
\right\|_{L^q}
\leq C_Xh^{p},
\end{align*}
for every $q\ge2$.
There $X_s^{t,x}$ is the unique strong solution to the SDE~\eqref{eq:SDE(with BDSDE)}
and $\{X_k^{(h)}\}_{k=0}^N$ is the discrete-time approximation generated by
\begin{align*}
  X_{k+1}^{(h)}=\mathcal{S}_X\left(X_k^{(h)},t_k,h;\Delta W_{t_k}\right),
\quad X_0^{(h)}=x.
\end{align*}
\end{definition}

\begin{definition}[Strong-error order of $S_\Phi$]
  \label{def:Strong error of S Phi}
The discretization operator $S_\Phi$ is said to have  strong-error order $p>0$
\begin{align*}
\left\|
\sup_{0\leq k\leq N}
\left|
\Phi(t,t_k)-\Phi^{(h)}_k
\right|
\right\|_{L^q}
\leq C_\Phi h^{p},
\end{align*}
for every $q\ge2$.
Here  $\{\Phi_k^{(h)}\}_{k=0}^N$ is the discrete-time approximation generated by
\begin{align*}
  \Phi_{k+1}^{(h)}=
  \mathcal{S}_\Phi\left(
  \Phi_k^{(h)},t_k,h;\Delta W_{t_k},\Delta\overleftarrow B_{t_k}
  \right),
\quad
\Phi_0^{(h)}=1.
\end{align*}
\end{definition}

\begin{definition}[Strong-error order of $S_{\mathrm{int}}$]
  \label{def:Strong error of S int}
The discretization operator $S_{\mathrm{int}}$ is said to have strong-error order $p>0$
if there exists a constant $C_{\mathrm{int}}>0$ such that
\begin{align*}
\left\|
\sup_{0\leq k\leq N}
\left|
Y_{t_k}-\widetilde{Y}_k^{(h)}
\right|
\right\|_{L^q}
\leq C_{\mathrm{int}} h^{p},
\end{align*}
for every $q\ge2$.
There  $\{\widetilde{Y}_k^{(h)}\}_{k=0}^N$ is the discrete-time payoff approximation generated by
\begin{align*}
  \widetilde{Y}_{k+1}^{(h)}
  = \widetilde{Y}_k^{(h)}+
  \mathcal{S}_{\mathrm{int}}\left(
 \Phi(t,t_k),
  X_{t_{k}}^{t,x},t_{k},h;\Delta W_{t_k},\Delta\overleftarrow B_{t_k}
  \right),
\quad
\widetilde{Y}_0^{(h)}=0,
\end{align*}
with $X,\Phi$  evaluated exactly,
and $Y_{t_k}$ denotes the exact accumulated payoff process at time $t_k$, i.e.,
\begin{align*}
  Y_{t_k}=\int_t^{t_k}\Phi(t,r)\left(F(r,X_r^{t,x})+d(r) H(r,X_r^{t,x})\right)\,\mathrm{d}r 
    +\int_t^{t_k}\Phi(t,r)H(r,X_r^{t,x})\mathrm d \overleftarrow{B}_r  .
\end{align*}

\end{definition}
\begin{remark}
  For $0<q<2$, the corresponding estimates follow from Jensen's inequality
whenever the $q=2$ estimate is available.
\end{remark}

After defining the strong-error order associated with each individual discretization step,
we now turn to the study of the overall strong-error of the full discretized scheme.
Before carrying out the detailed analysis, we first introduce the following
accumulated stability property.

\begin{definition}[Accumulated stability of $\mathcal S_{\mathrm{int}}$]
\label{def:accumulated stability of Sint}

The discretization operator
$\mathcal S_{\mathrm{int}}$
is said to satisfy the accumulated stability estimate
if there exists a constant
$L_{\mathrm{int}}>0$,
independent of $h$, such that
\begin{align*}
\left\|
\sup_{0\leq k\leq N}
\left|
\widetilde Y_k^{(h)}-Y_k^{(h)}
\right|
\right\|_{L^2}
\leq L_{\mathrm{int}}
\bigg(
&
\left\|
\sup_{0\leq j\leq N}
\left|
\Phi(t,t_j)-\Phi_j^{(h)}
\right|
\right\|_{L^4}+
\left\|
\sup_{0\leq j\leq N}
\left|
X_{t_j}^{t,x}-X_j^{(h)}
\right|
\right\|_{L^4}
\bigg),
\end{align*}
where $\widetilde{Y}_k^{(h)}$ is defined in
Definition~\ref{def:Strong error of S int}.
\end{definition}

\begin{proposition}\label{prop:error-u}
Fix $(t,x)\in[0,T]\times\mathbb{R}^d$ and a uniform grid $\{t_k\}_{k=0}^N$ with $h=(T-t)/N$.
Let
\begin{align*}
 P_{t_k}=\Phi(t,t_k)G\left(X_{t_k}^{t,x}\right)+Y_{t_k},
\end{align*}
where $Y_{t_k}$ is the exact accumulated payoff
\begin{align*}
  Y_{t_k}=\int_t^{t_k}\Phi(t,r)\left(F(r,X_r^{t,x})+d(r)H(r,X_r^{t,x})\right)\,\mathrm{d} r
+\int_t^{t_k}\Phi(t,r)H(r,X_r^{t,x})\,\mathrm d \overleftarrow{B}_r.
\end{align*}

Let $\left\{\left(X_k^{(h)},\Phi_k^{(h)}\right)\right\}_{k=0}^N$ be the approximations generated by the schemes
$S_X,S_\Phi$ as in Definition~\ref{def:Strong error of S X},~\ref{def:Strong error of S Phi}
 and define
 \begin{align*}
   Y_{k+1}^{(h)}
=&
Y_k^{(h)}
+
\mathcal{S}_{\mathrm{int}}
\left(
\Phi_k^{(h)},
X_k^{(h)},t_k,h;
\Delta W_{t_k},
\Delta\overleftarrow B_{t_k}
\right),
\quad
Y_0^{(h)}=0,\\
  P_k^{(h)}=&\Phi^{(h)}_k G\left(X_k^{(h)}\right)+Y_k^{(h)}.
\end{align*}

Assume:
\begin{enumerate}
    \item The strong-error orders of $\mathcal{S}_X$, $\mathcal{S}_\Phi$ and $\mathcal{S}_{\mathrm{int}}$
are $p_X,p_\Phi,p_{\mathrm{int}}$ respectively.
Moreover, $\mathcal{S}_{\mathrm{int}}$
satisfies the accumulated stability estimate
defined in Definition~\ref{def:accumulated stability of Sint}.

\item 
There exists $M>0$, independent of $h$, such that 
$\Phi_k^{(h)}$
 and 
$G(X_{t_k}^{t,x})$
are uniformly bounded in
$L^{4}_{W,B}$
by $M$.

\item 
$G$ is globally Lipschitz, i.e., there exists $L_G>0$ such that $|G(x)-G(y)|\leq L_G|x-y|$.
\end{enumerate}

Then the payoff approximation satisfies the joint strong-error bound
\begin{align*}
\left\|
\sup_{0\leq k\leq N}
\left|
P_{t_k}-P^{(h)}_k
\right|
\right\|_{L^2}
=
\mathcal{O}(h^p),
\quad
p=\min\{p_X,p_\Phi,p_{\mathrm{int}}\}.
\end{align*}
\end{proposition}
A detailed proof is provided in Appendix~\ref{app:strong-error-framework}, Paragraph~\ref{para:proof of prop:error-u}.
Consequently, on a dyadic grid with $h_\ell=2^{-\ell}$ and $N_\ell=(T-t)2^\ell$,
Lemma~\ref{lem:EB to EW} implies the corresponding fixed-$B$ conditional
$L_W^2$ strong-error estimate, up to the logarithmic factor
$(1+\ell)^{a/2}$ for any $a>1$.

For the forward diffusion process $X$, we employ the Milstein
discretization operator
\begin{align}
\mathcal S_X^{\mathrm{mil}}
(x,s,h;\Delta W_s)
=
x
+b(x)h
+\sigma(x)\Delta W_s
+
\sum_{i=1}^{d}\sum_{j=1}^{d}
L_i\sigma_j(x)
\int_s^{s+h}
(W_r^i-W_s^i)\,d\mathrm W_r^j,
\label{eq:SX-Mil}
\end{align}
where
$L_i
=
\sum_{n=1}^{d}
\sigma_{ni}(x)\partial_{x_n}.$

For the scalar weight process $\Phi$, we use the exponential-type
discretization
\begin{align}
\mathcal S_\Phi
(\Phi,s,h;
\Delta W_s,
\Delta\overleftarrow B_s)
=
\Phi
\exp\!\Big(
C_{s,h}
+\widetilde c(s+\tfrac h2)\cdot\Delta W_s
+d(s+\tfrac h2)\cdot\Delta\overleftarrow B_s
-\tfrac12 Q_{s,h}
\Big),
\label{eq:Sphi}
\end{align}
where
\[
C_{s,h}
=
\frac h6
\bigl(
c(s)+4c(s+\tfrac h2)+c(s+h)
\bigr),
\]
and
\[
Q_{s,h}
=
\frac h6
\bigl(
q(s)+4q(s+\tfrac h2)+q(s+h)
\bigr),
\quad
q(r)=\|\widetilde c(r)\|^2-\|d(r)\|^2.
\]

These discretizations are standard and achieve strong-error order one for the
corresponding stochastic differential equations.
The main additional difficulty in the present setting arises from the
accumulated payoff term
\[
Y_{t_k}
=
\int_t^{t_k}
\Phi(t,r)
\bigl(
F(r,X_r^{t,x})
+
d(r)H(r,X_r^{t,x})
\bigr)\,\mathrm dr
+
\int_t^{t_k}
\Phi(t,r)H(r,X_r^{t,x})
\mathrm d \overleftarrow{B}_r,
\]
whose discretization involves a backward stochastic integral.
A naive Euler-type approximation yields only global strong-error order
$1/2$ and is therefore insufficient for the multilevel framework
developed in this work.
To overcome this difficulty,
we introduce a Forward--Backward Taylor discretization operator
$\mathcal S_{\mathrm{int}}^{\mathrm{FBT}}$,
which incorporates suitable higher-order corrections for the backward
stochastic integral.

Consider the integral discretization operator
$\mathcal{S}_{\mathrm{int}}^{\mathrm{FBT}}$ defined by
\begin{align}
\mathcal{S}_{\mathrm{int}}^{\mathrm{FBT}}
(\Phi,X,s,h;\Delta W_s,\Delta \overleftarrow B_s)
=&\,
\Phi h\bigl(F(s,X)+d(s) H(s,X)\bigr)
+\Phi H(s,X)\cdot \Delta\overleftarrow B_s
\notag\\
&+\Phi
\sum_{i=1}^{\ell}\sum_{j=1}^{d}
\bigl(H_x(s,X)\sigma(X)+\widetilde c(s)H(s,X)\bigr)_{ij}
J_{ij}^{WB}(s,h)
+\Phi
\sum_{i=1}^{\ell}\sum_{j=1}^{\ell}
d_j(s)H_i(s,X)
J_{ij}^{BB}(s,h),
\label{eq:Sint-FBT}
\end{align}
where
$J^{WB}_{s,h}$ and $J^{BB}_{s,h}$
are defined in
(\ref{eq:def-JWB})--(\ref{eq:def-JBB}).

Note that the first two terms in \eqref{eq:Sint-FBT} 
coincide with the direct discretization approximation, while two additional terms
$J_{ij}^{WB}(s,h)$ and $J_{ij}^{BB}(s,h)$ appears.
These extra terms serve as higher-order corrections to the backward
stochastic integral
\[
\int_s^{s+h}\Phi(t,r)H(r,X_r)\cdot \mathrm d \overleftarrow{B}_r,
\]
and are obtained from the taylor expansion of the
integrand.

The following two properties explain why the integral discretization operator
is defined in this way.
Detailed proofs are provided in
Appendix~\ref{app:strong-error-framework},
Paragraphs~\ref{para:proof of prop:int-strong-error}
and~\ref{para:proof of prop:int stable}.

\begin{proposition}[Strong-error order of the integral discretization]
\label{prop:int-strong-error}
Assume that
\(F,H,d,\widetilde c\) and the coefficients of \(X\) and \(\Phi\) are sufficiently
smooth with polynomial growth, and that the corresponding moments of
\(X\) and \(\Phi\) are uniformly bounded. 
Then the approximation generated by
\begin{align}
\widetilde Y_{k+1}^{(h)}
=
\widetilde Y_k^{(h)}
+
\mathcal S_{\mathrm{int}}^{\mathrm{FBT}}
\left(\Phi(t,t_k),X_{t_k}^{t,x},
t_k,h;\Delta W_{t_k},\Delta\overleftarrow B_{t_k}\right),
\quad
\widetilde Y_0^{(h)}=0,
\end{align}
satisfies
\begin{align}
\left\|
\sup_{0\le k\le N}
\left|
Y_{t_k}-\widetilde Y_k^{(h)}
\right|
\right\|_{L^q}
\le Ch,
\end{align}
for every $q\ge2$. Consequently, by Jensen's inequality, the same estimate
also holds for every $0<q<2$. Hence
$\mathcal S_{\mathrm{int}}^{\mathrm{FBT}}$ has strong-error order $1$ in the sense
of Definition~\ref{def:Strong error of S int}.
\end{proposition}

\begin{proposition}
\label{prop:int stable}
Assume that 
$F,H$ are globally Lipschitz and have at most linear growth. 
Moreover, assume
that the coefficient
$H_x(t,x)\sigma(x)+\widetilde c(t)H(t,x)$
is globally Lipschitz and has at most linear growth.
Assume further that the exact and numerical input processes satisfy the uniform
moment bound
\[
\left\|
\sup_{0\le j\le N}
|\Phi(t,t_j)|
\right\|_{L^8}
+
\left\|
\sup_{0\le j\le N}
|\Phi_j^{(h)}|
\right\|_{L^8}
+
\left\|
\sup_{0\le j\le N}
|X_{t_j}^{t,x}|
\right\|_{L^8}
+
\left\|
\sup_{0\le j\le N}
|X_j^{(h)}|
\right\|_{L^8}
<\infty .
\]
Then
$\mathcal{S}_{\mathrm{int}}^{\mathrm{FBT}}$
satisfies the accumulated stability condition in
Definition~\ref{def:accumulated stability of Sint}. More precisely,
there exists $L_{\mathrm{int}}>0$, independent of $h$, such that
\begin{align*}
\left\|
\sup_{0\leq k\leq N}
\left|
\widetilde Y_k^{(h)}-Y_k^{(h)}
\right|
\right\|_{L^2}
\le
L_{\mathrm{int}}
\bigg(
&
\left\|
\sup_{0\leq j\leq N}
\left|
\Phi(t,t_j)-\Phi_j^{(h)}
\right|
\right\|_{L^4}
+
\left\|
\sup_{0\leq j\leq N}
\left|
X_{t_j}^{t,x}-X_j^{(h)}
\right|
\right\|_{L^4}
\bigg).
\end{align*}
\end{proposition}

\subsection{First-Order Greek Estimators}

The first-order Greek representation involves, in addition to the forward
diffusion $X$ and the weight process $\Phi$, the Jacobian flow $J^{t,x}$
associated with the forward SDE and a corresponding accumulated payoff
term.

The purpose of this subsection is twofold.
We first develop a general strong-error framework showing that the
convergence rate of the first-order Greek estimator is determined by the
approximation properties of the discretization operators for
$J^{t,x}$ and the accumulated payoff term.
We then construct suitable discretization operators satisfying the
requirements of the framework, leading to a global strong-error order one
approximation for the first-order Greek estimator.

\begin{definition}[Strong-error order of $S_J$]
\label{def:Strong error of S J}
The discretization operator $S_J$ is said to have strong-error order $p>0$
if there exists a constant $C_J>0$ such that
\[
\left\|
\sup_{0\le k\le N}
\left|
J_{t_k}^{t,x}-J_k^{(h)}
\right|
\right\|_{L^q}
\le C_J h^p
\]
for every $q\ge2$.
Here $J^{t,x}$ is the Jacobian flow defined in
\eqref{eq:app-jacobian}--\eqref{eq:app-jacobian-sde}, and
$\{J_k^{(h)}\}_{k=0}^N$ is generated by
$$
J_{k+1}^{(h)}=\mathcal S_J
\left(J_k^{(h)},X_k^{(h)},t_k,h;\Delta W_{t_k},\Delta\overleftarrow B_{t_k}\right),
\quad J_0^{(h)}=I_d .
$$
\end{definition}

\begin{definition}[Strong-error order of $S_{\mathrm{int}}^{(i)}$]
\label{def:Strong error of S int Greek}
Fix $1\le i\le d$. The discretization operator
$S_{\mathrm{int}}^{(i)}$ is said to have strong-error order $p>0$
if there exists a constant $C_{\mathrm{int}}^{(i)}>0$ such that
\[
\left\|
\sup_{0\le k\le N}
\left|
Y_{t_k}^{(i)}-\widetilde Y_k^{(i,h)}
\right|
\right\|_{L^q}
\le C_{\mathrm{int}}^{(i)} h^p
\]
for every $q\ge2$.
Here
$$
Y_{t_k}^{(i)}=\int_t^{t_k}\Phi(t,s)
\nabla_x\!\Big(F(s,X_s^{t,x})+d(s)H(s,X_s^{t,x})\Big)^\top
J_s^{t,x}e_i\,\mathrm ds+
\int_t^{t_k}\Phi(t,s)\nabla_xH(s,X_s^{t,x})^\top
J_s^{t,x}e_i\,\mathrm d \overleftarrow{B}_s.
$$
The approximation $\{\widetilde Y_k^{(i,h)}\}_{k=0}^N$ is generated by
$$
\widetilde Y_{k+1}^{(i,h)}
=\widetilde Y_k^{(i,h)}
+\mathcal S_{\mathrm{int}}^{(i)}
\left(\Phi(t,t_k),X_{t_k}^{t,x},J_{t_k}^{t,x},t_k,h;
\Delta W_{t_k},\Delta\overleftarrow B_{t_k}\right),
\quad
\widetilde Y_0^{(i,h)}=0,
$$
with $X,\Phi,J$ evaluated exactly.
\end{definition}

The following proposition establishes a general strong-error estimate for
the first-order Greek payoff estimator.
It shows that the convergence rate of the payoff approximation is determined
by the approximation properties of the underlying discretization operators
together with a suitable stability property of
$\mathcal S_{\mathrm{int}}^{(i)}$.
A detailed proof is provided in
Appendix~\ref{app:strong-error-framework},
Paragraph~\ref{para:proof of prop:error-first-order-greek-payoff}.

\begin{proposition}[Strong-error order for the first-order Greek payoff]
\label{prop:error-first-order-greek-payoff}
Fix $(t,x)\in[0,T]\times\mathbb R^d$ and $1\le i\le d$
and define
\begin{align*}
P_{t_k}^{(i)}
=&\Phi(t,t_k)\,\nabla G(X_{t_k}^{t,x})^\top J_{t_k}^i
+Y_{t_k}^{(i)},
\end{align*}
where
\begin{align*}
Y_{t_k}^{(i)}
=&
\int_t^{t_k}
\Phi(t,s)\,
\nabla_x\!\Big(F(s,X_s^{t,x})+d(s)H(s,X_s^{t,x})\Big)^\top
J_s^{t,x}e_i\,\mathrm ds  
+\int_t^{t_k}\Phi(t,s)\,\nabla_x H(s,X_s^{t,x})^\top J_s^{t,x}e_i\,
\mathrm d \overleftarrow{B}_s.
\end{align*}

Let $\{X_k^{(h)},\Phi_k^{(h)},J_k^{(h)}\}_{k=0}^N$
be generated by
$S_X,S_\Phi,S_J$
as in
Definition~\ref{def:Strong error of S X},
\ref{def:Strong error of S Phi}
and \ref{def:Strong error of S J},
and define
\begin{align*}
Y_{k+1}^{(i,h)}
=&Y_k^{(i,h)}
+\mathcal S_{\mathrm{int}}^{(i)}
\left(\Phi_k^{(h)},X_k^{(h)},J_k^{(h)},t_k,h;
\Delta W_{t_k},\Delta\overleftarrow B_{t_k}\right),
\quad
Y_0^{(i,h)}=0,\\
P_k^{(i,h)}
=&
\Phi_k^{(h)}
\nabla G(X_k^{(h)})^\top J_k^{(h)}e_i
+
Y_k^{(i,h)}.
\end{align*}

Assume:

\begin{enumerate}
\item
The strong-error orders of $S_X,S_\Phi,S_J$ and
$\mathcal S_{\mathrm{int}}^{(i)}$ are $p_X,p_\Phi,p_J,p_{\mathrm{int}}^{(i)}$, respectively.

\item
The operator $\mathcal S_{\mathrm{int}}^{(i)}$ satisfies the
accumulated stability estimate:
there exists $L_{\mathrm{int}}^{(i)}>0$, independent of $h$, such that
\[
\left\|
\sup_{0\le k\le N}
|\widetilde Y_k^{(i,h)}-Y_k^{(i,h)}|
\right\|_{L^2}
\le
L_{\mathrm{int}}^{(i)}
\bigg(
\left\|
\sup_{0\le j\le N}
|\Phi(t,t_j)-\Phi_j^{(h)}|
\right\|_{L^4}
+
\left\|
\sup_{0\le j\le N}
|X_{t_j}^{t,x}-X_j^{(h)}|
\right\|_{L^4}
+
\left\|
\sup_{0\le j\le N}
|J_{t_j}^{t,x}-J_j^{(h)}|
\right\|_{L^4}
\bigg),
\]
where $\widetilde Y_k^{(i,h)}$ is generated by
$\mathcal S_{\mathrm{int}}^{(i)}$ with the exact inputs
$\Phi(t,t_k),X_{t_k}^{t,x},J_{t_k}^{t,x}$
as in Definition~\ref{def:Strong error of S int Greek}.

\item
There exists $M>0$, independent of $h$, such that 
$\Phi_k^{(h)},\
\Phi(t,t_k),\ 
J_k^{(h)},\ 
J_{t_k}^{t,x}$
 and 
$\nabla G(X_{t_k}^{t,x}),\ \nabla G(X_k^{(h)})$
are uniformly bounded in
$L^{8}_{W,B}$
by $M$.

\item
$\nabla G$ is globally Lipschitz, i.e., there exists $L_{\nabla G}>0$ such that
\begin{align*}
\left|\nabla G(x)-\nabla G(y)\right|
\le L_{\nabla G}|x-y|,\quad x,y\in\mathbb R^d .
\end{align*}
\end{enumerate}

Then
\[
\left\|
\sup_{0\le k\le N}
\left|
P_{t_k}^{(i)}-P_k^{(i,h)}
\right|
\right\|_{L^2}
=
\mathcal{O}(h^p),
\quad
p=\min\{p_X,p_\Phi,p_J,p_{\mathrm{int}}^{(i)}\}.
\]
\end{proposition}
Applying Lemma~\ref{lem:EB to EW} with
\[
\Pi_{\ell,k}
=
P_{t_k^\ell}^{(i)}
-
P_{\ell,k}^{(i,h_\ell)}
\]
yields the fixed-$B$ conditional first-order Greek estimate with logarithmic
loss.

We now construct discretization operators satisfying the assumptions of
Proposition~\ref{prop:error-first-order-greek-payoff}.

For the Jacobian flow $J^{t,x}$, we use the standard Milstein
discretization. Since
\[
dJ_s^{t,x}
=
\nabla b(X_s^{t,x})J_s^{t,x}\,ds
+
\sum_{a=1}^{d}
\nabla \sigma_{\cdot a}(X_s^{t,x})J_s^{t,x}\,dW_s^a ,
\quad
J_t^{t,x}=I_d,
\]
we define
\begin{align}
\mathcal S_J^{\mathrm{mil}}
(J,X,s,h;\Delta W_s,\Delta\overleftarrow B_s)
=J
+\nabla b(X)J\,h
+\sum_{a=1}^{d}
\nabla\sigma_{\cdot a}(X)J\,\Delta W_s^a
+\sum_{a=1}^{d}\sum_{b=1}^{d}
L_a\!\big(\nabla\sigma_{\cdot b}J\big)(X,J)
\int_s^{s+h}
(W_r^a-W_s^a)\,\mathrm dW_r^b .
\end{align}
Under the regularity assumptions imposed on $b$ and $\sigma$,
this discretization has strong-error order one.
The remaining task is the construction of a suitable integral discretization
operator
$\mathcal S_{\mathrm{int}}^{(i),\mathrm{FBT}}$
for the accumulated payoff term in the first-order Greek representation.

As in the direct pricing case, the construction is based on a
first-order stochastic Taylor expansion of the corresponding integrands.
The resulting forward--backward Taylor correction terms involve both mixed forward--backward iterated integrals and purely backward iterated integrals.

Fix $1\le i\le d$. Recall that the first-order Greek integral is given by
\begin{align}
Y_{t_k}^{(i)}
=&\int_t^{t_k}\Phi(t,s)
\nabla_x\!\Big(F(s,X_s^{t,x})+d(s)H(s,X_s^{t,x})\Big)^\top
J_s^{t,x}e_i\,\mathrm ds
+\int_t^{t_k}\Phi(t,s)\nabla_xH(s,X_s^{t,x})^\top
J_s^{t,x}e_i\,\mathrm d \overleftarrow{B}_s.
\label{eq:first-order-greek-integral}
\end{align}
Here $J_s^{t,x}=\nabla_x X_s^{t,x}$ is the Jacobian flow and
$e_i$ is the $i$-th unit vector in $\mathbb R^d$.

Consider the first-order Greek integral discretization operator 
\begin{align}
&\mathcal S_{\mathrm{int}}^{(i),\mathrm{FBT}}
\left(\Phi,X,J,s,h;\Delta W_s,\Delta\overleftarrow B_s\right)\notag\\
=& \Phi h\,\nabla_x\!\Big(F(s,X)+d(s)H(s,X)\Big)^\top Je_i
+\Phi\left(\nabla_xH(s,X)^\top Je_i\right)\cdot\Delta\overleftarrow B_s
+\Phi\sum_{j=1}^{\ell}\sum_{\alpha=1}^{\ell}d_\alpha(s)
\left(\nabla_xH_j(s,X)^\top Je_i\right)
J_{j\alpha}^{BB}(s,h)\notag\\
&+\Phi\sum_{j=1}^{\ell}\sum_{a=1}^{d}
\bigg[\nabla_x^2H_j(s,X)[\sigma_{\cdot a}(X),Je_i]
+\nabla_xH_j(s,X)^\top\big(\nabla_x\sigma_{\cdot a}(X)Je_i\big)
+\widetilde c_a(s)\nabla_xH_j(s,X)^\top Je_i\bigg]J_{ja}^{WB}(s,h).
\label{eq:Sint-first-greek-mil}
\end{align}

The form of the operator is motivated by a first-order expansion of the
integrands with respect to the forward diffusion, the Jacobian flow, and
the weight process.
A detailed derivation is provided in
Appendix~\ref{app:strong-error-framework},
Paragraph~\ref{para:construction-first-greek}.

The following propositions show that the above Forward--Backward Taylor discretization achieves strong-error order one and satisfies the accumulated
stability estimate required in
Proposition~\ref{prop:error-first-order-greek-payoff}.
Detailed proofs are provided in
Appendix~\ref{app:strong-error-framework},
Paragraphs~\ref{para:proof of prop:first-order-greek-strong-error}
and~\ref{para:proof of prop:first-order-greek accumulated stability}.

\begin{proposition}[Strong-error order of the first-order Greek integral discretization]
\label{prop:first-order-greek-strong-error}
Assume that the coefficients $b,\sigma,F,H,d,\widetilde c$ are sufficiently
smooth with bounded derivatives up to the order used above, and assume that
$X$, $J$, and $\Phi$ have uniformly bounded moments of all required orders.
Then the approximation generated by
\begin{align}
\widetilde Y_{k+1}^{(i,h)}
=
\widetilde Y_k^{(i,h)}
+
\mathcal S_{\mathrm{int}}^{(i),\mathrm{FBT}}
\left(\Phi(t,t_k),X_{t_k}^{t,x},J_{t_k}^{t,x},
t_k,h;\Delta W_{t_k},\Delta\overleftarrow B_{t_k}\right),
\quad
\widetilde Y_0^{(i,h)}=0,
\end{align}
satisfies
\begin{align}
\left\|
\sup_{0\le k\le N}
\left|
Y_{t_k}^{(i)}
-
\widetilde Y_k^{(i,h)}
\right|
\right\|_{L^q}
\le C_qh,
\end{align}
for every $q\ge2$. Consequently, by Jensen's inequality, the same estimate
also holds for every $0<q<2$. Hence
$\mathcal S_{\mathrm{int}}^{(i),\mathrm{FBT}}$ has strong-error order $1$ in the sense
of Definition~\ref{def:Strong error of S int Greek}.
\end{proposition}

\begin{proposition}
\label{prop:first-order-greek accumulated stability}
Assume that
$F\in C_b^2$, $H\in C_b^3$, $\sigma\in C_b^2$.
Assume also that $d,\widetilde c$ are bounded. Moreover, assume that the exact
and numerical input processes satisfy the uniform moment bound
\small
\begin{align*}
&
\left\|\sup_{0\le j\le N}|\Phi(t,t_j)|\right\|_{L^{12}}
+
\left\|\sup_{0\le j\le N}|\Phi_j^{(h)}|\right\|_{L^{12}}
+
\left\|\sup_{0\le j\le N}|X_{t_j}^{t,x}|\right\|_{L^{12}}
+
\left\|\sup_{0\le j\le N}|X_j^{(h)}|\right\|_{L^{12}}
+
\left\|\sup_{0\le j\le N}|J_{t_j}^{t,x}|\right\|_{L^{12}}
+
\left\|\sup_{0\le j\le N}|J_j^{(h)}|\right\|_{L^{12}}
<\infty .
\end{align*}
\normalsize
Then $\mathcal S_{\mathrm{int}}^{(i),\mathrm{FBT}}$ satisfies the accumulated
stability estimate. More precisely, there exists
$L_{\mathrm{int}}^{(i)}>0$, independent of $h$, such that
\begin{align*}
\left\|
\sup_{0\le k\le N}
\left|
\widetilde Y_k^{(i,h)}-Y_k^{(i,h)}
\right|
\right\|_{L^2}
\le L_{\mathrm{int}}^{(i)}
\bigg(
&
\left\|
\sup_{0\le j\le N}
\left|
\Phi(t,t_j)-\Phi_j^{(h)}
\right|
\right\|_{L^4}+
\left\|
\sup_{0\le j\le N}
\left|
X_{t_j}^{t,x}-X_j^{(h)}
\right|
\right\|_{L^4}+
\left\|
\sup_{0\le j\le N}
\left|
J_{t_j}^{t,x}-J_j^{(h)}
\right|
\right\|_{L^4}
\bigg).
\end{align*}
\end{proposition}

\subsection{Second-Order Greek Estimators}
The second-order Greek representation involves, in addition to the
forward diffusion $X^{t,x}$ and the Jacobian flow $J^{t,x}$,
the second-order variational process
$K^{(ij),t,x}$ associated with the forward SDE.
For fixed $1\le i,j\le d$,
the process $K^{(ij),t,x}$ satisfies
\begin{align*}
dK_s^{(ij),t,x}
=\Big[\nabla_x b(X_s^{t,x})K_s^{(ij),t,x}
+\nabla_x^2 b(X_s^{t,x})
[J_s^i,J_s^j]\Big]\,ds
+\sum_{a=1}^{d}
\Big[\nabla_x\sigma_{\cdot a}(X_s^{t,x})
K_s^{(ij),t,x}
+\nabla_x^2\sigma_{\cdot a}(X_s^{t,x})
[J_s^i,J_s^j]
\Big]\,dW_s^a,
\end{align*}
where $J_s^i=J_s^{t,x}e_i.$

The purpose of this subsection is twofold.
We first develop a general strong-error framework showing that the
convergence rate of the second-order Greek estimator is determined by the
approximation properties of the discretization operators associated with
$K^{(ij),t,x}$ and the corresponding accumulated payoff term.
We then construct suitable discretization operators satisfying the
requirements of the framework, leading to a global strong-error order one
approximation for the second-order Greek estimator.

\begin{definition}[Strong-error order of $S_K$]
\label{def:Strong error of S K}
The discretization operator $S_K$ is said to have strong-error order $p>0$
if there exists a constant $C_K>0$ such that
\[
\left\|
\sup_{0\le k\le N}
\left|
K_{t_k}^{(ij),t,x}-K_k^{(ij,h)}
\right|
\right\|_{L^q}
\le C_K h^p
\]
for every $q\ge2$.
Here $K^{(ij),t,x}$ is the second tangent process defined in
\eqref{eq:app-second-tangent}--\eqref{eq:app-second-tangent-sde}, and
$\{K_k^{(ij,h)}\}_{k=0}^N$ is generated by
\begin{align*}
K_{k+1}^{(ij,h)}
=\mathcal S_K
\left(K_k^{(ij,h)},J_k^{(h)},X_k^{(h)},t_k,h;\Delta W_{t_k},\Delta\overleftarrow B_{t_k}\right),
\quad
K_0^{(ij,h)}=0.
\end{align*}
\end{definition}

\begin{definition}[Strong-error order of $S_{\mathrm{int}}^{(ij)}$]
\label{def:Strong error of S int second Greek}
The discretization operator $S_{\mathrm{int}}^{(ij)}$ is said to have strong-error order
$p>0$ if there exists a constant $C_{\mathrm{int}}^{(ij)}>0$ such that
\[
\left\|
\sup_{0\le k\le N}
\left|
Y_{t_k}^{(ij)}-\widetilde Y_k^{(ij,h)}
\right|
\right\|_{L^q}
\le C_{\mathrm{int}}^{(ij)}h^p
\]
for every $q\ge2$.
Here
\begin{align*}
\begin{aligned}
Y_{t_k}^{(ij)}
=&\int_t^{t_k}\Phi(t,s)
\Big((J_s^{t,x}e_j)^\top\nabla_x^2(F+dH)(s,X_s^{t,x})J_s^{t,x}e_i
+\nabla_x(F+dH)(s,X_s^{t,x})^\top K_s^{(ij),t,x}\Big)\,\mathrm ds  \\
&+\int_t^{t_k}\Phi(t,s)
\Big((J_s^{t,x}e_j)^\top\nabla_x^2H(s,X_s^{t,x})J_s^{t,x}e_i
+\nabla_xH(s,X_s^{t,x})^\top K_s^{(ij),t,x}\Big)\mathrm d \overleftarrow{B}_s.
\end{aligned}
\end{align*}
The approximation $\{\widetilde Y_k^{(ij,h)}\}_{k=0}^N$ is generated by
\begin{align*}
\widetilde Y_{k+1}^{(ij,h)}
=
\widetilde Y_k^{(ij,h)}
+
\mathcal S_{\mathrm{int}}^{(ij)}
\left(
\Phi(t,t_k),
X_{t_k}^{t,x},
J_{t_k}^{t,x},
K_{t_k}^{(ij),t,x},
t_k,h;
\Delta W_{t_k},\Delta\overleftarrow B_{t_k}
\right),
\quad
\widetilde Y_0^{(ij,h)}=0,
\end{align*}
with $X,\Phi,J,K$ evaluated exactly.
\end{definition}
The following proposition establishes a general strong-error estimate for
the second-order Greek payoff estimator.
In particular, if the underlying discretization operators satisfy suitable
strong-error and stability properties, then the resulting payoff
approximation inherits the same convergence rate.
The proof is deferred to
Appendix~\ref{app:strong-error-framework},
Paragraph~\ref{para:proof of prop:error-second-order-greek-payoff}.

\begin{proposition}[Strong-error order for the second-order Greek payoff]
\label{prop:error-second-order-greek-payoff}
Fix $(t,x)\in[0,T]\times\mathbb R^d$ and $1\le i,j\le d$.
Define
\begin{align*}
P_{t_k}^{(ij)}
=\Phi(t,t_k)
\Big((J_{t_k}^j)^\top\nabla_x^2G(X_{t_k}^{t,x})J_{t_k}^i
+\nabla_xG(X_{t_k}^{t,x})^\top K_{t_k}^{(ij),t,x}\Big)
+Y_{t_k}^{(ij)}.
\end{align*}
Let $\{X_k^{(h)},\Phi_k^{(h)},J_k^{(h)},K_k^{(ij,h)}\}_{k=0}^N$
be generated by $S_X,S_\Phi,S_J,S_K$ as in
Definition~\ref{def:Strong error of S X},
Definition~\ref{def:Strong error of S Phi},
Definition~\ref{def:Strong error of S J},
and Definition~\ref{def:Strong error of S K}.
Define
\begin{align*}
\begin{aligned}
Y_{k+1}^{(ij,h)}
=&
Y_k^{(ij,h)}
+
\mathcal S_{\mathrm{int}}^{(ij)}
\left(
\Phi_k^{(h)},
X_k^{(h)},
J_k^{(h)},
K_k^{(ij,h)},
t_k,h;
\Delta W_{t_k},\Delta\overleftarrow B_{t_k}
\right),
\quad
Y_0^{(ij,h)}=0,\\
P_k^{(ij,h)}
=&
\Phi_k^{(h)}
\Big(
(J_k^{j,h})^\top\nabla_x^2G(X_k^{(h)})J_k^{i,(h)}
+
\nabla_xG(X_k^{(h)})^\top K_k^{(ij,h)}
\Big)
+
Y_k^{(ij,h)}.
\end{aligned}
\end{align*}
Assume:

\begin{enumerate}
\item
The strong-error orders of $S_X,S_\Phi,S_J,S_K$ and
$\mathcal S_{\mathrm{int}}^{(ij)}$ are
$p_X,p_\Phi,p_J,p_K,p_{\mathrm{int}}^{(ij)}$, respectively.

\item
$\mathcal S_{\mathrm{int}}^{(ij)}$ satisfies the accumulated stability estimate:
there exists $L_{\mathrm{int}}^{(ij)}>0$, independent of $h$, such that
\begin{align*}
\left\|
\sup_{0\le k\le N}
|\widetilde Y_k^{(ij,h)}-Y_k^{(ij,h)}|
\right\|_{L^2}
\le
L_{\mathrm{int}}^{(ij)}
\bigg(&
\left\|
\sup_{0\le r\le N}
|\Phi(t,t_r)-\Phi_r^{(h)}|
\right\|_{L^4}
+
\left\|
\sup_{0\le r\le N}
|X_{t_r}^{t,x}-X_r^{(h)}|
\right\|_{L^4}\\
&+
\left\|
\sup_{0\le r\le N}
|J_{t_r}^{t,x}-J_r^{(h)}|
\right\|_{L^4}
+
\left\|
\sup_{0\le r\le N}
|K_{t_r}^{(ij),t,x}-K_r^{(ij,h)}|
\right\|_{L^4}
\bigg).
\end{align*}

\item
There exists $M>0$, independent of $h$, such that 
\[
\Phi(t,t_k),\ \Phi_k^{(h)},\
J_{t_k}^{t,x},\ J_k^{(h)},\
K_{t_k}^{(ij),t,x},\ K_k^{(ij,h)}
\quad
\text{ and }\quad
\nabla G(X_{t_k}^{t,x}),\ \nabla G(X_k^{(h)}),\
\nabla_x^2G(X_{t_k}^{t,x}),\ \nabla_x^2G(X_k^{(h)})
\]
are uniformly bounded in
$L^{12}_{W,B}$
by $M$.

\item
$\nabla G$ and $\nabla_x^2G$ are globally Lipschitz with constants
$L_{\nabla G}$ and $L_{\nabla^2 G}$.
\end{enumerate}

Then
\begin{align*}
\left\|
\sup_{0\le k\le N}
\left|
P_{t_k}^{(ij)}-P_k^{(ij,h)}
\right|
\right\|_{L^2}
=
\mathcal{O}(h^p),
\quad
p=\min\{p_X,p_\Phi,p_J,p_K,p_{\mathrm{int}}^{(ij)}\}.
\end{align*}
\end{proposition}
Applying Lemma~\ref{lem:EB to EW} with
\[
\Pi_{\ell,k}
=
P_{t_k^\ell}^{(ij)}
-
P_{\ell,k}^{(ij,h_\ell)}
\]
yields the fixed-$B$ conditional second-order Greek estimate with logarithmic
loss.

We now construct discretization operators satisfying the assumptions of
Proposition~\ref{prop:error-second-order-greek-payoff}.

For the second-order variational process $K^{(ij),t,x}$, we use the
Milstein discretization
\begin{align}
\mathcal S_K^{\mathrm{mil}}
\left(K,J^i,J^j,X,s,h;\Delta W_s,\Delta\overleftarrow B_s\right)
= K
&+A_K(X,J^i,J^j,K)\,h
+\sum_{a=1}^{d}
B_{K,a}(X,J^i,J^j,K)\,\Delta W_s^a\notag\\
&+\sum_{a=1}^{d}\sum_{b=1}^{d}
\mathcal L_a^{X,J,K}B_{K,b}(X,J^i,J^j,K)
\int_s^{s+h}
(W_r^a-W_s^a)\,\mathrm dW_r^b .
\end{align}
Here
\begin{align*}
A_K(X,J^i,J^j,K)
=\nabla_x b(X)K
+\nabla_x^2 b(X)[J^i,J^j],
\quad
B_{K,a}(X,J^i,J^j,K)
=\nabla_x\sigma_{\cdot a}(X)K
+\nabla_x^2\sigma_{\cdot a}(X)[J^i,J^j].
\end{align*}
The differential operator
$\mathcal L_a^{X,J,K}$
is the diffusion vector field of the extended process
$(X,J^i,J^j,K)$ associated with the $a$-th Brownian component, namely
\begin{align*}
\mathcal L_a^{X,J,K}
=\sigma_{\cdot a}(X)\cdot\nabla_X
+\big(\nabla_x\sigma_{\cdot a}(X)J^i\big)\cdot\nabla_{J^i}
+\big(\nabla_x\sigma_{\cdot a}(X)J^j\big)\cdot\nabla_{J^j}
+\Big(\nabla_x\sigma_{\cdot a}(X)K
+\nabla_x^2\sigma_{\cdot a}(X)[J^i,J^j]
\Big)\cdot\nabla_K .
\end{align*}
Thus the approximation is generated by
\[
K_{k+1}^{(ij,h)}
=
\mathcal S_K^{\mathrm{mil}}
\left(
K_k^{(ij,h)},J_k^{(i,h)},J_k^{(j,h)},X_k^{(h)},t_k,h;
\Delta W_{t_k},\Delta\overleftarrow B_{t_k}
\right),
\quad
K_0^{(ij,h)}=0 .
\]
Under the regularity assumptions imposed on the coefficients,
this Milstein approximation achieves strong-error order one.

Similarly to the first order case,
the integral discretization operator
$\mathcal S_{\mathrm{int}}^{(ij),\mathrm{FBT}}$
is obtained from a first order stochastic Taylor expansion of the
integrands appearing in the second-order Greek representation.
resulting discretization.
A detailed derivation is provided in
Appendix~\ref{app:strong-error-framework},
Paragraph~\ref{para:construction-second-greek}.

Fix $1\le i,j\le d$. For simplicity, we write
\begin{align*}
X_s=X_s^{t,x},
\quad
J_s=J_s^{t,x},
\quad
K_s^{(ij)}=K_s^{(ij),t,x},
\quad
\Phi_s=\Phi(t,s),
\quad
J_s^j=J_s^{t,x}e_j,
\end{align*}
For the second-order Greek integral 
\begin{align*}
Y_{t_k}^{(ij)}
=&\int_t^{t_k}\Phi_s
\Big(\nabla_x^2(F+dH)(s,X_s)[J_s^j,J_s^i]
+\nabla_x(F+dH)(s,X_s)^\top K_s^{(ij)}\Big)\,\mathrm ds\notag\\
&+\int_t^{t_k}\Phi_s\Big(
\nabla_x^2H(s,X_s)[J_s^j,J_s^i]
+\nabla_xH(s,X_s)^\top K_s^{(ij)}\Big)
\mathrm d \overleftarrow{B}_s,
\end{align*}
we define the discretization operator by
\begin{align}
&\mathcal S_{\mathrm{int}}^{(ij),\mathrm{FBT}}
\left(\Phi,X,J,K,s,h;
\Delta W_s,\Delta\overleftarrow B_s\right)\notag\\
=&\Phi h
\Big(\nabla_x^2(F+dH)(s,X)[J^j,J^i]
+\nabla_x(F+dH)(s,X)^\top K\Big)
+\Phi\Big(\nabla_x^2H(s,X)[J^j,J^i]
+\nabla_xH(s,X)^\top K\Big)
\cdot\Delta\overleftarrow B_s\notag\\
&+\Phi\sum_{\nu=1}^{\ell}\sum_{a=1}^{d}
\bigg[\nabla_x^3H_\nu(s,X)
[\sigma_{\cdot a}(X),J^i,J^j]
+\nabla_x^2H_\nu(s,X)[(\nabla_x\sigma_{\cdot a}(X)J)^j,J^i]
+\nabla_x^2H_\nu(s,X)
[J^j,(\nabla_x\sigma_{\cdot a}(X)J)^i]\notag\\
&\hspace{3cm}
+\sigma_{\cdot a}(X)^\top\nabla_x^2H_\nu(s,X)K
+\nabla_xH_\nu(s,X)^\top
\big(\nabla_x\sigma_{\cdot a}(X)K
+\nabla_x^2\sigma_{\cdot a}(X)[J^i,J^j]
\big)\notag\\
&\hspace{3cm}
+\widetilde c_a(s)
\Big(\nabla_x^2H_\nu(s,X)[J^j,J^i]
+\nabla_xH_\nu(s,X)^\top K\Big)\bigg]
J_{\nu a}^{WB}(s,h)\notag\\
&+\Phi\sum_{\nu=1}^{\ell}\sum_{\alpha=1}^{\ell}
d_\alpha(s)\Big((J^j)^\top
\nabla_x^2H_\nu(s,X)J^i
+\nabla_xH_\nu(s,X)^\top K\Big)
J_{\nu\alpha}^{BB}(s,h),
\label{eq:Sint-second-greek-mil-full}
\end{align}
where
$J^i=Je_i$, $J^j:=Je_j$,
and
$(\nabla_x\sigma_{\cdot a}(X)J)^i
=\nabla_x\sigma_{\cdot a}(X)J^i.$

The following two propositions verify the assumptions of
Proposition~\ref{prop:error-second-order-greek-payoff}.
More precisely, we show that
$\mathcal S_{\mathrm{int}}^{(ij),\mathrm{FBT}}$
has strong-error order one and satisfies the required accumulated stability
estimate.
Detailed proofs are provided in
Appendix~\ref{app:strong-error-framework},
Paragraphs~\ref{para:proof of prop:second-order-greek-strong-error}
and~\ref{para:proof of prop:second-order-greek accumulated stability}.

\begin{proposition}[Strong-error order of the second-order Greek integral discretization]
\label{prop:second-order-greek-strong-error}
Assume that the coefficients are sufficiently smooth with bounded derivatives
up to the order used above, and assume that $X$, $J$, $K$, and $\Phi$ have
uniformly bounded moments of all required orders. Then the approximation
generated by
\begin{align*}
\widetilde Y_{k+1}^{(ij,h)}
=
\widetilde Y_k^{(ij,h)}
+
\mathcal S_{\mathrm{int}}^{(ij),\mathrm{FBT}}
\left(\Phi(t,t_k),X_{t_k}^{t,x},J_{t_k}^{t,x},
K_{t_k}^{(ij),t,x},t_k,h;\Delta W_{t_k},\Delta\overleftarrow B_{t_k}
\right),
\quad
\widetilde Y_0^{(ij,h)}=0,
\end{align*}
satisfies, for every fixed $q\ge2$,
\begin{align*}
\left\|
\sup_{0\le k\le N}
\left|
Y_{t_k}^{(ij)}-\widetilde Y_k^{(ij,h)}
\right|
\right\|_{L^q}
\le C_qh.
\end{align*}
Consequently, by Jensen's inequality, the same estimate also holds for
$0<q<2$. Hence
$\mathcal S_{\mathrm{int}}^{(ij),\mathrm{FBT}}$ has strong-error order $1$ in the
sense of Definition~\ref{def:Strong error of S int second Greek}.
\end{proposition}

\begin{proposition}
\label{prop:second-order-greek accumulated stability}
Assume that
$F\in C_b^3$, $H\in C_b^4$, $\sigma\in C_b^3$
and $d$, $\widetilde c$ are bounded. Moreover, assume that
$\Phi(t,s)$, $\Phi^{(h)}$,
$X_{s}^{t,x}$, $X^{(h)}$,
$J_{s}^{t,x}$, $J^{(h)}$,
$K_{s}^{(ij),t,x}$, $K^{(ij,h)}$
are bounded in $L^{16}_{W,B}$.
Then $\mathcal S_{\mathrm{int}}^{(ij),\mathrm{FBT}}$ satisfies the accumulated
stability estimate. More precisely, there exists
$L_{\mathrm{int}}^{(ij)}>0$, independent of $h$, such that
\begin{align*}
\left\|
\sup_{0\le k\le N}
\left|
\widetilde Y_k^{(ij,h)}-Y_k^{(ij,h)}
\right|
\right\|_{L^2}
\le L_{\mathrm{int}}^{(ij)}
\bigg(
&
\left\|
\sup_{0\le r\le N}
\left|
\Phi(t,t_r)-\Phi_r^{(h)}
\right|
\right\|_{L^4}+
\left\|
\sup_{0\le r\le N}
\left|
X_{t_r}^{t,x}-X_r^{(h)}
\right|
\right\|_{L^4}\\
&+
\left\|
\sup_{0\le r\le N}
\left|
J_{t_r}^{t,x}-J_r^{(h)}
\right|
\right\|_{L^4}+
\left\|
\sup_{0\le r\le N}
\left|
K_{t_r}^{(ij),t,x}-K_r^{(ij,h)}
\right|
\right\|_{L^4}
\bigg).
\end{align*}
\end{proposition}

\section{Experiments on Conditional MLMC}
\label{sec:experiment}
In this section, we perform classical Monte Carlo experiments to
validate the discretization operators developed in
Section~\ref{sec:Strong-Order-One Numerical Schemes for Pricing and Greek Estimators}.

More specifically, we numerically investigate the multilevel
convergence rates
$\alpha$, $\beta$, $\gamma$
defined through
\begin{align*}
  |\mathbb E[P_\ell-P]|
=\mathcal O(h_\ell^\alpha),
\quad
\mathrm{Var}(P_\ell-P_{\ell-1})
=\mathcal O(h_\ell^\beta),
\quad
C_\ell=\mathcal O(h_\ell^{-\gamma}),
\end{align*}
where $P_\ell$ denotes the level-$\ell$ approximation and
$C_\ell$ is the average cost of generating one sample on level $\ell$.
According to Theorem~\ref{thm:QAMLMC for u when given B}
and Corollary~\ref{cor:app-qmlmc-greeks},
the proposed numerical scheme is expected to
yield
\begin{align*}
  \alpha=1,\quad\beta=2,\quad\gamma=1.
\end{align*}
The purpose of the experiments is therefore to verify these predicted
rates for the pricing estimator, the first-order Greek estimator, and
the second-order Greek estimator.

We consider the SPDE~\eqref{eq:spde0} with coefficients
\begin{align*}
c(s)&=-0.30+0.3s, \quad d(s)=0.15+0.2s, \quad \tilde c(s)=0.20+0.2s,\\
F(s,x)&=0.05e^{-s}\sin x,\quad
H(s,x)=0.08e^{-s}\cos x.
\end{align*}
The forward diffusion is given by
\begin{align*}
dX_s=b(X_s)\,ds+\sigma(X_s)\,dW_s,
\end{align*}
where
\begin{align*}
b(x)=\kappa(\mu-x),
\quad
\sigma(x)=\nu(1+0.5x),
\end{align*}
with parameters
\begin{align*}
\kappa=1.2,\quad \mu=0,\quad \nu=0.35.
\end{align*}
The terminal payoff is chosen as
\begin{align*}
G(x)=\sin(x).
\end{align*}
Unless otherwise specified, we take
\begin{align*}
t=0,\quad T=1, \quad x_0=0.3.
\end{align*}

For each experiment, we fix the maximal refinement level
\begin{align*}
  L_{\max}=12
\end{align*}
and use
\begin{align*}
N=2\times 10^4
\end{align*}
Monte Carlo samples on each level.

For the bias estimate, the exact target $P$ is approximated by the
finest-level approximation $P_{L_{\max}}$. Thus, we estimate
\begin{align*}
\left|\mathbb E[P_\ell-P_{L_{\max}}]\right|
\end{align*}
for $\ell<L_{\max}$, and obtain $\alpha$ from the regression
\begin{align*}
\left|
\mathbb E[P_\ell-P_{L_{\max}}]
\right|
=\mathcal O(h_\ell^\alpha).
\end{align*}

For the variance estimate, we use the coupled level difference
\begin{align*}
\Delta_\ell=P_\ell-P_{\ell-1},
\end{align*}
and estimate
\begin{align*}
\mathrm{Var}(\Delta_\ell)
=\mathcal O(h_\ell^\beta).
\end{align*}

Finally, $C_\ell$ denotes the average computational cost of generating
one level-$\ell$ sample, and $\gamma$ is estimated from
\begin{align*}
C_\ell
=\mathcal O(h_\ell^{-\gamma}).
\end{align*}

For each realization of the backward Brownian path $B$, we estimate
the exponents $\alpha$, $\beta$, and $\gamma$.
The experiment is repeated for ten independent realizations of $B$,
and the reported values correspond to the empirical mean of the
resulting estimates.

In particular, the mixed forward--backward iterated integral
\begin{align*}
J^{WB}_{s,h}
=\int_s^{s+h}
(W_r-W_s)\,\mathrm d \overleftarrow{B}_r
\end{align*}
is sampled using the Gaussian approximation
\begin{align*}
J^{WB}_{s,h}
\approx\frac12\,\Delta W\,\Delta B
+\sqrt{\frac{h\Delta W^2+h \Delta B^2+h^2}{12}}\,Z,
\quad
Z\sim N(0,1),
\end{align*}
where $Z$ is independent of $\Delta W$ and $\Delta B$.
This approximation matches the conditional mean and variance of
$J^{WB}_{s,h}$ given $(\Delta W,\Delta B)$.

\begin{figure}[ht]
\centering
\includegraphics[width=0.32\textwidth]{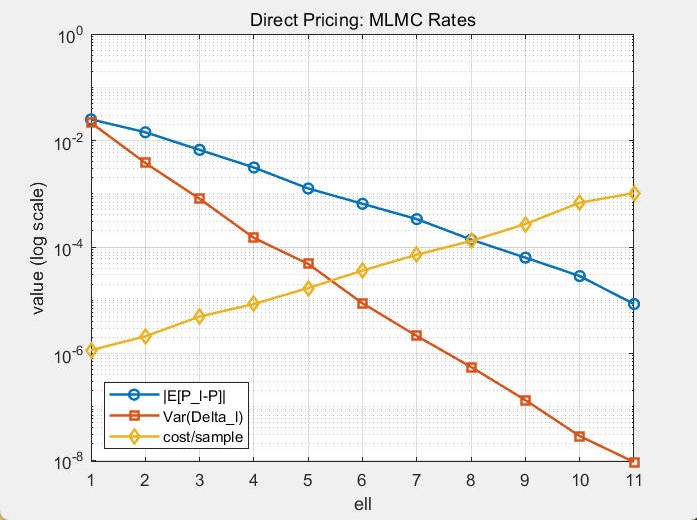}
\includegraphics[width=0.32\textwidth]{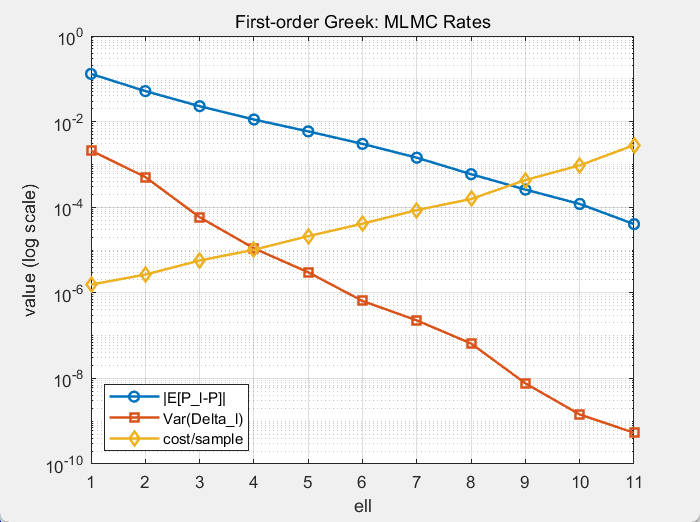}
\includegraphics[width=0.32\textwidth]{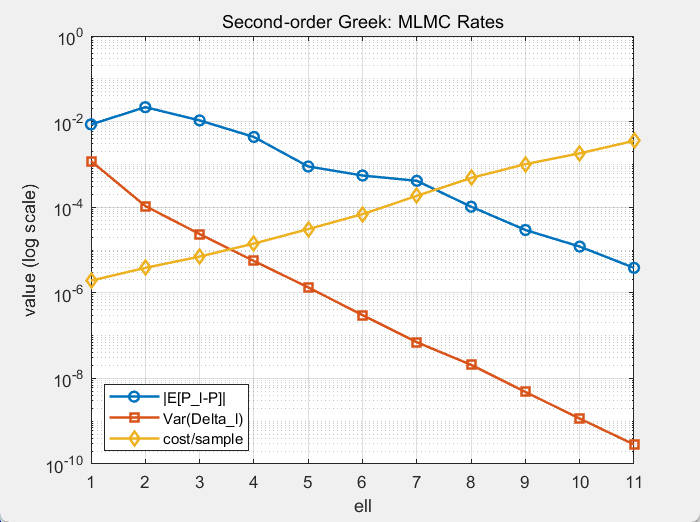}
\caption{Multilevel bias, variance, and cost for the pricing estimator
(left), first-order Greek estimator (middle), and second-order Greek
estimator (right).}
\label{fig:mlmc-rates}
\end{figure}
\begin{table}[ht]
\centering
\setlength{\tabcolsep}{8pt}
\begin{tabular}{c|ccccccccccc}
\toprule
& $B_1$ & $B_2$ & $B_3$ & $B_4$ & $B_5$ & $B_6$ & $B_7$ & $B_8$ & $B_9$ & $B_{10}$ & \textbf{Mean} \\
\midrule
$\alpha$ & 1.087 & 1.188 & 1.050 & 1.286 & 1.291 & 1.135 & 1.107 & 1.012 & 1.033 & 1.257 & \textbf{1.145} \\
$\beta$  & 2.008 & 2.053 & 1.993 & 2.043 & 2.075 & 2.065 & 2.045 & 2.075 & 2.042 & 2.057 & \textbf{2.045} \\
$\gamma$ & 0.891 & 0.999 & 0.990 & 0.996 & 0.987 & 1.008 & 0.999 & 0.996 & 1.003 & 1.006 & \textbf{0.987} \\
\bottomrule
\end{tabular}
\caption{Estimated multilevel exponents $(\alpha,\beta,\gamma)$ for the direct pricing estimator.}
\label{tab:direct-pricing-exponents}
\end{table}

\begin{table}[ht]
\centering
\setlength{\tabcolsep}{8pt}
\begin{tabular}{c|ccccccccccc}
\toprule
& $B_1$ & $B_2$ & $B_3$ & $B_4$ & $B_5$ & $B_6$ & $B_7$ & $B_8$ & $B_9$ & $B_{10}$ & \textbf{Mean} \\
\midrule
$\alpha$ & 1.152 & 1.071 & 1.174 & 1.054 & 1.071 & 1.090 & 1.104 & 1.050 & 1.370 & 1.045 & \textbf{1.118} \\
$\beta$  & 2.057 & 2.027 & 2.074 & 1.998 & 1.896 & 2.036 & 1.962 & 1.898 & 2.080 & 1.972 & \textbf{2.000} \\
$\gamma$ & 1.004 & 0.996 & 0.995 & 0.997 & 1.000 & 1.002 & 1.000 & 0.992 & 1.000 & 0.995 & \textbf{0.998} \\
\bottomrule
\end{tabular}
\caption{Estimated multilevel exponents $(\alpha,\beta,\gamma)$ for the first-order Greek estimator.}
\label{tab:first-greek-exponents}
\end{table}

\begin{table}[ht]
\centering
\setlength{\tabcolsep}{8pt}
\begin{tabular}{c|ccccccccccc}
\toprule
& $B_1$ & $B_2$ & $B_3$ & $B_4$ & $B_5$ & $B_6$ & $B_7$ & $B_8$ & $B_9$ & $B_{10}$ & \textbf{Mean} \\
\midrule
$\alpha$ & 1.101 & 0.914 & 0.970 & 0.872 & 1.041 & 1.076 & 1.061 & 0.984 & 0.991 & 0.998 & \textbf{1.001} \\
$\beta$  & 2.037 & 2.067 & 2.012 & 2.000 & 2.025 & 2.071 & 2.014 & 2.009 & 2.083 & 2.059 & \textbf{2.038} \\
$\gamma$ & 0.993 & 0.980 & 0.997 & 0.992 & 0.994 & 1.006 & 1.006 & 0.996 & 1.004 & 0.996 & \textbf{0.996} \\
\bottomrule
\end{tabular}
\caption{Estimated multilevel exponents $(\alpha,\beta,\gamma)$ for the second-order Greek estimator.}
\label{tab:second-greek-exponents}
\end{table}

The empirical rates shown in Figure~\ref{fig:mlmc-rates} and
Tables~\ref{tab:direct-pricing-exponents}--\ref{tab:second-greek-exponents}
report the estimated multilevel exponents for the direct pricing estimator,
the first-order Greek estimator, and the second-order Greek estimator.

For the direct pricing estimator, the averaged exponents over ten independent
realizations of the backward Brownian motion are
\begin{align*}
\alpha=1.145,\quad
\beta=2.045,\quad
\gamma=0.987.
\end{align*}
The first-order Greek estimator gives
\begin{align*}
\alpha=1.118,\quad
\beta=2.000,\quad
\gamma=0.998,
\end{align*}
while the second-order Greek estimator gives
\begin{align*}
\alpha=1.001,\quad
\beta=2.038,\quad
\gamma=0.996.
\end{align*}

Across all three estimators, the variance and cost exponents are particularly
stable with respect to the realization of the backward Brownian motion.
Moreover, the estimated exponents remain close to
\begin{align*}
(\alpha,\beta,\gamma)=(1,2,1),
\end{align*}
which provides strong numerical evidence for the expected multilevel scaling
relations
\begin{align*}
\alpha=1,
\quad
\beta=2,
\quad
\gamma=1.
\end{align*}
These results are in excellent agreement with the theoretical predictions
obtained from the global strong-error framework and the associated
Forward--Backward Taylor schemes developed in
Section~\ref{sec:Strong-Order-One Numerical Schemes for Pricing and Greek Estimators}.
Consequently, the assumptions required by
Theorem~\ref{thm:QAMLMC for u when given B}
and
Corollary~\ref{cor:app-qmlmc-greeks}
are numerically supported.

\section{Discussion and Future Work}
\label{sec:discussion}

In this paper, we developed a quantum-accelerated multilevel Monte
Carlo (QA-MLMC) framework for stochastic partial differential equations (SPDEs) arising
from stochastic-environment financial models.
The main idea is to exploit the BDSDE representation of the SPDE solution
to reformulate pricing and sensitivity estimation as conditional and
nested expectation estimation problems.
These expectation problems can then be estimated by QA-MLMC, leading to
quadratic quantum speedups for derivative pricing and Greek
estimation in stochastic environments.

The proposed framework contains two main components.
First, for a fixed realization of the backward Brownian motion, we
construct a conditional QA-MLMC estimator for the SPDE quantity of
interest.
Second, for quantities involving an additional average over the random
environment, we develop a nested QA-MLMC estimator.
Under the multilevel assumptions on bias, variance, and cost, these
estimators achieve quantum sampling complexity of order
\[
\widetilde{\mathcal O}(\epsilon^{-1}),
\]
for an additive error tolerance
$\epsilon$.
We also show that the same algorithmic structure applies not only to
prices, but also to first-order and second-order Greeks.

A key numerical ingredient is the Forward--Backward Taylor
scheme developed in this work.
Unlike discretization methods that eliminate the forward Brownian
randomness through conditional expectations at each time step, the
Forward--Backward Taylor discretization keeps the joint pathwise
dependence on the forward and backward Brownian motions.
This feature is essential for constructing coupled level differences in
the conditional multilevel estimator.
The strong-error order one convergence results for pricing and Greek
estimators provide the numerical foundation for the complexity analysis
of the proposed quantum algorithms.

There are several natural directions for future work.
First, although the present paper focuses on SPDEs arising from
stochastic-environment financial models, the conditional and nested
estimation structure is more general.
It would be interesting to extend the proposed QA-MLMC framework to
broader classes of SPDEs and stochastic systems with common noise,
random coefficients, or random media, where BDSDE-type representations
can be used to connect SPDE solutions with expectation estimation
problems.

Second, the Forward--Backward Taylor discretization introduced in this
work suggests several further numerical developments.
One direction is to construct higher-order or adaptive
forward--backward schemes that remain compatible with conditional
multilevel couplings.
Another is to analyze such schemes under weaker regularity assumptions,
more general noise structures, or path-dependent functionals.
A better understanding of the relation between strong convergence,
level variance, and sample cost would also help optimize the resulting
MLMC and QA-MLMC complexity.

Finally, the nested estimator developed here could be extended to more
general nested or nonlinear quantities, such as risk measures
 and multi-layer conditional expectations.
These problems arise naturally in financial applications and in
stochastic systems with random environments, and provide a promising
setting for further applications of quantum computing.

\begin{acknowledgments}
JPL acknowledges support from Quantum Science and Technology--National Science and Technology Major Project (Grant No.~2024ZD0300500), Excellent Young Scientists Fund Program, start-up funding from Tsinghua University and Beijing Institute of Mathematical Sciences and Applications. 
Z.L. was supported by the Beijing Natural Science Foundation Key Program (Grant No.~Z220002). R.L. and Z.L. were supported by BMSTC and ACZSP (Grant No.~Z221100002722017). 
\end{acknowledgments}

\bibliography{ref.bib}

\begin{thebibliography}{76}%
\makeatletter
\providecommand \@ifxundefined [1]{%
 \@ifx{#1\undefined}
}%
\providecommand \@ifnum [1]{%
 \ifnum #1\expandafter \@firstoftwo
 \else \expandafter \@secondoftwo
 \fi
}%
\providecommand \@ifx [1]{%
 \ifx #1\expandafter \@firstoftwo
 \else \expandafter \@secondoftwo
 \fi
}%
\providecommand \natexlab [1]{#1}%
\providecommand \enquote  [1]{``#1''}%
\providecommand \bibnamefont  [1]{#1}%
\providecommand \bibfnamefont [1]{#1}%
\providecommand \citenamefont [1]{#1}%
\providecommand \href@noop [0]{\@secondoftwo}%
\providecommand \href [0]{\begingroup \@sanitize@url \@href}%
\providecommand \@href[1]{\@@startlink{#1}\@@href}%
\providecommand \@@href[1]{\endgroup#1\@@endlink}%
\providecommand \@sanitize@url [0]{\catcode `\\12\catcode `\$12\catcode `\&12\catcode `\#12\catcode `\^12\catcode `\_12\catcode `\%12\relax}%
\providecommand \@@startlink[1]{}%
\providecommand \@@endlink[0]{}%
\providecommand \url  [0]{\begingroup\@sanitize@url \@url }%
\providecommand \@url [1]{\endgroup\@href {#1}{\urlprefix }}%
\providecommand \urlprefix  [0]{URL }%
\providecommand \Eprint [0]{\href }%
\providecommand \doibase [0]{https://doi.org/}%
\providecommand \selectlanguage [0]{\@gobble}%
\providecommand \bibinfo  [0]{\@secondoftwo}%
\providecommand \bibfield  [0]{\@secondoftwo}%
\providecommand \translation [1]{[#1]}%
\providecommand \BibitemOpen [0]{}%
\providecommand \bibitemStop [0]{}%
\providecommand \bibitemNoStop [0]{.\EOS\space}%
\providecommand \EOS [0]{\spacefactor3000\relax}%
\providecommand \BibitemShut  [1]{\csname bibitem#1\endcsname}%
\let\auto@bib@innerbib\@empty
\bibitem [{\citenamefont {Black}\ and\ \citenamefont {Scholes}(1973)}]{BlackScholes1973}%
  \BibitemOpen
  \bibfield  {author} {\bibinfo {author} {\bibfnamefont {F.}~\bibnamefont {Black}}\ and\ \bibinfo {author} {\bibfnamefont {M.}~\bibnamefont {Scholes}},\ }\bibfield  {title} {\bibinfo {title} {The pricing of options and corporate liabilities},\ }\href {https://doi.org/10.1086/260062} {\bibfield  {journal} {\bibinfo  {journal} {Journal of Political Economy}\ }\textbf {\bibinfo {volume} {81}},\ \bibinfo {pages} {637} (\bibinfo {year} {1973})}\BibitemShut {NoStop}%
\bibitem [{\citenamefont {Heston}(1993)}]{Heston1993}%
  \BibitemOpen
  \bibfield  {author} {\bibinfo {author} {\bibfnamefont {S.~L.}\ \bibnamefont {Heston}},\ }\bibfield  {title} {\bibinfo {title} {A closed-form solution for options with stochastic volatility with applications to bond and currency options},\ }\href {https://doi.org/10.1093/rfs/6.2.327} {\bibfield  {journal} {\bibinfo  {journal} {The Review of Financial Studies}\ }\textbf {\bibinfo {volume} {6}},\ \bibinfo {pages} {327} (\bibinfo {year} {1993})}\BibitemShut {NoStop}%
\bibitem [{\citenamefont {Fouque}\ \emph {et~al.}(2000)\citenamefont {Fouque}, \citenamefont {Papanicolaou},\ and\ \citenamefont {Sircar}}]{FouquePapanicolaouSircar2000}%
  \BibitemOpen
  \bibfield  {author} {\bibinfo {author} {\bibfnamefont {J.-P.}\ \bibnamefont {Fouque}}, \bibinfo {author} {\bibfnamefont {G.}~\bibnamefont {Papanicolaou}},\ and\ \bibinfo {author} {\bibfnamefont {R.}~\bibnamefont {Sircar}},\ }\href@noop {} {\emph {\bibinfo {title} {Derivatives in Financial Markets with Stochastic Volatility}}}\ (\bibinfo  {publisher} {Cambridge University Press},\ \bibinfo {address} {Cambridge},\ \bibinfo {year} {2000})\BibitemShut {NoStop}%
\bibitem [{\citenamefont {Bergomi}(2015)}]{Bergomi2015}%
  \BibitemOpen
  \bibfield  {author} {\bibinfo {author} {\bibfnamefont {L.}~\bibnamefont {Bergomi}},\ }\href {https://www.lorenzobergomi.com/} {\emph {\bibinfo {title} {Stochastic Volatility Modeling}}},\ Financial Mathematics Series\ (\bibinfo  {publisher} {Chapman and Hall/CRC},\ \bibinfo {year} {2015})\BibitemShut {NoStop}%
\bibitem [{\citenamefont {Heath}\ \emph {et~al.}(1992)\citenamefont {Heath}, \citenamefont {Jarrow},\ and\ \citenamefont {Morton}}]{HeathJarrowMorton1992}%
  \BibitemOpen
  \bibfield  {author} {\bibinfo {author} {\bibfnamefont {D.}~\bibnamefont {Heath}}, \bibinfo {author} {\bibfnamefont {R.}~\bibnamefont {Jarrow}},\ and\ \bibinfo {author} {\bibfnamefont {A.}~\bibnamefont {Morton}},\ }\bibfield  {title} {\bibinfo {title} {Bond pricing and the term structure of interest rates: A new methodology},\ }\href {https://doi.org/10.2307/2951677} {\bibfield  {journal} {\bibinfo  {journal} {Econometrica}\ }\textbf {\bibinfo {volume} {60}},\ \bibinfo {pages} {77} (\bibinfo {year} {1992})}\BibitemShut {NoStop}%
\bibitem [{\citenamefont {Santa-Clara}\ and\ \citenamefont {Sornette}(2001)}]{santa2001dynamics}%
  \BibitemOpen
  \bibfield  {author} {\bibinfo {author} {\bibfnamefont {P.}~\bibnamefont {Santa-Clara}}\ and\ \bibinfo {author} {\bibfnamefont {D.}~\bibnamefont {Sornette}},\ }\bibfield  {title} {\bibinfo {title} {The dynamics of the forward interest rate curve with stochastic string shocks},\ }\href@noop {} {\bibfield  {journal} {\bibinfo  {journal} {The Review of Financial Studies}\ }\textbf {\bibinfo {volume} {14}},\ \bibinfo {pages} {149} (\bibinfo {year} {2001})},\ \Eprint {https://arxiv.org/abs/arXiv:cond-mat/9801321} {arXiv:cond-mat/9801321} \BibitemShut {NoStop}%
\bibitem [{\citenamefont {Cont}(2005)}]{cont2005modeling}%
  \BibitemOpen
  \bibfield  {author} {\bibinfo {author} {\bibfnamefont {R.}~\bibnamefont {Cont}},\ }\bibfield  {title} {\bibinfo {title} {Modeling term structure dynamics: an infinite dimensional approach},\ }\href@noop {} {\bibfield  {journal} {\bibinfo  {journal} {International Journal of theoretical and applied finance}\ }\textbf {\bibinfo {volume} {8}},\ \bibinfo {pages} {357} (\bibinfo {year} {2005})},\ \Eprint {https://arxiv.org/abs/arXiv:cond-mat/9902018} {arXiv:cond-mat/9902018} \BibitemShut {NoStop}%
\bibitem [{\citenamefont {Benth}\ \emph {et~al.}(2008)\citenamefont {Benth}, \citenamefont {Benth},\ and\ \citenamefont {Koekebakker}}]{benth2008stochastic}%
  \BibitemOpen
  \bibfield  {author} {\bibinfo {author} {\bibfnamefont {F.~E.}\ \bibnamefont {Benth}}, \bibinfo {author} {\bibfnamefont {J.~S.}\ \bibnamefont {Benth}},\ and\ \bibinfo {author} {\bibfnamefont {S.}~\bibnamefont {Koekebakker}},\ }\href@noop {} {\emph {\bibinfo {title} {Stochastic modelling of electricity and related markets}}},\ Vol.~\bibinfo {volume} {11}\ (\bibinfo  {publisher} {World Scientific},\ \bibinfo {year} {2008})\BibitemShut {NoStop}%
\bibitem [{\citenamefont {Hull}\ and\ \citenamefont {Basu}(2016)}]{hull2016options}%
  \BibitemOpen
  \bibfield  {author} {\bibinfo {author} {\bibfnamefont {J.~C.}\ \bibnamefont {Hull}}\ and\ \bibinfo {author} {\bibfnamefont {S.}~\bibnamefont {Basu}},\ }\href@noop {} {\emph {\bibinfo {title} {Options, futures, and other derivatives}}}\ (\bibinfo  {publisher} {Pearson Education India},\ \bibinfo {year} {2016})\BibitemShut {NoStop}%
\bibitem [{\citenamefont {Printems}(2001)}]{printems2001discretization}%
  \BibitemOpen
  \bibfield  {author} {\bibinfo {author} {\bibfnamefont {J.}~\bibnamefont {Printems}},\ }\bibfield  {title} {\bibinfo {title} {On the discretization in time of parabolic stochastic partial differential equations},\ }\href@noop {} {\bibfield  {journal} {\bibinfo  {journal} {ESAIM: Mathematical Modelling and Numerical Analysis}\ }\textbf {\bibinfo {volume} {35}},\ \bibinfo {pages} {1055} (\bibinfo {year} {2001})}\BibitemShut {NoStop}%
\bibitem [{\citenamefont {Larsson}\ and\ \citenamefont {Thom{\'e}e}(2003)}]{larsson2003partial}%
  \BibitemOpen
  \bibfield  {author} {\bibinfo {author} {\bibfnamefont {S.}~\bibnamefont {Larsson}}\ and\ \bibinfo {author} {\bibfnamefont {V.}~\bibnamefont {Thom{\'e}e}},\ }\href@noop {} {\emph {\bibinfo {title} {Partial differential equations with numerical methods}}}\ (\bibinfo  {publisher} {Springer},\ \bibinfo {year} {2003})\BibitemShut {NoStop}%
\bibitem [{\citenamefont {Gy{\"o}ngy}(1999)}]{gyongy1999lattice}%
  \BibitemOpen
  \bibfield  {author} {\bibinfo {author} {\bibfnamefont {I.}~\bibnamefont {Gy{\"o}ngy}},\ }\bibfield  {title} {\bibinfo {title} {Lattice approximations for stochastic quasi-linear parabolic partial differential equations driven by space-time white noise {II}},\ }\href@noop {} {\bibfield  {journal} {\bibinfo  {journal} {Potential Analysis}\ }\textbf {\bibinfo {volume} {11}},\ \bibinfo {pages} {1} (\bibinfo {year} {1999})}\BibitemShut {NoStop}%
\bibitem [{\citenamefont {Lord}\ and\ \citenamefont {Rougemont}(2004)}]{lord2004numerical}%
  \BibitemOpen
  \bibfield  {author} {\bibinfo {author} {\bibfnamefont {G.~J.}\ \bibnamefont {Lord}}\ and\ \bibinfo {author} {\bibfnamefont {J.}~\bibnamefont {Rougemont}},\ }\bibfield  {title} {\bibinfo {title} {A numerical scheme for stochastic pdes with gevrey regularity},\ }\href@noop {} {\bibfield  {journal} {\bibinfo  {journal} {IMA journal of numerical analysis}\ }\textbf {\bibinfo {volume} {24}},\ \bibinfo {pages} {587} (\bibinfo {year} {2004})}\BibitemShut {NoStop}%
\bibitem [{\citenamefont {Brenner}\ and\ \citenamefont {Scott}(2008)}]{brenner2008mathematical}%
  \BibitemOpen
  \bibfield  {author} {\bibinfo {author} {\bibfnamefont {S.~C.}\ \bibnamefont {Brenner}}\ and\ \bibinfo {author} {\bibfnamefont {L.~R.}\ \bibnamefont {Scott}},\ }\href@noop {} {\emph {\bibinfo {title} {The mathematical theory of finite element methods}}}\ (\bibinfo  {publisher} {Springer},\ \bibinfo {year} {2008})\BibitemShut {NoStop}%
\bibitem [{\citenamefont {Bungartz}\ and\ \citenamefont {Griebel}(2004)}]{BungartzGriebel2004}%
  \BibitemOpen
  \bibfield  {author} {\bibinfo {author} {\bibfnamefont {H.-J.}\ \bibnamefont {Bungartz}}\ and\ \bibinfo {author} {\bibfnamefont {M.}~\bibnamefont {Griebel}},\ }\bibfield  {title} {\bibinfo {title} {Sparse grids},\ }\href {https://www.cambridge.org/core/journals/acta-numerica/article/sparse-grids/47EA2993DB84C9D231BB96ECB26F615C} {\bibfield  {journal} {\bibinfo  {journal} {Acta Numerica}\ }\textbf {\bibinfo {volume} {13}},\ \bibinfo {pages} {147} (\bibinfo {year} {2004})}\BibitemShut {NoStop}%
\bibitem [{\citenamefont {Hout}\ \emph {et~al.}(2010)\citenamefont {Hout}, \citenamefont {Foulon} \emph {et~al.}}]{hout2010adi}%
  \BibitemOpen
  \bibfield  {author} {\bibinfo {author} {\bibfnamefont {I.}~\bibnamefont {Hout}}, \bibinfo {author} {\bibfnamefont {S.}~\bibnamefont {Foulon}}, \emph {et~al.},\ }\bibfield  {title} {\bibinfo {title} {{ADI} finite difference schemes for option pricing in the heston model with correlation.},\ }\href@noop {} {\bibfield  {journal} {\bibinfo  {journal} {International Journal of Numerical Analysis \& Modeling}\ }\textbf {\bibinfo {volume} {7}} (\bibinfo {year} {2010})},\ \Eprint {https://arxiv.org/abs/arXiv:0811.3427v1} {arXiv:0811.3427v1} \BibitemShut {NoStop}%
\bibitem [{\citenamefont {Pardoux}\ and\ \citenamefont {Peng}(1990)}]{Pardoux1990AdaptedSO}%
  \BibitemOpen
  \bibfield  {author} {\bibinfo {author} {\bibfnamefont {E.}~\bibnamefont {Pardoux}}\ and\ \bibinfo {author} {\bibfnamefont {S.}~\bibnamefont {Peng}},\ }\bibfield  {title} {\bibinfo {title} {Adapted solution of a backward stochastic differential equation},\ }\href@noop {} {\bibfield  {journal} {\bibinfo  {journal} {Systems \& control letters}\ }\textbf {\bibinfo {volume} {14}},\ \bibinfo {pages} {55} (\bibinfo {year} {1990})}\BibitemShut {NoStop}%
\bibitem [{\citenamefont {Pardoux}\ and\ \citenamefont {Peng}(2005)}]{Pardoux1992BackwardSD}%
  \BibitemOpen
  \bibfield  {author} {\bibinfo {author} {\bibfnamefont {E.}~\bibnamefont {Pardoux}}\ and\ \bibinfo {author} {\bibfnamefont {S.}~\bibnamefont {Peng}},\ }\bibfield  {title} {\bibinfo {title} {Backward stochastic differential equations and quasilinear parabolic partial differential equations},\ }in\ \href@noop {} {\emph {\bibinfo {booktitle} {Stochastic Partial Differential Equations and Their Applications: Proceedings of IFIP WG 7/1 International Conference University of North Carolina at Charlotte, NC June 6--8, 1991}}}\ (\bibinfo {organization} {Springer},\ \bibinfo {year} {2005})\ pp.\ \bibinfo {pages} {200--217}\BibitemShut {NoStop}%
\bibitem [{\citenamefont {Kobylanski}(2000)}]{kobylanski2000backward}%
  \BibitemOpen
  \bibfield  {author} {\bibinfo {author} {\bibfnamefont {M.}~\bibnamefont {Kobylanski}},\ }\bibfield  {title} {\bibinfo {title} {Backward stochastic differential equations and partial differential equations with quadratic growth},\ }\href@noop {} {\bibfield  {journal} {\bibinfo  {journal} {Annals of probability}\ ,\ \bibinfo {pages} {558}} (\bibinfo {year} {2000})}\BibitemShut {NoStop}%
\bibitem [{\citenamefont {Peng}(1992)}]{peng1992generalized}%
  \BibitemOpen
  \bibfield  {author} {\bibinfo {author} {\bibfnamefont {S.}~\bibnamefont {Peng}},\ }\bibfield  {title} {\bibinfo {title} {A generalized dynamic programming principle and {Hamilton-Jacobi-Bellman} equation},\ }\href@noop {} {\bibfield  {journal} {\bibinfo  {journal} {Stochastics: An International Journal of Probability and Stochastic Processes}\ }\textbf {\bibinfo {volume} {38}},\ \bibinfo {pages} {119} (\bibinfo {year} {1992})}\BibitemShut {NoStop}%
\bibitem [{\citenamefont {Yong}\ and\ \citenamefont {Zhou}(1999)}]{yong1999stochastic}%
  \BibitemOpen
  \bibfield  {author} {\bibinfo {author} {\bibfnamefont {J.}~\bibnamefont {Yong}}\ and\ \bibinfo {author} {\bibfnamefont {X.~Y.}\ \bibnamefont {Zhou}},\ }\href@noop {} {\emph {\bibinfo {title} {Stochastic controls: {H}amiltonian systems and {HJB} equations}}},\ Vol.~\bibinfo {volume} {43}\ (\bibinfo  {publisher} {Springer Science \& Business Media},\ \bibinfo {year} {1999})\BibitemShut {NoStop}%
\bibitem [{\citenamefont {El~Karoui}\ \emph {et~al.}(1997)\citenamefont {El~Karoui}, \citenamefont {Peng},\ and\ \citenamefont {Quenez}}]{Karoui1997BackwardSD}%
  \BibitemOpen
  \bibfield  {author} {\bibinfo {author} {\bibfnamefont {N.}~\bibnamefont {El~Karoui}}, \bibinfo {author} {\bibfnamefont {S.}~\bibnamefont {Peng}},\ and\ \bibinfo {author} {\bibfnamefont {M.~C.}\ \bibnamefont {Quenez}},\ }\bibfield  {title} {\bibinfo {title} {Backward stochastic differential equations in finance},\ }\href@noop {} {\bibfield  {journal} {\bibinfo  {journal} {Mathematical finance}\ }\textbf {\bibinfo {volume} {7}},\ \bibinfo {pages} {1} (\bibinfo {year} {1997})}\BibitemShut {NoStop}%
\bibitem [{\citenamefont {Peng}(1997)}]{peng1997backward}%
  \BibitemOpen
  \bibfield  {author} {\bibinfo {author} {\bibfnamefont {S.}~\bibnamefont {Peng}},\ }\bibfield  {title} {\bibinfo {title} {Backward {SDE} and related g-expectation},\ }\href@noop {} {\bibfield  {journal} {\bibinfo  {journal} {Pitman research notes in mathematics series}\ ,\ \bibinfo {pages} {141}} (\bibinfo {year} {1997})}\BibitemShut {NoStop}%
\bibitem [{\citenamefont {Peng}(2004)}]{peng2004nonlinear}%
  \BibitemOpen
  \bibfield  {author} {\bibinfo {author} {\bibfnamefont {S.}~\bibnamefont {Peng}},\ }\bibfield  {title} {\bibinfo {title} {Nonlinear expectations, nonlinear evaluations and risk measures},\ }in\ \href@noop {} {\emph {\bibinfo {booktitle} {Stochastic Methods in Finance: Lectures given at the CIME-EMS Summer School held in Bressanone/Brixen, Italy, July 6-12, 2003}}}\ (\bibinfo  {publisher} {Springer},\ \bibinfo {year} {2004})\ pp.\ \bibinfo {pages} {165--253}\BibitemShut {NoStop}%
\bibitem [{\citenamefont {Pardoux}\ and\ \citenamefont {Peng}(1994)}]{Pardoux1994BackwardDS}%
  \BibitemOpen
  \bibfield  {author} {\bibinfo {author} {\bibfnamefont {{\'E}.}~\bibnamefont {Pardoux}}\ and\ \bibinfo {author} {\bibfnamefont {S.}~\bibnamefont {Peng}},\ }\bibfield  {title} {\bibinfo {title} {Backward doubly stochastic differential equations and systems of quasilinear{ SPDE}s},\ }\href@noop {} {\bibfield  {journal} {\bibinfo  {journal} {Probability theory and related fields}\ }\textbf {\bibinfo {volume} {98}},\ \bibinfo {pages} {209} (\bibinfo {year} {1994})}\BibitemShut {NoStop}%
\bibitem [{\citenamefont {Zhang}(2004)}]{zhang2004numerical}%
  \BibitemOpen
  \bibfield  {author} {\bibinfo {author} {\bibfnamefont {J.}~\bibnamefont {Zhang}},\ }\bibfield  {title} {\bibinfo {title} {A numerical scheme for {BSDE}s},\ }\href@noop {} {\bibfield  {journal} {\bibinfo  {journal} {The annals of applied probability}\ }\textbf {\bibinfo {volume} {14}},\ \bibinfo {pages} {459} (\bibinfo {year} {2004})}\BibitemShut {NoStop}%
\bibitem [{\citenamefont {Bouchard}\ and\ \citenamefont {Touzi}(2004)}]{bouchard2004discrete}%
  \BibitemOpen
  \bibfield  {author} {\bibinfo {author} {\bibfnamefont {B.}~\bibnamefont {Bouchard}}\ and\ \bibinfo {author} {\bibfnamefont {N.}~\bibnamefont {Touzi}},\ }\bibfield  {title} {\bibinfo {title} {Discrete-time approximation and {Monte-Carlo} simulation of backward stochastic differential equations},\ }\href@noop {} {\bibfield  {journal} {\bibinfo  {journal} {Stochastic Processes and their applications}\ }\textbf {\bibinfo {volume} {111}},\ \bibinfo {pages} {175} (\bibinfo {year} {2004})}\BibitemShut {NoStop}%
\bibitem [{\citenamefont {Gobet}\ \emph {et~al.}(2005)\citenamefont {Gobet}, \citenamefont {Lemor},\ and\ \citenamefont {Warin}}]{gobet2005regression}%
  \BibitemOpen
  \bibfield  {author} {\bibinfo {author} {\bibfnamefont {E.}~\bibnamefont {Gobet}}, \bibinfo {author} {\bibfnamefont {J.-P.}\ \bibnamefont {Lemor}},\ and\ \bibinfo {author} {\bibfnamefont {X.}~\bibnamefont {Warin}},\ }\bibfield  {title} {\bibinfo {title} {A regression-based {Monte Carlo} method to solve backward stochastic differential equations},\ }\href@noop {} {\bibfield  {journal} {\bibinfo  {journal} {The Annals of Applied Probability}\ }\textbf {\bibinfo {volume} {15}},\ \bibinfo {pages} {2172} (\bibinfo {year} {2005})},\ \Eprint {https://arxiv.org/abs/arXiv:math/0508491} {arXiv:math/0508491} \BibitemShut {NoStop}%
\bibitem [{\citenamefont {Han}\ \emph {et~al.}(2017)\citenamefont {Han}, \citenamefont {Jentzen} \emph {et~al.}}]{han2017deep}%
  \BibitemOpen
  \bibfield  {author} {\bibinfo {author} {\bibfnamefont {J.}~\bibnamefont {Han}}, \bibinfo {author} {\bibfnamefont {A.}~\bibnamefont {Jentzen}}, \emph {et~al.},\ }\bibfield  {title} {\bibinfo {title} {Deep learning-based numerical methods for high-dimensional parabolic partial differential equations and backward stochastic differential equations},\ }\href@noop {} {\bibfield  {journal} {\bibinfo  {journal} {Communications in mathematics and statistics}\ }\textbf {\bibinfo {volume} {5}},\ \bibinfo {pages} {349} (\bibinfo {year} {2017})},\ \Eprint {https://arxiv.org/abs/arXiv:1706.04702} {arXiv:1706.04702} \BibitemShut {NoStop}%
\bibitem [{\citenamefont {Han}\ \emph {et~al.}(2018)\citenamefont {Han}, \citenamefont {Jentzen},\ and\ \citenamefont {E}}]{han2018solving}%
  \BibitemOpen
  \bibfield  {author} {\bibinfo {author} {\bibfnamefont {J.}~\bibnamefont {Han}}, \bibinfo {author} {\bibfnamefont {A.}~\bibnamefont {Jentzen}},\ and\ \bibinfo {author} {\bibfnamefont {W.}~\bibnamefont {E}},\ }\bibfield  {title} {\bibinfo {title} {Solving high-dimensional partial differential equations using deep learning},\ }\href@noop {} {\bibfield  {journal} {\bibinfo  {journal} {Proceedings of the National Academy of Sciences}\ }\textbf {\bibinfo {volume} {115}},\ \bibinfo {pages} {8505} (\bibinfo {year} {2018})},\ \Eprint {https://arxiv.org/abs/arXiv:1707.02568} {arXiv:1707.02568} \BibitemShut {NoStop}%
\bibitem [{\citenamefont {Beck}\ \emph {et~al.}(2019)\citenamefont {Beck}, \citenamefont {E},\ and\ \citenamefont {Jentzen}}]{beck2019machine}%
  \BibitemOpen
  \bibfield  {author} {\bibinfo {author} {\bibfnamefont {C.}~\bibnamefont {Beck}}, \bibinfo {author} {\bibfnamefont {W.}~\bibnamefont {E}},\ and\ \bibinfo {author} {\bibfnamefont {A.}~\bibnamefont {Jentzen}},\ }\bibfield  {title} {\bibinfo {title} {Machine learning approximation algorithms for high-dimensional fully nonlinear partial differential equations and second-order backward stochastic differential equations},\ }\href@noop {} {\bibfield  {journal} {\bibinfo  {journal} {Journal of Nonlinear Science}\ }\textbf {\bibinfo {volume} {29}},\ \bibinfo {pages} {1563} (\bibinfo {year} {2019})},\ \Eprint {https://arxiv.org/abs/arXiv:1709.05963} {arXiv:1709.05963} \BibitemShut {NoStop}%
\bibitem [{\citenamefont {Aman}(2013)}]{aman2013numerical}%
  \BibitemOpen
  \bibfield  {author} {\bibinfo {author} {\bibfnamefont {A.}~\bibnamefont {Aman}},\ }\bibfield  {title} {\bibinfo {title} {A numerical scheme for backward doubly stochastic differential equations},\ }\href@noop {} {\bibfield  {journal} {\bibinfo  {journal} {Bernoulli}\ }\textbf {\bibinfo {volume} {19}},\ \bibinfo {pages} {93} (\bibinfo {year} {2013})},\ \Eprint {https://arxiv.org/abs/arXiv:1011.6170} {arXiv:1011.6170} \BibitemShut {NoStop}%
\bibitem [{\citenamefont {Bachouch}\ \emph {et~al.}(2016)\citenamefont {Bachouch}, \citenamefont {Gobet},\ and\ \citenamefont {Matoussi}}]{bachouch2016empirical}%
  \BibitemOpen
  \bibfield  {author} {\bibinfo {author} {\bibfnamefont {A.}~\bibnamefont {Bachouch}}, \bibinfo {author} {\bibfnamefont {E.}~\bibnamefont {Gobet}},\ and\ \bibinfo {author} {\bibfnamefont {A.}~\bibnamefont {Matoussi}},\ }\bibfield  {title} {\bibinfo {title} {Empirical regression method for backward doubly stochastic differential equations},\ }\href@noop {} {\bibfield  {journal} {\bibinfo  {journal} {SIAM/ASA Journal on Uncertainty Quantification}\ }\textbf {\bibinfo {volume} {4}},\ \bibinfo {pages} {358} (\bibinfo {year} {2016})}\BibitemShut {NoStop}%
\bibitem [{\citenamefont {Bao}\ \emph {et~al.}(2016)\citenamefont {Bao}, \citenamefont {Cao}, \citenamefont {Meir},\ and\ \citenamefont {Zhao}}]{bao2016first}%
  \BibitemOpen
  \bibfield  {author} {\bibinfo {author} {\bibfnamefont {F.}~\bibnamefont {Bao}}, \bibinfo {author} {\bibfnamefont {Y.}~\bibnamefont {Cao}}, \bibinfo {author} {\bibfnamefont {A.}~\bibnamefont {Meir}},\ and\ \bibinfo {author} {\bibfnamefont {W.}~\bibnamefont {Zhao}},\ }\bibfield  {title} {\bibinfo {title} {A first order scheme for backward doubly stochastic differential equations},\ }\href@noop {} {\bibfield  {journal} {\bibinfo  {journal} {SIAM/ASA Journal on Uncertainty Quantification}\ }\textbf {\bibinfo {volume} {4}},\ \bibinfo {pages} {413} (\bibinfo {year} {2016})}\BibitemShut {NoStop}%
\bibitem [{\citenamefont {Giles}(2008)}]{giles2008multilevel}%
  \BibitemOpen
  \bibfield  {author} {\bibinfo {author} {\bibfnamefont {M.~B.}\ \bibnamefont {Giles}},\ }\bibfield  {title} {\bibinfo {title} {Multilevel {Monte Carlo} path simulation},\ }\href@noop {} {\bibfield  {journal} {\bibinfo  {journal} {Operations research}\ }\textbf {\bibinfo {volume} {56}},\ \bibinfo {pages} {607} (\bibinfo {year} {2008})}\BibitemShut {NoStop}%
\bibitem [{\citenamefont {Giles}(2009{\natexlab{a}})}]{giles2009multilevel}%
  \BibitemOpen
  \bibfield  {author} {\bibinfo {author} {\bibfnamefont {M.~B.}\ \bibnamefont {Giles}},\ }\bibfield  {title} {\bibinfo {title} {Multilevel {Monte Carlo} for basket options},\ }in\ \href@noop {} {\emph {\bibinfo {booktitle} {Proceedings of the 2009 Winter Simulation Conference (WSC)}}}\ (\bibinfo {organization} {IEEE},\ \bibinfo {year} {2009})\ pp.\ \bibinfo {pages} {1283--1290}\BibitemShut {NoStop}%
\bibitem [{\citenamefont {Burgos}\ and\ \citenamefont {Giles}(2012)}]{burgos2012computing}%
  \BibitemOpen
  \bibfield  {author} {\bibinfo {author} {\bibfnamefont {S.}~\bibnamefont {Burgos}}\ and\ \bibinfo {author} {\bibfnamefont {M.~B.}\ \bibnamefont {Giles}},\ }\bibfield  {title} {\bibinfo {title} {Computing {Greeks} using multilevel path simulation},\ }in\ \href@noop {} {\emph {\bibinfo {booktitle} {Monte Carlo and Quasi-Monte Carlo Methods 2010}}}\ (\bibinfo  {publisher} {Springer},\ \bibinfo {year} {2012})\ pp.\ \bibinfo {pages} {281--296}\BibitemShut {NoStop}%
\bibitem [{\citenamefont {Giles}\ and\ \citenamefont {Szpruch}(2018)}]{giles2018multilevel}%
  \BibitemOpen
  \bibfield  {author} {\bibinfo {author} {\bibfnamefont {M.~B.}\ \bibnamefont {Giles}}\ and\ \bibinfo {author} {\bibfnamefont {L.}~\bibnamefont {Szpruch}},\ }\bibfield  {title} {\bibinfo {title} {Multilevel {Monte Carlo} methods for applications in finance},\ }\href@noop {} {\bibfield  {journal} {\bibinfo  {journal} {High-Performance Computing in Finance}\ ,\ \bibinfo {pages} {197}} (\bibinfo {year} {2018})},\ \Eprint {https://arxiv.org/abs/arXiv:1212.1377} {arXiv:1212.1377} \BibitemShut {NoStop}%
\bibitem [{\citenamefont {Giles}\ and\ \citenamefont {Haji-Ali}(2019)}]{giles2019multilevel}%
  \BibitemOpen
  \bibfield  {author} {\bibinfo {author} {\bibfnamefont {M.~B.}\ \bibnamefont {Giles}}\ and\ \bibinfo {author} {\bibfnamefont {A.-L.}\ \bibnamefont {Haji-Ali}},\ }\bibfield  {title} {\bibinfo {title} {Multilevel nested simulation for efficient risk estimation},\ }\href@noop {} {\bibfield  {journal} {\bibinfo  {journal} {SIAM/ASA Journal on Uncertainty Quantification}\ }\textbf {\bibinfo {volume} {7}},\ \bibinfo {pages} {497} (\bibinfo {year} {2019})},\ \Eprint {https://arxiv.org/abs/arXiv:1802.05016} {arXiv:1802.05016} \BibitemShut {NoStop}%
\bibitem [{\citenamefont {Barth}\ \emph {et~al.}(2013)\citenamefont {Barth}, \citenamefont {Lang},\ and\ \citenamefont {Schwab}}]{barth2013multilevel}%
  \BibitemOpen
  \bibfield  {author} {\bibinfo {author} {\bibfnamefont {A.}~\bibnamefont {Barth}}, \bibinfo {author} {\bibfnamefont {A.}~\bibnamefont {Lang}},\ and\ \bibinfo {author} {\bibfnamefont {C.}~\bibnamefont {Schwab}},\ }\bibfield  {title} {\bibinfo {title} {Multilevel {Monte Carlo} method for parabolic stochastic partial differential equations},\ }\href@noop {} {\bibfield  {journal} {\bibinfo  {journal} {BIT Numerical Mathematics}\ }\textbf {\bibinfo {volume} {53}},\ \bibinfo {pages} {3} (\bibinfo {year} {2013})}\BibitemShut {NoStop}%
\bibitem [{\citenamefont {Iliev}\ \emph {et~al.}(2017)\citenamefont {Iliev}, \citenamefont {Mohring},\ and\ \citenamefont {Shegunov}}]{iliev2017renormalization}%
  \BibitemOpen
  \bibfield  {author} {\bibinfo {author} {\bibfnamefont {O.}~\bibnamefont {Iliev}}, \bibinfo {author} {\bibfnamefont {J.}~\bibnamefont {Mohring}},\ and\ \bibinfo {author} {\bibfnamefont {N.}~\bibnamefont {Shegunov}},\ }\bibfield  {title} {\bibinfo {title} {Renormalization based {MLMC} method for scalar elliptic{ SPDE}},\ }in\ \href@noop {} {\emph {\bibinfo {booktitle} {International Conference on Large-Scale Scientific Computing}}}\ (\bibinfo {organization} {Springer},\ \bibinfo {year} {2017})\ pp.\ \bibinfo {pages} {295--303}\BibitemShut {NoStop}%
\bibitem [{\citenamefont {Chada}\ \emph {et~al.}(2022)\citenamefont {Chada}, \citenamefont {Hoel}, \citenamefont {Jasra},\ and\ \citenamefont {Zouraris}}]{chada2022improved}%
  \BibitemOpen
  \bibfield  {author} {\bibinfo {author} {\bibfnamefont {N.~K.}\ \bibnamefont {Chada}}, \bibinfo {author} {\bibfnamefont {H.}~\bibnamefont {Hoel}}, \bibinfo {author} {\bibfnamefont {A.}~\bibnamefont {Jasra}},\ and\ \bibinfo {author} {\bibfnamefont {G.~E.}\ \bibnamefont {Zouraris}},\ }\bibfield  {title} {\bibinfo {title} {Improved efficiency of multilevel{ Monte Carlo} for stochastic {PDE }through strong pairwise coupling},\ }\href@noop {} {\bibfield  {journal} {\bibinfo  {journal} {Journal of Scientific Computing}\ }\textbf {\bibinfo {volume} {93}},\ \bibinfo {pages} {62} (\bibinfo {year} {2022})},\ \Eprint {https://arxiv.org/abs/arXiv:2108.00794} {arXiv:2108.00794} \BibitemShut {NoStop}%
\bibitem [{\citenamefont {Giles}\ and\ \citenamefont {Reisinger}(2012)}]{giles2012stochastic}%
  \BibitemOpen
  \bibfield  {author} {\bibinfo {author} {\bibfnamefont {M.~B.}\ \bibnamefont {Giles}}\ and\ \bibinfo {author} {\bibfnamefont {C.}~\bibnamefont {Reisinger}},\ }\bibfield  {title} {\bibinfo {title} {Stochastic finite differences and multilevel {Monte Carlo }for a class of {SPDEs} in finance},\ }\href@noop {} {\bibfield  {journal} {\bibinfo  {journal} {SIAM journal on financial mathematics}\ }\textbf {\bibinfo {volume} {3}},\ \bibinfo {pages} {572} (\bibinfo {year} {2012})},\ \Eprint {https://arxiv.org/abs/arXiv:1204.1442} {arXiv:1204.1442} \BibitemShut {NoStop}%
\bibitem [{\citenamefont {Brassard}\ \emph {et~al.}(2000)\citenamefont {Brassard}, \citenamefont {Hoyer}, \citenamefont {Mosca},\ and\ \citenamefont {Tapp}}]{brassard2000quantum}%
  \BibitemOpen
  \bibfield  {author} {\bibinfo {author} {\bibfnamefont {G.}~\bibnamefont {Brassard}}, \bibinfo {author} {\bibfnamefont {P.}~\bibnamefont {Hoyer}}, \bibinfo {author} {\bibfnamefont {M.}~\bibnamefont {Mosca}},\ and\ \bibinfo {author} {\bibfnamefont {A.}~\bibnamefont {Tapp}},\ }\bibfield  {title} {\bibinfo {title} {Quantum amplitude amplification and estimation},\ }\href {https://arxiv.org/abs/quant-ph/0005055} {\bibfield  {journal} {\bibinfo  {journal} {arXiv preprint quant-ph/0005055}\ } (\bibinfo {year} {2000})}\BibitemShut {NoStop}%
\bibitem [{\citenamefont {Heinrich}(2002)}]{heinrich2002quantum}%
  \BibitemOpen
  \bibfield  {author} {\bibinfo {author} {\bibfnamefont {S.}~\bibnamefont {Heinrich}},\ }\bibfield  {title} {\bibinfo {title} {Quantum summation with an application to integration},\ }\href@noop {} {\bibfield  {journal} {\bibinfo  {journal} {Journal of Complexity}\ }\textbf {\bibinfo {volume} {18}},\ \bibinfo {pages} {1} (\bibinfo {year} {2002})},\ \Eprint {https://arxiv.org/abs/arXiv:quant-ph/0105116} {arXiv:quant-ph/0105116} \BibitemShut {NoStop}%
\bibitem [{\citenamefont {Montanaro}(2015)}]{Montanaro2015}%
  \BibitemOpen
  \bibfield  {author} {\bibinfo {author} {\bibfnamefont {A.}~\bibnamefont {Montanaro}},\ }\bibfield  {title} {\bibinfo {title} {Quantum speedup of {M}onte {C}arlo methods},\ }\href@noop {} {\bibfield  {journal} {\bibinfo  {journal} {Proceedings of the Royal Society A: Mathematical, Physical and Engineering Sciences}\ }\textbf {\bibinfo {volume} {471}},\ \bibinfo {pages} {20150301} (\bibinfo {year} {2015})},\ \Eprint {https://arxiv.org/abs/arXiv:1504.06987} {arXiv:1504.06987} \BibitemShut {NoStop}%
\bibitem [{\citenamefont {Kothari}\ and\ \citenamefont {O'Donnell}(2023)}]{Kothari2022MeanEW}%
  \BibitemOpen
  \bibfield  {author} {\bibinfo {author} {\bibfnamefont {R.}~\bibnamefont {Kothari}}\ and\ \bibinfo {author} {\bibfnamefont {R.}~\bibnamefont {O'Donnell}},\ }\bibfield  {title} {\bibinfo {title} {Mean estimation when you have the source code; or, quantum {Monte Carlo} methods},\ }in\ \href@noop {} {\emph {\bibinfo {booktitle} {Proceedings of the 2023 Annual ACM-SIAM Symposium on Discrete Algorithms (SODA)}}}\ (\bibinfo {organization} {SIAM},\ \bibinfo {year} {2023})\ pp.\ \bibinfo {pages} {1186--1215},\ \Eprint {https://arxiv.org/abs/arXiv:2208.07544} {arXiv:2208.07544} \BibitemShut {NoStop}%
\bibitem [{\citenamefont {An}\ \emph {et~al.}(2021)\citenamefont {An}, \citenamefont {Linden}, \citenamefont {Liu}, \citenamefont {Montanaro}, \citenamefont {Shao},\ and\ \citenamefont {Wang}}]{An2021quantumaccelerated}%
  \BibitemOpen
  \bibfield  {author} {\bibinfo {author} {\bibfnamefont {D.}~\bibnamefont {An}}, \bibinfo {author} {\bibfnamefont {N.}~\bibnamefont {Linden}}, \bibinfo {author} {\bibfnamefont {J.-P.}\ \bibnamefont {Liu}}, \bibinfo {author} {\bibfnamefont {A.}~\bibnamefont {Montanaro}}, \bibinfo {author} {\bibfnamefont {C.}~\bibnamefont {Shao}},\ and\ \bibinfo {author} {\bibfnamefont {J.}~\bibnamefont {Wang}},\ }\bibfield  {title} {\bibinfo {title} {Quantum-accelerated multilevel {Monte Carlo} methods for stochastic differential equations in mathematical finance},\ }\href@noop {} {\bibfield  {journal} {\bibinfo  {journal} {Quantum}\ }\textbf {\bibinfo {volume} {5}},\ \bibinfo {pages} {481} (\bibinfo {year} {2021})},\ \Eprint {https://arxiv.org/abs/arXiv:2012.06283} {arXiv:2012.06283} \BibitemShut {NoStop}%
\bibitem [{\citenamefont {Blanchet}\ \emph {et~al.}(2024)\citenamefont {Blanchet}, \citenamefont {Szegedy},\ and\ \citenamefont {Wang}}]{blanchet2024quadratic}%
  \BibitemOpen
  \bibfield  {author} {\bibinfo {author} {\bibfnamefont {J.}~\bibnamefont {Blanchet}}, \bibinfo {author} {\bibfnamefont {M.}~\bibnamefont {Szegedy}},\ and\ \bibinfo {author} {\bibfnamefont {G.}~\bibnamefont {Wang}},\ }\bibfield  {title} {\bibinfo {title} {Quadratic speed-up in infinite variance quantum {Monte Carlo}},\ }\href {https://arxiv.org/abs/2401.07497} {\bibfield  {journal} {\bibinfo  {journal} {arXiv preprint arXiv:2401.07497}\ } (\bibinfo {year} {2024})}\BibitemShut {NoStop}%
\bibitem [{\citenamefont {Blanchet}\ \emph {et~al.}(2026)\citenamefont {Blanchet}, \citenamefont {Hamoudi}, \citenamefont {Szegedy},\ and\ \citenamefont {Wang}}]{Blanchet2025NonlinearQM}%
  \BibitemOpen
  \bibfield  {author} {\bibinfo {author} {\bibfnamefont {J.}~\bibnamefont {Blanchet}}, \bibinfo {author} {\bibfnamefont {Y.}~\bibnamefont {Hamoudi}}, \bibinfo {author} {\bibfnamefont {M.}~\bibnamefont {Szegedy}},\ and\ \bibinfo {author} {\bibfnamefont {G.}~\bibnamefont {Wang}},\ }\bibfield  {title} {\bibinfo {title} {Quantum speedup of non-linear{ Monte Carlo} problems},\ }\href@noop {} {\bibfield  {journal} {\bibinfo  {journal} {Advances in Neural Information Processing Systems}\ }\textbf {\bibinfo {volume} {38}},\ \bibinfo {pages} {18736} (\bibinfo {year} {2026})},\ \Eprint {https://arxiv.org/abs/arXiv:2502.05094} {arXiv:2502.05094} \BibitemShut {NoStop}%
\bibitem [{\citenamefont {Ozgul}\ \emph {et~al.}(2025)\citenamefont {Ozgul}, \citenamefont {Li}, \citenamefont {Mahdavi},\ and\ \citenamefont {Wang}}]{ozgul2025quantum}%
  \BibitemOpen
  \bibfield  {author} {\bibinfo {author} {\bibfnamefont {G.}~\bibnamefont {Ozgul}}, \bibinfo {author} {\bibfnamefont {X.}~\bibnamefont {Li}}, \bibinfo {author} {\bibfnamefont {M.}~\bibnamefont {Mahdavi}},\ and\ \bibinfo {author} {\bibfnamefont {C.}~\bibnamefont {Wang}},\ }\bibfield  {title} {\bibinfo {title} {Quantum speedups for markov chain {Monte Carlo} methods with application to optimization},\ }\href {https://arxiv.org/abs/2504.03626} {\bibfield  {journal} {\bibinfo  {journal} {arXiv preprint arXiv:2504.03626}\ } (\bibinfo {year} {2025})}\BibitemShut {NoStop}%
\bibitem [{\citenamefont {Li}\ and\ \citenamefont {Liu}(2026)}]{li2026quantum}%
  \BibitemOpen
  \bibfield  {author} {\bibinfo {author} {\bibfnamefont {X.}~\bibnamefont {Li}}\ and\ \bibinfo {author} {\bibfnamefont {J.-P.}\ \bibnamefont {Liu}},\ }\bibfield  {title} {\bibinfo {title} {Quantum algorithms for {Gibbs} expectation of non-log-concave and heavy-tailed distributions},\ }\href {https://arxiv.org/abs/2604.00656} {\bibfield  {journal} {\bibinfo  {journal} {arXiv preprint arXiv:2604.00656}\ } (\bibinfo {year} {2026})}\BibitemShut {NoStop}%
\bibitem [{\citenamefont {Rebentrost}\ \emph {et~al.}(2018)\citenamefont {Rebentrost}, \citenamefont {Gupt},\ and\ \citenamefont {Bromley}}]{rebentrost2018quantum}%
  \BibitemOpen
  \bibfield  {author} {\bibinfo {author} {\bibfnamefont {P.}~\bibnamefont {Rebentrost}}, \bibinfo {author} {\bibfnamefont {B.}~\bibnamefont {Gupt}},\ and\ \bibinfo {author} {\bibfnamefont {T.~R.}\ \bibnamefont {Bromley}},\ }\bibfield  {title} {\bibinfo {title} {Quantum computational finance: {Monte Carlo} pricing of financial derivatives},\ }\href@noop {} {\bibfield  {journal} {\bibinfo  {journal} {Physical Review A}\ }\textbf {\bibinfo {volume} {98}},\ \bibinfo {pages} {022321} (\bibinfo {year} {2018})},\ \Eprint {https://arxiv.org/abs/arXiv:1805.00109} {arXiv:1805.00109} \BibitemShut {NoStop}%
\bibitem [{\citenamefont {Stamatopoulos}\ \emph {et~al.}(2020)\citenamefont {Stamatopoulos}, \citenamefont {Egger}, \citenamefont {Sun}, \citenamefont {Zoufal}, \citenamefont {Iten}, \citenamefont {Shen},\ and\ \citenamefont {Woerner}}]{stamatopoulos2020option}%
  \BibitemOpen
  \bibfield  {author} {\bibinfo {author} {\bibfnamefont {N.}~\bibnamefont {Stamatopoulos}}, \bibinfo {author} {\bibfnamefont {D.~J.}\ \bibnamefont {Egger}}, \bibinfo {author} {\bibfnamefont {Y.}~\bibnamefont {Sun}}, \bibinfo {author} {\bibfnamefont {C.}~\bibnamefont {Zoufal}}, \bibinfo {author} {\bibfnamefont {R.}~\bibnamefont {Iten}}, \bibinfo {author} {\bibfnamefont {N.}~\bibnamefont {Shen}},\ and\ \bibinfo {author} {\bibfnamefont {S.}~\bibnamefont {Woerner}},\ }\bibfield  {title} {\bibinfo {title} {Option pricing using quantum computers},\ }\href@noop {} {\bibfield  {journal} {\bibinfo  {journal} {Quantum}\ }\textbf {\bibinfo {volume} {4}},\ \bibinfo {pages} {291} (\bibinfo {year} {2020})},\ \Eprint {https://arxiv.org/abs/arXiv:1905.02666} {arXiv:1905.02666} \BibitemShut {NoStop}%
\bibitem [{\citenamefont {Herman}\ \emph {et~al.}(2026)\citenamefont {Herman}, \citenamefont {Sun}, \citenamefont {Liu}, \citenamefont {Pistoia}, \citenamefont {Che}, \citenamefont {Otter}, \citenamefont {Chakrabarti},\ and\ \citenamefont {Harrow}}]{herman2026quantum}%
  \BibitemOpen
  \bibfield  {author} {\bibinfo {author} {\bibfnamefont {D.}~\bibnamefont {Herman}}, \bibinfo {author} {\bibfnamefont {Y.}~\bibnamefont {Sun}}, \bibinfo {author} {\bibfnamefont {J.-P.}\ \bibnamefont {Liu}}, \bibinfo {author} {\bibfnamefont {M.}~\bibnamefont {Pistoia}}, \bibinfo {author} {\bibfnamefont {C.}~\bibnamefont {Che}}, \bibinfo {author} {\bibfnamefont {R.}~\bibnamefont {Otter}}, \bibinfo {author} {\bibfnamefont {S.}~\bibnamefont {Chakrabarti}},\ and\ \bibinfo {author} {\bibfnamefont {A.}~\bibnamefont {Harrow}},\ }\bibfield  {title} {\bibinfo {title} {Quantum speedups for derivative pricing beyond {Black-Scholes}},\ }\href {https://arxiv.org/abs/2602.03725} {\bibfield  {journal} {\bibinfo  {journal} {arXiv preprint arXiv:2602.03725}\ } (\bibinfo {year} {2026})}\BibitemShut {NoStop}%
\bibitem [{\citenamefont {Guseynov}\ \emph {et~al.}(2026)\citenamefont {Guseynov}, \citenamefont {Liu}, \citenamefont {Pun},\ and\ \citenamefont {Vaidya}}]{guseynov2026end}%
  \BibitemOpen
  \bibfield  {author} {\bibinfo {author} {\bibfnamefont {N.}~\bibnamefont {Guseynov}}, \bibinfo {author} {\bibfnamefont {N.}~\bibnamefont {Liu}}, \bibinfo {author} {\bibfnamefont {C.~S.}\ \bibnamefont {Pun}},\ and\ \bibinfo {author} {\bibfnamefont {T.}~\bibnamefont {Vaidya}},\ }\bibfield  {title} {\bibinfo {title} {End-to-end {PDE}-based quantum algorithms for multi-asset option pricing under local and stochastic volatility},\ }\href {https://arxiv.org/abs/2605.26610} {\bibfield  {journal} {\bibinfo  {journal} {arXiv preprint arXiv:2605.26610}\ } (\bibinfo {year} {2026})}\BibitemShut {NoStop}%
\bibitem [{\citenamefont {Fujita}\ \emph {et~al.}(2024)\citenamefont {Fujita}, \citenamefont {Miyamoto},\ and\ \citenamefont {Sekine}}]{fujita2024application}%
  \BibitemOpen
  \bibfield  {author} {\bibinfo {author} {\bibfnamefont {M.}~\bibnamefont {Fujita}}, \bibinfo {author} {\bibfnamefont {K.}~\bibnamefont {Miyamoto}},\ and\ \bibinfo {author} {\bibfnamefont {J.}~\bibnamefont {Sekine}},\ }\bibfield  {title} {\bibinfo {title} {Application of quantum {M}onte {C}arlo integration to {M}arkovian backward stochastic differential equations},\ }\href@noop {} {\bibfield  {journal} {\bibinfo  {journal} {JSIAM Letters}\ }\textbf {\bibinfo {volume} {16}},\ \bibinfo {pages} {105} (\bibinfo {year} {2024})}\BibitemShut {NoStop}%
\bibitem [{\citenamefont {Jin}\ \emph {et~al.}(2025)\citenamefont {Jin}, \citenamefont {Liu},\ and\ \citenamefont {Wei}}]{jin2025quantum}%
  \BibitemOpen
  \bibfield  {author} {\bibinfo {author} {\bibfnamefont {S.}~\bibnamefont {Jin}}, \bibinfo {author} {\bibfnamefont {N.}~\bibnamefont {Liu}},\ and\ \bibinfo {author} {\bibfnamefont {W.}~\bibnamefont {Wei}},\ }\bibfield  {title} {\bibinfo {title} {Quantum algorithms for stochastic differential equations: A {S}chr{\"o}dingerisation approach},\ }\href@noop {} {\bibfield  {journal} {\bibinfo  {journal} {Journal of Scientific Computing}\ }\textbf {\bibinfo {volume} {104}},\ \bibinfo {pages} {56} (\bibinfo {year} {2025})},\ \Eprint {https://arxiv.org/abs/arXiv:2412.14868} {arXiv:2412.14868} \BibitemShut {NoStop}%
\bibitem [{\citenamefont {Bravyi}\ \emph {et~al.}(2025)\citenamefont {Bravyi}, \citenamefont {Manson-Sawko}, \citenamefont {Zayats},\ and\ \citenamefont {Zhuk}}]{bravyi2025quantum}%
  \BibitemOpen
  \bibfield  {author} {\bibinfo {author} {\bibfnamefont {S.}~\bibnamefont {Bravyi}}, \bibinfo {author} {\bibfnamefont {R.}~\bibnamefont {Manson-Sawko}}, \bibinfo {author} {\bibfnamefont {M.}~\bibnamefont {Zayats}},\ and\ \bibinfo {author} {\bibfnamefont {S.}~\bibnamefont {Zhuk}},\ }\bibfield  {title} {\bibinfo {title} {Quantum simulation of a noisy classical nonlinear dynamics},\ }\href {https://arxiv.org/abs/2507.06198} {\bibfield  {journal} {\bibinfo  {journal} {arXiv preprint arXiv:2507.06198}\ } (\bibinfo {year} {2025})}\BibitemShut {NoStop}%
\bibitem [{\citenamefont {Yang}\ and\ \citenamefont {Liu}(2025)}]{yang2025circuitefficientrandomizedquantumsimulation}%
  \BibitemOpen
  \bibfield  {author} {\bibinfo {author} {\bibfnamefont {S.}~\bibnamefont {Yang}}\ and\ \bibinfo {author} {\bibfnamefont {J.-P.}\ \bibnamefont {Liu}},\ }\bibfield  {title} {\bibinfo {title} {Circuit-efficient randomized quantum simulation of non-unitary dynamics with observable-driven and symmetry-aware designs},\ }\href {https://arxiv.org/abs/2509.08030} {\bibfield  {journal} {\bibinfo  {journal} {arXiv preprint arXiv:2509.08030}\ } (\bibinfo {year} {2025})}\BibitemShut {NoStop}%
\bibitem [{\citenamefont {Li}\ \emph {et~al.}(2026)\citenamefont {Li}, \citenamefont {Catli}, \citenamefont {Lim}, \citenamefont {Pocrnic}, \citenamefont {An}, \citenamefont {Liu},\ and\ \citenamefont {Wiebe}}]{li2026efficientquantumsimulationnonlinear}%
  \BibitemOpen
  \bibfield  {author} {\bibinfo {author} {\bibfnamefont {X.}~\bibnamefont {Li}}, \bibinfo {author} {\bibfnamefont {A.~B.}\ \bibnamefont {Catli}}, \bibinfo {author} {\bibfnamefont {H.~K.}\ \bibnamefont {Lim}}, \bibinfo {author} {\bibfnamefont {M.}~\bibnamefont {Pocrnic}}, \bibinfo {author} {\bibfnamefont {D.}~\bibnamefont {An}}, \bibinfo {author} {\bibfnamefont {J.-P.}\ \bibnamefont {Liu}},\ and\ \bibinfo {author} {\bibfnamefont {N.}~\bibnamefont {Wiebe}},\ }\bibfield  {title} {\bibinfo {title} {Efficient quantum simulation for nonlinear stochastic differential equations},\ }\href {https://arxiv.org/abs/2603.12398} {\bibfield  {journal} {\bibinfo  {journal} {arXiv preprint arXiv:2603.12398}\ } (\bibinfo {year} {2026})}\BibitemShut {NoStop}%
\bibitem [{\citenamefont {Bravyi}\ \emph {et~al.}(2026)\citenamefont {Bravyi}, \citenamefont {Byrne}, \citenamefont {Zayats},\ and\ \citenamefont {Zhuk}}]{bravyi2026quantumalgorithmsstochasticnonlinear}%
  \BibitemOpen
  \bibfield  {author} {\bibinfo {author} {\bibfnamefont {S.}~\bibnamefont {Bravyi}}, \bibinfo {author} {\bibfnamefont {A.}~\bibnamefont {Byrne}}, \bibinfo {author} {\bibfnamefont {M.}~\bibnamefont {Zayats}},\ and\ \bibinfo {author} {\bibfnamefont {S.}~\bibnamefont {Zhuk}},\ }\bibfield  {title} {\bibinfo {title} {Quantum algorithms for stochastic nonlinear differential equations},\ }\href {https://arxiv.org/abs/2606.08349} {\bibfield  {journal} {\bibinfo  {journal} {arXiv preprint arXiv:2606.08349}\ } (\bibinfo {year} {2026})}\BibitemShut {NoStop}%
\bibitem [{\citenamefont {Sun}\ \emph {et~al.}(2026)\citenamefont {Sun}, \citenamefont {Wang},\ and\ \citenamefont {Blanchet}}]{sun2026optimalquantumspeedupsrepeatedly}%
  \BibitemOpen
  \bibfield  {author} {\bibinfo {author} {\bibfnamefont {Y.}~\bibnamefont {Sun}}, \bibinfo {author} {\bibfnamefont {G.}~\bibnamefont {Wang}},\ and\ \bibinfo {author} {\bibfnamefont {J.}~\bibnamefont {Blanchet}},\ }\href@noop {} {\bibinfo {title} {Optimal quantum speedups for repeatedly nested expectation estimation}} (\bibinfo {year} {2026}),\ \Eprint {https://arxiv.org/abs/2602.08120} {arXiv:2602.08120 [quant-ph]} \BibitemShut {NoStop}%
\bibitem [{\citenamefont {Bally}\ and\ \citenamefont {Matoussi}(2001)}]{BallyMatoussi2001}%
  \BibitemOpen
  \bibfield  {author} {\bibinfo {author} {\bibfnamefont {V.}~\bibnamefont {Bally}}\ and\ \bibinfo {author} {\bibfnamefont {A.}~\bibnamefont {Matoussi}},\ }\bibfield  {title} {\bibinfo {title} {Weak solutions for {SPDEs} and backward doubly stochastic differential equations},\ }\href {https://doi.org/10.1023/A:1007825232513} {\bibfield  {journal} {\bibinfo  {journal} {Journal of Theoretical Probability}\ }\textbf {\bibinfo {volume} {14}},\ \bibinfo {pages} {125} (\bibinfo {year} {2001})}\BibitemShut {NoStop}%
\bibitem [{\citenamefont {Merton}(1973)}]{Merton1973RationalOptionPricing}%
  \BibitemOpen
  \bibfield  {author} {\bibinfo {author} {\bibfnamefont {R.~C.}\ \bibnamefont {Merton}},\ }\bibfield  {title} {\bibinfo {title} {Theory of rational option pricing},\ }\href {https://doi.org/10.2307/3003143} {\bibfield  {journal} {\bibinfo  {journal} {The Bell Journal of Economics and Management Science}\ }\textbf {\bibinfo {volume} {4}},\ \bibinfo {pages} {141} (\bibinfo {year} {1973})}\BibitemShut {NoStop}%
\bibitem [{\citenamefont {Giles}(2015)}]{giles2015multilevel}%
  \BibitemOpen
  \bibfield  {author} {\bibinfo {author} {\bibfnamefont {M.~B.}\ \bibnamefont {Giles}},\ }\bibfield  {title} {\bibinfo {title} {Multilevel {Monte Carlo} methods},\ }\href@noop {} {\bibfield  {journal} {\bibinfo  {journal} {Acta numerica}\ }\textbf {\bibinfo {volume} {24}},\ \bibinfo {pages} {259} (\bibinfo {year} {2015})},\ \Eprint {https://arxiv.org/abs/arXiv:1304.5472} {arXiv:1304.5472} \BibitemShut {NoStop}%
\bibitem [{\citenamefont {Jerrum}\ \emph {et~al.}(1986)\citenamefont {Jerrum}, \citenamefont {Valiant},\ and\ \citenamefont {Vazirani}}]{Jerrum1986RandomGO}%
  \BibitemOpen
  \bibfield  {author} {\bibinfo {author} {\bibfnamefont {M.~R.}\ \bibnamefont {Jerrum}}, \bibinfo {author} {\bibfnamefont {L.~G.}\ \bibnamefont {Valiant}},\ and\ \bibinfo {author} {\bibfnamefont {V.~V.}\ \bibnamefont {Vazirani}},\ }\bibfield  {title} {\bibinfo {title} {Random generation of combinatorial structures from a uniform distribution},\ }\href@noop {} {\bibfield  {journal} {\bibinfo  {journal} {Theoretical computer science}\ }\textbf {\bibinfo {volume} {43}},\ \bibinfo {pages} {169} (\bibinfo {year} {1986})}\BibitemShut {NoStop}%
\bibitem [{\citenamefont {Broadie}\ and\ \citenamefont {Glasserman}(1996)}]{BroadieGlasserman1996}%
  \BibitemOpen
  \bibfield  {author} {\bibinfo {author} {\bibfnamefont {M.}~\bibnamefont {Broadie}}\ and\ \bibinfo {author} {\bibfnamefont {P.}~\bibnamefont {Glasserman}},\ }\bibfield  {title} {\bibinfo {title} {Estimating security price derivatives using simulation},\ }\href {https://doi.org/10.1287/mnsc.42.2.269} {\bibfield  {journal} {\bibinfo  {journal} {Management Science}\ }\textbf {\bibinfo {volume} {42}},\ \bibinfo {pages} {269} (\bibinfo {year} {1996})}\BibitemShut {NoStop}%
\bibitem [{\citenamefont {Fourni{\'e}}\ \emph {et~al.}(1999)\citenamefont {Fourni{\'e}}, \citenamefont {Lasry}, \citenamefont {Lebuchoux}, \citenamefont {Lions},\ and\ \citenamefont {Touzi}}]{Fournie1999}%
  \BibitemOpen
  \bibfield  {author} {\bibinfo {author} {\bibfnamefont {E.}~\bibnamefont {Fourni{\'e}}}, \bibinfo {author} {\bibfnamefont {J.-M.}\ \bibnamefont {Lasry}}, \bibinfo {author} {\bibfnamefont {J.}~\bibnamefont {Lebuchoux}}, \bibinfo {author} {\bibfnamefont {P.-L.}\ \bibnamefont {Lions}},\ and\ \bibinfo {author} {\bibfnamefont {N.}~\bibnamefont {Touzi}},\ }\bibfield  {title} {\bibinfo {title} {Applications of malliavin calculus to monte carlo methods in finance},\ }\href {https://doi.org/10.1007/s007800050068} {\bibfield  {journal} {\bibinfo  {journal} {Finance and Stochastics}\ }\textbf {\bibinfo {volume} {3}},\ \bibinfo {pages} {391} (\bibinfo {year} {1999})}\BibitemShut {NoStop}%
\bibitem [{\citenamefont {Giles}(2009{\natexlab{b}})}]{Giles2009Vibrato}%
  \BibitemOpen
  \bibfield  {author} {\bibinfo {author} {\bibfnamefont {M.~B.}\ \bibnamefont {Giles}},\ }\bibfield  {title} {\bibinfo {title} {Vibrato monte carlo sensitivities},\ }in\ \href {https://doi.org/10.1007/978-3-642-04107-5_23} {\emph {\bibinfo {booktitle} {Monte Carlo and Quasi-Monte Carlo Methods 2008}}}\ (\bibinfo  {publisher} {Springer},\ \bibinfo {year} {2009})\ pp.\ \bibinfo {pages} {369--382}\BibitemShut {NoStop}%
\bibitem [{\citenamefont {Bismut}(1984)}]{Bismut1984}%
  \BibitemOpen
  \bibfield  {author} {\bibinfo {author} {\bibfnamefont {J.-M.}\ \bibnamefont {Bismut}},\ }\href@noop {} {\emph {\bibinfo {title} {Large Deviations and the Malliavin Calculus}}},\ \bibinfo {series} {Progress in Mathematics}, Vol.~\bibinfo {volume} {45}\ (\bibinfo  {publisher} {Birkh{\"a}user Boston},\ \bibinfo {year} {1984})\BibitemShut {NoStop}%
\bibitem [{\citenamefont {Elworthy}\ and\ \citenamefont {Li}(1994)}]{ElworthyLi1994}%
  \BibitemOpen
  \bibfield  {author} {\bibinfo {author} {\bibfnamefont {K.~D.}\ \bibnamefont {Elworthy}}\ and\ \bibinfo {author} {\bibfnamefont {X.-M.}\ \bibnamefont {Li}},\ }\bibfield  {title} {\bibinfo {title} {Formulae for the derivatives of heat semigroups},\ }\href {https://doi.org/10.1006/jfan.1994.1124} {\bibfield  {journal} {\bibinfo  {journal} {Journal of Functional Analysis}\ }\textbf {\bibinfo {volume} {125}},\ \bibinfo {pages} {252} (\bibinfo {year} {1994})}\BibitemShut {NoStop}%
\bibitem [{\citenamefont {Fourni{\'e}}\ \emph {et~al.}(2001)\citenamefont {Fourni{\'e}}, \citenamefont {Lasry}, \citenamefont {Lebuchoux},\ and\ \citenamefont {Lions}}]{Fournie2001}%
  \BibitemOpen
  \bibfield  {author} {\bibinfo {author} {\bibfnamefont {E.}~\bibnamefont {Fourni{\'e}}}, \bibinfo {author} {\bibfnamefont {J.-M.}\ \bibnamefont {Lasry}}, \bibinfo {author} {\bibfnamefont {J.}~\bibnamefont {Lebuchoux}},\ and\ \bibinfo {author} {\bibfnamefont {P.-L.}\ \bibnamefont {Lions}},\ }\bibfield  {title} {\bibinfo {title} {Applications of malliavin calculus to monte carlo methods in finance. {II}},\ }\href {https://doi.org/10.1007/PL00013529} {\bibfield  {journal} {\bibinfo  {journal} {Finance and Stochastics}\ }\textbf {\bibinfo {volume} {5}},\ \bibinfo {pages} {201} (\bibinfo {year} {2001})}\BibitemShut {NoStop}%
\bibitem [{\citenamefont {Gobet}\ and\ \citenamefont {Kohatsu-Higa}(2003)}]{GobetKohatsuHiga2003}%
  \BibitemOpen
  \bibfield  {author} {\bibinfo {author} {\bibfnamefont {E.}~\bibnamefont {Gobet}}\ and\ \bibinfo {author} {\bibfnamefont {A.}~\bibnamefont {Kohatsu-Higa}},\ }\bibfield  {title} {\bibinfo {title} {Computation of greeks for barrier and look-back options using malliavin calculus},\ }\href@noop {} {\bibfield  {journal} {\bibinfo  {journal} {Electronic Communications in Probability}\ }\textbf {\bibinfo {volume} {8}},\ \bibinfo {pages} {51} (\bibinfo {year} {2003})}\BibitemShut {NoStop}%
\bibitem [{\citenamefont {Lord}\ \emph {et~al.}(2010)\citenamefont {Lord}, \citenamefont {Koekkoek},\ and\ \citenamefont {van Dijk}}]{LordKoekkoekVanDijk2010}%
  \BibitemOpen
  \bibfield  {author} {\bibinfo {author} {\bibfnamefont {R.}~\bibnamefont {Lord}}, \bibinfo {author} {\bibfnamefont {R.}~\bibnamefont {Koekkoek}},\ and\ \bibinfo {author} {\bibfnamefont {D.}~\bibnamefont {van Dijk}},\ }\bibfield  {title} {\bibinfo {title} {A comparison of biased simulation schemes for stochastic volatility models},\ }\href {https://doi.org/10.1080/14697680802392496} {\bibfield  {journal} {\bibinfo  {journal} {Quantitative Finance}\ }\textbf {\bibinfo {volume} {10}},\ \bibinfo {pages} {177} (\bibinfo {year} {2010})}\BibitemShut {NoStop}%
\bibitem [{\citenamefont {Kac}(1949)}]{Kac1949OnDO}%
  \BibitemOpen
  \bibfield  {author} {\bibinfo {author} {\bibfnamefont {M.}~\bibnamefont {Kac}},\ }\bibfield  {title} {\bibinfo {title} {On distributions of certain wiener functionals},\ }\href@noop {} {\bibfield  {journal} {\bibinfo  {journal} {Transactions of the American Mathematical Society}\ }\textbf {\bibinfo {volume} {65}},\ \bibinfo {pages} {1} (\bibinfo {year} {1949})}\BibitemShut {NoStop}%
\end{thebibliography}%

\onecolumngrid
\clearpage

\appendix
\begin{center}
    \textbf{Supplementary Materials}
\end{center}

\section{Terminal-value SPDE and Its BDSDE Representation}\label{sec:SPDE and BDSDE}

It is well known that the Feynman-Kac formula provides a classical probabilistic representation
for solutions of linear parabolic PDEs~\cite{Kac1949OnDO}.
Given a terminal-value problem of the form
$$
\partial_t u(t,x) + \mathcal{L} u(t,x) - r(t,x)u(t,x) + F(t,x) = 0,
\quad u(T,x)=G(x),
$$
the Feynman--Kac formula represents the solution $u(t,x)$
as the conditional expectation of a functional of the solution to an associated SDE.
More precisely, $u(t,x)$ can be written as
$$
u(t,x)=\mathbb{E}\left[\Phi\left(X^{t,x}_{T}\right)\right],
$$
where $X^{t,x}$ solves the SDE whose infinitesimal generator coincides with
the differential operator $\mathcal{L}$.

An important extension of the Feynman-Kac framework was achieved through the theory of
backward stochastic differential equations (BSDEs).
The seminal work of Pardoux and Peng established that BSDEs provide probabilistic
representations for semilinear terminal-value PDEs~\cite{Pardoux1990AdaptedSO,Pardoux1992BackwardSD}.
In particular, \cite{Karoui1997BackwardSD} discusses applications of BSDEs in mathematical
finance, including the pricing of European options.

However, When the evolution equation contains an additional noise term of the form
$G(t,x)\mathrm{d}W_t$, the solution becomes a stochastic process in both time and space, leading to
a stochastic partial differential equation.
In this setting, neither the classical Feynman--Kac formula nor the BSDE framework is
applicable.

Backward doubly stochastic differential equations (BDSDEs)
~\cite{Pardoux1994BackwardDS,BallyMatoussi2001}
were introduced precisely to overcome
this limitation.
By incorporating an additional backward It\^o integral with respect to a time-reversed Brownian
motion, BDSDEs extend the BSDE framework and provide a rigorous probabilistic representation for
terminal-value SPDEs.

In contrast to classical BSDEs, which are associated with deterministic or random-coefficient
partial differential equations, BDSDEs involve two sources of randomness: a forward Brownian
motion and a backward stochastic integral.
This additional backward component allows BDSDEs to faithfully capture the intrinsic randomness
of SPDE solutions while preserving the backward-in-time structure imposed by the terminal
condition.

Our goal is to give a probabilistic representation for the solution of
the quasilinear backward SPDEs:
\begin{equation}\label{eq:SPDE(with BDSDE)}
  \begin{cases}
\partial_t u(t,x)
=
\left[\mathcal{L}u(t,x)
+ f\left(t,x,u(t,x),(\sigma^\top\nabla u)(t,x)\right)\right]\,\mathrm{d}t
+ h(t,x,u(t,x),(\sigma^\top\nabla u)(t,x))\,\mathrm{d}B_t,\quad t\in[0,T],\\
u(T,x)=G(x).
\end{cases}
\end{equation}
There $u:\mathbb{R_+}\times \mathbb{R}^d\rightarrow \mathbb{R}^k$ and
  $\mathcal{L}u=\left( Lu_1,\cdots ,Lu_k \right)^\top$
with $L=\frac{1}{2}\sum\limits_{i,j=1}^d(\sigma\sigma^\top)_{ij}\frac{\partial^2}{\partial x_i\partial x_j}+\sum\limits_{i=1}^db_i\frac{\partial}{\partial x_i}.$
Alternatively, we can write the SPDE~\eqref{eq:SPDE(with BDSDE)} in the integral form
\begin{align}\label{eq:SPDE(with BDSDE)-integral}
  u(t,x)=G(x)+\int_t^T \left[\mathcal{L}u(s,x)
+ f\left(s,x,u(s,x),(\sigma^\top\nabla u)(s,x)\right)\right]\,\mathrm{d}s
+ \int_t^T h(s,x,u(s,x),(\sigma^\top\nabla u)(s,x))\mathrm d \overleftarrow{B}_s.
\end{align}
The backward It\^o integral with respect to $\mathrm d \overleftarrow{B}_s$ is defined as
\begin{align*}
  \int_t^T \phi(s) \mathrm d \overleftarrow{B}_s
:= \int_0^{T-t} \phi_{T-s}\,\mathrm{d}B_s',
\end{align*}
where $B_s' = B_T - B_{T-s}$ is the time-reversed Brownian motion.
Equivalently, for any partition $t=t_0<\cdots<t_n=T$,
\begin{align*}
  \int_t^T \phi(s)\mathrm d \overleftarrow{B}_s
= \lim_{|\Pi|\to 0}\sum_{i=0}^{n-1}
\phi_{t_{i+1}}\left(B_{t_i}-B_{t_{i+1}}\right),
\end{align*}
with convergence in $L^2(\Omega)$.

The connection with BDSDEs is obtained as follows.
For each $(t,x)\in\mathbb{R}_+\times \mathbb{R}^d$, let
$\{X_s^{t,x};t\leq s\leq T\}$ be the solution of the SDE:
\begin{equation}\label{eq:SDE(with BDSDE)}
\begin{cases}
\mathrm{d}X_s^{t,x}=b(X_s^{t,x})\,\mathrm{d}s
+\sigma(X_s^{t,x})\,\mathrm{d}W_s,\quad s\in[t,T],\\
X_t^{t,x}=x.
\end{cases}
\end{equation}
Assume that the SPDE~\eqref{eq:SPDE(with BDSDE)} has a classical solution.
Then the couple $\left(Y_s^{t,x},Z_s^{t,x}\right)$ where 
\begin{align*}
  Y_s^{t,x}=u(s,X_s^{t,x}),\quad Z_s^{t,x}=(\sigma^\top\nabla u)(s,X_s^{t,x})
\end{align*}
verify the following BDSDE:
\begin{align}\label{eq:BDSDE}
  Y_s^{t,x}=G(X_T^{t,x})+\int_s^T f(r,X_r^{t,x},Y_r^{t,x},Z_r^{t,x})\,\mathrm{d}r 
  +\int_s^T h(r,X_r^{t,x},Y_r^{t,x},Z_r^{t,x})\mathrm d \overleftarrow{B}_r-\int_s^T Z_r^{t,x}\,\mathrm{d}W_r,\quad t\leq s\leq T, 
\end{align}
or alternatively
\begin{equation}\label{eq:BDSDE-d}
\begin{cases}
\mathrm{d}Y_s^{t,x}=
- f(s,X_s^{t,x},Y_s^{t,x},Z_s^{t,x})\,\mathrm{d}s
- h(s,X_s^{t,x},Y_s^{t,x},Z_s^{t,x})\mathrm d \overleftarrow{B}_s
+ Z_s^{t,x}\,\mathrm{d}W_s,
\quad s\in[t,T],\\
Y_T^{t,x}=G(X_T^{t,x}).
\end{cases}
\end{equation}

In the following, 
based on~\cite{Pardoux1994BackwardDS,BallyMatoussi2001},
we provide a more detailed exposition of results related to BDSDEs.

\paragraph{Fundamentals of BDSDEs}

Let $(\Omega,\mathcal F,\mathbb P)$ be a probability space
 and $T>0$ be a fixed terminal time.
 Let $\{W_{t}, 0\leq t\leq T\}$ and $\{B_{t},0\leq t\leq T\}$ be two
  mutually independent standard Brownian motion processes 
  with values in $\mathbb{R}^{d}$ and in $\mathbb{R}^{l}$, respectively.
For each $t\in [0,T]$, we define
 \begin{align*}
 \mathcal{F}_{t} :=  \mathcal{F}_{0,t}^{W} \vee \mathcal{F}_{t,T}^{B},
 \end{align*}
  where for any process $\{\eta_{t}\}, \ \mathcal{F}_{s,t}^{\eta}=\sigma\{\eta_{r}-\eta_{s};s\leq r\leq t\}\vee \mathcal{N}$ and
 $\mathcal{N}$ is the class of $\mathbb{P}$-null sets of $\mathcal{F}$.
Note that the collection $\{ \mathcal{F}_{t}; t\in[0,T] \}$ is neither increasing nor decreasing, and it does not constitute a filtration.

We recall some notation from Pardoux and Peng~\cite{Pardoux1994BackwardDS}.
For $k\in\mathbb{N}$, we denote by
$C^k(\mathbb{R}^p;\mathbb{R}^q)$ the space of $C^k$ functions from $\mathbb{R}^p$
to $\mathbb{R}^q$,
by $C^k_b(\mathbb{R}^p;\mathbb{R}^q)$ the subspace of functions whose partial
derivatives up to order $k$ are bounded,
and by $C^k_p(\mathbb{R}^p;\mathbb{R}^q)$ the space of functions whose partial
derivatives up to order $k$ have at most polynomial growth at infinity.

Let $M^{2}([0,T]; \mathbb{R}^{n})$ be the set of 
$n$-dimensional jointly measurable stochastic processes  $\{ \varphi_{t}; t\in[0,T] \}$ which satisfy:\\
(i)  $\|\varphi \|_{M^{2}}^{2} =\mathbb{E}\left[\int_0^T |\varphi_{t}|^{2} \mathrm dt\right]<\infty$;\\
(ii)  $\varphi_{t}$ is $\mathcal{F}_{t}$-measurable  for a.e. $t\in[0,T].$

Similarly, let $S^{2}([0,T]; \mathbb{R}^{n})$ be the set of $n$-dimensional continuous stochastic processes, which satisfy:\\
(i)  $\|\varphi \|_{S^{2}}^{2} =\mathbb{E}\left[\sup\limits_{0\leq t\leq T} |\varphi_{t}|^{2}\right]<\infty$;\\
(ii) $\varphi_{t}$ is $\mathcal{F}_{t}$-measurable for a.e. $t\in[0,T].$\\

\begin{assumption}\label{ass:H1}
  Assume that
\begin{align*}
  f& : \Omega \times [0,T] \times \mathbb{R}^{k} \times \mathbb{R}^{k \times d} \to \mathbb{R}^{k},\\
  h &: \Omega \times [0,T] \times \mathbb{R}^{k} \times \mathbb{R}^{k \times d} \to \mathbb{R}^{k \times l},
\end{align*}
are jointly measurable and such that for any fixed $(y,z)$,
$$
f(\cdot,\cdot,y,z)\in M^{2}([0,T];\mathbb{R}^{k}),\quad
 h(\cdot,\cdot,y,z)\in   M^{2}([0,T];\mathbb{R}^{k \times l}).
$$
We assume moreover that there exist some constants $c>0$ and $0< \alpha<1$ such that for every
$(\omega,t)\in \Omega\times[0,T]$ and 
$(y_1,z_1),(y_2,z_2)\in \mathbb{R}^{k}\times\mathbb{R}^{k\times d},$
the following inequalities hold:
\begin{align}
| f(t,y_1,z_1)-f(t,y_2,z_2)|^{2}
&\leq c\left(| y_1-y_2|^{2}+\| z_1-z_2\|^{2}\right), \label{eq:f_lipschitz}\\
\|h(t,y_1,z_1)-h(t,y_2,z_2)\|^{2}
&\leq c| y_1-y_2|^{2}+\alpha \| z_1-z_2\|^{2}. \label{eq:h_lipschitz}
\end{align}
Here $|y|$ denotes the Euclidean norm and
$\|z\|^{2} = \operatorname{Tr}(z z^{*}).$
\end{assumption}

\begin{remark}
Note that, in the SPDE~\eqref{eq:SPDE(with BDSDE)} and the BDSDE~\eqref{eq:BDSDE},
the coefficient functions under consideration are of the form
\begin{align*}
f &: [0,T]\times \mathbb{R}^d \times \mathbb{R}^k \times \mathbb{R}^{k\times d} \to \mathbb{R}^k,\\
h &: [0,T]\times \mathbb{R}^d \times \mathbb{R}^k \times \mathbb{R}^{k\times d} \to \mathbb{R}^{k\times \ell}.
\end{align*}

Let $b\in C_b^3(\mathbb{R}^d;\mathbb{R}^d)$ and
$\sigma\in C_b^3(\mathbb{R}^d;\mathbb{R}^{d\times d})$.
For each $t\in[0,T]$ and $x\in\mathbb{R}^d$, we denote by
$\{X_s^{t,x};\, t\le s\le T\}$ the unique strong solution of the
SDE~\eqref{eq:SDE(with BDSDE)}.
Then the two formulations are then naturally linked through the forward diffusion
process by adopting the shorthand notation
$$
f(s,y,z) = f\left(s, X_s^{t,x}, y, z\right),
\quad
h(s,y,z) = h\left(s, X_s^{t,x}, y, z\right).
$$
Moreover, we assume that for any $s\in[0,T]$, 
$(x,y,z)\rightarrow \left(f(s,x,y,z),h(s,x,y,z)\right)$ is of class $C^3$.
\end{remark}

\begin{proposition}(Theorem 1.1 in~\cite{Pardoux1994BackwardDS} )
Under the Assumption~\ref{ass:H1}, the BDSDE
\begin{align*}
  Y_t=\xi+\int_t^Tf(s,Y_s,Z_s)\,\mathrm{d}s
  +\int_t^Th(s,Y_s,Z_s)\mathrm d \overleftarrow{B}_s
  -\int_t^T Z_s\,\mathrm{d}W_s,\quad 0\leq t\leq T,
\end{align*}
has  unique solution
$$(Y,Z)\in S^2([0,T];\mathbb{R}^k)\times M^2([0,T];\mathbb{R}^{k\times d})$$
for any $\xi\in L^2\left(\Omega,\mathcal{F}_T,\mathbb{P};\mathbb{R}^k\right)$.
\end{proposition}

\begin{assumption}\label{ass:H2}
There exists a constant $c>0$ such that, for all
$(t,y,z)\in[0,T]\times \mathbb{R}^k\times\mathbb{R}^{k\times d}$,
$$
h(t,y,z)h(t,y,z)^{*}
\leq zz^{*}+c\left(\|h(t,0,0)\|^{2}+|y|^{2}\right)I.
$$
\end{assumption}

\begin{assumption}\label{ass:H3}
For all $t\in[0,T]$, $x\in\mathbb{R}^d$, $y\in\mathbb{R}^k$ and
$z,\theta\in\mathbb{R}^{k\times d}$, it holds that
$$
h_z'(t,x,y,z)\,\theta\theta^{*}\,h_z'(t,x,y,z)^{*}
\leq \theta\theta^{*}.
$$
\end{assumption}

The following two theorems establish the connection between solutions of BDSDEs and SPDEs in the general setting.
\begin{theorem}[Theorem~3.1 in~\cite{Pardoux1994BackwardDS}]
Assume that $f$ and $h$ satisfy Assumptions~\ref{ass:H1} and~\ref{ass:H2}, and that $G\in C^2$.
Let $u$ be a solution to the SPDE~\eqref{eq:SPDE(with BDSDE)}.
Then
$$
u(t,x)=Y_t^{t,x},
$$
where $\{(Y_s^{t,x},Z_s^{t,x});\, t\leq s\leq T\}$
is the unique solution of the BDSDE~\eqref{eq:BDSDE}.
\end{theorem}

\begin{theorem}[Theorem 3.2 in~\cite{Pardoux1994BackwardDS}]
  Let $f$, $G$ and $h$ satisfy Assumption~\ref{ass:H1},~\ref{ass:H2} and ~\ref{ass:H3},
   then
  $$\{u(t,x)\triangleq Y_t^{t,x};0\leq t\leq T,x\in\mathbb{R}^d\}$$
  is the unique classical solution of the system of backward SPDEs~\eqref{eq:SDE(with BDSDE)},
  where $\{(Y_s^{t,x},Z_s^{t,x});t\leq s\leq T\}$ is the unique solution of the BDSDE~\eqref{eq:BDSDE}.
\end{theorem}

\paragraph{Linear SPDEs and BDSDEs}
In our applications of SPDEs, we are mainly interested in the linear case.
Therefore, following~\cite{BallyMatoussi2001}, we present the linear setting
that will be used in this work.

In the linear setting, consider
 \begin{align*}
    f(t,x,y,z)=F(t,x)+c(t)y+\widetilde{c}(t)z,\quad h(t,x,y,z)=H(t,x)+d(t)y,
  \end{align*}
  where $c,\widetilde{c}$ are bounded deterministic functions.
Then the BDSDE~\eqref{eq:BDSDE} becomes
\begin{align}
\label{eq:linear BDSDE}
  Y_s^{t,x}=G(X_T^{t,x})+\int_s^T \left[F(r,X_r^{t,x})+c(r)Y_r^{t,x}+\widetilde{c}(r)Z_r^{t,x}\right]\,\mathrm{d}r 
  +\int_s^T \left[H(r,X_r^{t,x})+d(r)Y_r^{t,x}\right]\mathrm d \overleftarrow{B}_r-
  \int_s^T Z_r^{t,x}\,\mathrm{d}W_r,\quad t\leq s\leq T, 
\end{align}
and the SPDE~\eqref{eq:SPDE(with BDSDE)} becomes
\begin{equation}\label{eq:linear SPDE(with BDSDE)}
  \begin{cases}
\mathrm{d} u(t,x)
=
\Bigl[\mathcal{L}u(t,x)
+ F(t,x)
+ c(t)u(t,x)
+ \widetilde{c}(t)(\sigma^\top\nabla u)(t,x)\Bigr]\,\mathrm{d}t
+ \Bigl[H(t,x)+d(t)u(t,x)\Bigr]\,\mathrm{d}B_t,\quad t\in[0,T],\\
u(T,x)=G(x).
\end{cases}
\end{equation}

\begin{proposition}[Linear BDSDE representation, Proposition~2.1 in \cite{BallyMatoussi2001}]
Assume that $G\in C_p^3(\mathbb{R}^d)$ and that
$F, H\in C_b^3([0,T]\times\mathbb{R}^d)$.
Let $\{X_s^{t,x};t\leq s\leq T\}$ be the solution of the SDE~\eqref{eq:SDE(with BDSDE)}
and define the stochastic exponential $\Phi(s,r)$ as
  \begin{align*}
    \Phi(s,\tau)=\exp\left(\int_s^\tau c(r)\,\mathrm{d}r 
    +\int_s^\tau d(r)\mathrm d \overleftarrow{B}_r
    +\int_s^\tau \widetilde{c}(r)\,\mathrm{d}W_r
    -\frac{1}{2}\int_s^\tau\left(|\widetilde{c}(r)|^2-|d(r)|^2\right)\mathrm d r\right).
  \end{align*}
Then,

(i)The unique solution $\{(Y_s^{t,x},Z_s^{t,x});\, t\le s\le T\}$ of
\eqref{eq:linear BDSDE} is given by
\begin{align*}
Y_s^{t,x}
= \Phi(s,T) G\bigl(X_T^{t,x}\bigr)
+ \int_s^T \Phi(s,r)
\Bigl( F\bigl(r,X_r^{t,x}\bigr) + d(r) H\bigl(r,X_r^{t,x}\bigr) \Bigr)\,\mathrm{d}r 
 + \int_s^T \Phi(s,r) H\bigl(r,X_r^{t,x}\bigr)
\mathrm d \overleftarrow{B}_r
- \int_s^T \Phi(s,r) Z_r^{t,x}\,\mathrm{d}W_r ,
\end{align*}
and it can be written as
  \begin{align*}
    Y_s^{t,x}=\mathbb{E}\left[\Phi(s,T)G(X_T^{t,x})
    +\int_s^T\Phi(s,r)\left(F(r,X_r^{t,x})+d(r) H(r,X_r^{t,x})\right)\,\mathrm{d}r 
    +\int_s^T\Phi(s,r)H(r,X_r^{t,x})\mathrm d \overleftarrow{B}_r\,|\,\mathcal{F}_{t,s}^W\vee \mathcal{F}_{t,T}^B  \right].
  \end{align*}

  (ii)The SPDE~\eqref{eq:linear SPDE(with BDSDE)} has a unique solution $u$ and it can 
  be written as
  \begin{align*}
    u(t,x)=\mathbb{E}\left[\Phi(t,T)G(X_T^{t,x})
    +\int_t^T\Phi(t,r)\left(F(r,X_r^{t,x})+d(r) H(r,X_r^{t,x})\right)\,\mathrm{d}r 
    +\int_t^T\Phi(t,r)H(r,X_r^{t,x})\mathrm d \overleftarrow{B}_r\,|\,\mathcal{F}_{t,T}^B  \right].
  \end{align*}
\end{proposition}

\section{Proofs for the Strong-error Framework}
\label{app:strong-error-framework}

This appendix collects the technical proofs underlying the strong-error
analysis developed in
Section~\ref{sec:Strong-Order-One Numerical Schemes for Pricing and Greek Estimators}.
The corresponding discretization operators are introduced in the main text,
and here we establish the local consistency, strong-error, and accumulated
stability estimates required for the global strong-error order one convergence
results.

\paragraph{Notation and conventions.}
Throughout this appendix, \(C\) and \(C_q\) denote generic constants,
independent of the time step \(h\), whose values may change from line to line.
Whenever a uniform grid is used, we write
\[
t_k=t_0+kh,\quad k=0,\ldots,N,\quad t_N=T,
\]
with \(h=(T-t_0)/N\).
For simplicity, denote
$$A(r,x):=F(r,x)+d(r)H(r,x).$$

Unless explicitly stated otherwise, all \(L^p\) norms are taken with respect to
the joint law of the Brownian motions \((W,B)\). That is,
\[
\|\cdot\|_{L^p}:=\|\cdot\|_{L^p_{W,B}} .
\]
When a fixed realization of the backward Brownian motion \(B\) is considered,
we write the conditional norm explicitly as \(L_W^p\).

\paragraph{Forward--backward information and martingale estimates.}
In the BDSDE setting, the natural information at a grid point \(t_k\) consists
of the forward \(W\)-information up to \(t_k\) and the backward \(B\)-information
from \(t_k\) to \(T\). At the grid level we write this information as
\[
\mathcal G_k
:=
\sigma(W_{t_0},\ldots,W_{t_k})
\vee
\sigma(B_{t_k}-B_{t_m}: k\le m\le N)
\vee\mathcal N .
\]
The family \((\mathcal G_k)_{k=0}^N\) is not a filtration in the ordinary
increasing-time sense: the \(W\)-part is increasing in \(k\), while the
\(B\)-part is decreasing in \(k\). Thus, when estimating accumulated local
fluctuations below, we do not regard \((\mathcal G_k)\) itself as a filtration.

Instead, each martingale-type contribution is treated by splitting it into its
forward and backward parts. The \(W\)-terms are estimated as ordinary forward
martingale differences
and the \(B\)-terms are estimated as
backward martingale differences, or equivalently as ordinary martingale
differences after introducing the reversed Brownian motion
\[
\widehat B_u:=B_T-B_{T-u},\quad 0\le u\le T .
\]
The Burkholder--Davis--Gundy estimates used below are always applied
in this split sense, and the resulting forward and reversed-time estimates are
combined by the triangle and Minkowski inequalities.

\paragraph{A joint-to-conditional estimate.}
Most of the strong-error estimates in this appendix are first proved under the
joint law of the two Brownian motions \((W,B)\). However, in conditional setting,
we need to interpret the estimate after fixing a realization of the
backward Brownian motion \(B\), so that the remaining randomness comes only
from the forward Brownian motion \(W\). The following elementary lemma provides
this passage from joint \(L^p_{W,B}\)-bounds to conditional \(L^p_W\)-bounds for
\(\mathbb P_B\)-almost every realization of \(B\).

The point of the lemma is that a deterministic joint strong-error estimate on
a sequence of dyadic time steps can be converted into an almost-sure-in-\(B\)
conditional estimate, at the cost of a harmless logarithmic-type factor. More
precisely, if an error family satisfies
\[
\mathbb E_{W,B}
\left[
\sup_{0\le k\le N_\ell}
|\Pi_{\ell,k}|^p
\right]
\le C_p2^{-p\ell},
\]
then for almost every fixed \(B\), the conditional \(W\)-error has the same
dyadic decay rate up to the factor \((1+\ell)^a\), with a finite random
constant depending on \(B\). This loss is mild and is sufficient for the
pathwise-in-\(B\) estimates used below.

\begin{lemma}
\label{lem:EB to EW}
Let \(p\ge2\). Assume that we are given a jointly measurable family of
random variables
\[
\{\Pi_{\ell,k}\}_{0\le k\le N_\ell,\,\ell\ge1}
\]
depending on both \(W\) and \(B\), such that
\[
\mathbb E_{W,B}
\left[
\sup_{0\le k\le N_\ell}
\left|
\Pi_{\ell,k}
\right|^p
\right]
\le C_p2^{-p\ell},
\quad \ell\ge1.
\]
Then, for every $a>1$, there exists a finite random
constant $C_{p,a}(B)<\infty$ for $\mathbb P_B$-almost every realization of $B$, such that
for all $\ell\ge1$,
\[
\mathbb E_W
\left[
\sup_{0\le k\le N_\ell}
\left|
\Pi_{\ell,k}
\right|^p
\,\middle|\,B
\right]
\le
C_{p,a}(B)(1+\ell)^a2^{-p\ell}.
\]

\end{lemma}

\begin{proof}
For each $\ell\ge1$, define
\[
V_\ell^{(p)}(B)
:=
\mathbb E_W
\left[
\sup_{0\le k\le N_\ell}
\left|
\Pi_{\ell,k}
\right|^p
\,\middle|\,B
\right].
\]
Then
\[
\mathbb E_B\left[V_\ell^{(p)}(B)\right]
=
\mathbb E_{W,B}
\left[
\sup_{0\le k\le N_\ell}
\left|
\Pi_{\ell,k}
\right|^p
\right]
\le
C_p2^{-p\ell}.
\]

Fix $a>1$ and define
\[
C_{p,a}(B)
:=
\sum_{\ell=1}^{\infty}
\frac{2^{p\ell}}{(1+\ell)^a}
V_\ell^{(p)}(B).
\]
Since all terms in the series are nonnegative,
\[
\mathbb E_B\left[C_{p,a}(B)\right]
=
\sum_{\ell=1}^{\infty}
\frac{2^{p\ell}}{(1+\ell)^a}
\mathbb E_B\left[V_\ell^{(p)}(B)\right]
\le
C_p
\sum_{\ell=1}^{\infty}
\frac{1}{(1+\ell)^a}
<\infty,
\]
where the last inequality follows from $a>1$. Hence
\[
C_{p,a}(B)<\infty
\]
for $\mathbb P_B$-almost every realization of $B$.

For such a realization of $B$, since every term in the defining series of
$C_{p,a}(B)$ is nonnegative, we have, for every $\ell\ge1$,
\[
\frac{2^{p\ell}}{(1+\ell)^a}
V_\ell^{(p)}(B)
\le
C_{p,a}(B).
\]
Therefore,
\[
V_\ell^{(p)}(B)
\le
C_{p,a}(B)(1+\ell)^a2^{-p\ell},
\]
i.e.
\[
\mathbb E_W
\left[
\sup_{0\le k\le N_\ell}
\left|
\Pi_{\ell,k}
\right|^p
\,\middle|\,B
\right]
\le
C_{p,a}(B)(1+\ell)^a2^{-p\ell},
\]
for all $\ell\ge1$ and for $\mathbb P_B$-almost every realization of $B$.
\end{proof}

\subsection{Direct Pricing Estimator}

\paragraph{Proof of Proposition~\ref{prop:error-u}}\label{para:proof of prop:error-u}

\begin{proposition*}[Strong-error order for the direct pricing payoff]
Fix $(t,x)\in[0,T]\times\mathbb{R}^d$ and a uniform grid $\{t_k\}_{k=0}^N$ with $h=(T-t)/N$.
Let
\begin{align*}
 P_{t_k}=\Phi(t,t_k)G\left(X_{t_k}^{t,x}\right)+Y_{t_k},
\end{align*}
where $Y_{t_k}$ is the exact accumulated payoff
\begin{align*}
  Y_{t_k}=\int_t^{t_k}\Phi(t,r)\left(F(r,X_r^{t,x})+d(r)H(r,X_r^{t,x})\right)\,\mathrm{d} r
+\int_t^{t_k}\Phi(t,r)H(r,X_r^{t,x})\,\mathrm d \overleftarrow{B}_r.
\end{align*}

Let $\left\{\left(X_k^{(h)},\Phi_k^{(h)}\right)\right\}_{k=0}^N$ be the approximations generated by the schemes
$S_X,S_\Phi$ as in Definition~\ref{def:Strong error of S X},~\ref{def:Strong error of S Phi}
 and define
 \begin{align*}
   Y_{k+1}^{(h)}
=&
Y_k^{(h)}
+
\mathcal{S}_{\mathrm{int}}
\left(
\Phi_k^{(h)},
X_k^{(h)},t_k,h;
\Delta W_{t_k},
\Delta\overleftarrow B_{t_k}
\right),
\quad
Y_0^{(h)}=0,\\
  P_k^{(h)}=&\Phi^{(h)}_k G\left(X_k^{(h)}\right)+Y_k^{(h)}.
\end{align*}

Assume:
\begin{enumerate}
    \item The strong-error orders of $\mathcal{S}_X$, $\mathcal{S}_\Phi$ and $\mathcal{S}_{\mathrm{int}}$
are $p_X,p_\Phi,p_{\mathrm{int}}$ respectively.
Moreover, $\mathcal{S}_{\mathrm{int}}$
satisfies the accumulated stability estimate
defined in Definition~\ref{def:accumulated stability of Sint}.

\item 
There exists $M>0$, independent of $h$, such that 
$\Phi_k^{(h)}$
 and 
$G(X_{t_k}^{t,x})$
are uniformly bounded in
$L^{4}_{W,B}$
by $M$.

\item 
$G$ is globally Lipschitz, i.e., there exists $L_G>0$ such that $|G(x)-G(y)|\leq L_G|x-y|$.
\end{enumerate}

Then the payoff approximation satisfies the joint strong-error bound
\begin{align*}
\left\|
\sup_{0\leq k\leq N}
\left|
P_{t_k}-P^{(h)}_k
\right|
\right\|_{L^2}
=
\mathcal{O}(h^p),
\quad
p=\min\{p_X,p_\Phi,p_{\mathrm{int}}\}.
\end{align*}
\end{proposition*}

\begin{proof}
First, for each $0\le k\le N$, we have
\begin{align*}
 P_{t_k}-P_k^{(h)}
 =&
 \left(
 \Phi(t,t_k)G\left(X_{t_k}^{t,x}\right)+Y_{t_k}
 \right)
 -
 \left(
 \Phi_k^{(h)}G\left(X_k^{(h)}\right)+Y_k^{(h)}
 \right)\\
=&
\left(\Phi(t,t_k)-\Phi_k^{(h)}\right)
G\left(X_{t_k}^{t,x}\right)
+
\Phi_k^{(h)}
\left(
G\left(X_{t_k}^{t,x}\right)-G\left(X_k^{(h)}\right)
\right)+
\left(Y_{t_k}-\widetilde Y_k^{(h)}\right)
+
\left(\widetilde Y_k^{(h)}-Y_k^{(h)}\right).
\end{align*}

Taking the supremum over $0\le k\le N$ and then using the triangle
inequality and H\"older's inequality, we obtain
\begin{align*}
&
\left\|
\sup_{0\le k\le N}
\left|
P_{t_k}-P_k^{(h)}
\right|
\right\|_{L^2}\\
\le&
\left\|
\sup_{0\le k\le N}
\left|
\left(\Phi(t,t_k)-\Phi_k^{(h)}\right)
G\left(X_{t_k}^{t,x}\right)
\right|
\right\|_{L^2}+
\left\|
\sup_{0\le k\le N}
\left|
\Phi_k^{(h)}
\left(
G\left(X_{t_k}^{t,x}\right)-G\left(X_k^{(h)}\right)
\right)
\right|
\right\|_{L^2}\\
&+
\left\|
\sup_{0\le k\le N}
\left|
Y_{t_k}-\widetilde Y_k^{(h)}
\right|
\right\|_{L^2}
+
\left\|
\sup_{0\le k\le N}
\left|
\widetilde Y_k^{(h)}-Y_k^{(h)}
\right|
\right\|_{L^2}\\
\le&
\left\|
\sup_{0\le k\le N}
\left|
\Phi(t,t_k)-\Phi_k^{(h)}
\right|
\right\|_{L^4}
\left\|
\sup_{0\le k\le N}
\left|
G\left(X_{t_k}^{t,x}\right)
\right|
\right\|_{L^4}+
\left\|
\sup_{0\le k\le N}
\left|
\Phi_k^{(h)}
\right|
\right\|_{L^4}
\left\|
\sup_{0\le k\le N}
\left|
G\left(X_{t_k}^{t,x}\right)-G\left(X_k^{(h)}\right)
\right|
\right\|_{L^4}+
C_{\mathrm{int}}h^{p_{\mathrm{int}}}\\
&+
L_{\mathrm{int}}
\bigg(
\left\|
\sup_{0\le j\le N}
\left|
\Phi(t,t_j)-\Phi_j^{(h)}
\right|
\right\|_{L^4}
+
\left\|
\sup_{0\le j\le N}
\left|
X_{t_j}^{t,x}-X_j^{(h)}
\right|
\right\|_{L^4}
\bigg).
\end{align*}
By the moment assumption and the Lipschitz continuity of $G$,
\begin{align*}
\left\|
\sup_{0\le k\le N}
\left|
P_{t_k}-P_k^{(h)}
\right|
\right\|_{L^2}
\le&
M
\left\|
\sup_{0\le k\le N}
\left|
\Phi(t,t_k)-\Phi_k^{(h)}
\right|
\right\|_{L^4}
+
ML_G
\left\|
\sup_{0\le k\le N}
\left|
X_{t_k}^{t,x}-X_k^{(h)}
\right|
\right\|_{L^4}\\
&+
C_{\mathrm{int}}h^{p_{\mathrm{int}}}
+
L_{\mathrm{int}}
\bigg(
\left\|
\sup_{0\le j\le N}
\left|
\Phi(t,t_j)-\Phi_j^{(h)}
\right|
\right\|_{L^4}
+
\left\|
\sup_{0\le j\le N}
\left|
X_{t_j}^{t,x}-X_j^{(h)}
\right|
\right\|_{L^4}
\bigg).
\end{align*}
Using the pathwise-in-time strong-error order estimates for $S_X$ and $S_\Phi$, we
therefore get
\begin{align*}
&
\left\|
\sup_{0\le k\le N}
\left|
P_{t_k}-P_k^{(h)}
\right|
\right\|_{L^2}
\le
C_\Phi(M+L_{\mathrm{int}})h^{p_\Phi}
+
C_X(ML_G+L_{\mathrm{int}})h^{p_X}
+
C_{\mathrm{int}}h^{p_{\mathrm{int}}}.
\end{align*}
Hence, for all sufficiently small $h$,
\begin{align*}
\left\|
\sup_{0\le k\le N}
\left|
P_{t_k}-P_k^{(h)}
\right|
\right\|_{L^2}
\le
C h^p,
\quad
p=\min\{p_X,p_\Phi,p_{\mathrm{int}}\}.
\end{align*}
This proves
\[
\left\|
\sup_{0\leq k\leq N}
\left|
P_{t_k}-P^{(h)}_k
\right|
\right\|_{L^2}
=
\mathcal{O}(h^{p}).
\]
\end{proof}

\paragraph{Proof of Proposition~\ref{prop:int-strong-error}}
\label{para:proof of prop:int-strong-error}
\begin{proposition*}[Strong error of the integral discretization]
  Assume that
\(F,H,d,\widetilde c\) and the coefficients of \(X\) and \(\Phi\) are sufficiently
smooth with polynomial growth, and that the corresponding moments of
\(X\) and \(\Phi\) are uniformly bounded. 
Then the approximation generated by
\begin{align}
\widetilde Y_{k+1}^{(h)}
=
\widetilde Y_k^{(h)}
+
\mathcal S_{\mathrm{int}}^{\mathrm{FBT}}
\left(\Phi(t,t_k),X_{t_k}^{t,x},
t_k,h;\Delta W_{t_k},\Delta\overleftarrow B_{t_k}\right),
\quad
\widetilde Y_0^{(h)}=0,
\end{align}
satisfies
\begin{align}
\left\|
\sup_{0\le k\le N}
\left|
Y_{t_k}-\widetilde Y_k^{(h)}
\right|
\right\|_{L^q}
\le Ch,
\end{align}
for every $q\ge2$. Consequently, by Jensen's inequality, the same estimate
also holds for every $0<q<2$. Hence
$\mathcal S_{\mathrm{int}}^{\mathrm{FBT}}$ has strong-error order $1$ in the sense
of Definition~\ref{def:Strong error of S int}.
\end{proposition*}

\begin{proof}

For $k=0,\ldots,N-1$, write the exact one-step integral as
\begin{align*}
I_k
:=
\int_{t_k}^{t_{k+1}}
\Phi(t,r)\bigl(F(r,X_r^{t,x})+d(r)H(r,X_r^{t,x})\bigr)\,\mathrm dr
+
\int_{t_k}^{t_{k+1}}
\Phi(t,r)H(r,X_r^{t,x})\cdot
\mathrm d \overleftarrow{B}_r.
\end{align*}
Then, for each $0\le n\le N$,
\begin{align*}
Y_{t_n}-\widetilde Y_n^{(h)}
=
\sum_{k=0}^{n-1}
\left[
I_k-
\mathcal S_{\mathrm{int}}^{\mathrm{FBT}}
\left(
\Phi(t,t_k),X_{t_k}^{t,x},
t_k,h;
\Delta W_{t_k},\Delta\overleftarrow B_{t_k}
\right)
\right].
\end{align*}

We first record the one-step consistency decomposition. 
By the smoothness assumptions and the local It\^o--Taylor expansions of
$X$ and $\Phi$, for $r\in[t_k,t_{k+1}]$,
\begin{align*}
\Phi(t,r)A(r,X_r^{t,x})
=\Phi(t,t_k)A(t_k,X_{t_k}^{t,x})
+&\Phi(t,t_k)
\sum_{a=1}^d
\left(
\nabla_xA(t_k,X_{t_k}^{t,x})^\top
\sigma_{\cdot a}(X_{t_k}^{t,x})+\widetilde c_a(t_k)A(t_k,X_{t_k}^{t,x})\right)
\left(W_r^a-W_{t_k}^a\right)\\
+&\Phi(t,t_k)\,\sum_{\alpha=1}^{\ell}
d_\alpha(t_k)\,
A(t_k,X_{t_k}^{t,x})
\left(B_{t_{k}}^\alpha-B_r^\alpha\right)
+
R_{k,r}^{D},
\end{align*}
where
$\|R_{k,r}^{D}\|_{L^q}
\le
C_q(r-t_k).$

Set
$A_{k,h}^{D}:=\int_{t_k}^{t_{k+1}}R_{k,r}^{D}\,\mathrm dr$
and
\small
\begin{align*}
M_{k,h}^{D}
:=&
\Phi(t,t_k)
\sum_{a=1}^d
\left(
\nabla_xA(t_k,X_{t_k}^{t,x})^\top
\sigma_{\cdot a}(X_{t_k}^{t,x})+\widetilde c_a(t_k)A(t_k,X_{t_k}^{t,x})\right)
\int_{t_k}^{t_{k+1}}
\left(W_r^a-W_{t_k}^a\right)\,\mathrm dr\\
&\quad+\Phi(t,t_k)\,\sum_{\alpha=1}^{\ell}
d_\alpha(t_k)\,
A(t_k,X_{t_k}^{t,x})
\int_{t_k}^{t_{k+1}}
\left(B_{t_{k}}^\alpha-B_r^\alpha\right)\,\mathrm dr .
\end{align*}
\normalsize
Then 
\begin{align}
  \label{eq:direct-A-D-bound}
  \|A_{k,h}^{D}\|_{L^q}\le C_qh^2,
\end{align}
and \(M_{k,h}^{D}\) is a martingale-type local fluctuation
in the sense of the discrete forward--backward information convention above:
its \(W\)-part is a
forward martingale difference and its \(B\)-part is a reverse martingale
difference, equivalently a martingale difference after reversing the \(B\)-time.
Moreover,
\begin{align}
\|M_{k,h}^{D}\|_{L^q}
\le C_qh^{3/2}.
\label{eq:direct-M-D-bound}
\end{align}
We have
\begin{align*}
\int_{t_k}^{t_{k+1}}
\Phi(t,r)A(r,X_r^{t,x})\,\mathrm dr
-
h\Phi(t,t_k)A(t_k,X_{t_k}^{t,x})
=
A_{k,h}^{D}+M_{k,h}^{D}.
\end{align*}

Next, for the backward stochastic integral, the local expansions of
$X$, $H$, and $\Phi$ yield
\begin{align*}
\Phi(t,r)H(r,X_r^{t,x})
=
\Phi(t,t_k)H(t_k,X_{t_k}^{t,x})
&+
\Phi(t,t_k)
\bigl(
H_x(t_k,X_{t_k}^{t,x})\sigma(X_{t_k}^{t,x})
+\widetilde c(t_k)H(t_k,X_{t_k}^{t,x})
\bigr)
\left(W_r-W_{t_k}\right)\\
&+
\Phi(t,t_k)d(t_k)H(t_k,X_{t_k}^{t,x})
\left(B_{t_{k}}-B_r\right)
+
R_{k,r}^{B},
\end{align*}
with
\begin{align}
\|R_{k,r}^{B}\|_{L^q}
\le
C_q\Big((r-t_k)+(r-t_k)^{1/2}h^{1/2}\Big).
\label{eq:direct-backward-integrand-remainder}
\end{align}
Therefore, after subtracting the corresponding Forward--Backward Taylor
correction terms in
$\mathcal S_{\mathrm{int}}^{\mathrm{FBT}}$, the one-step error has the form
\begin{align}
I_k-
\mathcal S_{\mathrm{int}}^{\mathrm{FBT}}
\left(
\Phi(t,t_k),X_{t_k}^{t,x},
t_k,h;
\Delta W_{t_k},\Delta\overleftarrow B_{t_k}
\right)
=
A_{k,h}^{D}+M_{k,h}^{D}+R_{k,h}^{B},
\label{eq:direct-one-step-consistency}
\end{align}
where
\begin{align*}
R_{k,h}^{B}
:=
\int_{t_k}^{t_{k+1}}
R_{k,r}^{B}\cdot\mathrm d \overleftarrow{B}_r.
\end{align*}
Moreover, by \eqref{eq:direct-backward-integrand-remainder},
\begin{align}
\int_{t_k}^{t_{k+1}}
\|R_{k,r}^{B}\|_{L^q}^2\,\mathrm dr
\le C_qh^3.
\label{eq:direct-integrated-RB-bound}
\end{align}

Using \eqref{eq:direct-one-step-consistency}, we have, for every
$0\le n\le N$,
\begin{align*}
Y_{t_n}-\widetilde Y_n^{(h)}
=
\sum_{k=0}^{n-1}A_{k,h}^{D}
+
\sum_{k=0}^{n-1}M_{k,h}^{D}
+
\sum_{k=0}^{n-1}R_{k,h}^{B}.
\end{align*}

We estimate the three accumulated terms separately. For the finite-variation
part, by the pathwise bound
\[
\sup_{0\le n\le N}
\left|
\sum_{k=0}^{n-1}A_{k,h}^{D}
\right|
\le
\sum_{k=0}^{N-1}
|A_{k,h}^{D}|,
\]
and \eqref{eq:direct-A-D-bound}, we get
\begin{align}
\left\|
\sup_{0\le n\le N}
\left|
\sum_{k=0}^{n-1}A_{k,h}^{D}
\right|
\right\|_{L^q}
\le
\sum_{k=0}^{N-1}
\|A_{k,h}^{D}\|_{L^q}
\le
C_qNh^2
\le
C_qh.
\label{eq:direct-global-A-bound}
\end{align}

For the martingale-type contribution, the discrete
Burkholder--Davis--Gundy inequality gives
\begin{align}
\left\|
\sup_{0\le n\le N}
\left|
\sum_{k=0}^{n-1}M_{k,h}^{D}
\right|
\right\|_{L^q}
&\le
C_q
\left\|
\left(
\sum_{k=0}^{N-1}
|M_{k,h}^{D}|^2
\right)^{1/2}
\right\|_{L^q}
\notag\\
&=
C_q
\left\|
\sum_{k=0}^{N-1}
|M_{k,h}^{D}|^2
\right\|_{L^{q/2}}^{1/2}
\notag\\
&\le
C_q
\left(
\sum_{k=0}^{N-1}
\|M_{k,h}^{D}\|_{L^q}^2
\right)^{1/2}
\notag\\
&\le
C_q(Nh^3)^{1/2}
\le
C_qh.
\label{eq:direct-global-M-bound}
\end{align}
Here we used Minkowski's inequality in $L^{q/2}$, which is valid because
$q\ge2$.

Finally, since
\[
\sum_{k=0}^{n-1}R_{k,h}^{B}
=
\int_t^{t_n}
R_r^{B,h}\cdot\mathrm d \overleftarrow{B}_r,
\quad
R_r^{B,h}:=R_{k,r}^{B}\quad
\text{for }r\in[t_k,t_{k+1}],
\]
the Burkholder--Davis--Gundy inequality for backward It\^o integrals yields
\begin{align}
&
\left\|
\sup_{0\le n\le N}
\left|
\sum_{k=0}^{n-1}R_{k,h}^{B}
\right|
\right\|_{L^q}
\notag\\
\le &
C_q
\left\|
\left(
\sum_{k=0}^{N-1}
\int_{t_k}^{t_{k+1}}
|R_{k,r}^{B}|^2\,\mathrm dr
\right)^{1/2}
\right\|_{L^q}
\notag\\
\le &
C_q
\left(
\sum_{k=0}^{N-1}
\int_{t_k}^{t_{k+1}}
\|R_{k,r}^{B}\|_{L^q}^2\,\mathrm dr
\right)^{1/2}
\notag\\
\le &
C_q(Nh^3)^{1/2}
\le
C_qh.
\label{eq:direct-global-RB-bound}
\end{align}

Combining
\eqref{eq:direct-global-A-bound},
\eqref{eq:direct-global-M-bound}, and
\eqref{eq:direct-global-RB-bound}, we obtain
\begin{align*}
\left\|
\sup_{0\le n\le N}
\left|
Y_{t_n}-\widetilde Y_n^{(h)}
\right|
\right\|_{L^q}
\le
C_qh,
\quad q\ge2.
\end{align*}
This proves the claimed estimate for $q\ge2$.

For $0<q<2$, Jensen's inequality gives
\begin{align*}
\left\|
\sup_{0\le n\le N}
\left|
Y_{t_n}-\widetilde Y_n^{(h)}
\right|
\right\|_{L^q}
\le
\left\|
\sup_{0\le n\le N}
\left|
Y_{t_n}-\widetilde Y_n^{(h)}
\right|
\right\|_{L^2}
\le C_2h.
\end{align*}
Hence $\mathcal S_{\mathrm{int}}^{\mathrm{FBT}}$ has strong-error order $1$ in the
sense of Definition~\ref{def:Strong error of S int}.
\end{proof}

\paragraph{Proof of Proposition~\ref{prop:int stable}}
\label{para:proof of prop:int stable}
\begin{proposition*}
Assume that 
$F,H$ are globally Lipschitz and have at most linear growth. 
Moreover, assume
that the coefficient
$H_x(t,x)\sigma(x)+\widetilde c(t)H(t,x)$
is globally Lipschitz and has at most linear growth.
Assume further that the exact and numerical input processes satisfy the uniform
moment bound
\[
\left\|
\sup_{0\le j\le N}
|\Phi(t,t_j)|
\right\|_{L^8}
+
\left\|
\sup_{0\le j\le N}
|\Phi_j^{(h)}|
\right\|_{L^8}
+
\left\|
\sup_{0\le j\le N}
|X_{t_j}^{t,x}|
\right\|_{L^8}
+
\left\|
\sup_{0\le j\le N}
|X_j^{(h)}|
\right\|_{L^8}
<\infty .
\]
Then
$\mathcal{S}_{\mathrm{int}}^{\mathrm{FBT}}$
satisfies the accumulated stability condition in
Definition~\ref{def:accumulated stability of Sint}. More precisely,
there exists $L_{\mathrm{int}}>0$, independent of $h$, such that
\begin{align*}
\left\|
\sup_{0\leq k\leq N}
\left|
\widetilde Y_k^{(h)}-Y_k^{(h)}
\right|
\right\|_{L^2}
\le
L_{\mathrm{int}}
\bigg(
&
\left\|
\sup_{0\leq j\leq N}
\left|
\Phi(t,t_j)-\Phi_j^{(h)}
\right|
\right\|_{L^4}
+
\left\|
\sup_{0\leq j\leq N}
\left|
X_{t_j}^{t,x}-X_j^{(h)}
\right|
\right\|_{L^4}
\bigg).
\end{align*}
\end{proposition*}

\begin{proof}
Set
$\Delta\Phi_j:=\Phi(t,t_j)-\Phi_j^{(h)}$ and
$\Delta X_j:=X_{t_j}^{t,x}-X_j^{(h)}$.
Define
\[
\delta_\Phi
:=
\left\|
\sup_{0\le j\le N}
|\Delta\Phi_j|
\right\|_{L^4},
\quad
\delta_X
:=
\left\|
\sup_{0\le j\le N}
|\Delta X_j|
\right\|_{L^4},
\quad
\delta:=\delta_\Phi+\delta_X.
\]

By the definitions of $\widetilde Y_k^{(h)}$ and $Y_k^{(h)}$, we have
\begin{align*}
\widetilde Y_k^{(h)}-Y_k^{(h)}
=
\sum_{j=0}^{k-1}
\Bigg[
&\mathcal{S}_{\mathrm{int}}^{\mathrm{FBT}}
\left(\Phi(t,t_j),X_{t_j}^{t,x},t_j,h;
\Delta W_{t_j},\Delta\overleftarrow B_{t_j}\right)
-
\mathcal{S}_{\mathrm{int}}^{\mathrm{FBT}}
\left(\Phi_j^{(h)},X_j^{(h)},t_j,h;
\Delta W_{t_j},\Delta\overleftarrow B_{t_j}\right)
\Bigg].
\end{align*}

We estimate the four contributions in
$\mathcal{S}_{\mathrm{int}}^{\mathrm{FBT}}$ separately.

First consider the time-integral term. 
Since $F$ and $H$ are globally Lipschitz with at most linear growth, and
$d$ is bounded, $A$ is globally Lipschitz in $x$ and has at most linear
growth. Hence, using H\"older's inequality and the uniform moment bounds,
\begin{align*}
&
\left\|
\sup_{0\le k\le N}
\left|
\sum_{j=0}^{k-1}
h\Big[
\Phi(t,t_j)A(t_j,X_{t_j}^{t,x})
-
\Phi_j^{(h)}A(t_j,X_j^{(h)})
\Big]
\right|
\right\|_{L^2}
\\
\le&
\sum_{j=0}^{N-1}
h
\left\|
\Phi(t,t_j)A(t_j,X_{t_j}^{t,x})
-
\Phi_j^{(h)}A(t_j,X_j^{(h)})
\right\|_{L^2}
\\
\le&
\sum_{j=0}^{N-1}
h
\bigg(
\left\|
\Delta\Phi_j A(t_j,X_{t_j}^{t,x})
\right\|_{L^2}
+
\left\|
\Phi_j^{(h)}
\big[
A(t_j,X_{t_j}^{t,x})-A(t_j,X_j^{(h)})
\big]
\right\|_{L^2}
\bigg)
\\
\le&
C
\sum_{j=0}^{N-1}
h
\left(
\|\Delta\Phi_j\|_{L^4}
+
\|\Delta X_j\|_{L^4}
\right)
\\
\le&
CT(\delta_\Phi+\delta_X).
\end{align*}

Next consider the backward stochastic increment term. 
By the same Lipschitz and growth estimates,
\[
\|\Phi(t,t_j)H(t_j,X_{t_j}^{t,x})
-
\Phi_j^{(h)}H(t_j,X_j^{(h)})\|_{L^2}
\le
C
\left(
\|\Delta\Phi_j\|_{L^4}
+
\|\Delta X_j\|_{L^4}
\right).
\]
Using the discrete Burkholder--Davis--Gundy inequality for the backward
increments, equivalently after reversing time, we obtain
\begin{align*}
&
\left\|
\sup_{0\le k\le N}
\left|
\sum_{j=0}^{k-1}
\left(\Phi(t,t_j)H(t_j,X_{t_j}^{t,x})
-
\Phi_j^{(h)}H(t_j,X_j^{(h)})\right)\cdot \Delta\overleftarrow B_{t_j}
\right|
\right\|_{L^2}
\\
\le&
C
\left\|
\left(
\sum_{j=0}^{N-1}
h\left|\Phi(t,t_j)H(t_j,X_{t_j}^{t,x})
-
\Phi_j^{(h)}H(t_j,X_j^{(h)})\right|^2
\right)^{1/2}
\right\|_{L^2}
\\
\le&
C
\left(
\sum_{j=0}^{N-1}
h
\left\|\Phi(t,t_j)H(t_j,X_{t_j}^{t,x})
-
\Phi_j^{(h)}H(t_j,X_j^{(h)})\right\|_{L^2}^2
\right)^{1/2}
\\
\le&
C\sqrt T(\delta_\Phi+\delta_X).
\end{align*}

Now consider the mixed \(WB\) correction term. For
\(1\le i\le \ell\) and \(1\le j\le d\), set
\[
K_{i j}(s,x)
:=
\nabla_xH_i(s,x)^\top\sigma_{\cdot j}(x)
+
\widetilde c_j(s)H_i(s,x).
\]
By assumption, each \(K_{i j}\) is globally Lipschitz in \(x\) and has at
most linear growth, uniformly in \(s\).

The \(WB\) contribution to the stability difference is
\[
\sum_{r=0}^{k-1}
\sum_{i=1}^{\ell}\sum_{j=1}^{d}
\left(\Phi(t,t_r)K_{i j}(t_r,X_{t_r}^{t,x})
-
\Phi_r^{(h)}K_{i j}(t_r,X_r^{(h)})\right)
J_{i j}^{WB}(t_r,h).
\]
Decompose
\[
\Phi(t,t_r)K_{i j}(t_r,X_{t_r}^{t,x})
-
\Phi_r^{(h)}K_{i j}(t_r,X_r^{(h)})
=
\Delta\Phi_r K_{i j}(t_r,X_{t_r}^{t,x})
+
\Phi_r^{(h)}
\Big(
K_{i j}(t_r,X_{t_r}^{t,x})
-
K_{i j}(t_r,X_r^{(h)})
\Big).
\]
For every \(p\ge2\), the Brownian iterated increment satisfies
\[
\|J_{i j}^{WB}(t_r,h)\|_{L^p}\le C_p h.
\]
Hence, by H\"older's inequality with exponents \(4,8,8\), the uniform
\(L^8\) moment bounds, and the linear growth of \(K_{i j}\),
\begin{align*}
&
\left\|
\Delta\Phi_r
K_{i j}(t_r,X_{t_r}^{t,x})
J_{i j}^{WB}(t_r,h)
\right\|_{L^2}
\le
\|\Delta\Phi_r\|_{L^4}
\|K_{i j}(t_r,X_{t_r}^{t,x})\|_{L^8}
\|J_{i j}^{WB}(t_r,h)\|_{L^8}
\le
Ch\|\Delta\Phi_r\|_{L^4}.
\end{align*}
Similarly, using the Lipschitz continuity of \(K_{i j}\),
\begin{align*}
&
\left\|
\Phi_r^{(h)}
\Big(
K_{i j}(t_r,X_{t_r}^{t,x})
-
K_{i j}(t_r,X_r^{(h)})
\Big)
J_{i j}^{WB}(t_r,h)
\right\|_{L^2}
\\
\le&
\|\Phi_r^{(h)}\|_{L^8}
\left\|
K_{i j}(t_r,X_{t_r}^{t,x})
-
K_{i j}(t_r,X_r^{(h)})
\right\|_{L^4}
\|J_{i j}^{WB}(t_r,h)\|_{L^8}
\\
\le&
Ch\|\Delta X_r\|_{L^4}.
\end{align*}
Therefore, for every \(r,i,j\),
\[
\left\|
\left(\Phi(t,t_r)K_{i j}(t_r,X_{t_r}^{t,x})
-
\Phi_r^{(h)}K_{i j}(t_r,X_r^{(h)})\right)
J_{i j}^{WB}(t_r,h)
\right\|_{L^2}
\le
Ch
\left(
\|\Delta\Phi_r\|_{L^4}
+
\|\Delta X_r\|_{L^4}
\right).
\]

Since \(\ell\) and \(d\) are fixed, summing over the component indices gives
\begin{align*}
&
\left\|
\sup_{0\le k\le N}
\left|
\sum_{r=0}^{k-1}
\sum_{i=1}^{\ell}\sum_{j=1}^{d}
\left(\Phi(t,t_r)K_{i j}(t_r,X_{t_r}^{t,x})
-
\Phi_r^{(h)}K_{i j}(t_r,X_r^{(h)})\right)
J_{i j}^{WB}(t_r,h)
\right|
\right\|_{L^2}
\\
\le&
\sum_{r=0}^{N-1}
\sum_{i=1}^{\ell}\sum_{j=1}^{d}
\left\|
\left(\Phi(t,t_r)K_{i j}(t_r,X_{t_r}^{t,x})
-
\Phi_r^{(h)}K_{i j}(t_r,X_r^{(h)})\right)
J_{i j}^{WB}(t_r,h)
\right\|_{L^2}
\\
\le&
C
\sum_{r=0}^{N-1}
h
\left(
\|\Delta\Phi_r\|_{L^4}
+
\|\Delta X_r\|_{L^4}
\right)
\\
\le&
CT(\delta_\Phi+\delta_X).
\end{align*}

Finally, consider the backward--backward correction term appearing in
$\mathcal S_{\mathrm{int}}^{\mathrm{FBT}}$:
\[
\Phi
\sum_{i=1}^{\ell}\sum_{j=1}^{\ell}
d_j(t_r)H_i(t_r,X)
J_{ij}^{BB}(t_r,h).
\]
Since $d$ is bounded and $H$ is globally Lipschitz with at most linear growth,
H\"older's inequality and the uniform moment bounds give
\[
\left\|
d_j(t_r)
\Big[
\Phi(t,t_r)H_i(t_r,X_{t_r}^{t,x})
-
\Phi_r^{(h)}H_i(t_r,X_r^{(h)})
\Big]
\right\|_{L^2}
\le
C
\left(
\|\Delta\Phi_r\|_{L^4}
+
\|\Delta X_r\|_{L^4}
\right).
\]
Moreover, for every $1\le i,j\le\ell$,
\[
\left\|
J_{ij}^{BB}(t_r,h)
\right\|_{L^2}
\le Ch.
\]
Therefore,
\begin{align*}
&
\left\|
\sup_{0\le k\le N}
\left|
\sum_{r=0}^{k-1}
\sum_{i=1}^{\ell}\sum_{j=1}^{\ell}
d_j(t_r)
\Big[
\Phi(t,t_r)H_i(t_r,X_{t_r}^{t,x})
-
\Phi_r^{(h)}H_i(t_r,X_r^{(h)})
\Big]
J_{ij}^{BB}(t_r,h)
\right|
\right\|_{L^2}
\\
\le&
\sum_{r=0}^{N-1}
\sum_{i=1}^{\ell}\sum_{j=1}^{\ell}
\left\|
d_j(t_r)
\Big[
\Phi(t,t_r)H_i(t_r,X_{t_r}^{t,x})
-
\Phi_r^{(h)}H_i(t_r,X_r^{(h)})
\Big]
J_{ij}^{BB}(t_r,h)
\right\|_{L^2}
\\
\le&
C
\sum_{r=0}^{N-1}
h
\left(
\|\Delta\Phi_r\|_{L^4}
+
\|\Delta X_r\|_{L^4}
\right)
\\
\le&
CT
\left(
\delta_\Phi+\delta_X
\right).
\end{align*}

Combining the four estimates, we conclude that
\begin{align*}
\left\|
\sup_{0\leq k\leq N}
\left|
\widetilde Y_k^{(h)}-Y_k^{(h)}
\right|
\right\|_{L^2}
\le
L_{\mathrm{int}}
\left(
\delta_\Phi+\delta_X
\right).
\end{align*}
Substituting the definitions of $\delta_\Phi$ and $\delta_X$ gives
\begin{align*}
\left\|
\sup_{0\leq k\leq N}
\left|
\widetilde Y_k^{(h)}-Y_k^{(h)}
\right|
\right\|_{L^2}
\le
L_{\mathrm{int}}
\bigg(
&
\left\|
\sup_{0\leq j\leq N}
\left|
\Phi(t,t_j)-\Phi_j^{(h)}
\right|
\right\|_{L^4}
+
\left\|
\sup_{0\leq j\leq N}
\left|
X_{t_j}^{t,x}-X_j^{(h)}
\right|
\right\|_{L^4}
\bigg).
\end{align*}
This is the desired accumulated stability estimate.
\end{proof}

\subsection{First-order Greek Estimators}
\paragraph{Proof of Proposition~\ref{prop:error-first-order-greek-payoff}}
\label{para:proof of prop:error-first-order-greek-payoff}
\begin{proposition*}[Strong-error order for the first-order Greek payoff]
Fix $(t,x)\in[0,T]\times\mathbb R^d$ and $1\le i\le d$
and define
\begin{align*}
P_{t_k}^{(i)}
=&\Phi(t,t_k)\,\nabla G(X_{t_k}^{t,x})^\top J_{t_k}^{t,x}e_i
+Y_{t_k}^{(i)},
\end{align*}
where
\begin{align*}
Y_{t_k}^{(i)}
=&
\int_t^{t_k}
\Phi(t,s)\,
\nabla_x\!\Big(F(s,X_s^{t,x})+d(s)H(s,X_s^{t,x})\Big)^\top
J_s^{t,x}e_i\,\mathrm ds  
+\int_t^{t_k}\Phi(t,s)\,\nabla_x H(s,X_s^{t,x})^\top J_s^{t,x}e_i\,
\mathrm d \overleftarrow{B}_s.
\end{align*}

Let $\{X_k^{(h)},\Phi_k^{(h)},J_k^{(h)}\}_{k=0}^N$
be generated by
$S_X,S_\Phi,S_J$
as in
Definition~\ref{def:Strong error of S X},
\ref{def:Strong error of S Phi}
and \ref{def:Strong error of S J},
and define
\begin{align*}
Y_{k+1}^{(i,h)}
=&Y_k^{(i,h)}
+\mathcal S_{\mathrm{int}}^{(i)}
\left(\Phi_k^{(h)},X_k^{(h)},J_k^{(h)},t_k,h;
\Delta W_{t_k},\Delta\overleftarrow B_{t_k}\right),
\quad
Y_0^{(i,h)}=0,\\
P_k^{(i,h)}
=&
\Phi_k^{(h)}
\nabla G(X_k^{(h)})^\top J_k^{(h)}e_i
+
Y_k^{(i,h)}.
\end{align*}

Assume:

\begin{enumerate}
\item
The strong-error orders of $S_X,S_\Phi,S_J$ and
$\mathcal S_{\mathrm{int}}^{(i)}$ are $p_X,p_\Phi,p_J,p_{\mathrm{int}}^{(i)}$, respectively.

\item
The operator $\mathcal S_{\mathrm{int}}^{(i)}$ satisfies the
accumulated stability estimate:
there exists $L_{\mathrm{int}}^{(i)}>0$, independent of $h$, such that
\[
\left\|
\sup_{0\le k\le N}
|\widetilde Y_k^{(i,h)}-Y_k^{(i,h)}|
\right\|_{L^2}
\le
L_{\mathrm{int}}^{(i)}
\bigg(
\left\|
\sup_{0\le j\le N}
|\Phi(t,t_j)-\Phi_j^{(h)}|
\right\|_{L^4}
+
\left\|
\sup_{0\le j\le N}
|X_{t_j}^{t,x}-X_j^{(h)}|
\right\|_{L^4}
+
\left\|
\sup_{0\le j\le N}
|J_{t_j}^{t,x}-J_j^{(h)}|
\right\|_{L^4}
\bigg),
\]
where $\widetilde Y_k^{(i,h)}$ is generated by
$\mathcal S_{\mathrm{int}}^{(i)}$ with the exact inputs
$\Phi(t,t_k),X_{t_k}^{t,x},J_{t_k}^{t,x}$
as in Definition~\ref{def:Strong error of S int Greek}.

\item
There exists $M>0$, independent of $h$, such that 
$\Phi_k^{(h)},\
\Phi(t,t_k),\ 
J_k^{(h)},\ 
J_{t_k}^{t,x}$
 and 
$\nabla G(X_{t_k}^{t,x}),\ \nabla G(X_k^{(h)})$
are uniformly bounded in
$L^{8}_{W,B}$
by $M$.

\item
$\nabla G$ is globally Lipschitz, i.e., there exists $L_{\nabla G}>0$ such that
\begin{align*}
\left|\nabla G(x)-\nabla G(y)\right|
\le L_{\nabla G}|x-y|,\quad x,y\in\mathbb R^d .
\end{align*}
\end{enumerate}

Then
\[
\left\|
\sup_{0\le k\le N}
\left|
P_{t_k}^{(i)}-P_k^{(i,h)}
\right|
\right\|_{L^2}
=
\mathcal{O}(h^p),
\quad
p=\min\{p_X,p_\Phi,p_J,p_{\mathrm{int}}^{(i)}\}.
\]
\end{proposition*}

\begin{proof}
For each $0\le k\le N$, we have
\begin{align*}
P_{t_k}^{(i)}-P_k^{(i,h)}
=&
\left(
\Phi(t,t_k)\nabla G(X_{t_k}^{t,x})^\top J_{t_k}^{t,x}e_i
+
Y_{t_k}^{(i)}
\right)
-
\left(
\Phi_k^{(h)}\nabla G(X_k^{(h)})^\top J_k^{(h)}e_i
+
Y_k^{(i,h)}
\right)
\\
=&
\left(\Phi(t,t_k)-\Phi_k^{(h)}\right)
\nabla G(X_{t_k}^{t,x})^\top J_{t_k}^{t,x}e_i
+
\Phi_k^{(h)}
\left(
\nabla G(X_{t_k}^{t,x})
-
\nabla G(X_k^{(h)})
\right)^\top
J_{t_k}^{t,x}e_i
\\
&+
\Phi_k^{(h)}
\nabla G(X_k^{(h)})^\top
\left(
J_{t_k}^{t,x}-J_k^{(h)}
\right)e_i
+
\left(
Y_{t_k}^{(i)}
-
\widetilde Y_k^{(i,h)}
\right)
+
\left(
\widetilde Y_k^{(i,h)}
-
Y_k^{(i,h)}
\right).
\end{align*}

Taking the supremum over $0\le k\le N$ and then the
$L^2$ norm, the triangle inequality gives
\begin{align*}
&
\left\|
\sup_{0\le k\le N}
\left|
P_{t_k}^{(i)}
-
P_k^{(i,h)}
\right|
\right\|_{L^2}
\\
\le&
\left\|
\sup_{0\le k\le N}
\left|
\left(\Phi(t,t_k)-\Phi_k^{(h)}\right)
\nabla G(X_{t_k}^{t,x})^\top J_{t_k}^{t,x}e_i
\right|
\right\|_{L^2}
+
\left\|
\sup_{0\le k\le N}
\left|
\Phi_k^{(h)}
\left(
\nabla G(X_{t_k}^{t,x})
-
\nabla G(X_k^{(h)})
\right)^\top
J_{t_k}^{t,x}e_i
\right|
\right\|_{L^2}
\\
&+
\left\|
\sup_{0\le k\le N}
\left|
\Phi_k^{(h)}
\nabla G(X_k^{(h)})^\top
\left(
J_{t_k}^{t,x}
-
J_k^{(h)}
\right)e_i
\right|
\right\|_{L^2}
+
\left\|
\sup_{0\le k\le N}
\left|
Y_{t_k}^{(i)}
-
\widetilde Y_k^{(i,h)}
\right|
\right\|_{L^2}
+
\left\|
\sup_{0\le k\le N}
\left|
\widetilde Y_k^{(i,h)}
-
Y_k^{(i,h)}
\right|
\right\|_{L^2}.
\end{align*}

We estimate the three payoff terms separately. By H\"older's inequality,
\begin{align*}
&
\left\|
\sup_{0\le k\le N}
\left|
\left(\Phi(t,t_k)-\Phi_k^{(h)}\right)
\nabla G(X_{t_k}^{t,x})^\top J_{t_k}^{t,x}e_i
\right|
\right\|_{L^2}
\\
\le&
\left\|
\sup_{0\le k\le N}
\left|
\Phi(t,t_k)-\Phi_k^{(h)}
\right|
\right\|_{L^4}
\left\|
\sup_{0\le k\le N}
\left|
\nabla G(X_{t_k}^{t,x})
\right|
\right\|_{L^8}
\left\|
\sup_{0\le k\le N}
\left|
J_{t_k}^{t,x}
\right|
\right\|_{L^8}
\\
\le&
C_\Phi M^2 h^{p_\Phi}.
\end{align*}
Similarly, using the Lipschitz continuity of $\nabla G$,
\begin{align*}
&
\left\|
\sup_{0\le k\le N}
\left|
\Phi_k^{(h)}
\left(
\nabla G(X_{t_k}^{t,x})
-
\nabla G(X_k^{(h)})
\right)^\top
J_{t_k}^{t,x}e_i
\right|
\right\|_{L^2}
\\
\le &
\left\|
\sup_{0\le k\le N}
\left|
\Phi_k^{(h)}
\right|
\right\|_{L^8}
\left\|
\sup_{0\le k\le N}
\left|
\nabla G(X_{t_k}^{t,x})
-
\nabla G(X_k^{(h)})
\right|
\right\|_{L^4}
\left\|
\sup_{0\le k\le N}
\left|
J_{t_k}^{t,x}
\right|
\right\|_{L^8}
\\
\le &
L_{\nabla G}M^2
\left\|
\sup_{0\le k\le N}
\left|
X_{t_k}^{t,x}
-
X_k^{(h)}
\right|
\right\|_{L^4}
\\
\le&
C_XL_{\nabla G}M^2h^{p_X}.
\end{align*}
For the Jacobian term,
\begin{align*}
&
\left\|
\sup_{0\le k\le N}
\left|
\Phi_k^{(h)}
\nabla G(X_k^{(h)})^\top
\left(
J_{t_k}^{t,x}
-
J_k^{(h)}
\right)e_i
\right|
\right\|_{L^2}
\\
\le &
\left\|
\sup_{0\le k\le N}
\left|
\Phi_k^{(h)}
\right|
\right\|_{L^8}
\left\|
\sup_{0\le k\le N}
\left|
\nabla G(X_k^{(h)})
\right|
\right\|_{L^8}
\left\|
\sup_{0\le k\le N}
\left|
J_{t_k}^{t,x}
-
J_k^{(h)}
\right|
\right\|_{L^4}
\\
\le &
C_JM^2h^{p_J}.
\end{align*}

By the exact-input integral discretization estimate,
\begin{align*}
\left\|
\sup_{0\le k\le N}
\left|
Y_{t_k}^{(i)}
-
\widetilde Y_k^{(i,h)}
\right|
\right\|_{L^2}
\le
C_{\mathrm{int}}^{(i)}h^{p_{\mathrm{int}}^{(i)}}.
\end{align*}
By the accumulated stability assumption,
\begin{align*}
&
\left\|
\sup_{0\le k\le N}
\left|
\widetilde Y_k^{(i,h)}
-
Y_k^{(i,h)}
\right|
\right\|_{L^2}
\\
\le &
L_{\mathrm{int}}^{(i)}
\bigg(
\left\|
\sup_{0\le j\le N}
\left|
\Phi(t,t_j)-\Phi_j^{(h)}
\right|
\right\|_{L^4}
+
\left\|
\sup_{0\le j\le N}
\left|
X_{t_j}^{t,x}-X_j^{(h)}
\right|
\right\|_{L^4}
+
\left\|
\sup_{0\le j\le N}
\left|
J_{t_j}^{t,x}-J_j^{(h)}
\right|
\right\|_{L^4}
\bigg)
\\
\le &
L_{\mathrm{int}}^{(i)}
\left(
C_\Phi h^{p_\Phi}
+
C_Xh^{p_X}
+
C_Jh^{p_J}
\right).
\end{align*}

Combining the above estimates, we obtain
\begin{align*}
&
\left\|
\sup_{0\le k\le N}
\left|
P_{t_k}^{(i)}
-
P_k^{(i,h)}
\right|
\right\|_{L^2}
\le
C_\Phi M^2h^{p_\Phi}
+
C_XL_{\nabla G}M^2h^{p_X}
+
C_JM^2h^{p_J}
+
C_{\mathrm{int}}^{(i)}h^{p_{\mathrm{int}}^{(i)}}
+
L_{\mathrm{int}}^{(i)}
\left(
C_\Phi h^{p_\Phi}
+
C_Xh^{p_X}
+
C_Jh^{p_J}
\right).
\end{align*}
Therefore, 
\begin{align*}
\left\|
\sup_{0\le k\le N}
\left|
P_{t_k}^{(i)}
-
P_k^{(i,h)}
\right|
\right\|_{L^2}
\le
C h^p,
\quad
p=\min\{p_X,p_\Phi,p_J,p_{\mathrm{int}}^{(i)}\}.
\end{align*}
\end{proof}

\paragraph{Construction of
$\mathcal S_{\mathrm{int}}^{(i),\mathrm{FBT}}$}
\label{para:construction-first-greek}
Fix $1\le i\le d$. Recall that the first-order Greek integral is given by
\begin{align}
Y_{t_k}^{(i)}
=&\int_t^{t_k}\Phi(t,s)
\nabla_x\!\Big(F(s,X_s^{t,x})+d(s)H(s,X_s^{t,x})\Big)^\top
J_s^{t,x}e_i\,\mathrm ds
+\int_t^{t_k}\Phi(t,s)\nabla_xH(s,X_s^{t,x})^\top
J_s^{t,x}e_i\,\mathrm d \overleftarrow{B}_s.
\label{eq:first-order-greek-integral}
\end{align}
Here $J_s^{t,x}=\nabla_x X_s^{t,x}$ is the Jacobian flow and
$e_i$ is the $i$-th unit vector in $\mathbb R^d$.

For notational simplicity, throughout this subsection we write
\[
X_s=X_s^{t,x},\quad
J_s=J_s^{t,x},\quad
\Phi_s=\Phi(t,s).
\]

Recall the mixed iterated integral
\begin{align}
J^{WB}_{s,h}
=
\Big(J_{ij}^{WB}(s,h)\Big)_{1\le i\le \ell,\ 1\le j\le d},
\quad
J_{ij}^{WB}(s,h)
:=
\int_s^{s+h}
(W_r^j-W_s^j)\,
\mathrm d \overleftarrow{B}_r^i.
\label{eq:def-JWB-first-greek}
\end{align}
Similarly, define the backward iterated integral
\begin{align}
J^{BB}_{s,h}
=
\Big(J_{ij}^{BB}(s,h)\Big)_{1\le i,j\le \ell},
\quad
J_{ij}^{BB}(s,h)
:=
\int_s^{s+h}
(B_{s}^j-B_r^j)\,
\mathrm d \overleftarrow{B}_r^i.
\label{eq:def-JBB-first-greek}
\end{align}

We define the first-order Greek integral discretization operator by
\begin{align}
&\mathcal S_{\mathrm{int}}^{(i),\mathrm{FBT}}
\left(\Phi,X,J,s,h;\Delta W_s,\Delta\overleftarrow B_s\right)\notag\\
=& \Phi h\,\nabla_x\!\Big(F(s,X)+d(s)H(s,X)\Big)^\top Je_i
+\Phi\left(\nabla_xH(s,X)^\top Je_i\right)\cdot\Delta\overleftarrow B_s
+\Phi\sum_{j=1}^{\ell}\sum_{\alpha=1}^{\ell}d_\alpha(s)
\left(\nabla_xH_j(s,X)^\top Je_i\right)
J_{j\alpha}^{BB}(s,h)\notag\\
&+\Phi\sum_{j=1}^{\ell}\sum_{a=1}^{d}
\bigg[\nabla_x^2H_j(s,X)\big[\sigma_{\cdot a}(X),Je_i\big]
+\nabla_xH_j(s,X)^\top\big(\nabla_x\sigma_{\cdot a}(X)Je_i\big)
+\widetilde c_a(s)\nabla_xH_j(s,X)^\top Je_i\bigg]J_{ja}^{WB}(s,h).
\label{eq:Sint-first-greek-mil}
\end{align}

The construction of this operator is natural, 
as it is obtained from first-order Taylor expansions of the integrands.
For $r\in[s,s+h]$, the local expansions are
\begin{align}
X_r-X_s
=\sigma(X_s)(W_r-W_s)+R^X_{s,r},
\label{eq:X-expansion-first-greek}
\end{align}
and
\begin{align}
J_r-J_s
=\sum_{a=1}^d\nabla_x\sigma_{\cdot a}(X_s)J_s
(W_r^a-W_s^a)+R^J_{s,r}.
\label{eq:J-expansion-first-greek}
\end{align}
Moreover,
\begin{align}
\Phi(t,r)
=\Phi(t,s)+\Phi(t,s)\widetilde c(s)\cdot(W_r-W_s)
+\Phi(t,s)d(s)\cdot(B_{s}-B_r)
+R^\Phi_{s,r}.
\label{eq:Phi-expansion-first-greek}
\end{align}
Under the standard smoothness and moment assumptions, for every $q\ge2$,
\begin{align}
\|R^X_{s,r}\|_{L^q}
+\|R^J_{s,r}\|_{L^q}
+\|R^\Phi_{s,r}\|_{L^q}
\le C(r-s).
\label{eq:basic-remainder-first-greek}
\end{align}

For each $1\le j\le \ell$, expanding with respect to $X_r$, we obtain
\begin{align}
&\nabla_xH_j(r,X_r)^\top J_se_i
=\nabla_xH_j(s,X_s)^\top J_se_i
+(X_r-X_s)^\top\nabla_x^2H_j(s,X_s)
J_se_i+R^{X,H,i}_{j,s,r}.
\label{eq:X-part-taylor-first-greek}
\end{align}
For the Taylor remainder in the spatial expansion of
$\nabla_xH_j(r,X_r)^\top J_se_i$, we have
\begin{align}
|R^{X,H,i}_{j,s,r}|
\le C\left(|r-s|+|X_r-X_s|^2\right)|J_se_i|.
\label{eq:RXH-pointwise-bound-first-greek}
\end{align}
Hence, using
$\|X_r-X_s\|_{L^{2q}}\le C(r-s)^{1/2}$
and the uniform moment bounds of $J_s$, we obtain
\begin{align}
\|R^{X,H,i}_{j,s,r}\|_{L^q}
\le C(r-s).
\label{eq:RXH-bound-first-greek}
\end{align}
Using \eqref{eq:X-expansion-first-greek},
we deduce
\begin{align}
&(X_r-X_s)^\top\nabla_x^2H_j(s,X_s)J_se_i
=\sum_{a=1}^d\sigma_{\cdot a}(X_s)^\top
\nabla_x^2H_j(s,X_s)J_se_i(W_r^a-W_s^a)
+\widetilde R^{X,H,i}_{j,s,r}.
\label{eq:X-H-expanded-no-Psi-first-greek}
\end{align}

Next, expanding with respect to the Jacobian flow $J_r$.
Using~\eqref{eq:J-expansion-first-greek}
we obtain
\begin{align}
&\nabla_xH_j(s,X_s)^\top
(J_r-J_s)e_i
=\sum_{a=1}^d\nabla_xH_j(s,X_s)^\top
\Big(\nabla_x\sigma_{\cdot a}(X_s)J_se_i\Big)
(W_r^a-W_s^a)+R^{J,H,i}_{j,s,r}.
\label{eq:J-H-expanded-no-Psi-first-greek}
\end{align}
There
$R^{J,H,i}_{j,s,r}
=\nabla_xH_j(s,X_s)^\top R^J_{s,r}e_i,$
and hence
$\|R^{J,H,i}_{j,s,r}\|_{L^q}
\le C(r-s).$

Combining
\eqref{eq:X-part-taylor-first-greek}
and
\eqref{eq:J-H-expanded-no-Psi-first-greek},
we finally obtain
\begin{align}
&\nabla_xH_j(r,X_r)^\top J_re_i\notag\\
=&\nabla_xH_j(s,X_s)^\top J_se_i
+\sum_{a=1}^d\bigg[\sigma_{\cdot a}(X_s)^\top
\nabla_x^2H_j(s,X_s)J_se_i
+\nabla_xH_j(s,X_s)^\top
\Big(\nabla_x\sigma_{\cdot a}(X_s)J_se_i\Big)\bigg](W_r^a-W_s^a)
+R^{H,i}_{j,s,r},
\label{eq:H-gradient-expansion-no-Psi-first-greek}
\end{align}
where $\|R^{H,i}_{j,s,r}\|_{L^q}\le C(r-s)$.

Combining \eqref{eq:Phi-expansion-first-greek} and
\eqref{eq:H-gradient-expansion-no-Psi-first-greek}, we obtain
\begin{align}
&\Phi(t,r)\nabla_xH_j(r,X_r)^\top J_re_i\notag\\
=&\Phi(t,s)\nabla_xH_j(s,X_s)^\top J_se_i\notag\\
&+\Phi(t,s)\sum_{a=1}^d
\bigg[\sigma_{\cdot a}(X_s)^\top
\nabla_x^2H_j(s,X_s)J_se_i
+\nabla_xH_j(s,X_s)^\top
\Big(\nabla_x\sigma_{\cdot a}(X_s)J_se_i\Big)
+\widetilde c_a(s)
\nabla_xH_j(s,X_s)^\top J_se_i
\bigg](W_r^a-W_s^a)\notag\\
&+\Phi(t,s)\sum_{\alpha=1}^{\ell}
d_\alpha(s)
\left(\nabla_xH_j(s,X_s)^\top J_se_i\right)
(B_{s}^{\alpha}-B_r^\alpha)
+R^{\Phi H,i}_{j,s,r}.
\label{eq:Phi-H-gradient-expansion-first-greek}
\end{align}
There
\begin{align}
|R^{\Phi H,i}_{j,s,r}|
\le C|R^\Phi_{s,r}|
+C|R^{H,i}_{j,s,r}|
+C\Big(|W_r-W_s|+|B_{s}-B_r|\Big)|W_r-W_s|.
\label{eq:RPhiH-pointwise-bound-first-greek}
\end{align}
Using
\[
\|W_r-W_s\|_{L^{2q}}\le C(r-s)^{1/2},
\quad\|B_{s}-B_r\|_{L^{2q}}\le C(r-s)^{1/2}\le Ch^{1/2},
\]
we obtain
\begin{align}
\|R^{\Phi H,i}_{j,s,r}\|_{L^q}
\le C(r-s).
\label{eq:RPhiH-bound-sharp-first-greek}
\end{align}
In particular,
\begin{align}
\int_s^{s+h}
\|R^{\Phi H,i}_{j,s,r}\|_{L^q}^2\,\mathrm dr
\le Ch^3.
\label{eq:RPhiH-integrated-bound-first-greek}
\end{align}

Finally, substituting \eqref{eq:Phi-H-gradient-expansion-first-greek} into the backward
stochastic integral gives
\begin{align}
&\int_s^{s+h}\Phi(t,r)\nabla_xH(r,X_r)^\top J_re_i\,
\mathrm d \overleftarrow{B}_r\notag\\
=&\Phi(t,s)
\left(\nabla_xH(s,X_s)^\top J_se_i\right)
\cdot\Delta\overleftarrow B_s\notag\\
&+\Phi(t,s)
\sum_{j=1}^{\ell}\sum_{a=1}^{d}
\bigg[\sigma_{\cdot a}(X_s)^\top
\nabla_x^2H_j(s,X_s)J_se_i
+\nabla_xH_j(s,X_s)^\top
\big(\nabla_x\sigma_{\cdot a}(X_s)J_se_i\big)
+\widetilde c_a(s)
\nabla_xH_j(s,X_s)^\top J_se_i
\bigg]
J_{ja}^{WB}(s,h)
\notag\\
&+
\Phi(t,s)
\sum_{j=1}^{\ell}\sum_{\alpha=1}^{\ell}
d_\alpha(s)
\left(
\nabla_xH_j(s,X_s)^\top J_se_i
\right)
J_{ j\alpha}^{BB}(s,h)
+R^{B,i}_{s,h}.
\label{eq:backward-expansion-first-greek}
\end{align}
Here the remainder is given by
\begin{align}
R^{B,i}_{s,h}:=\sum_{j=1}^{\ell}
\int_s^{s+h}R^{\Phi H,i}_{j,s,r}
\mathrm d \overleftarrow{B}_r^j.
\label{eq:RB-definition-first-greek}
\end{align}

Similarly, for the time integral part,
we can expand the integrand as
\begin{align}
&\Phi(t,r)
\nabla_x\!\Big(F(r,X_r)+d(r)H(r,X_r)\Big)^\top J_re_i\notag\\
=&\Phi(t,s)\nabla_x\!\Big(F(s,X_s)+d(s)H(s,X_s)\Big)^\top J_se_i\notag\\
&+\Phi(t,s)\sum_{a=1}^{d}
\bigg[\sigma_{\cdot a}(X_s)^\top
\nabla_x^2\!\Big(F(s,X_s)+d(s)H(s,X_s)\Big)J_se_i
+\nabla_x\!\Big(F(s,X_s)+d(s)H(s,X_s)\Big)^\top
\big(\nabla_x\sigma_{\cdot a}(X_s)J_se_i\big)\notag\\
&\hspace{4cm}
+\widetilde c_a(s)
\nabla_x\!\Big(F(s,X_s)+d(s)H(s,X_s)\Big)^\top J_se_i\bigg]
(W_r^a-W_s^a)\notag\\
&+\Phi(t,s)\sum_{\alpha=1}^{\ell}
d_\alpha(s)\nabla_x\!\Big(F(s,X_s)+d(s)H(s,X_s)\Big)^\top J_se_i\,
(B_{s}^\alpha-B_r^\alpha)+R^{D,i}_{s,r}
\label{eq:drift-integrand-expansion-first-greek}
\end{align}
with
$\|R^{D,i}_{s,r}\|_{L^q}
\le
C_q(r-s).$

Integrating \eqref{eq:drift-integrand-expansion-first-greek} over
$[s,s+h]$, we obtain
\begin{align}
&\int_s^{s+h}\Phi(t,r)
\nabla_x\!\Big(F(r,X_r)+d(r)H(r,X_r)\Big)^\top J_re_i
\,\mathrm dr
=\Phi(t,s)h
\nabla_x\!\Big(F(s,X_s)+d(s)H(s,X_s)\Big)^\top J_se_i
+A^{D,i}_{s,h}+M^{D,i}_{s,h},
\label{eq:drift-expansion-first-greek}
\end{align}
where
\begin{align}
A^{D,i}_{s,h}=\int_s^{s+h}R^{D,i}_{s,r}\,\mathrm dr,
\label{eq:AD-explicit-first-greek}
\end{align}
and
\begin{align}
M^{D,i}_{s,h}
=&\Phi(t,s)
\sum_{a=1}^{d}
\bigg[
\sigma_{\cdot a}(X_s)^\top
\nabla_x^2\!\Big(F(s,X_s)+d(s)H(s,X_s)\Big)
J_se_i
+\nabla_x\!\Big(F(s,X_s)+d(s)H(s,X_s)\Big)^\top
\big(\nabla_x\sigma_{\cdot a}(X_s)J_se_i\big)\notag\\
&\hspace{2.8cm}+\widetilde c_a(s)
\nabla_x\!\Big(F(s,X_s)+d(s)H(s,X_s)\Big)^\top J_se_i\bigg]
\int_s^{s+h}(W_r^a-W_s^a)\,\mathrm dr\notag\\
&+\Phi(t,s)\sum_{\alpha=1}^{\ell}
d_\alpha(s)\nabla_x\!\Big(F(s,X_s)+d(s)H(s,X_s)\Big)^\top J_se_i
\int_s^{s+h}(B_{s}^\alpha-B_r^\alpha)\,\mathrm dr.
\label{eq:MD-explicit-first-greek}
\end{align}
Here $A^{D,i}_{s,h}$ is the accumulated Taylor remainder in the
time integral and is of finite-variation type. The term
$M^{D,i}_{s,h}$ collects the first-order stochastic fluctuations of
$X$, $J$, and $\Phi$ inside the time integral. It is centered conditionally
on the information at time $s$ and is therefore treated as the local
martingale-type contribution.
We have
\begin{align}
\|A^{D,i}_{s,h}\|_{L^q}
\le\int_s^{s+h}\|R^{D,i}_{s,r}\|_{L^q}\,\mathrm dr
\le C\int_s^{s+h}(r-s)\,\mathrm dr
\le Ch^2.
\label{eq:AD-bound-first-greek}
\end{align}
Moreover, since
\begin{align}
\left\|\int_s^{s+h}(W_r^a-W_s^a)\,\mathrm dr\right\|_{L^q}
\le Ch^{3/2},
\quad
\left\|
\int_s^{s+h}(B_{s}^\alpha-B_r^\alpha)\,\mathrm dr
\right\|_{L^q}
\le Ch^{3/2},
\end{align}
and the coefficients have uniformly bounded moments, we also have
\begin{align}
\|M^{D,i}_{s,h}\|_{L^q}
\le Ch^{3/2}.
\label{eq:MD-bound-first-greek}
\end{align}

Combining \eqref{eq:backward-expansion-first-greek} and
\eqref{eq:drift-expansion-first-greek}, we have the one-step consistency
relation
\begin{align}
&\int_s^{s+h}\Phi(t,r)
\nabla_x\!\Big(F(r,X_r)+d(r)H(r,X_r)\Big)^\top J_re_i\,\mathrm dr+
\int_s^{s+h}\Phi(t,r)\nabla_xH(r,X_r)^\top J_re_i\,
\mathrm d \overleftarrow{B}_r\notag\\
=&\mathcal S_{\mathrm{int}}^{(i),\mathrm{FBT}}
\left(\Phi(t,s),X_s,J_s,s,h;
\Delta W_s,\Delta\overleftarrow B_s\right)
+R^{B,i}_{s,h}+A^{D,i}_{s,h}+M^{D,i}_{s,h}.
\label{eq:one-step-consistency-first-greek}
\end{align}

\paragraph{Proof of Proposition~\ref{prop:first-order-greek-strong-error}}
\label{para:proof of prop:first-order-greek-strong-error}
\begin{proposition*}[Strong-error order of the first-order Greek integral discretization]
Assume that the coefficients $b,\sigma,F,H,d,\widetilde c$ are sufficiently
smooth with bounded derivatives up to the order used above, and assume that
$X$, $J$, and $\Phi$ have uniformly bounded moments of all required orders.
Then the approximation generated by
\begin{align}
\widetilde Y_{k+1}^{(i,h)}
=
\widetilde Y_k^{(i,h)}
+
\mathcal S_{\mathrm{int}}^{(i),\mathrm{FBT}}
\left(\Phi(t,t_k),X_{t_k}^{t,x},J_{t_k}^{t,x},
t_k,h;\Delta W_{t_k},\Delta\overleftarrow B_{t_k}\right),
\quad
\widetilde Y_0^{(i,h)}=0,
\end{align}
satisfies
\begin{align}
\left\|
\sup_{0\le k\le N}
\left|
Y_{t_k}^{(i)}
-
\widetilde Y_k^{(i,h)}
\right|
\right\|_{L^q}
\le C_qh,
\end{align}
for every $q\ge2$. Consequently, by Jensen's inequality, the same estimate
also holds for every $0<q<2$. Hence
$\mathcal S_{\mathrm{int}}^{(i),\mathrm{FBT}}$ has strong-error order $1$ in the sense
of Definition~\ref{def:Strong error of S int Greek}.
\end{proposition*}

\begin{proof}

By \eqref{eq:one-step-consistency-first-greek}, for each $0\le k\le N$,
\begin{align}
Y_{t_k}^{(i)}-\widetilde Y_k^{(i,h)}
=
\sum_{j=0}^{k-1}A_{t_j,h}^{D,i}
+
\sum_{j=0}^{k-1}M_{t_j,h}^{D,i}
+
\sum_{j=0}^{k-1}R_{t_j,h}^{B,i}.
\label{eq:global-error-decomposition-first-greek}
\end{align}
For $k=0$, the sums are 0 and the identity is consistent with
$Y_{t_0}^{(i)}=\widetilde Y_0^{(i,h)}=0$.

For the finite-variation part, using the pathwise bound
\[
\sup_{0\le k\le N}
\left|
\sum_{j=0}^{k-1}A_{t_j,h}^{D,i}
\right|
\le
\sum_{j=0}^{N-1}
\left|
A_{t_j,h}^{D,i}
\right|,
\]
we obtain
\begin{align}
\left\|
\sup_{0\le k\le N}
\left|
\sum_{j=0}^{k-1}A_{t_j,h}^{D,i}
\right|
\right\|_{L^q}
&\le
\sum_{j=0}^{N-1}
\left\|
A_{t_j,h}^{D,i}
\right\|_{L^q}
\le
C_qNh^2
\le
C_qh.
\label{eq:global-A-bound-first-greek}
\end{align}

For the martingale-type part, by the discrete forward--backward information
convention above, the \(W\)-part of
\(\{M_{t_j,h}^{D,i}\}_{j=0}^{N-1}\) is a forward martingale-difference
sequence, while the \(B\)-part is a reverse martingale-difference sequence.
Hence, applying the discrete Burkholder--Davis--Gundy inequality,
$\left\|M_{t_j,h}^{D,i}\right\|_{L^q}
\le C_qh^{3/2}$
give
\begin{align}
\left\|
\sup_{0\le k\le N}
\left|
\sum_{j=0}^{k-1}M_{t_j,h}^{D,i}
\right|
\right\|_{L^q}
\le&
C_q
\left\|
\left(
\sum_{j=0}^{N-1}
\left|
M_{t_j,h}^{D,i}
\right|^2
\right)^{1/2}
\right\|_{L^q}
\notag\\
=&
C_q
\left\|
\sum_{j=0}^{N-1}
\left|
M_{t_j,h}^{D,i}
\right|^2
\right\|_{L^{q/2}}^{1/2}
\notag\\
\le&
C_q
\left(
\sum_{j=0}^{N-1}
\left\|
M_{t_j,h}^{D,i}
\right\|_{L^q}^2
\right)^{1/2}
\notag\\
\le&
C_q(Nh^3)^{1/2}
\le
C_qh.
\label{eq:global-MD-bound-first-greek}
\end{align}
Here the third inequality uses Minkowski's inequality in $L^{q/2}$, which is
valid because $q\ge2$.

It remains to estimate the accumulated backward stochastic remainder. By
definition,
\begin{align}
\sum_{j=0}^{k-1}R_{t_j,h}^{B,i}
=
\sum_{j=0}^{k-1}
\sum_{l=1}^{\ell}
\int_{t_j}^{t_{j+1}}
R^{\Phi H,i}_{l,t_j,r}
\,\mathrm d \overleftarrow{B}_r^l.
\end{align}
Applying the Burkholder--Davis--Gundy inequality for backward stochastic
integrals, equivalently after reversing time, yields
\begin{align}
&
\left\|
\sup_{0\le k\le N}
\left|
\sum_{j=0}^{k-1}R_{t_j,h}^{B,i}
\right|
\right\|_{L^q}
\notag\\
\le &
C_q
\left\|
\left(
\sum_{j=0}^{N-1}
\sum_{l=1}^{\ell}
\int_{t_j}^{t_{j+1}}
\left|
R^{\Phi H,i}_{l,t_j,r}
\right|^2
\,\mathrm dr
\right)^{1/2}
\right\|_{L^q}
\notag\\
\le &
C_q
\left(
\sum_{j=0}^{N-1}
\sum_{l=1}^{\ell}
\int_{t_j}^{t_{j+1}}
\left\|
R^{\Phi H,i}_{l,t_j,r}
\right\|_{L^q}^2
\,\mathrm dr
\right)^{1/2}.
\label{eq:global-RB-BDG-first-greek}
\end{align}
Using the integrated remainder estimate
\begin{align}
\sum_{l=1}^{\ell}
\int_{t_j}^{t_{j+1}}
\left\|
R^{\Phi H,i}_{l,t_j,r}
\right\|_{L^q}^2
\,\mathrm dr
\le
C_qh^3,
\end{align}
we get
\begin{align}
\left\|
\sup_{0\le k\le N}
\left|
\sum_{j=0}^{k-1}R_{t_j,h}^{B,i}
\right|
\right\|_{L^q}
\le
C_q(Nh^3)^{1/2}
\le
C_qh.
\label{eq:global-RB-bound-first-greek}
\end{align}

Combining
\eqref{eq:global-A-bound-first-greek},
\eqref{eq:global-MD-bound-first-greek},
and
\eqref{eq:global-RB-bound-first-greek} with
\eqref{eq:global-error-decomposition-first-greek}, we conclude that
\begin{align}
\left\|
\sup_{0\le k\le N}
\left|
Y_{t_k}^{(i)}
-
\widetilde Y_k^{(i,h)}
\right|
\right\|_{L^q}
\le
C_qh,
\quad q\ge2.
\end{align}

Finally, for $0<q<2$, Jensen's inequality gives
\begin{align*}
\left\|
\sup_{0\le k\le N}
\left|
Y_{t_k}^{(i)}
-
\widetilde Y_k^{(i,h)}
\right|
\right\|_{L^q}
\le
\left\|
\sup_{0\le k\le N}
\left|
Y_{t_k}^{(i)}
-
\widetilde Y_k^{(i,h)}
\right|
\right\|_{L^2}
\le
C_2h.
\end{align*}
This proves the claimed strong-error order one estimate.
\end{proof}

\paragraph{Proof of Proposition~\ref{prop:first-order-greek accumulated stability}}
\label{para:proof of prop:first-order-greek accumulated stability}
\begin{proposition*}
Assume that
$F\in C_b^2$, $H\in C_b^3$, $\sigma\in C_b^2$.
Assume also that $d,\widetilde c$ are bounded. Moreover, assume that the exact
and numerical input processes satisfy the uniform moment bound
\small
\begin{align*}
&
\left\|\sup_{0\le j\le N}|\Phi(t,t_j)|\right\|_{L^{12}}
+
\left\|\sup_{0\le j\le N}|\Phi_j^{(h)}|\right\|_{L^{12}}
+
\left\|\sup_{0\le j\le N}|X_{t_j}^{t,x}|\right\|_{L^{12}}
+
\left\|\sup_{0\le j\le N}|X_j^{(h)}|\right\|_{L^{12}}
+
\left\|\sup_{0\le j\le N}|J_{t_j}^{t,x}|\right\|_{L^{12}}
+
\left\|\sup_{0\le j\le N}|J_j^{(h)}|\right\|_{L^{12}}
<\infty .
\end{align*}
\normalsize
Then $\mathcal S_{\mathrm{int}}^{(i),\mathrm{FBT}}$ satisfies the accumulated
stability estimate. More precisely, there exists
$L_{\mathrm{int}}^{(i)}>0$, independent of $h$, such that
\begin{align*}
\left\|
\sup_{0\le k\le N}
\left|
\widetilde Y_k^{(i,h)}-Y_k^{(i,h)}
\right|
\right\|_{L^2}
\le L_{\mathrm{int}}^{(i)}
\bigg(
&
\left\|
\sup_{0\le j\le N}
\left|
\Phi(t,t_j)-\Phi_j^{(h)}
\right|
\right\|_{L^4}+
\left\|
\sup_{0\le j\le N}
\left|
X_{t_j}^{t,x}-X_j^{(h)}
\right|
\right\|_{L^4}+
\left\|
\sup_{0\le j\le N}
\left|
J_{t_j}^{t,x}-J_j^{(h)}
\right|
\right\|_{L^4}
\bigg).
\end{align*}
\end{proposition*}
\begin{proof}
Set
\[
\Delta\Phi_r:=\Phi(t,t_r)-\Phi_r^{(h)},\quad
\Delta X_r:=X_{t_r}^{t,x}-X_r^{(h)},\quad
\Delta J_r:=J_{t_r}^{t,x}-J_r^{(h)}
\]
and define
\[
\delta_\Phi
:=
\left\|
\sup_{0\le r\le N}
|\Delta\Phi_r|
\right\|_{L^4},
\quad
\delta_X
:=
\left\|
\sup_{0\le r\le N}
|\Delta X_r|
\right\|_{L^4},
\quad
\delta_J
:=
\left\|
\sup_{0\le r\le N}
|\Delta J_r|
\right\|_{L^4}.
\]
Write $\delta:=\delta_\Phi+\delta_X+\delta_J$.

By the definitions of $\widetilde Y_k^{(i,h)}$ and $Y_k^{(i,h)}$, we have
\begin{align}
\widetilde Y_k^{(i,h)}-Y_k^{(i,h)}
=\sum_{r=0}^{k-1}\Bigg[
&\mathcal S_{\mathrm{int}}^{(i),\mathrm{FBT}}
\left(\Phi(t,t_r),X_{t_r}^{t,x},J_{t_r}^{t,x},t_r,h;
\Delta W_{t_r},\Delta\overleftarrow B_{t_r}\right)
-
\mathcal S_{\mathrm{int}}^{(i),\mathrm{FBT}}
\left(\Phi_r^{(h)},X_r^{(h)},J_r^{(h)},t_r,h;
\Delta W_{t_r},\Delta\overleftarrow B_{t_r}\right)
\Bigg].
\label{eq:first-greek-stability-decomposition}
\end{align}
We estimate the contribution of each term in
$\mathcal S_{\mathrm{int}}^{(i),\mathrm{FBT}}$.

First, consider the time-integral term
$\Phi h\,\nabla_x\!\Big(F(s,X)+d(s)H(s,X)\Big)^\top Je_i,$
since $F\in C_b^2$, $H\in C_b^3$, and $d$ is bounded, the map
$x\mapsto \nabla_x(F+dH)(t_r,x)$
is bounded and globally Lipschitz, uniformly in $r$. Hence, by adding and
subtracting intermediate terms and using H\"older's inequality together with
the uniform moment bounds,
\begin{align*}
&
\left\|
\Phi(t,t_r)
\nabla_x(F+dH)(t_r,X_{t_r}^{t,x})^\top J_{t_r}^{t,x}e_i
-
\Phi_r^{(h)}
\nabla_x(F+dH)(t_r,X_r^{(h)})^\top J_r^{(h)}e_i
\right\|_{L^2}
\\
\le&
C\left(
\|\Delta\Phi_r\|_{L^4}
+
\|\Delta X_r\|_{L^4}
+
\|\Delta J_r\|_{L^4}
\right)
\\
\le& C\delta.
\end{align*}
Therefore, using the pathwise bound
\[
\sup_{0\le k\le N}
\left|
\sum_{r=0}^{k-1}h A_r
\right|
\le
\sum_{r=0}^{N-1}h|A_r|,
\]
we obtain
\begin{align}
&
\left\|
\sup_{0\le k\le N}
\left|
\sum_{r=0}^{k-1}h
\Big[
\Phi(t,t_r)
\nabla_x(F+dH)(t_r,X_{t_r}^{t,x})^\top J_{t_r}^{t,x}e_i
-
\Phi_r^{(h)}
\nabla_x(F+dH)(t_r,X_r^{(h)})^\top J_r^{(h)}e_i
\Big]
\right|
\right\|_{L^2}
\notag\\
\le&
\sum_{r=0}^{N-1}
h
\left\|
\Phi(t,t_r)
\nabla_x(F+dH)(t_r,X_{t_r}^{t,x})^\top J_{t_r}^{t,x}e_i
-
\Phi_r^{(h)}
\nabla_x(F+dH)(t_r,X_r^{(h)})^\top J_r^{(h)}e_i
\right\|_{L^2}
\notag\\
\le&
CT\delta.
\label{eq:first-greek-stability-time-term}
\end{align}

Next, consider the backward stochastic increment term
$\Phi\left(\nabla_xH(s,X)^\top Je_i\right)\cdot\Delta\overleftarrow B_s.$
Since $H\in C_b^3$, the map $x\mapsto\nabla_xH(t_r,x)$ is bounded and
globally Lipschitz. The same H\"older argument gives
\[
\|\Phi(t,t_r)
\nabla_xH(t_r,X_{t_r}^{t,x})^\top J_{t_r}^{t,x}e_i
-
\Phi_r^{(h)}
\nabla_xH(t_r,X_r^{(h)})^\top J_r^{(h)}e_i\|_{L^2}
\le
C\left(
\|\Delta\Phi_r\|_{L^4}
+
\|\Delta X_r\|_{L^4}
+
\|\Delta J_r\|_{L^4}
\right)
\le C\delta .
\]
By the discrete Burkholder--Davis--Gundy inequality for backward stochastic
increments, equivalently after reversing time,
\begin{align}
&\left\|
\sup_{0\le k\le N}
\left|
\sum_{r=0}^{k-1}
\left(\Phi(t,t_r)
\nabla_xH(t_r,X_{t_r}^{t,x})^\top J_{t_r}^{t,x}e_i
-
\Phi_r^{(h)}
\nabla_xH(t_r,X_r^{(h)})^\top J_r^{(h)}e_i\right)\cdot\Delta\overleftarrow B_{t_r}
\right|
\right\|_{L^2}\notag\\
\le&
C
\left\|
\left(
\sum_{r=0}^{N-1}
h\left|\Phi(t,t_r)
\nabla_xH(t_r,X_{t_r}^{t,x})^\top J_{t_r}^{t,x}e_i
-
\Phi_r^{(h)}
\nabla_xH(t_r,X_r^{(h)})^\top J_r^{(h)}e_i\right|^2
\right)^{1/2}
\right\|_{L^2}\notag\\
\le&
C
\left(
\sum_{r=0}^{N-1}
h\left\|\Phi(t,t_r)
\nabla_xH(t_r,X_{t_r}^{t,x})^\top J_{t_r}^{t,x}e_i
-
\Phi_r^{(h)}
\nabla_xH(t_r,X_r^{(h)})^\top J_r^{(h)}e_i\right\|_{L^2}^2
\right)^{1/2}\notag\\
\le&
C\sqrt T\,\delta .
\label{eq:first-greek-stability-B-term}
\end{align}

Now consider the $BB$ correction term. With the convention
\[
J_{j\alpha}^{BB}(s,h)
=
\int_s^{s+h}
(B_{s}^{\alpha}-B_u^\alpha)\,
\mathrm d \overleftarrow{B}_u^j,
\]
the $BB$ contribution is
\[
\Phi
\sum_{j=1}^{\ell}\sum_{\alpha=1}^{\ell}
d_\alpha(s)
\left(\nabla_xH_j(s,X)^\top Je_i\right)
J_{j\alpha}^{BB}(s,h).
\]
Set
\[
D_{j\alpha,r}^{BB}
:=
d_\alpha(t_r)
\Big[
\Phi(t,t_r)
\nabla_xH_j(t_r,X_{t_r}^{t,x})^\top J_{t_r}^{t,x}e_i
-
\Phi_r^{(h)}
\nabla_xH_j(t_r,X_r^{(h)})^\top J_r^{(h)}e_i
\Big].
\]
Since $d$ is bounded and $H\in C_b^3$,
\[
\|D_{j\alpha,r}^{BB}\|_{L^2}
\le C\delta.
\]
Moreover, for every $j,\alpha$,
\[
\|J_{j\alpha}^{BB}(t_r,h)\|_{L^2}\le Ch.
\]
Therefore,
\begin{align}
&
\left\|
\sup_{0\le k\le N}
\left|
\sum_{r=0}^{k-1}
\sum_{j=1}^{\ell}\sum_{\alpha=1}^{\ell}
D_{j\alpha,r}^{BB}
J_{j\alpha}^{BB}(t_r,h)
\right|
\right\|_{L^2}
\le
\sum_{r=0}^{N-1}
\sum_{j=1}^{\ell}\sum_{\alpha=1}^{\ell}
\left\|
D_{j\alpha,r}^{BB}
J_{j\alpha}^{BB}(t_r,h)
\right\|_{L^2}
\le
C\sum_{r=0}^{N-1}h\delta
\le
CT\delta .
\label{eq:first-greek-stability-BB-term}
\end{align}

It remains to estimate the $WB$ correction term. With the convention
\[
J_{ja}^{WB}(s,h)
=
\int_s^{s+h}
(W_u^a-W_s^a)\,
\mathrm d \overleftarrow{B}_u^j,
\]
the $WB$ contribution is
\[
\Phi
\sum_{j=1}^{\ell}\sum_{a=1}^{d}
Q_{ja}(s,X,J)
J_{ja}^{WB}(s,h),
\]
where
\begin{align*}
Q_{ja}(s,X,J)
:=
&\sigma_{\cdot a}(X)^\top\nabla_x^2H_j(s,X)Je_i
+
\nabla_xH_j(s,X)^\top
\big(\nabla_x\sigma_{\cdot a}(X)Je_i\big)+
\widetilde c_a(s)
\nabla_xH_j(s,X)^\top Je_i .
\end{align*}
For fixed $a$ and $j$, define the coefficient difference
\[
D_{ja,r}^{WB}
:=
\Phi(t,t_r)Q_{ja}(t_r,X_{t_r}^{t,x},J_{t_r}^{t,x})
-
\Phi_r^{(h)}Q_{ja}(t_r,X_r^{(h)},J_r^{(h)}).
\]
Since $H\in C_b^3$, $\sigma\in C_b^2$, and $\widetilde c$ is bounded, we have
the Lipschitz-type estimate
\[
\left|
Q_{ja}(t_r,X_{t_r}^{t,x},J_{t_r}^{t,x})
-
Q_{ja}(t_r,X_r^{(h)},J_r^{(h)})
\right|
\le
C(1+|J_{t_r}^{t,x}|)
|\Delta X_r|
+
C|\Delta J_r|.
\]
Consequently, by adding and subtracting
$\Phi_r^{(h)}Q_{ja}(t_r,X_{t_r}^{t,x},J_{t_r}^{t,x})$, and using
H\"older's inequality with the uniform $L^{12}$ moment bound,
\[
\|D_{ja,r}^{WB}\|_{L^2}
\le
C\left(
\|\Delta\Phi_r\|_{L^4}
+
\|\Delta X_r\|_{L^4}
+
\|\Delta J_r\|_{L^4}
\right)
\le C\delta.
\]
Moreover,
\[
\|J_{ja}^{WB}(t_r,h)\|_{L^2}\le Ch.
\]
Again  we obtain
\begin{align}
&
\left\|
\sup_{0\le k\le N}
\left|
\sum_{r=0}^{k-1}
\sum_{j=1}^{\ell}\sum_{a=1}^{d}
D_{ja,r}^{WB}
J_{ja}^{WB}(t_r,h)
\right|
\right\|_{L^2}
\le
\sum_{r=0}^{N-1}
\sum_{j=1}^{\ell}\sum_{a=1}^{d}
\left\|
D_{ja,r}^{WB}
J_{ja}^{WB}(t_r,h)
\right\|_{L^2}
\le
C\sum_{r=0}^{N-1}h\delta
\le
CT\delta .
\label{eq:first-greek-stability-WB-term}
\end{align}

Combining
\eqref{eq:first-greek-stability-time-term},
\eqref{eq:first-greek-stability-B-term},
\eqref{eq:first-greek-stability-BB-term}, and
\eqref{eq:first-greek-stability-WB-term} in
\eqref{eq:first-greek-stability-decomposition}, we obtain
\[
\left\|
\sup_{0\le k\le N}
\left|
\widetilde Y_k^{(i,h)}-Y_k^{(i,h)}
\right|
\right\|_{L^2}
\le
L_{\mathrm{int}}^{(i)}
(\delta_\Phi+\delta_X+\delta_J).
\]
Substituting the definitions of
$\delta_\Phi,\delta_X,\delta_J$ gives
\begin{align*}
\left\|
\sup_{0\le k\le N}
\left|
\widetilde Y_k^{(i,h)}-Y_k^{(i,h)}
\right|
\right\|_{L^2}
\le
L_{\mathrm{int}}^{(i)}
\bigg(
&
\left\|
\sup_{0\le r\le N}
\left|
\Phi(t,t_r)-\Phi_r^{(h)}
\right|
\right\|_{L^4}
+
\left\|
\sup_{0\le r\le N}
\left|
X_{t_r}^{t,x}-X_r^{(h)}
\right|
\right\|_{L^4}
+
\left\|
\sup_{0\le r\le N}
\left|
J_{t_r}^{t,x}-J_r^{(h)}
\right|
\right\|_{L^4}
\bigg).
\end{align*}
This is the desired accumulated stability estimate.
\end{proof}

\subsection{Second-order Greek Estimators}

\paragraph{Proof of Proposition~\ref{prop:error-second-order-greek-payoff}}
\label{para:proof of prop:error-second-order-greek-payoff}
\begin{proposition*}[Strong-error order for the second-order Greek payoff]
Fix $(t,x)\in[0,T]\times\mathbb R^d$ and $1\le i,j\le d$.
Define
\begin{align*}
P_{t_k}^{(ij)}
=\Phi(t,t_k)
\Big((J_{t_k}^j)^\top\nabla_x^2G(X_{t_k}^{t,x})J_{t_k}^i
+\nabla_xG(X_{t_k}^{t,x})^\top K_{t_k}^{(ij),t,x}\Big)
+Y_{t_k}^{(ij)}.
\end{align*}
Let $\{X_k^{(h)},\Phi_k^{(h)},J_k^{(h)},K_k^{(ij,h)}\}_{k=0}^N$
be generated by $S_X,S_\Phi,S_J,S_K$ as in
Definition~\ref{def:Strong error of S X},
Definition~\ref{def:Strong error of S Phi},
Definition~\ref{def:Strong error of S J},
and Definition~\ref{def:Strong error of S K}.
Define
\begin{align*}
\begin{aligned}
Y_{k+1}^{(ij,h)}
=&
Y_k^{(ij,h)}
+
\mathcal S_{\mathrm{int}}^{(ij)}
\left(
\Phi_k^{(h)},
X_k^{(h)},
J_k^{(h)},
K_k^{(ij,h)},
t_k,h;
\Delta W_{t_k},\Delta\overleftarrow B_{t_k}
\right),
\quad
Y_0^{(ij,h)}=0,\\
P_k^{(ij,h)}
=&
\Phi_k^{(h)}
\Big(
(J_k^{j,h})^\top\nabla_x^2G(X_k^{(h)})J_k^{i,(h)}
+
\nabla_xG(X_k^{(h)})^\top K_k^{(ij,h)}
\Big)
+
Y_k^{(ij,h)}.
\end{aligned}
\end{align*}
Assume:

\begin{enumerate}
\item
The strong-error orders of $S_X,S_\Phi,S_J,S_K$ and
$\mathcal S_{\mathrm{int}}^{(ij)}$ are
$p_X,p_\Phi,p_J,p_K,p_{\mathrm{int}}^{(ij)}$, respectively.

\item
$\mathcal S_{\mathrm{int}}^{(ij)}$ satisfies the accumulated stability estimate:
there exists $L_{\mathrm{int}}^{(ij)}>0$, independent of $h$, such that
\begin{align*}
\left\|
\sup_{0\le k\le N}
|\widetilde Y_k^{(ij,h)}-Y_k^{(ij,h)}|
\right\|_{L^2}
\le
L_{\mathrm{int}}^{(ij)}
\bigg(&
\left\|
\sup_{0\le r\le N}
|\Phi(t,t_r)-\Phi_r^{(h)}|
\right\|_{L^4}
+
\left\|
\sup_{0\le r\le N}
|X_{t_r}^{t,x}-X_r^{(h)}|
\right\|_{L^4}\\
&+
\left\|
\sup_{0\le r\le N}
|J_{t_r}^{t,x}-J_r^{(h)}|
\right\|_{L^4}
+
\left\|
\sup_{0\le r\le N}
|K_{t_r}^{(ij),t,x}-K_r^{(ij,h)}|
\right\|_{L^4}
\bigg).
\end{align*}

\item
There exists $M>0$, independent of $h$, such that all factors appearing in
the payoff decomposition, namely
\[
\Phi(t,t_k),\ \Phi_k^{(h)},\
J_{t_k}^{t,x},\ J_k^{(h)},\
K_{t_k}^{(ij),t,x},\ K_k^{(ij,h)}
\quad
\text{ and }\quad
\nabla G(X_{t_k}^{t,x}),\ \nabla G(X_k^{(h)}),\
\nabla_x^2G(X_{t_k}^{t,x}),\ \nabla_x^2G(X_k^{(h)})
\]
are uniformly bounded in
$L^{12}_{W,B}$
by $M$.

\item
$\nabla G$ and $\nabla_x^2G$ are globally Lipschitz with constants
$L_{\nabla G}$ and $L_{\nabla^2 G}$.
\end{enumerate}

Then
\begin{align*}
\left\|
\sup_{0\le k\le N}
\left|
P_{t_k}^{(ij)}-P_k^{(ij,h)}
\right|
\right\|_{L^2}
=
\mathcal{O}(h^p),
\quad
p=\min\{p_X,p_\Phi,p_J,p_K,p_{\mathrm{int}}^{(ij)}\}.
\end{align*}
\end{proposition*}

\begin{proof}

For each $0\le k\le N$, we have
\begin{align*}
P_{t_k}^{(ij)}-P_k^{(ij,h)}
=&\Bigg[\Phi(t,t_k)
\Big((J_{t_k}^j)^\top\nabla_x^2G(X_{t_k}^{t,x})J_{t_k}^i
+\nabla_xG(X_{t_k}^{t,x})^\top K_{t_k}^{(ij),t,x}\Big)
+Y_{t_k}^{(ij)}\Bigg] \\
&-\Bigg[\Phi_k^{(h)}
\Big((J_k^{j,h})^\top\nabla_x^2G(X_k^{(h)})J_k^{i,(h)}
+\nabla_xG(X_k^{(h)})^\top K_k^{(ij,h)}\Big)
+Y_k^{(ij,h)}\Bigg]\\
=&\left(\Phi(t,t_k)-\Phi_k^{(h)}\right)
\Big((J_{t_k}^j)^\top\nabla_x^2G(X_{t_k}^{t,x})J_{t_k}^i
+\nabla_xG(X_{t_k}^{t,x})^\top K_{t_k}^{(ij),t,x}\Big) \\
&+\Phi_k^{(h)}\Big((J_{t_k}^j)^\top
\big(\nabla_x^2G(X_{t_k}^{t,x})-\nabla_x^2G(X_k^{(h)})\big)
J_{t_k}^i\Big) 
+\Phi_k^{(h)}\Big((J_{t_k}^j-J_k^{j,h})^\top
\nabla_x^2G(X_k^{(h)})J_{t_k}^i\Big) \\
&+\Phi_k^{(h)}\Big((J_k^{j,h})^\top\nabla_x^2G(X_k^{(h)})
(J_{t_k}^i-J_k^{i,(h)})\Big) 
+\Phi_k^{(h)}\Big(\big(\nabla_xG(X_{t_k}^{t,x})-\nabla_xG(X_k^{(h)})
\big)^\top
K_{t_k}^{(ij),t,x}\Big) \\
&+\Phi_k^{(h)}\nabla_xG(X_k^{(h)})^\top
\left(K_{t_k}^{(ij),t,x}-K_k^{(ij,h)}\right) 
+\left(Y_{t_k}^{(ij)}-\widetilde Y_k^{(ij,h)}\right)
+\left(\widetilde Y_k^{(ij,h)}-Y_k^{(ij,h)}\right).
\end{align*}

Taking the supremum over $0\le k\le N$ and then the
$L^2$ norm, the triangle inequality gives
\begin{align*}
&
\left\|
\sup_{0\le k\le N}
\left|
P_{t_k}^{(ij)}-P_k^{(ij,h)}
\right|
\right\|_{L^2}
\le
T_\Phi+T_{\nabla^2G,X}+T_{J,j}+T_{J,i}
+T_{\nabla G,X}+T_K
+T_{\mathrm{int}}+T_{\mathrm{stab}},
\end{align*}
where the terms correspond to the eight summands in the decomposition above.

We estimate them one by one. By H\"older's inequality and the uniform
$L^{12}$ moment bounds,
\begin{align*}
T_\Phi
&:=
\left\|
\sup_{0\le k\le N}
\left|
\left(\Phi(t,t_k)-\Phi_k^{(h)}\right)
\Big(
(J_{t_k}^j)^\top\nabla_x^2G(X_{t_k}^{t,x})J_{t_k}^i
+
\nabla_xG(X_{t_k}^{t,x})^\top K_{t_k}^{(ij),t,x}
\Big)
\right|
\right\|_{L^2}
\le
C M^3
\left\|
\sup_{0\le k\le N}
\left|
\Phi(t,t_k)-\Phi_k^{(h)}
\right|
\right\|_{L^4}.
\end{align*}
Indeed, the quadratic $J^\top \nabla^2G J$ term is estimated with exponents
$4,12,12,12$, while the $\nabla G^\top K$ term is estimated with exponents
$4,8,8$.

For the Hessian-difference term, using the global Lipschitz continuity of
$\nabla_x^2G$,
\begin{align*}
T_{\nabla^2G,X}
&:=
\left\|
\sup_{0\le k\le N}
\left|
\Phi_k^{(h)}
(J_{t_k}^j)^\top
\big(
\nabla_x^2G(X_{t_k}^{t,x})
-
\nabla_x^2G(X_k^{(h)})
\big)
J_{t_k}^i
\right|
\right\|_{L^2}
\le
M^3L_{\nabla^2G}
\left\|
\sup_{0\le k\le N}
\left|
X_{t_k}^{t,x}-X_k^{(h)}
\right|
\right\|_{L^4}.
\end{align*}

For the two Jacobian-difference terms, H\"older's inequality gives
\begin{align*}
T_{J,j}
&:=
\left\|
\sup_{0\le k\le N}
\left|
\Phi_k^{(h)}
(J_{t_k}^j-J_k^{j,h})^\top
\nabla_x^2G(X_k^{(h)})J_{t_k}^i
\right|
\right\|_{L^2}
\le
M^3
\left\|
\sup_{0\le k\le N}
\left|
J_{t_k}^{t,x}-J_k^{(h)}
\right|
\right\|_{L^4},
\end{align*}
and similarly,
\begin{align*}
T_{J,i}
&:=
\left\|
\sup_{0\le k\le N}
\left|
\Phi_k^{(h)}
(J_k^{j,h})^\top
\nabla_x^2G(X_k^{(h)})
(J_{t_k}^i-J_k^{i,(h)})
\right|
\right\|_{L^2}
\le
M^3
\left\|
\sup_{0\le k\le N}
\left|
J_{t_k}^{t,x}-J_k^{(h)}
\right|
\right\|_{L^4}.
\end{align*}

For the gradient-difference term, using the global Lipschitz continuity of
$\nabla G$,
\begin{align*}
T_{\nabla G,X}
&:=
\left\|
\sup_{0\le k\le N}
\left|
\Phi_k^{(h)}
\big(
\nabla_xG(X_{t_k}^{t,x})
-
\nabla_xG(X_k^{(h)})
\big)^\top
K_{t_k}^{(ij),t,x}
\right|
\right\|_{L^2}
\le
M^2L_{\nabla G}
\left\|
\sup_{0\le k\le N}
\left|
X_{t_k}^{t,x}-X_k^{(h)}
\right|
\right\|_{L^4}.
\end{align*}

For the second-variational-process term,
\begin{align*}
T_K
&:=
\left\|
\sup_{0\le k\le N}
\left|
\Phi_k^{(h)}
\nabla_xG(X_k^{(h)})^\top
\left(
K_{t_k}^{(ij),t,x}-K_k^{(ij,h)}
\right)
\right|
\right\|_{L^2}
\le
M^2
\left\|
\sup_{0\le k\le N}
\left|
K_{t_k}^{(ij),t,x}-K_k^{(ij,h)}
\right|
\right\|_{L^4}.
\end{align*}

By the exact-input integral strong-error estimate for
$\mathcal S_{\mathrm{int}}^{(ij)}$,
\begin{align*}
T_{\mathrm{int}}
&:=
\left\|
\sup_{0\le k\le N}
\left|
Y_{t_k}^{(ij)}-\widetilde Y_k^{(ij,h)}
\right|
\right\|_{L^2}
\le
C_{\mathrm{int}}^{(ij)}h^{p_{\mathrm{int}}^{(ij)}}.
\end{align*}
By the accumulated stability assumption,
\begin{align*}
&T_{\mathrm{stab}}
:=
\left\|
\sup_{0\le k\le N}
\left|
\widetilde Y_k^{(ij,h)}-Y_k^{(ij,h)}
\right|
\right\|_{L^2}
\\
&\le
L_{\mathrm{int}}^{(ij)}
\bigg(
\left\|
\sup_{0\le r\le N}
\left|
\Phi(t,t_r)-\Phi_r^{(h)}
\right|
\right\|_{L^4}
+
\left\|
\sup_{0\le r\le N}
\left|
X_{t_r}^{t,x}-X_r^{(h)}
\right|
\right\|_{L^4}
+
\left\|
\sup_{0\le r\le N}
\left|
J_{t_r}^{t,x}-J_r^{(h)}
\right|
\right\|_{L^4}
+
\left\|
\sup_{0\le r\le N}
\left|
K_{t_r}^{(ij),t,x}-K_r^{(ij,h)}
\right|
\right\|_{L^4}
\bigg).
\end{align*}

Using the strong-error order of
$S_X,S_\Phi,S_J,S_K$, 
we obtain
\begin{align*}
&
\left\|
\sup_{0\le k\le N}
\left|
P_{t_k}^{(ij)}-P_k^{(ij,h)}
\right|
\right\|_{L^2}
\\
&\le
C M^3 C_\Phi h^{p_\Phi}
+
C M^2(ML_{\nabla^2G}+L_{\nabla G})C_Xh^{p_X}
+
C M^3 C_Jh^{p_J}
+
C M^2 C_Kh^{p_K}
+
C_{\mathrm{int}}^{(ij)}h^{p_{\mathrm{int}}^{(ij)}}
\\
&\quad+
L_{\mathrm{int}}^{(ij)}
\left(
C_\Phi h^{p_\Phi}
+
C_Xh^{p_X}
+
C_Jh^{p_J}
+
C_Kh^{p_K}
\right).
\end{align*}
Therefore, for all sufficiently small $h$,
\begin{align*}
\left\|
\sup_{0\le k\le N}
\left|
P_{t_k}^{(ij)}-P_k^{(ij,h)}
\right|
\right\|_{L^2}
\le
C h^p,
\quad
p=\min\{p_X,p_\Phi,p_J,p_K,p_{\mathrm{int}}^{(ij)}\}.
\end{align*}
\end{proof}

\paragraph{Construction of
$\mathcal S_{\mathrm{int}}^{(ij),\mathrm{FBT}}$}
\label{para:construction-second-greek}

For the discretization operator 
\begin{align}
&\mathcal S_{\mathrm{int}}^{(ij),\mathrm{FBT}}
\left(\Phi,X,J,K,s,h;
\Delta W_s,\Delta\overleftarrow B_s\right)\notag\\
=&\Phi h
\Big((J^j)^\top\nabla_x^2(F+dH)(s,X)J^i
+\nabla_x(F+dH)(s,X)^\top K\Big)
+\Phi\Big((J^j)^\top\nabla_x^2H(s,X)J^i
+\nabla_xH(s,X)^\top K\Big)
\cdot\Delta\overleftarrow B_s\notag\\
&+\Phi\sum_{\nu=1}^{\ell}\sum_{a=1}^{d}
\bigg[\nabla_x^3H_\nu(s,X)
[\sigma_{\cdot a}(X),J^i,J^j]
+\big((\nabla_x\sigma_{\cdot a}(X)J)^j\big)^\top
\nabla_x^2H_\nu(s,X)J^i
+(J^j)^\top\nabla_x^2H_\nu(s,X)
(\nabla_x\sigma_{\cdot a}(X)J)^i\notag\\
&\hspace{3cm}
+\sigma_{\cdot a}(X)^\top\nabla_x^2H_\nu(s,X)K
+\nabla_xH_\nu(s,X)^\top
\big(\nabla_x\sigma_{\cdot a}(X)K
+(J^i)^\top\nabla_x^2\sigma_{\cdot a}(X)J^j
\big)\notag\\
&\hspace{3cm}
+\widetilde c_a(s)
\Big((J^j)^\top\nabla_x^2H_\nu(s,X)J^i
+\nabla_xH_\nu(s,X)^\top K\Big)\bigg]
J_{\nu a}^{WB}(s,h)\notag\\
&+\Phi\sum_{\nu=1}^{\ell}\sum_{\alpha=1}^{\ell}
d_\alpha(s)\Big((J^j)^\top
\nabla_x^2H_\nu(s,X)J^i
+\nabla_xH_\nu(s,X)^\top K\Big)
J_{\nu\alpha}^{BB}(s,h),
\label{eq:Sint-second-greek-mil-full}
\end{align}
where
$J^i=Je_i$, $J^j:=Je_j$,
and
$(\nabla_x\sigma_{\cdot a}(X)J)^i
=\nabla_x\sigma_{\cdot a}(X)J^i,$
the construction is natural, as it is obtained from
first-order Taylor expansions of the integrands. 
For $r\in[s,s+h]$, we will repeatedly use the following first-order expansions:
\begin{align}
X_r-X_s
=&\sum_{a=1}^d\sigma_{\cdot a}(X_s)(W_r^a-W_s^a)
+R^X_{s,r},
\label{eq:X-expansion-second-greek}\\
J_r-J_s
=&\sum_{a=1}^d
\nabla_x\sigma_{\cdot a}(X_s)J_s
(W_r^a-W_s^a)
+R^J_{s,r},
\label{eq:J-expansion-second-greek}\\
K_r^{(ij)}-K_s^{(ij)}
=&\sum_{a=1}^d
\Big(\nabla_x\sigma_{\cdot a}(X_s)K_s^{(ij)}
+\left(J_s^i\right)^\top\nabla_x^2\sigma_{\cdot a}(X_s)J_s^j
\Big)(W_r^a-W_s^a)
+R^K_{s,r},
\label{eq:K-expansion-second-greek}\\
\Phi(t,r)
=&\Phi(t,s)
+\Phi(t,s)\widetilde c(s)\cdot(W_r-W_s)
+\Phi(t,s)d(s)\cdot(B_{s}-B_r)
+R^\Phi_{s,r}.
\label{eq:Phi-expansion-second-greek}
\end{align}
Moreover, under the standard smoothness and moment assumptions, for every $q\ge2$,
\begin{align*}
\|R^X_{s,r}\|_{L^q}
+\|R^J_{s,r}\|_{L^q}
+\|R^K_{s,r}\|_{L^q}
+\|R^\Phi_{s,r}\|_{L^q}
\le C(r-s).
\end{align*}

For each $1\le \nu\le \ell$, we expand the backward integrand
\begin{align*}
(J_r^j)^\top\nabla_x^2H_\nu(r,X_r)J_r^i
+\nabla_xH_\nu(r,X_r)^\top K_r^{(ij)}
\end{align*}
around $(s,X_s,J_s,K_s^{(ij)})$.

First, by Taylor expansion,
\begin{align*}
\nabla_x^2H_\nu(r,X_r)
=\nabla_x^2H_\nu(s,X_s)
+\nabla_x^3H_\nu(s,X_s)[X_r-X_s]
+R^{\nabla^2H_\nu}_{s,r},
\end{align*}
where $\nabla_x^3H_\nu(s,X_s)[X_r-X_s]$ denotes the derivative of the
Hessian matrix in the direction $X_r-X_s$. 
Since
\begin{align*}
  \|X_r-X_s\|_{L^{2q}}\le C(r-s)^{1/2}
\end{align*}
and the fourth spatial derivatives of $H_\nu$ are bounded,
the remainder satisfies
\begin{align*}
\left\|R^{\nabla^2H_\nu}_{s,r}\right\|_{L^q}
\le C(r-s).
\end{align*}
Using the first-order expansions of $X_r$ and $J_r$, we have
\begin{align*}
(J_r^j)^\top\nabla_x^2H_\nu(r,X_r)J_r^i
=&(J_s^j)^\top\nabla_x^2H_\nu(s,X_s)J_s^i
+(J_s^j)^\top
\big(\nabla_x^3H_\nu(s,X_s)[X_r-X_s]\big)
J_s^i\notag\\
&+(J_r^j-J_s^j)^\top
\nabla_x^2H_\nu(s,X_s)J_s^i
+(J_s^j)^\top
\nabla_x^2H_\nu(s,X_s)(J_r^i-J_s^i)
+R^{(1),H,(ij)}_{\nu,s,r}.
\end{align*}
Here the terms containing products such as
$(J_r^j-J_s^j)^\top\nabla_x^2H_\nu(s,X_s)(J_r^i-J_s^i),$
as well as the products involving
$R^{\nabla^2H_\nu}_{s,r}$, are absorbed into
$R^{(1),H,(ij)}_{\nu,s,r}$.
Substituting the expression for
$X_r-X_s$ given in~\eqref{eq:X-expansion-second-greek},
we get
\begin{align*}
&(J_s^j)^\top
\big(\nabla_x^3H_\nu(s,X_s)[X_r-X_s]\big)J_s^i
=\sum_{a=1}^{d}
\nabla_x^3H_\nu(s,X_s)
[\sigma_{\cdot a}(X_s),J_s^i,J_s^j]
(W_r^a-W_s^a)
+\widetilde R^{(1),H,(ij)}_{\nu,s,r}.
\end{align*}
Similarly, using~\eqref{eq:J-expansion-second-greek},
we obtain
\begin{align}
&(J_r^j)^\top\nabla_x^2H_\nu(r,X_r)J_r^i\notag\\
=&(J_s^j)^\top\nabla_x^2H_\nu(s,X_s)J_s^i
+\sum_{a=1}^{d}
\bigg[\nabla_x^3H_\nu(s,X_s)[\sigma_{\cdot a}(X_s),J_s^i,J_s^j]
+\big((\nabla_x\sigma_{\cdot a}(X_s)J_s)^j\big)^\top
\nabla_x^2H_\nu(s,X_s)J_s^i\notag\\
&\hspace{5cm}
+(J_s^j)^\top\nabla_x^2H_\nu(s,X_s)
(\nabla_x\sigma_{\cdot a}(X_s)J_s)^i
\bigg](W_r^a-W_s^a)
+R^{(2),H,(ij)}_{\nu,s,r}.
\label{eq:Hessian-part-expanded-second-greek}
\end{align}
The remainder satisfies
\begin{align*}
\|R^{(2),H,(ij)}_{\nu,s,r}\|_{L^q}\le C(r-s).
\end{align*}

Next, consider the term
$\nabla_xH_\nu(r,X_r)^\top K_r^{(ij)}.$
A first-order Taylor expansion gives
\begin{align}
\nabla_xH_\nu(r,X_r)^\top K_r^{(ij)}
=&\nabla_xH_\nu(s,X_s)^\top K_s^{(ij)}
+(X_r-X_s)^\top\nabla_x^2H_\nu(s,X_s)K_s^{(ij)}
+\nabla_xH_\nu(s,X_s)^\top(K_r^{(ij)}-K_s^{(ij)})
+R^{(3),H,(ij)}_{\nu,s,r}.
\label{eq:gradient-K-term-Taylor-second-greek}
\end{align}
Using~\eqref{eq:X-expansion-second-greek}, 
we get
\begin{align}
(X_r-X_s)^\top\nabla_x^2H_\nu(s,X_s)K_s^{(ij)}
=\sum_{a=1}^d
\sigma_{\cdot a}(X_s)^\top
\nabla_x^2H_\nu(s,X_s)K_s^{(ij)}
(W_r^a-W_s^a)
+\widetilde R^{(3),H,(ij)}_{\nu,s,r}.
\label{eq:X-K-term-second-greek}
\end{align}
By the expansion of $K_r^{(ij)}-K_s^{(ij)}$ given in~\eqref{eq:K-expansion-second-greek},
\begin{align}
\nabla_xH_\nu(s,X_s)^\top(K_r^{(ij)}-K_s^{(ij)})
=&\sum_{a=1}^d
\nabla_xH_\nu(s,X_s)^\top
\Big(\nabla_x\sigma_{\cdot a}(X_s)K_s^{(ij)}
+(J_s^i)^\top\nabla_x^2\sigma_{\cdot a}(X_s)J_s^j\Big)
(W_r^a-W_s^a)
+\widetilde R^{(4),H,(ij)}_{\nu,s,r}.
\label{eq:K-H-term-second-greek}
\end{align}
Combining
\eqref{eq:gradient-K-term-Taylor-second-greek},
\eqref{eq:X-K-term-second-greek}, and
\eqref{eq:K-H-term-second-greek}, we obtain
\begin{align}
&\nabla_xH_\nu(r,X_r)^\top K_r^{(ij)}
=\nabla_xH_\nu(s,X_s)^\top K_s^{(ij)}\notag\\
&+\sum_{a=1}^d
\bigg[\sigma_{\cdot a}(X_s)^\top
\nabla_x^2H_\nu(s,X_s)K_s^{(ij)}
+\nabla_xH_\nu(s,X_s)^\top
\Big(\nabla_x\sigma_{\cdot a}(X_s)K_s^{(ij)}
+(J_s^i)^\top\nabla_x^2\sigma_{\cdot a}(X_s)J_s^j\Big)\bigg]
(W_r^a-W_s^a)
+R^{(4),H,(ij)}_{\nu,s,r}.
\label{eq:gradient-K-expanded-second-greek}
\end{align}

Adding \eqref{eq:Hessian-part-expanded-second-greek} and
\eqref{eq:gradient-K-expanded-second-greek}, we finally get
\begin{align}
&(J_r^j)^\top\nabla_x^2H_\nu(r,X_r)J_r^i
+\nabla_xH_\nu(r,X_r)^\top K_r^{(ij)}\notag\\
=&(J_s^j)^\top\nabla_x^2H_\nu(s,X_s)J_s^i
+\nabla_xH_\nu(s,X_s)^\top K_s^{(ij)}\notag\\
&+\sum_{a=1}^{d}\bigg[\nabla_x^3H_\nu(s,X_s)
[\sigma_{\cdot a}(X_s),J_s^i,J_s^j]
+\big((\nabla_x\sigma_{\cdot a}(X_s)J_s)^j\big)^\top
\nabla_x^2H_\nu(s,X_s)J_s^i
+(J_s^j)^\top
\nabla_x^2H_\nu(s,X_s)
(\nabla_x\sigma_{\cdot a}(X_s)J_s)^i\notag\\
&\hspace{1.2cm}
+\sigma_{\cdot a}(X_s)^\top\nabla_x^2H_\nu(s,X_s)K_s^{(ij)}
+\nabla_xH_\nu(s,X_s)^\top
\Big(\nabla_x\sigma_{\cdot a}(X_s)K_s^{(ij)}
+(J_s^i)^\top\nabla_x^2\sigma_{\cdot a}(X_s)J_s^j
\Big)\bigg]
(W_r^a-W_s^a)+R^{H,(ij)}_{\nu,s,r}.
\label{eq:H-second-integrand-expansion}
\end{align}

It remains to estimate the remainder. The terms collected in
$R^{H,(ij)}_{\nu,s,r}$ are of the following types:
\begin{align*}
&|r-s|,
\quad|X_r-X_s|^2,
\quad|J_r-J_s|^2,
\quad|K_r^{(ij)}-K_s^{(ij)}|^2,\\
&|X_r-X_s|\,|J_r-J_s|,
\quad|X_r-X_s|\,|K_r^{(ij)}-K_s^{(ij)}|,
\quad|J_r-J_s|\,|K_r^{(ij)}-K_s^{(ij)}|,\\
&|R^X_{s,r}|,
\quad|R^J_{s,r}|,
\quad|R^K_{s,r}|.
\end{align*}
Using
\begin{align*}
\|X_r-X_s\|_{L^{2q}}
+\|J_r-J_s\|_{L^{2q}}
+\|K_r^{(ij)}-K_s^{(ij)}\|_{L^{2q}}
\le C(r-s)^{1/2},
\end{align*}
and
\begin{align*}
\|R^X_{s,r}\|_{L^q}
+\|R^J_{s,r}\|_{L^q}
+\|R^K_{s,r}\|_{L^q}
\le C(r-s),
\end{align*}
we obtain
\begin{align*}
\|R^{H,(ij)}_{\nu,s,r}\|_{L^q}
\le C(r-s).
\end{align*}

Combining \eqref{eq:Phi-expansion-second-greek} and
\eqref{eq:H-second-integrand-expansion}, we obtain
\begin{align}
&\Phi(t,r)
\Big((J_r^j)^\top\nabla_x^2H_\nu(r,X_r)J_r^i
+\nabla_xH_\nu(r,X_r)^\top K_r^{(ij)}\Big)\notag\\
=&\Phi(t,s)
\Big((J_s^j)^\top\nabla_x^2H_\nu(s,X_s)J_s^i
+\nabla_xH_\nu(s,X_s)^\top K_s^{(ij)}\Big)\notag\\
&+\Phi(t,s)
\sum_{a=1}^{d}
\bigg[\nabla_x^3H_\nu(s,X_s)
[\sigma_{\cdot a}(X_s),J_s^i,J_s^j]
+\big((\nabla_x\sigma_{\cdot a}(X_s)J_s)^j\big)^\top
\nabla_x^2H_\nu(s,X_s)J_s^i
+(J_s^j)^\top\nabla_x^2H_\nu(s,X_s)
(\nabla_x\sigma_{\cdot a}(X_s)J_s)^i\notag\\
&\hspace{3cm}
+\sigma_{\cdot a}(X_s)^\top
\nabla_x^2H_\nu(s,X_s)K_s^{(ij)}
+\nabla_xH_\nu(s,X_s)^\top
\Big(\nabla_x\sigma_{\cdot a}(X_s)K_s^{(ij)}
+\nabla_x^2\sigma_{\cdot a}(X_s)[J_s^i,J_s^j]\Big)\notag\\
&\hspace{3cm}
+\widetilde c_a(s)
\Big((J_s^j)^\top\nabla_x^2H_\nu(s,X_s)J_s^i
+\nabla_xH_\nu(s,X_s)^\top K_s^{(ij)}\Big)\bigg](W_r^a-W_s^a)\notag\\
&+\Phi(t,s)\sum_{\alpha=1}^{\ell}d_\alpha(s)
\Big((J_s^j)^\top\nabla_x^2H_\nu(s,X_s)J_s^i
+\nabla_xH_\nu(s,X_s)^\top K_s^{(ij)}\Big)
(B_{s}^{\alpha}-B_r^\alpha)
+R^{\Phi H,(ij)}_{\nu,s,r}.
\label{eq:PhiH-second-expansion}
\end{align}
The product remainder satisfies
\begin{align*}
\|R^{\Phi H,(ij)}_{\nu,s,r}\|_{L^q}
\le C(r-s).
\end{align*}
Consequently,
\begin{align*}
\sum_{\nu=1}^{\ell}
\int_s^{s+h}
\|R^{\Phi H,(ij)}_{\nu,s,r}\|_{L^q}^2
\,\mathrm dr
\le Ch^3.
\end{align*}

Substituting \eqref{eq:PhiH-second-expansion} into the backward
stochastic integral gives
\begin{align}
&\int_s^{s+h}\Phi(t,r)
\Big((J_r^j)^\top\nabla_x^2H(r,X_r)J_r^i
+\nabla_xH(r,X_r)^\top K_r^{(ij)}\Big)
\mathrm d \overleftarrow{B}_r\notag\\
=&\Phi(t,s)
\Big((J_s^j)^\top\nabla_x^2H(s,X_s)J_s^i
+\nabla_xH(s,X_s)^\top K_s^{(ij)}\Big)
\cdot\Delta\overleftarrow B_s\notag\\
&+\Phi(t,s)\sum_{\nu=1}^{\ell}\sum_{a=1}^{d}
\bigg[\nabla_x^3H_\nu(s,X_s)
[\sigma_{\cdot a}(X_s),J_s^i,J_s^j]
+\big((\nabla_x\sigma_{\cdot a}(X_s)J_s)^j\big)^\top
\nabla_x^2H_\nu(s,X_s)J_s^i\notag\\
&\hspace{3cm}
+(J_s^j)^\top
\nabla_x^2H_\nu(s,X_s)
(\nabla_x\sigma_{\cdot a}(X_s)J_s)^i
+\sigma_{\cdot a}(X_s)^\top
\nabla_x^2H_\nu(s,X_s)K_s^{(ij)}\notag\\
&\hspace{3cm}
+\nabla_xH_\nu(s,X_s)^\top
\Big(\nabla_x\sigma_{\cdot a}(X_s)K_s^{(ij)}
+(J_s^i)^\top\nabla_x^2\sigma_{\cdot a}(X_s)J_s^j\Big)\notag\\
&\hspace{3cm}
+\widetilde c_a(s)
\Big((J_s^j)^\top\nabla_x^2H_\nu(s,X_s)J_s^i
+\nabla_xH_\nu(s,X_s)^\top K_s^{(ij)}\Big)\bigg]
J_{\nu a}^{WB}(s,h)\notag\\
&+\Phi(t,s)
\sum_{\nu=1}^{\ell}\sum_{\alpha=1}^{\ell}
d_\alpha(s)
\Big((J_s^j)^\top\nabla_x^2H_\nu(s,X_s)J_s^i
+\nabla_xH_\nu(s,X_s)^\top K_s^{(ij)}\Big)
J_{\nu\alpha}^{BB}(s,h)
+R^{B,(ij)}_{s,h},
\label{eq:backward-second-expansion}
\end{align}
where
\begin{align*}
R^{B,(ij)}_{s,h}
=\sum_{\nu=1}^{\ell}\int_s^{s+h}
R^{\Phi H,(ij)}_{\nu,s,r}
\mathrm d \overleftarrow{B}_r^\nu.
\end{align*}

For the time integral part, set
\begin{align}
\Lambda^{(ij)}(r)
:=(J_r^j)^\top\nabla_x^2(F+dH)(r,X_r)J_r^i
+\nabla_x(F+dH)(r,X_r)^\top K_r^{(ij)} .
\end{align}
By the same Taylor expansion as for the backward integrand, we have
\begin{align}
\Lambda^{(ij)}(r)
=\Lambda^{(ij)}(s)
+\sum_{a=1}^d \lambda_a^{(ij)}(s)(W_r^a-W_s^a)
+R^{\Lambda,(ij)}_{s,r},
\label{eq:Lambda-expansion-second-greek}
\end{align}
where
\begin{align}
\lambda_a^{(ij)}(s)
:=&\nabla_x^3(F+dH)(s,X_s)
[\sigma_{\cdot a}(X_s),J_s^i,J_s^j]
+\big((\nabla_x\sigma_{\cdot a}(X_s)J_s)^j\big)^\top
\nabla_x^2(F+dH)(s,X_s)J_s^i\notag\\
&+(J_s^j)^\top\nabla_x^2(F+dH)(s,X_s)
(\nabla_x\sigma_{\cdot a}(X_s)J_s)^i
+\sigma_{\cdot a}(X_s)^\top
\nabla_x^2(F+dH)(s,X_s)K_s^{(ij)}\notag\\
&+\nabla_x(F+dH)(s,X_s)^\top
\Big(\nabla_x\sigma_{\cdot a}(X_s)K_s^{(ij)}
+(J_s^i)^\top\nabla_x^2\sigma_{\cdot a}(X_s)J_s^j\Big).
\end{align}
The remainder satisfies
\begin{align}
\|R^{\Lambda,(ij)}_{s,r}\|_{L^q}
\le C(r-s).
\end{align}
Combining \eqref{eq:Lambda-expansion-second-greek} with
\eqref{eq:Phi-expansion-second-greek},
we obtain
\begin{align}
\Phi(t,r)\Lambda^{(ij)}(r)
=&\Phi(t,s)\Lambda^{(ij)}(s)
+\Phi(t,s)\sum_{a=1}^d
\Big(\lambda_a^{(ij)}(s)
+\widetilde c_a(s)\Lambda^{(ij)}(s)\Big)
(W_r^a-W_s^a)\notag\\
&+\Phi(t,s)\sum_{\alpha=1}^{\ell}
d_\alpha(s)\Lambda^{(ij)}(s)(B_{s}^\alpha-B_r^\alpha)
+R^{D,(ij)}_{s,r}.
\label{eq:PhiLambda-expansion-second-greek}
\end{align}
The product remainder satisfies
\begin{align}
\|R^{D,(ij)}_{s,r}\|_{L^q}
\le C(r-s)+C(r-s)^{1/2}(s+h-r)^{1/2}.
\end{align}
In particular,
\begin{align}
\int_s^{s+h}
\|R^{D,(ij)}_{s,r}\|_{L^q}\,\mathrm dr
\le Ch^2.
\label{eq:PhiLambda-integrated-remainder-second-greek}
\end{align}
Integrating \eqref{eq:PhiLambda-expansion-second-greek} over $[s,s+h]$ gives
\begin{align}
&\int_s^{s+h}\Phi(t,r)
\Big((J_r^j)^\top\nabla_x^2(F+dH)(r,X_r)J_r^i
+\nabla_x(F+dH)(r,X_r)^\top K_r^{(ij)}\Big)
\,\mathrm dr\notag\\
=&\Phi(t,s)h
\Big((J_s^j)^\top\nabla_x^2(F+dH)(s,X_s)J_s^i
+\nabla_x(F+dH)(s,X_s)^\top K_s^{(ij)}\Big)
+M^{D,(ij)}_{s,h}+A^{D,(ij)}_{s,h},
\end{align}
where
\begin{align}
A^{D,(ij)}_{s,h}
:=\int_s^{s+h}
R^{D,(ij)}_{s,r}\,\mathrm dr,
\end{align}
and
\begin{align}
M^{D,(ij)}_{s,h}:=&\Phi(t,s)
\sum_{a=1}^d\Big(\lambda_a^{(ij)}(s)
+\widetilde c_a(s)\Lambda^{(ij)}(s)\Big)
\int_s^{s+h}(W_r^a-W_s^a)\,\mathrm dr
+\Phi(t,s)\sum_{\alpha=1}^{\ell}
d_\alpha(s)\Lambda^{(ij)}(s)
\int_s^{s+h}(B_{s}^\alpha-B_r^\alpha)\,\mathrm dr.
\end{align}
Here $A^{D,(ij)}_{s,h}$ is the finite-variation remainder, while
$M^{D,(ij)}_{s,h}$ collects the first-order Brownian fluctuations of
$X$, $J$, $K$, and $\Phi$ inside the time integral.
By \eqref{eq:PhiLambda-integrated-remainder-second-greek},
\begin{align*}
\|A^{D,(ij)}_{s,h}\|_{L^q}
\le
\int_s^{s+h}
\|R^{D,(ij)}_{s,r}\|_{L^q}\,\mathrm dr
\le
Ch^2.
\end{align*}
Moreover, using the uniform moment bounds of the coefficients and
\begin{align}
\left\|
\int_s^{s+h}(W_r^a-W_s^a)\,\mathrm dr
\right\|_{L^q}
\le Ch^{3/2},
\quad
\left\|
\int_s^{s+h}(B_{s}^\alpha-B_r^\alpha)\,\mathrm dr
\right\|_{L^q}
\le Ch^{3/2},
\end{align}
we obtain
\begin{align*}
\|M^{D,(ij)}_{s,h}\|_{L^q}
\le Ch^{3/2}.
\end{align*}

Combining the time integral expansion and the backward stochastic integral
expansion, we obtain the one-step consistency relation
\begin{align}
&\int_s^{s+h}\Phi(t,r)
\Big((J_r^j)^\top\nabla_x^2(F+dH)(r,X_r)J_r^i
+\nabla_x(F+dH)(r,X_r)^\top K_r^{(ij)}\Big)
\,\mathrm dr\notag\\
&+\int_s^{s+h}\Phi(t,r)
\Big((J_r^j)^\top\nabla_x^2H(r,X_r)J_r^i
+\nabla_xH(r,X_r)^\top K_r^{(ij)}\Big)
\mathrm d\overleftarrow{B}_r\notag\\
&=\mathcal S_{\mathrm{int}}^{(ij),\mathrm{FBT}}
\left(\Phi(t,s),X_s,J_s,K_s^{(ij)},s,h;
\Delta W_s,\Delta\overleftarrow B_s\right)
+A^{D,(ij)}_{s,h}
+M^{D,(ij)}_{s,h}
+R^{B,(ij)}_{s,h}.
\label{eq:one-step-consistency-second-greek}
\end{align}

\paragraph{Proof of Proposition~\ref{prop:second-order-greek-strong-error}}
\label{para:proof of prop:second-order-greek-strong-error}
\begin{proposition*}[Strong-error order of the second-order Greek integral discretization]
Assume that the coefficients are sufficiently smooth with bounded derivatives
up to the order used above, and assume that $X$, $J$, $K$, and $\Phi$ have
uniformly bounded moments of all required orders. Then the approximation
generated by
\begin{align*}
\widetilde Y_{k+1}^{(ij,h)}
=
\widetilde Y_k^{(ij,h)}
+
\mathcal S_{\mathrm{int}}^{(ij),\mathrm{FBT}}
\left(\Phi(t,t_k),X_{t_k}^{t,x},J_{t_k}^{t,x},
K_{t_k}^{(ij),t,x},t_k,h;\Delta W_{t_k},\Delta\overleftarrow B_{t_k}
\right),
\quad
\widetilde Y_0^{(ij,h)}=0,
\end{align*}
satisfies, for every fixed $q\ge2$,
\begin{align*}
\left\|
\sup_{0\le k\le N}
\left|
Y_{t_k}^{(ij)}-\widetilde Y_k^{(ij,h)}
\right|
\right\|_{L^q}
\le C_qh.
\end{align*}
Consequently, by Jensen's inequality, the same estimate also holds for
$0<q<2$. Hence
$\mathcal S_{\mathrm{int}}^{(ij),\mathrm{FBT}}$ has strong-error order $1$ in the
sense of Definition~\ref{def:Strong error of S int second Greek}.
\end{proposition*}

\begin{proof}

By the one-step consistency relation
\eqref{eq:one-step-consistency-second-greek}, for each $0\le k\le N$,
\begin{align}
Y_{t_k}^{(ij)}-\widetilde Y_k^{(ij,h)}
=
\sum_{r=0}^{k-1}A^{D,(ij)}_{t_r,h}
+
\sum_{r=0}^{k-1}M^{D,(ij)}_{t_r,h}
+
\sum_{r=0}^{k-1}R^{B,(ij)}_{t_r,h}.
\label{eq:second-global-error-decomposition}
\end{align}
For $k=0$, the sums are 0 and the identity is consistent with
$Y_t^{(ij)}=\widetilde Y_0^{(ij,h)}=0$.

We estimate the three accumulated terms in
\eqref{eq:second-global-error-decomposition} separately.

For the finite-variation part, using the pathwise bound
\[
\sup_{0\le k\le N}
\left|
\sum_{r=0}^{k-1}
A^{D,(ij)}_{t_r,h}
\right|
\le
\sum_{r=0}^{N-1}
\left|
A^{D,(ij)}_{t_r,h}
\right|,
\]
we obtain
\begin{align}
\left\|
\sup_{0\le k\le N}
\left|
\sum_{r=0}^{k-1}
A^{D,(ij)}_{t_r,h}
\right|
\right\|_{L^q}
&\le
\sum_{r=0}^{N-1}
\left\|
A^{D,(ij)}_{t_r,h}
\right\|_{L^q}
\le
C_qNh^2
\le
C_qh.
\label{eq:second-global-A-bound}
\end{align}

For the martingale-type part, by the discrete forward--backward information
convention above, the \(W\)-part of
\(\{M^{D,(ij)}_{t_r,h}\}_{r=0}^{N-1}\) is a forward martingale-difference
sequence, while the \(B\)-part is a reverse martingale-difference sequence.
Hence, applying the discrete Burkholder--Davis--Gundy inequality, equivalently
after reversing the \(B\)-time for the backward part, 
\begin{align}
\left\|
\sup_{0\le k\le N}
\left|
\sum_{r=0}^{k-1}
M^{D,(ij)}_{t_r,h}
\right|
\right\|_{L^q}
&\le
C_q
\left\|
\left(
\sum_{r=0}^{N-1}
\left|
M^{D,(ij)}_{t_r,h}
\right|^2
\right)^{1/2}
\right\|_{L^q}
\notag\\
&=
C_q
\left\|
\sum_{r=0}^{N-1}
\left|
M^{D,(ij)}_{t_r,h}
\right|^2
\right\|_{L^{q/2}}^{1/2}
\notag\\
&\le
C_q
\left(
\sum_{r=0}^{N-1}
\left\|
M^{D,(ij)}_{t_r,h}
\right\|_{L^q}^2
\right)^{1/2}
\notag\\
&\le
C_q(Nh^3)^{1/2}
\le
C_qh.
\label{eq:second-global-MD-bound}
\end{align}
Here the third inequality uses Minkowski's inequality in $L^{q/2}$, which is
valid since $q\ge2$.

It remains to estimate the accumulated backward stochastic remainder. By
definition,
\begin{align}
\sum_{r=0}^{k-1}
R^{B,(ij)}_{t_r,h}
=
\sum_{r=0}^{k-1}
\sum_{\nu=1}^{\ell}
\int_{t_r}^{t_{r+1}}
R^{\Phi H,(ij)}_{\nu,t_r,u}
\,\mathrm d \overleftarrow{B}_u^\nu.
\end{align}
Applying the Burkholder--Davis--Gundy inequality for backward stochastic
integrals, equivalently after reversing time, gives
\begin{align}
&
\left\|
\sup_{0\le k\le N}
\left|
\sum_{r=0}^{k-1}
R^{B,(ij)}_{t_r,h}
\right|
\right\|_{L^q}
\notag\\
\le&
C_q
\left\|
\left(
\sum_{r=0}^{N-1}
\sum_{\nu=1}^{\ell}
\int_{t_r}^{t_{r+1}}
\left|
R^{\Phi H,(ij)}_{\nu,t_r,u}
\right|^2
\,\mathrm du
\right)^{1/2}
\right\|_{L^q}
\notag\\
=&
C_q
\left\|
\sum_{r=0}^{N-1}
\sum_{\nu=1}^{\ell}
\int_{t_r}^{t_{r+1}}
\left|
R^{\Phi H,(ij)}_{\nu,t_r,u}
\right|^2
\,\mathrm du
\right\|_{L^{q/2}}^{1/2}
\notag\\
\le&
C_q
\left(
\sum_{r=0}^{N-1}
\sum_{\nu=1}^{\ell}
\int_{t_r}^{t_{r+1}}
\left\|
R^{\Phi H,(ij)}_{\nu,t_r,u}
\right\|_{L^q}^2
\,\mathrm du
\right)^{1/2}
\notag\\
\le&
C_q(Nh^3)^{1/2}
\le
C_qh.
\label{eq:second-global-RB-bound}
\end{align}

Combining
\eqref{eq:second-global-A-bound},
\eqref{eq:second-global-MD-bound}, and
\eqref{eq:second-global-RB-bound} with
\eqref{eq:second-global-error-decomposition}, we conclude that
\begin{align}
\left\|
\sup_{0\le k\le N}
\left|
Y_{t_k}^{(ij)}
-
\widetilde Y_k^{(ij,h)}
\right|
\right\|_{L^q}
\le
C_qh,
\quad q\ge2.
\end{align}

Finally, for $0<q<2$, Jensen's inequality gives
\begin{align*}
\left\|
\sup_{0\le k\le N}
\left|
Y_{t_k}^{(ij)}
-
\widetilde Y_k^{(ij,h)}
\right|
\right\|_{L^q}
\le
\left\|
\sup_{0\le k\le N}
\left|
Y_{t_k}^{(ij)}
-
\widetilde Y_k^{(ij,h)}
\right|
\right\|_{L^2}
\le
C_2h.
\end{align*}
This proves the claimed strong-error order one estimate.
\end{proof}

\paragraph{Proof of Proposition~\ref{prop:second-order-greek accumulated stability}}
\label{para:proof of prop:second-order-greek accumulated stability}
\begin{proposition*}
Assume that
$F\in C_b^4$, $H\in C_b^4$, $\sigma\in C_b^3$
and $d$, $\widetilde c$ are bounded. 
Moreover, assume that
$\Phi(t,s)$, $\Phi^{(h)}$,
$X_{s}^{t,x}$, $X^{(h)}$,
$J_{s}^{t,x}$, $J^{(h)}$,
$K_{s}^{(ij),t,x}$, $K^{(ij,h)}$
are bounded in $L^{16}_{W,B}$.
Then $\mathcal S_{\mathrm{int}}^{(ij),\mathrm{FBT}}$ satisfies the accumulated
stability estimate. More precisely, there exists
$L_{\mathrm{int}}^{(ij)}>0$, independent of $h$, such that
\begin{align*}
\left\|
\sup_{0\le k\le N}
\left|
\widetilde Y_k^{(ij,h)}-Y_k^{(ij,h)}
\right|
\right\|_{L^2}
\le L_{\mathrm{int}}^{(ij)}
\bigg(
&
\left\|
\sup_{0\le r\le N}
\left|
\Phi(t,t_r)-\Phi_r^{(h)}
\right|
\right\|_{L^4}+
\left\|
\sup_{0\le r\le N}
\left|
X_{t_r}^{t,x}-X_r^{(h)}
\right|
\right\|_{L^4}\\
&+
\left\|
\sup_{0\le r\le N}
\left|
J_{t_r}^{t,x}-J_r^{(h)}
\right|
\right\|_{L^4}+
\left\|
\sup_{0\le r\le N}
\left|
K_{t_r}^{(ij),t,x}-K_r^{(ij,h)}
\right|
\right\|_{L^4}
\bigg).
\end{align*}
\end{proposition*}

\begin{proof}

Set
\[
\Delta\Phi_r:=\Phi(t,t_r)-\Phi_r^{(h)},\quad
\Delta X_r:=X_{t_r}^{t,x}-X_r^{(h)},\quad
\Delta J_r:=J_{t_r}^{t,x}-J_r^{(h)},\quad
\Delta K_r:=K_{t_r}^{(ij),t,x}-K_r^{(ij,h)}
\]
and define
\[
\delta_\Phi
:=
\left\|
\sup_{0\le r\le N}
|\Delta\Phi_r|
\right\|_{L^4},
\quad
\delta_X
:=
\left\|
\sup_{0\le r\le N}
|\Delta X_r|
\right\|_{L^4},
\quad
\delta_J
:=
\left\|
\sup_{0\le r\le N}
|\Delta J_r|
\right\|_{L^4},
\quad
\delta_K
:=
\left\|
\sup_{0\le r\le N}
|\Delta K_r|
\right\|_{L^4},
\]
and write
$\delta:=\delta_\Phi+\delta_X+\delta_J+\delta_K.$

By the definitions of $\widetilde Y_k^{(ij,h)}$ and $Y_k^{(ij,h)}$,
\begin{align}
\widetilde Y_k^{(ij,h)}-Y_k^{(ij,h)}
=\sum_{r=0}^{k-1}\Bigg[
&\mathcal S_{\mathrm{int}}^{(ij),\mathrm{FBT}}
\left(\Phi(t,t_r),X_{t_r}^{t,x},J_{t_r}^{t,x},
K_{t_r}^{(ij),t,x},t_r,h;
\Delta W_{t_r},\Delta\overleftarrow B_{t_r}\right)
\notag\\
&-
\mathcal S_{\mathrm{int}}^{(ij),\mathrm{FBT}}
\left(\Phi_r^{(h)},X_r^{(h)},J_r^{(h)},
K_r^{(ij,h)},t_r,h;
\Delta W_{t_r},\Delta\overleftarrow B_{t_r}\right)
\Bigg].
\label{eq:second-greek-stability-decomposition}
\end{align}
We estimate the contributions of the four terms in
$\mathcal S_{\mathrm{int}}^{(ij),\mathrm{FBT}}$.

First, define
\[
\Lambda_F(s,X,J,K)
:=
(J^j)^\top\nabla_x^2(F+dH)(s,X)J^i
+
\nabla_x(F+dH)(s,X)^\top K .
\]
Since $F\in C_b^3$, $H\in C_b^4$, and $d$ is bounded, the map
$(X,J,K)\mapsto \Lambda_F(s,X,J,K)$ is locally Lipschitz with at most
quadratic growth in $J$ and linear growth in $K$, uniformly in $s$.
Using H\"older's inequality and the uniform $L^{16}$ moment bounds, we get
\begin{align}
&
\left\|
\Phi(t,t_r)\Lambda_F(t_r,X_{t_r}^{t,x},J_{t_r}^{t,x},K_{t_r}^{(ij),t,x})
-
\Phi_r^{(h)}
\Lambda_F(t_r,X_r^{(h)},J_r^{(h)},K_r^{(ij,h)})
\right\|_{L^2}
\notag\\
\le&
C\left(
\|\Delta\Phi_r\|_{L^4}
+
\|\Delta X_r\|_{L^4}
+
\|\Delta J_r\|_{L^4}
+
\|\Delta K_r\|_{L^4}
\right)
\le C\delta .
\label{eq:second-stability-time-coeff}
\end{align}
Therefore, using the pathwise bound
\[
\sup_{0\le k\le N}
\left|
\sum_{r=0}^{k-1}hA_r
\right|
\le
\sum_{r=0}^{N-1}h|A_r|,
\]
we obtain
\begin{align}
&
\left\|
\sup_{0\le k\le N}
\left|
\sum_{r=0}^{k-1}
h\Big[
\Phi(t,t_r)\Lambda_F(t_r,X_{t_r}^{t,x},J_{t_r}^{t,x},K_{t_r}^{(ij),t,x})
-
\Phi_r^{(h)}
\Lambda_F(t_r,X_r^{(h)},J_r^{(h)},K_r^{(ij,h)})
\Big]
\right|
\right\|_{L^2}
\notag\\
\le&
\sum_{r=0}^{N-1}
h
\left\|
\Phi(t,t_r)\Lambda_F(t_r,X_{t_r}^{t,x},J_{t_r}^{t,x},K_{t_r}^{(ij),t,x})
-
\Phi_r^{(h)}
\Lambda_F(t_r,X_r^{(h)},J_r^{(h)},K_r^{(ij,h)})
\right\|_{L^2}
\le
CT\delta .
\label{eq:second-stability-time-term}
\end{align}

Next, define for $1\le \nu\le \ell$,
\[
\mathcal B_\nu(s,X,J,K)
:=
(J^j)^\top\nabla_x^2H_\nu(s,X)J^i
+
\nabla_xH_\nu(s,X)^\top K .
\]
The same smoothness and moment assumptions imply
\begin{align}
&
\left\|
\Phi(t,t_r)\mathcal B_\nu(t_r,X_{t_r}^{t,x},J_{t_r}^{t,x},K_{t_r}^{(ij),t,x})
-
\Phi_r^{(h)}
\mathcal B_\nu(t_r,X_r^{(h)},J_r^{(h)},K_r^{(ij,h)})
\right\|_{L^2}
\le C\delta .
\label{eq:second-stability-B-coeff}
\end{align}
Let
\[
D_{\nu,r}^{B}
:=
\Phi(t,t_r)\mathcal B_\nu(t_r,X_{t_r}^{t,x},J_{t_r}^{t,x},K_{t_r}^{(ij),t,x})
-
\Phi_r^{(h)}
\mathcal B_\nu(t_r,X_r^{(h)},J_r^{(h)},K_r^{(ij,h)}).
\]
Using the discrete Burkholder--Davis--Gundy inequality for backward stochastic
increments, equivalently after reversing time, we obtain
\begin{align}
&
\left\|
\sup_{0\le k\le N}
\left|
\sum_{r=0}^{k-1}
\sum_{\nu=1}^{\ell}
D_{\nu,r}^{B}\,
\Delta\overleftarrow B_{t_r}^{\nu}
\right|
\right\|_{L^2}
\le
C
\left\|
\left(
\sum_{r=0}^{N-1}
h\sum_{\nu=1}^{\ell}
|D_{\nu,r}^{B}|^2
\right)^{1/2}
\right\|_{L^2}
\le
C
\left(
\sum_{r=0}^{N-1}
h\sum_{\nu=1}^{\ell}
\|D_{\nu,r}^{B}\|_{L^2}^2
\right)^{1/2}
\le
C\sqrt T\,\delta .
\label{eq:second-stability-B-term}
\end{align}

It remains to estimate the $WB$ and $BB$ correction terms. For fixed
$1\le\nu\le\ell$ and $1\le a\le d$, set
\begin{align*}
\mathcal A_{\nu a}(s,X,J,K)
:=&
\nabla_x^3H_\nu(s,X)
[\sigma_{\cdot a}(X),J^i,J^j]
+\big((\nabla_x\sigma_{\cdot a}(X)J)^j\big)^\top
\nabla_x^2H_\nu(s,X)J^i\\
&+(J^j)^\top\nabla_x^2H_\nu(s,X)
(\nabla_x\sigma_{\cdot a}(X)J)^i
+\sigma_{\cdot a}(X)^\top\nabla_x^2H_\nu(s,X)K\\
&+\nabla_xH_\nu(s,X)^\top
\big(
\nabla_x\sigma_{\cdot a}(X)K
+(J^i)^\top
\nabla_x^2\sigma_{\cdot a}(X)J^j
\big)\\
&+\widetilde c_a(s)
\Big(
(J^j)^\top\nabla_x^2H_\nu(s,X)J^i
+
\nabla_xH_\nu(s,X)^\top K
\Big).
\end{align*}
Since $H\in C_b^4$, $\sigma\in C_b^3$, and $\widetilde c$ is bounded, we have
the pointwise Lipschitz-type estimate
\begin{align}
&
\left|
\mathcal A_{\nu a}(t_r,X_{t_r}^{t,x},J_{t_r}^{t,x},K_{t_r}^{(ij),t,x})
-
\mathcal A_{\nu a}(t_r,X_r^{(h)},J_r^{(h)},K_r^{(ij,h)})
\right|
\notag\\
\le&
C\Big(
1+|J_{t_r}^{t,x}|^2+|J_r^{(h)}|^2
+|K_{t_r}^{(ij),t,x}|+|K_r^{(ij,h)}|
\Big)
|\Delta X_r|
+
C\Big(
1+|J_{t_r}^{t,x}|+|J_r^{(h)}|
\Big)|\Delta J_r|
+
C|\Delta K_r|.
\label{eq:second-A-Lipschitz}
\end{align}
Hence, by adding and subtracting the intermediate term with
$\Phi_r^{(h)}$ and exact $(X,J,K)$, and then using H\"older's inequality
together with the uniform $L^{16}$ moment bounds,
\begin{align}
&
\left\|
\Phi(t,t_r)\mathcal A_{\nu a}
(t_r,X_{t_r}^{t,x},J_{t_r}^{t,x},K_{t_r}^{(ij),t,x})
-
\Phi_r^{(h)}\mathcal A_{\nu a}
(t_r,X_r^{(h)},J_r^{(h)},K_r^{(ij,h)})
\right\|_{L^2}
\le C\delta .
\label{eq:second-stability-WB-coeff}
\end{align}
Define
\[
D_{\nu a,r}^{WB}
:=
\Phi(t,t_r)\mathcal A_{\nu a}
(t_r,X_{t_r}^{t,x},J_{t_r}^{t,x},K_{t_r}^{(ij),t,x})
-
\Phi_r^{(h)}\mathcal A_{\nu a}
(t_r,X_r^{(h)},J_r^{(h)},K_r^{(ij,h)}).
\]
Since
\[
\|J_{\nu a}^{WB}(t_r,h)\|_{L^2}\le Ch,
\]
therefore,
\begin{align}
&
\left\|
\sup_{0\le k\le N}
\left|
\sum_{r=0}^{k-1}
\sum_{\nu=1}^{\ell}\sum_{a=1}^{d}
D_{\nu a,r}^{WB}J_{\nu a}^{WB}(t_r,h)
\right|
\right\|_{L^2}
\le
\sum_{r=0}^{N-1}
\sum_{\nu=1}^{\ell}\sum_{a=1}^{d}
\left\|
D_{\nu a,r}^{WB}J_{\nu a}^{WB}(t_r,h)
\right\|_{L^2}
\le
C\sum_{r=0}^{N-1}h\delta
\le
CT\delta .
\label{eq:second-stability-WB-term}
\end{align}

Finally, consider the $BB$ correction term. Define
\[
D_{\nu\alpha,r}^{BB}
:=
d_\alpha(t_r)
\Big[
\Phi(t,t_r)\mathcal B_\nu
(t_r,X_{t_r}^{t,x},J_{t_r}^{t,x},K_{t_r}^{(ij),t,x})
-
\Phi_r^{(h)}
\mathcal B_\nu(t_r,X_r^{(h)},J_r^{(h)},K_r^{(ij,h)})
\Big].
\]
Since $d$ is bounded, \eqref{eq:second-stability-B-coeff} gives
\[
\|D_{\nu\alpha,r}^{BB}\|_{L^2}
\le C\delta .
\]
Moreover,
\[
\|J_{\nu\alpha}^{BB}(t_r,h)\|_{L^2}\le Ch.
\]
We have
\begin{align}
&
\left\|
\sup_{0\le k\le N}
\left|
\sum_{r=0}^{k-1}
\sum_{\nu=1}^{\ell}\sum_{\alpha=1}^{\ell}
D_{\nu\alpha,r}^{BB}
J_{\nu\alpha}^{BB}(t_r,h)
\right|
\right\|_{L^2}
\le
\sum_{r=0}^{N-1}
\sum_{\nu=1}^{\ell}\sum_{\alpha=1}^{\ell}
\left\|
D_{\nu\alpha,r}^{BB}
J_{\nu\alpha}^{BB}(t_r,h)
\right\|_{L^2}
\le
C\sum_{r=0}^{N-1}h\delta
\le
CT\delta .
\label{eq:second-stability-BB-term}
\end{align}

Combining
\eqref{eq:second-stability-time-term},
\eqref{eq:second-stability-B-term},
\eqref{eq:second-stability-WB-term}, and
\eqref{eq:second-stability-BB-term} in
\eqref{eq:second-greek-stability-decomposition}, we obtain
\[
\left\|
\sup_{0\le k\le N}
\left|
\widetilde Y_k^{(ij,h)}-Y_k^{(ij,h)}
\right|
\right\|_{L^2}
\le
L_{\mathrm{int}}^{(ij)}
\delta.
\]
Substituting the definition of $\delta$ gives
\begin{align*}
\left\|
\sup_{0\le k\le N}
\left|
\widetilde Y_k^{(ij,h)}-Y_k^{(ij,h)}
\right|
\right\|_{L^2}
\le L_{\mathrm{int}}^{(ij)}
\bigg(
&
\left\|
\sup_{0\le r\le N}
\left|
\Phi(t,t_r)-\Phi_r^{(h)}
\right|
\right\|_{L^4}
+
\left\|
\sup_{0\le r\le N}
\left|
X_{t_r}^{t,x}-X_r^{(h)}
\right|
\right\|_{L^4}
\\
&+
\left\|
\sup_{0\le r\le N}
\left|
J_{t_r}^{t,x}-J_r^{(h)}
\right|
\right\|_{L^4}
+
\left\|
\sup_{0\le r\le N}
\left|
K_{t_r}^{(ij),t,x}-K_r^{(ij,h)}
\right|
\right\|_{L^4}
\bigg).
\end{align*}
This proves the accumulated stability estimate.
\end{proof}

\end{document}